\documentclass[11pt,letterpaper,table]{article}
\usepackage{amsmath,amssymb,amsfonts,amsthm}
\usepackage{thm-restate}
\usepackage{cellspace}
\usepackage{dsfont}
\usepackage{mathrsfs}
\usepackage{enumitem}
\usepackage{bbm}
\usepackage[dvipsnames]{xcolor}
\usepackage{float}
\usepackage[margin=1in]{geometry}
\usepackage[colorlinks, citecolor = blue, linkcolor = magenta]{hyperref}
\usepackage[nameinlink]{cleveref}
\Crefname{ntn}{Notation}{Notations}
\numberwithin{equation}{section}

\usepackage{tikz}
\usepackage{tikz-cd}
\usetikzlibrary{shapes.geometric, arrows, arrows.meta, fit, positioning, calc}
\tikzstyle{node} = [rectangle, rounded corners, minimum width=2cm, minimum height=1cm, text centered, draw=black, fill=blue!10]
\tikzstyle{hidden} = [draw=none, fill=none] 
\tikzstyle{single arrow} = [thick, distance=2pt, -{Latex[round]}]
\tikzstyle{dashed arrow} = [thick, dashed, distance=2pt, -{Latex[round]}]

\newtheorem{theorem}{Theorem}[section]
\newtheorem{lemma}[theorem]{Lemma}
\newtheorem{corollary}[theorem]{Corollary}
\newtheorem{conjecture}[theorem]{Conjecture}
\newtheorem{claim}[theorem]{Claim}

\theoremstyle{definition}
\newtheorem{remark}[theorem]{Remark}
\newtheorem{definition}[theorem]{Definition}
\newtheorem{proposition}[theorem]{Proposition}
\newtheorem{condition}[theorem]{Condition}
\newtheorem{notation}[theorem]{Notation}

\renewcommand{\leq}{\leqslant}
\renewcommand{\le}{\leqslant}
\renewcommand{\geq}{\geqslant}

\newcommand{\cA}{\ensuremath{\mathcal{A}}}
\newcommand{\cB}{\ensuremath{\mathcal{B}}}
\newcommand{\cC}{\ensuremath{\mathcal{C}}}
\newcommand{\cD}{\ensuremath{\mathcal{D}}}
\newcommand{\cE}{\ensuremath{\mathcal{E}}}
\newcommand{\cF}{\ensuremath{\mathcal{F}}}
\newcommand{\cG}{\ensuremath{\mathcal{G}}}

\newcommand{\cJ}{\ensuremath{\mathcal{J}}}

\newcommand{\cL}{\ensuremath{\mathcal{L}}}
\newcommand{\cM}{\ensuremath{\mathcal{M}}}
\newcommand{\cN}{\ensuremath{\mathcal{N}}}

\newcommand{\cQ}{\ensuremath{\mathcal{Q}}}
\newcommand{\cR}{\ensuremath{\mathcal{R}}}
\newcommand{\cS}{\ensuremath{\mathcal{S}}}
\newcommand{\cT}{\ensuremath{\mathcal{T}}}
\newcommand{\cU}{\ensuremath{\mathcal{U}}}
\newcommand{\cV}{\ensuremath{\mathcal{V}}}
\newcommand{\cW}{\ensuremath{\mathcal{W}}}

\newcommand{\bE}{\ensuremath{\mathbb{E}}}
\newcommand{\bF}{\ensuremath{\mathbb{F}}}
\newcommand{\bN}{\ensuremath{\mathbb{N}}}
\newcommand{\bP}{\ensuremath{\mathbb{P}}}
\newcommand{\bR}{\ensuremath{\mathbb{R}}}

\newcommand{\bZ}{\ensuremath{\mathbb{Z}}}

\newcommand{\bfv}{\ensuremath{\mathbf{v}}}

\newcommand{\Gapdone}[2]{\ensuremath{\mathsf{Gap}\textup{-}#1\textup{-}\mathsf{to}\textup{-}#2}}
\newcommand{\GapSdone}[2]{\ensuremath{\mathsf{GapS}\textup{-}#1\textup{-}\mathsf{to}\textup{-}#2}}

\newcommand{\GameAa}{\ensuremath{\mathsf{Game1a}}}
\newcommand{\GameAb}{\ensuremath{\mathsf{Game1b}}}
\newcommand{\GameBa}{\ensuremath{\mathsf{Game2a}}}
\newcommand{\GameBb}
{\ensuremath{\mathsf{Game2b}}}
\newcommand{\GameBc}
{\ensuremath{\mathsf{Game2c}}}
\newcommand{\GameCa}{\ensuremath{\mathsf{Game3}}}

\newcommand{\sfv}{\ensuremath{\mathsf{v}}}
\newcommand{\inn}{\ensuremath{\mathsf{in}}}
\newcommand{\midd}{\ensuremath{\mathsf{mid}}}
\newcommand{\outt}{\ensuremath{\mathsf{out}}}
\newcommand{\rev}{\ensuremath{\mathsf{rev}}}

\newcommand{\scrA}{\ensuremath{\mathscr{A}}}
\newcommand{\scrB}{\ensuremath{\mathscr{B}}}
\newcommand{\scrE}{\ensuremath{\mathscr{E}}}
\newcommand{\scrV}{\ensuremath{\mathscr{V}}}
\newcommand{\nil}{\textup{\texttt{nil}}}
\newcommand{\emp}{\textup{\texttt{empty}}}

\newcommand{\qbinom}[2]{\genfrac{[}{]}{0pt}{}{#1}{#2}_{2}}

\DeclareMathOperator{\toplus}{{\textstyle\bigoplus}}

\DeclareMathOperator{\Lspan}{Lspan}
\DeclareMathOperator{\Had}{Had}
\DeclareMathOperator{\Lvec}{Lvec}
\DeclareMathOperator{\Bit}{Bit}

\newcommand{\Varu}[2]{\underset{ #1 } {\mathrm{Var}} \left[ #2 \right]}

\newcommand{\Ex}[1]{\bE \left[ #1 \right]}
\newcommand{\Exu}[2]{\underset{#1} \bE \left[ #2 \right] }

\newcommand{\Exs}[2]{\bE_{#1}\left[ #2 \right]}
\renewcommand{\Pr}[1]{\bP \left[ #1 \right]} % Probability
\newcommand{\Prs}[2]{\bP_{#1}\left[ #2 \right]}
\newcommand{\Pru}[2]{\underset{ #1 }\bP \left[ #2 \right]}

\newcommand{\ind}[1]{\mathds{1} \left[ #1 \right] }

\newcommand{\inp}[2]{\left\langle #1, #2 \right\rangle}

\newcommand{\codim}{\ensuremath{\mathrm{codim}}}

\newcommand{\Gr}{\ensuremath{\mathsf{Gr}}}
\newcommand{\eps}{\varepsilon}
\newcommand{\inner}[2]{\langle{#1},{#2}\rangle}

\newcommand{\Vrf}{\operatorname{Vrf}}

\allowdisplaybreaks

\title{On the Hardness of $4$-to-$1$ Games with Perfect Completeness}
\author{Yumou Fei\thanks{Department of EECS, Massachusetts Institute of Technology. Supported by NSF CCF award 2310818 and NSF DMS award 2022448.}\and Dor Minzer\thanks{Department of Mathematics, Massachusetts Institute of Technology. Supported by NSF CCF award 2227876 and NSF CAREER award 2239160.}\and Shuo Wang\thanks{Department of Mathematics, Massachusetts Institute of Technology. Supported by NSF award 2239160.}}
\date{\vspace{-5ex}}
\begin{document}
\maketitle
\begin{abstract}
    We prove that for all $\eps>0$, there exists a positive integer $k$ such that given a $4$-to-$1$ game $\Psi$ with alphabet size at most $k$, it is $\mathbf{NP}$-hard to distinguish between the case that $\operatorname{val}(\Psi)=1$ and the case that $\operatorname{val}(\Psi)\leq \eps$. This confirms the $4$-to-$1$ Games Conjecture from [Khot, \textit{CCC 2002}]. Previously, the best known result, due to [Dinur, Khot, Kindler, Minzer, Safra], established the almost-perfect completeness version (but applied to the stricter problem of $2$-to-$1$ games). Using results from the literature, we get the following implications:
    \begin{enumerate}[label=(\arabic*)]
        \item For all $k\in \bN$, given a $3$-colorable graph $G$, it is $\mathbf{NP}$-hard to find a proper $k$-coloring. 
        \item For all $\delta>0$, given a $2$-colorable $3$-uniform hypergraph $G$, it is $\mathbf{NP}$-hard to find in it an independent set containing at least $\delta$ fraction of the vertices.
    \end{enumerate}
    
    Our proof is a three-step construction that builds on the two-step framework of [Dinur, Khot, Kindler, Minzer, Safra]. In the outer-PCP step, we use quadratic equations to gain perfect completeness. We then construct a new middle PCP that performs low-rank tests while preserving a key covering property. Finally, we construct a new inner PCP based on a tensor product of the standard Grassmann encoding with its low-rank variant due to [Golowich, FOCS 2023].
\end{abstract}
% \vspace{30ex}
% The current version of the manuscript is complete mathematically, but it is \textbf{not} in the shape we wished to share in. We have chosen to do so due to rumors that surfaced around Friday, September 11th. We will work on a more complete version of the manuscript in the close future.

\newpage 
\tableofcontents

\newpage 

\section{Introduction}
The PCP Theorem~\cite{FGLSS,AS,ALMSS} is a fundamental result in theoretical computer science with many consequences. Most relevant to the current paper are the applications in the area of hardness of approximation, in which the PCP theorem often serves as the starting point for reductions. The combination of Fourier analysis with the long-code (or, in some cases, other error correcting codes) has proven itself to be a very successful recipe, leading to many strong hardness of approximation results~\cite{Feige,Hastad,DinurSafra2005,Zuckerman2007}. However, this paradigm appears to be incapable of establishing optimal results for many other problems, and below are a few notable examples:
\begin{enumerate}[label=(\arabic*)]
    \item The Max-Cut problem: given a graph $G = (V,E)$, the goal is to partition the vertex set into two sides $A,B$ so as to maximize the number of edges crossing between $A$ and $B$. A polynomial time algorithm due to~\cite{GoemansWilliamson1995} achieves the approximation ratio $\alpha_{GW}\approx 0.878$, and the best known $\mathbf{NP}$-hardness result~\cite{TrevisanSorkinSudanWilliamson2000} has ratio $\frac{16}{17}+\eps$, for any $\eps>0$.
    \item The Vertex-Cover problem: given a graph $G = (V,E)$, the goal is to find the smallest set of vertices $A\subseteq V$ that touches all edges of $E$. Known polynomial time algorithms achieve approximation ratio $2$, and the best known $\mathbf{NP}$-hardness result~\cite{KMS1,DKKMS18,DKKMS2,KMS23} has ratio $\sqrt{2}-\eps$, for any $\eps>0$.
    \item The $3$-coloring problem: given a graph $G = (V,E)$ promised to be $3$-colorable, the goal is to find a proper $k$-coloring of the vertices of $G$ with $k$ as small as possible. Polynomial time algorithms for this problem have been studied extensively over the years~\cite{Wigderson1983,KargerMotwaniSudan1998,BlumKarger1997,AroraChlamtacCharikar2006,Chlamtac2007,KawarabayashiThorup2017,KawarabayashiThorupYoneda2024,BansalHuangLee2026}, and the state-of-the-art algorithm achieves $k=|V|^{0.2-\eps}$ for a small explicit $\eps>0$. The $\mathbf{NP}$-hardness of this problem has also been studied~\cite{KhannaLinialSafra2000,GuruswamiKhanna2004,BartoBulinKrokhinOprsal2021}, and the state-of-the-art $\mathbf{NP}$-hardness result is for $k=5$.
\end{enumerate}
To make progress on these problems, Khot~\cite{Khot} introduced the Unique-Games Conjecture and the $d$-to-$1$ Games conjectures, which we present next.
\subsection{Label-Cover, Unique Games and $d$-to-$1$ Games}
It is convenient to discuss PCPs in the language of the label-cover problem, defined as follows.
\begin{definition}\label{def:lc}
    An instance $\Psi$ of the label-cover problem consists of a bipartite graph $G = (L\sqcup R, E)$,\footnote{We always allow parallel edges, so $E$ is a multi-set.} alphabets $\Sigma_L$, $\Sigma_R$, and a collection of constraints $\{\Phi_{e}\subseteq \Sigma_L\times \Sigma_R\}_{e\in E}$. For any given positive integers $d$ and $d'$, a constraint $\Phi_e\subseteq \Sigma_{L}\times \Sigma_{R}$ is called a $d$-to-$d'$ constraint if there is a partition $\Sigma_{L}=\Sigma_{L,1}\sqcup \dots\sqcup\Sigma_{L,t}$ into sets of size $d$ and a partition $\Sigma_{R}=\Sigma_{R,1}\sqcup \dots\sqcup \Sigma_{R,t}$ into sets of size $d'$ such that
    \[
    \Phi_{e}=\bigcup_{i=1}^{t}\Sigma_{L,i}\times \Sigma_{R,i}.
    \]
\end{definition}
Given a label-cover instance $\Psi$ and assignments $A_L\colon L\to\Sigma_L$, $A_R\colon R\to\Sigma_R$, we define 
\[
\operatorname{val}_{\Psi}(A_L,A_R)
=\frac{\left|\{e=(u,v)\in E~|~(A_L(u),A_R(v))\in\Phi_e\}\right|}{|E|},
\]
that is, the fraction of constraints that are satisfied by $A_L,A_R$. We define the value of $\Psi$ as
\[
\operatorname{val}(\Psi)=\max_{A_L,A_R} \operatorname{val}_{\Psi}(A_L,A_R).
\]

For parameters $0<s<c\leq 1$, an input to the promise problem $\mathsf{GapLC}(c,s)$ is a label cover instance $\Psi$ that either has 
$\operatorname{val}(\Psi)\geq c$ or $\operatorname{val}(\Psi)\leq s$, and the goal is to distinguish between the two cases. The basic PCP theorem~\cite{FGLSS,AS,ALMSS} combined with the parallel repetition theorem~\cite{Raz} gives:
\begin{theorem}\label{thm:PCP}
    For all $\eps>0$, there exists a positive integer $r$ such that $\mathsf{GapLC}_r(1,\eps)$ is $\mathbf{NP}$-hard. Here $\mathsf{GapLC}_r$ is the $\mathsf{GapLC}$ problem in which both alphabets have size at most $r$.
\end{theorem}

Khot conjectured that versions of~\Cref{thm:PCP} continue to hold even when the constraints are $d$-to-$1$ for a small $d$. Towards presenting these conjectures, we make the following definition.
\begin{definition}\label{def:Gap4to1}
    Let $d$ and $d'$ be positive integers. An instance $\Psi$ of label-cover is called a $d$-to-$d'$ game if all of its constraints are $d$-to-$d'$. In the case that $d=d'=1$, $\Psi$ is called a unique-game. We write \(\mathsf{GapUG}_r(c,s)\) for the restriction of \(\mathsf{GapLC}_r(c,s)\) to unique-game instances, and \(\Gapdone{d}{d'}_r(c,s)\) for its restriction to \(d\)-to-\(d'\) game instances.
\end{definition}
It is easy to see that one cannot hope for~\Cref{thm:PCP} to hold for unique-games, and Khot \cite{Khot} conjectured that allowing for imperfect completeness remedies that: 
\begin{conjecture}\label{conj:UGC}
    For any $\eps>0$, there exists $r\in\mathbb{N}^{+}$, such that $\mathsf{GapUG}_r(1-\eps,\eps)$ is $\mathbf{NP}$-hard.
\end{conjecture}
For $d\geq 2$ and $d'=1$, he conjectured that the assertion of~\Cref{thm:PCP} should hold as is. More precisely:
\begin{conjecture}\label{conj:d-to-1}
    For any $\eps>0$ and any integer $d\geq 2$, there exists a positive integer $r$ such that $\Gapdone{d}{1}_r(1,\eps)$ is $\mathbf{NP}$-hard.
\end{conjecture}
These conjectures, and especially~\Cref{conj:UGC}, have been shown to have remarkable implications, see for example~\cite{KhotRegev,KKMO,Raghavendra,KhotVishnoi2015} and~\cite{Khot2010UniqueGamesSurvey,Khot2010ICM,Trevisan2012UGC,Khot2014Hardness} for surveys, and significant effort has gone into studying them. 
On the algorithmic front, connections with semi-definite programming, sum-of-squares algorithms and graph expansion have led to the development of very powerful techniques and a number of algorithms for unique-games with non-trivial guarantees~\cite{CharikarMakarychevMakarychev2006,AroraKhotKollaSteurerTulsianiVishnoi2008,AroraBarakSteurer2015}. On the hardness front, a line of work~\cite{KMS1,DKKMS18,DKKMS2,KMS23} (with contributions from~\cite{BKS19,KMMs}) established the imperfect completeness version of~\Cref{conj:d-to-1},  showing that for all $\eps>0$, there is $r\in\mathbb{N}^{+}$ such that $\Gapdone{2}{1}_r(1-\eps,\eps)$ is $\mathbf{NP}$-hard. This implies that for all $\eps>0$ there exists $r\in\mathbb{N}^{+}$ such that $\mathsf{GapUG}_r(1/2,\eps)$ is $\mathbf{NP}$-hard, which to date is the best evidence towards~\Cref{conj:UGC}.

The difference between~\Cref{conj:d-to-1} and its imperfect completeness version is significant in a number of ways:
\begin{enumerate}[label=(\arabic*)]
    \item \textbf{Implications:} there are a number of hardness of approximation results that crucially rely on perfect completeness,  most prominently graph coloring problems.
    \item \textbf{Techniques:} the approach taken by~\cite{KMS1,DKKMS18,DKKMS2,KMS23} inherently incurs imperfect completeness. Currently, constructions of low-soundness PCPs with perfect completeness follow composition of PCPs, for example~\cite{MoshkovitzRaz2010,DinurHarsha2013,DinurHarshaKindler2015}, which at the very least require constructions that already have perfect completeness and small soundness (with perhaps a large alphabet which is then reduced), or from parallel repetition~\cite{Raz,Holenstein,rao2008parallel}, which inherently increases the $d$-parameter rapidly.
\end{enumerate}

\subsection{Main Result}
Our main result is that~\Cref{conj:d-to-1} holds for $d=4$ and $d'=1$.
\begin{theorem}\label{thm:main}
    For all $\eps>0$, there exists $r\in\mathbb{N}^{+}$, such that 
    $\Gapdone{4}{1}_r(1,\eps)$ is $\mathbf{NP}$-hard.
\end{theorem}
%\begin{remark}\label{remark:2to2}
%    With minor modifications, our argument shows that~\Cref{thm:main} holds also for $2$-to-$2$ games, which are label-cover instances where each constraint is defined by a $2$-to-$2$ map. We do not know how 
%    to prove an $\mathbf{NP}$-hardness result for $2$-to-$1$ games, though.
%\end{remark}

\subsubsection{Implications}
We now discuss some implications of~\Cref{thm:main}. Using known connections between $d$-to-$1$ games and graph coloring/independent sets~\cite{DinurMosselRegev2009,krokhin2023topology,GuruswamiSandeep2020}, we get the following results.
\begin{corollary}\label{cor:coloring}
    For any integer $k\geq 3$, given a graph $G$, it is $\mathbf{NP}$-hard to distinguish between the following two cases:
    \begin{enumerate}[label=(\arabic*)]
        \item \textbf{Yes case:} $G$ is $3$-colorable.
        \item \textbf{No case:} $G$ is not $k$-colorable.
    \end{enumerate}
\end{corollary}
Previously,~\Cref{cor:coloring} was known~\cite{BartoBulinKrokhinOprsal2021} for $k=5$, or in the case the condition in the yes-case only required a proper $3$-coloring of $1-\eps$ fraction of the vertices~\cite{HechtMinzerSafra2023}. 
\begin{remark}
    Under slightly super-polynomial time reductions, our techniques give versions of~\Cref{cor:coloring} for $k$ growing with the instance size, but at a very slow rate. Namely, assuming $\mathbf{NP}$ has no algorithms running in time $n^{t(n)}$, we get that there are no polynomial time algorithms that color a given $3$-colorable graph with $O(\log\cdots\log t(n))$ colors, where the number of logs is finite (probably at most $10$, say).
\end{remark}

In the next corollary, the condition in the completeness case is weaker, but the condition in the soundness case is stronger:

\begin{corollary}\label{cor:IS_vs_coloring}
    For all $\delta>0$, given a graph $G$, it is $\mathbf{NP}$-hard to distinguish between the following two cases:
    \begin{enumerate}[label=(\arabic*)]
        \item \textbf{Yes case:} $G$ is $8$-colorable.
        %\textcolor{red}{Dor: this is $4$ if the remark I wrote about $2$-to-$2$ is correct.}
        \item \textbf{No case:} $G$ has no independent set of relative size at least $\delta$.
\end{enumerate}
\end{corollary}
Using~\cite{KhotSaket2014}, we get the following result for hypergraphs:
\begin{corollary}\label{cor:IS_vs_coloring_hypergraph}
    For all $\delta>0$, given a $3$-uniform hypergraph $G$, it is $\mathbf{NP}$-hard to distinguish between the following two cases:
    \begin{enumerate}[label=(\arabic*)]
        \item \textbf{Yes case:} $G$ is $2$-colorable.
        \item \textbf{No case:} $G$ has no independent set of relative size at least $\delta$.
\end{enumerate}
\end{corollary}

Using the results of~\cite{HardtMekaRaghavendraWeitz2014},~\Cref{cor:coloring,cor:IS_vs_coloring} show hardness results for low-rank matrix completion. An instance in the $(k,r,\eps,c)$-completion problem is a matrix $A\in(\mathbb{R}\cup\{\perp\})^{n\times n}$, such that there exists $M\in\mathbb{R}^{n\times n}$ satisfying:
\begin{enumerate}[label=(\arabic*)]
    \item $|M(i,j)|\leq c$ and $M$ has rank at most $r$.
    \item Letting $\Omega\subseteq[n]\times [n]$ be the set of $(i,j)$ such that $A(i,j)\neq \bot$, we have that $|\Omega|\geq 0.9 n^2$ and for each $(i,j)\in \Omega$ we have $M(i,j) = A(i,j)$.
\end{enumerate}
The goal is to find a matrix $B\in\mathbb{R}^{n\times n}$ of rank at most $k$, such that $|B(i,j)|\leq c$ for all $(i,j)$ and 
 \[
 \sum\limits_{(i,j)\in\Omega}|A(i,j)-B(i,j)|^2\leq \eps n.
 \]
 Applying~\cite[Theorem 1.8]{HardtMekaRaghavendraWeitz2014} on~\Cref{cor:coloring} gives:
\begin{corollary}
    For all $k > 3$ and $c>1$, the $(k,3,0,c)$-completion problem is $\mathbf{NP}$-hard.
\end{corollary}

Applying~\cite[Theorem 1.9]{HardtMekaRaghavendraWeitz2014} on~\Cref{cor:IS_vs_coloring} gives:
\begin{corollary}
    For all $k> 8$ and $c>1$, the $(k,8,\eps,c)$-completion problem is $\mathbf{NP}$-hard for all $0<\eps<\frac{1}{(2ck)^2}$.
\end{corollary}

\subsection{Our Techniques}
In this section we give a brief account of the components that go into the proof of~\Cref{thm:main}.
We refer the reader to~\Cref{sec:proof_overview} for a more thorough description.

For $2$-to-$1$ games with almost-perfect completeness, the proof of~\cite{KMS1,DKKMS18,DKKMS2,KMS23} combines the following two components:
\begin{enumerate}[label=(\arabic*)]
    \item An outer PCP with completeness close to $1$, vanishing soundness, and a feature called the ``covering property'' (originally appearing in~\cite{KhotSafra}). This game is very far from being a $2$-to-$1$ game, and to turn it into a $2$-to-$1$ game it is composed with an appropriate inner PCP. The goal of the covering property is to facilitate the composition step.
    \item An inner PCP with perfect completeness and a non-standard notion of soundness. The inner PCP is based on the Grassmann graph, the Grassmann encoding and the Grassmann agreement test. The Grassmann agreement test can be viewed as a local tester for the Grassmann code, but it turns out not to be sound in the traditional error correcting codes sense. Instead, it satisfies a weaker notion of soundness, which we refer to as non-standard soundness.
\end{enumerate}
The outer PCP and inner PCP are then composed, and the covering property of the outer PCP is used in the soundness analysis of the composed PCP (even though the inner PCP only satisfies a non-standard notion of soundness). 

The outer PCP construction of~\cite{KMS1,DKKMS18,DKKMS2,KMS23} is based on the hardness of approximation result of~\cite{Hastad} for solving linear equations, which inherently has imperfect completeness, and this propagates throughout their construction. To remedy that, a key difficulty is to come up with an outer PCP that has perfect completeness and simultaneously admits (some version of) the covering property. 

We start our outer PCP construction based on hardness of solving a system of quadratic equations (which has perfect completeness). However, the Grassmann agreement test (which is still basically the only tool we have to achieve $d$-to-$1$-ness in the inner PCP) is inherently a ``linear'' test and cannot incorporate quadratic equations. In order for quadratic equations to be accommodated in a $d$-to-$1$ way, we add a layer of construction called the ``middle PCP'' between the outer and inner PCPs. In the middle PCP, we carry out a low-rank test that allows us to convert the quadratic equations into linear ones.

However, this linearization of quadratic equations destroys the covering property that is essential for composing with a Grassmann-based inner PCP. To remedy that, we have to come up with another two key ideas: 
\begin{enumerate}[label=(\arabic*)]
\item We do a ``lazy parallel repetition'' in the middle PCP. This creates room for covering properties while preserving the soundness.
\item The covering property we get from lazy parallel repetition is still not sufficient for composing with the Grassmann agreement test of~\cite{KMS1,DKKMS18,DKKMS2,KMS23}. As a workaround, we base our inner PCP on a combination of standard Grassmann graphs over matrix spaces $\mathbb{F}_2^{n\times n}$ with a sparser version of them, known as the low-rank Grassmann graph over $\mathbb{F}_2^{n\times n}$, recently defined by~\cite{Golowich23}.
\end{enumerate}

\paragraph{Statement of AI use:} the mathematical ideas and the overall approach presented in this paper are due to the authors. ChatGPT 5.6, 6 and 6.1 were used for polishing and proofreading.

\section{Proof Overview}\label{sec:proof_overview}

In this section we give an overview of the proof of~\Cref{thm:main}. It will be useful to first discuss the construction of~\cite{KMS1,DKKMS18,DKKMS2,KMS23} which achieves imperfect completeness. We then discuss why it fails to have perfect completeness, and then in~\Cref{sec:techniques} we explain preliminary attempts at fixing that, building towards our construction. 

\subsection{The Imperfect Completeness Construction}
In this section we explain the components that go into the construction of~\cite{KMS1,DKKMS18,DKKMS2,KMS23}.
\subsubsection{The Outer PCP}\label{sec:outer_pcp_imperfect}
It will often be easier for us to describe label-cover instances in the equivalent, active language of $2$-prover-$1$-round games, which we introduce next. In a $2$-prover-$1$-round game, a verifier samples a challenge $e = (u,v)$, sends $u$ to a prover named Alice, sends $v$ to a prover named Bob. Upon receiving their challenges, they respond with labels (they are not allowed to communicate once they get their challenges), and based on their responses the verifier decides to accept or reject. The value of the game is the maximum, over all strategies of Alice and Bob, of the probability that the verifier accepts. Identifying the individual challenges the players receive as vertices of a bipartite graph, pairs of challenges the verifier may sample with its edges, and the verifier's acceptance criteria with the constraints, one gets the equivalence. It is easy to see that the value of a game is exactly the value of the resulting label cover instance. 

With this in mind, the reduction of~\cite{KMS1,DKKMS18,DKKMS2,KMS23} starts with the result of~\cite{Hastad} regarding the approximability of $3$\textsf{Lin}$_{\mathbb{F}_2}$.
\begin{definition}
    An instance of the $3$\textsf{Lin}$_{\mathbb{F}_2}$ problem consists of a set of variables $\scrV$ and a set of equations $\scrE$, where each $e\in \scrE$ is a linear equation over $\mathbb{F}_2$ containing $3$ distinct variables.
\end{definition}
H{\aa}stad proved that for all $\eps>0$, \textsf{Gap}$3$\textsf{Lin}$_{\mathbb{F}_2}(1-\eps,1/2+\eps)$ is $\mathbf{NP}$-hard, and the reduction of~\cite{KMS1,DKKMS18,DKKMS2,KMS23} starts off with (some regularized version of) this problem. More specifically, they use the parallel repetition theorem to prove $\mathbf{NP}$-hardness for a certain type of label-cover instance $\Psi_{q,\beta,k} = \Psi_{q,\beta,k}(\scrV,\scrE)$, which we present below in the active $2$-prover-$1$-round game view, and refer to as the $\beta$-smooth $k$-fold repeated game. For technical reasons, this game contains an additional feature called ``advice'', which we encourage the reader to ignore in a first reading.
 
\begin{enumerate}
    \item \textbf{Alice's challenge:} the verifier samples equations $e_1,\ldots,e_k$ from $\scrE$ uniformly and independently, and sets $E$ to be the set of all variables that appear in them.
    \item \textbf{Bob's challenge:} for each $i\in \{1,\dots,k\}$ independently, with probability $1-\beta$, the verifier takes $E_i'$ to be the set of variables in $e_i$, and with probability $\beta$ the verifier chooses $E_i'$ to consist of a single uniformly chosen variable from $e_i$. The verifier sets $E' = E_1'\cup\ldots\cup E_k'$.
    \item \textbf{Advice:} the verifier chooses $b^{(1)},\ldots,b^{(q)}\in \mathbb{F}_2^{E'}$ uniformly, and then takes $a^{(1)},\ldots,a^{(q)}\in\mathbb{F}_2^{E}$ where $a^{(i)}$ is obtained from $b^{(i)}$ by filling the missing coordinates with $0$'s. We write $B=\bigl(b^{(1)},\ldots,b^{(q)}\bigr)$ and $A=\bigl(a^{(1)},\ldots,a^{(q)}\bigr)$.
    \item \textbf{Acceptance criteria:} the verifier sends $(E,A)$ to Alice, 
    and $(E',B)$ to Bob. They respond with vectors $x\in \mathbb{F}_2^E$ and $y\in\mathbb{F}_2^{E'}$ respectively, and the verifier accepts if and only if $x$ satisfies all of the equations $e_1,\ldots,e_k$ and $x|_{E'}= y$.
\end{enumerate}
The works~\cite{KMS1,DKKMS18,DKKMS2,KMS23} prove that for all $\eps>0$, it is $\mathbf{NP}$-hard to distinguish between the case that $\operatorname{val}(\Psi_{q,\beta,k})\geq 1-\eps$, and the case that $\operatorname{val}(\Psi_{q,\beta,k})\leq 2^{-\Omega_{q}(\beta k)}$. Here and throughout, one should think of $q$ as an absolute constant, and of $\beta$ as slightly larger than $1/k$, say $\beta= k^{-3/4}$.
The game $\Psi_{q,\beta,k}$ has several useful features in the context of $2$-to-$1$ Games:
\begin{enumerate}[label=(\arabic*)]
    \item \textbf{Simplicity of the constraints:} the constraints of $\Psi_{q,\beta,k}$ can all be seen to be linear equations in the responses of the provers. This is because the starting point, the $3$\textsf{Lin}$_{\mathbb{F}_2}$, consisted of linear constraints, and on top of that we have the consistency constraints between the responses of Alice and Bob, which are also linear.
    \item \textbf{The covering property:} for a fixed $E$, the spaces $\mathbb{F}_2^{E'}$ living on Bob's side cover Alice's space $\mathbb{F}_2^{E}$ almost perfectly. For example, fixing $E$, taking $E'$ randomly as described above conditioned on $E$, then $y\in\mathbb{F}_2^{E'}$ and lifting it to $x\in\mathbb{F}_2^{E}$ by filling in the missing coordinates with $0$'s, we have that the distribution of $x$ is statistically close to the uniform distribution over $\mathbb{F}_2^{E}$ (this property first appeared in~\cite{KhotSafra} in a somewhat different context).
\end{enumerate}
The relevance of these properties will become clear in the next part, the inner PCP. 

\subsubsection{The Inner PCP}\label{sec:inner_pcp_imperfect}

The 2-prover 1-round game $\Psi_{q,\beta,k}$ constructed in~\Cref{sec:outer_pcp_imperfect} is far from $2$-to-$1$, because given a Bob-answer $y\in \bF_{2}^{E'}$, there are many possible Alice-answers $x\in \bF_{2}^{E}$ consistent with it. In order to make it a $2$-to-$1$ game, in a single query we can ask the provers only for a smaller piece of information.   

For each possible question $(E,A)$ to Alice, we ask her to prepare her answer in ``Grassmann encoding,'' and instead of requesting a full answer $x\in \bF_{2}^{E}$, we ask her to return only the information in one location of the Grassmann encoding of $x$. 

The Grassmann encoding scheme is based on the Grassmann graph
$\mathsf{Gr}(\bF_{2}^{E},\ell)$, whose vertices are the $\ell$-dimensional
linear subspaces of $\bF_{2}^{E}$. Two vertices $L$ and $L'$ are adjacent
if and only if $\dim(L\cap L')=\ell-1$. By a slight abuse of notation,
we also use $\mathsf{Gr}(\bF_{2}^{E},\ell)$ to denote the vertex set of
this graph. The Grassmann encoding of a vector $x\in\bF_{2}^{E}$ stores, at each
location $L\in\mathsf{Gr}(\bF_{2}^{E},\ell)$, the local linear function
\[
F[L]:L\rightarrow \bF_{2}\text{ defined by }z\mapsto \inp{z}{x}.
\]

Similarly, we ask Bob to prepare his answer $y\in \bF_{2}^{E'}$ in Grassmann encoding, where each location $L'\in\mathsf{Gr}(\bF_{2}^{E'},\ell-1)$ stores the local function
\[
G[L']:L'\rightarrow \bF_{2}\text{ defined by }z\mapsto \inp{z}{y}.
\]

The final $2$-to-$1$ game is the composition of the outer PCP with the Grassmann graph gadget. In the composed PCP, upon choosing $(E,A)$ and $(E',B)$ as in $\Psi_{q,\beta,k}$, the verifier chooses a space $L'$ from $\mathsf{Gr}(\mathbb{F}_2^{E'},\ell-1)$ and a space $L\in\mathsf{Gr}(\bF_{2}^{E'},\ell)$ containing $L'$. The subspace $L\subseteq \bF_{2}^{E'}$ will be regarded as a subspace of Alice's space $\bF_{2}^{E}$ by appending $0$'s in the missing coordinates. The verifier then sends $(E,A,L)$ to Alice and $(E',B,L')$ to Bob, expecting to get linear functions $F[E,A,L]\colon L\to\mathbb{F}_2$ and 
$G[E',B,L']\colon L'\to\mathbb{F}_2$. The verifier accepts only if $G[E',B,L']$ agrees with the restriction $F[E,A,L]|_{L'}$. This is essentially a $2$-to-$1$ game because for any possible choice of $G[E',B,L']$, there are at most two choices of $F[E,A,L]$ consistent with it.

\paragraph{Covering property and consistency:} consider a question $(E,A)$ for Alice in the outer PCP (think of $\ell$ as much smaller than $k$), and the Grassmann graph $\mathsf{Gr}(\bF_{2}^{E},\ell)$. The covering property above asserts that if we sample $E'$ as in the outer PCP and a vertex in the graph $\mathsf{\Gr}(\bF_{2}^{E'},\ell)$, lifting it to $\mathsf{Gr}(\bF_{2}^{E},\ell)$ yields a nearly uniform vertex in the latter graph. This implies that if Alice's assignment $F[E,A,\cdot]$ has agreement $\varepsilon'$ with an actual Grassmann codeword $F$, then with decent probability over $E'$, Bob's assignment $G[E',B,\cdot]$ agrees with the projection of $F$ to $\mathbb{F}_2^{E'}$. In short, the covering property ensures the consistency between decodings of assignments of the provers holds.

\paragraph{The linear equations side conditions:} recall that in the outer PCP we want to ensure that Alice's response $x\in\mathbb{F}_2^{E}$ satisfies a collection of $k$ linear equations. Note that these equations correspond to requiring that the function $z\mapsto\inner{z}{x}$ attains specific values on $k$ inputs. More precisely, letting the equations be $e_1,\ldots,e_k$, letting $s^{(i)}\in\mathbb{F}_{2}^{E}$ be the indicator vector of the variables appearing in $e_i$, and letting $t_i\in\mathbb{F}_2$ be the constant term of the equation $e_i$, we want to ensure that $\inp{s^{(i)}}{x}=t_i$. To do so, we define the space $N_{E} = \operatorname{span}(s^{(1)},\ldots,s^{(k)})$, and request Alice to give assignment $F[E,A,L+N_{E}]\colon (L+N_{E})\to\mathbb{F}_2$ (instead of just $L\rightarrow \bF_{2}$). The verifier then checks consistency between Alice and Bob similarly to the above, and additionally that $F[E,A,L+N_{E}]|_{N_{E}}= g_E$, where the linear function $g_{E}\colon N_{E}\to\mathbb{F}_2$ is defined by $g_E(s^{(i)})=t_i$.

\subsubsection{On the Non-Standard Notion of Soundness}\label{sec:nonstnadrad_soundness}

One would wish that the Grassmann gadget in the inner PCP has the following type of soundness: suppose $F[L]$ is a linear map $L\rightarrow \bF_{2}$ for each $L\in \mathsf{Gr}(\bF_{2}^{E},\ell)$, such that
\begin{equation}\label{eq:overview_Grassmann_test}
\Pru{L_{1},L_{2}\in \mathsf{Gr}(\bF_{2}^{E},\ell)}{F[L_{1}]|_{L_{1}\cap L_{2}}=F[L_{2}]|_{L_{1}\cap L_{2}}\,\Big|\,\dim(L_{1}\cap L_{2})=\ell-1}\geq \delta
\end{equation}
for some constant $\delta\in (0,1)$, then there exists a linear function $f:\bF_{2}^{E}\rightarrow \bF_{2}$ such that
\[
\Pru{L\in \mathsf{Gr}(\bF_{2}^{E},\ell)}{F[L]=f|_{L}\Big.}\geq \delta'
\]
for some $\delta'\in (0,1)$. However, this conclusion cannot be made, as we next explain. Take $\varepsilon>0$ to be a small constant, and take a random set $X\subseteq\mathbb{F}_2^{E}\setminus\{0\}$ of size $2^{-\ell}\eps  \cdot 2^{|E|}$. For each $x\in X$ choose a random linear function $f_x\colon \mathbb{F}_2^{E}\to\mathbb{F}_2$, and set a table $F$ on $\ell$-dimensional subspaces $L$ as follows:
\begin{itemize}
    \item If $L$ contains a unique point $x\in X$, define $F[L] = f_x|_{L}$.
    \item Otherwise, choose $F[L]$ as a random linear function on $L$.
\end{itemize}
Using the inclusion-exclusion formula, one can prove that for a random edge $\{L_1,L_2\}$ in the Grassmann graph, with probability at least $\Omega(\varepsilon)$ there is $x\in X$ such that $L_1\cap X = L_2\cap X = \{x\}$, and $F$ satisfies such constraints. Thus, $F$ passes the Grassmann test \Cref{eq:overview_Grassmann_test} with $\delta=\Omega(\eps)$. However, it only agrees on $o(1)$ locations with any given codeword from the Grassmann code. 

This example tells us that while there might be no linear function that agrees with $F$ globally, it may be possible to extract a linear function by focusing on certain substructures inside the Grassmann graph. Indeed, it turns out that what can be concluded from~\Cref{eq:overview_Grassmann_test} is something of this form: for at least $\exp(-O(\ell^{2}))$ fraction of $Q\in \mathsf{Gr}(\bF_{2}^{E},q)$, where $q$ depends only on $\delta$, there exists $W\in \mathsf{Gr}(\bF_{2}^{E},|E|-q)$ and a linear function $f:W\rightarrow \bF_{2}$ such that
\begin{equation}\label{eq:overview_Grassmann_non_standard}
\Pru{L\in \mathsf{Gr}(\bF_{2}^{E},\ell)}{F[L]\equiv f|_{L}\,\Big|\,Q\subseteq L\subseteq W}\geq \delta'.
\end{equation}
We refer to the spaces $Q$ and $W$ in~\Cref{eq:overview_Grassmann_non_standard} as zoom-ins and zoom-outs, respectively, and encourage the reader to ignore $W$ (it will be irrelevant for the majority of the discussion). In words,~\Cref{eq:overview_Grassmann_non_standard} asserts that if $F$ passes the Grassmann test with noticeable probability, then inside non-trivial fraction of zoom-ins of constant dimension $q$, it does have a global linear structure.

The proof of~\cite{KMS1,DKKMS18,DKKMS2,KMS23} has to deal with this non-standard notion of soundness, and this is the purpose of the advice feature in the game $\Psi_{q,\beta,k}$. More precisely, given Alice and Bob's tables of assignments $F$ and $G$ for the composed PCP, one needs to be able to derive a strategy for the original game $\Psi_{q,\beta,k}$. However, the strategy that one can derive using \Cref{eq:overview_Grassmann_non_standard} heavily depends on the choice of what zoom-in $Q$ is used. The advice feature allows the two provers in $\Psi_{q,\beta,k}$ to use (essentially) the same $Q$ for decoding --- namely the span of their advice vectors --- which guarantees consistency between their answers.

\subsubsection{Difficulties in Achieving Perfect Completeness}

The construction for $2$-to-$1$ games with imperfect completeness can be illustrated using~\Cref{fig:imperfect_completeness}.

\begin{figure}[H]
\centering
\begin{tikzpicture}
\node (A) [node, align = center] at (0,0) {\textbf{Outer PCP}\\ 1. equation vs. variable game\\
2. smooth parallel repetition\\
(for low soundness\\
and covering property)};
\node (B) [node, align = center, right=4cm of A]
{\textbf{Inner PCP}\\ Grassmann agreement test};

\draw [single arrow] (A) -- (B);
\end{tikzpicture}
\caption{The imperfect-completeness construction of~\cite{KMS1,DKKMS18,DKKMS2,KMS23}}
\label{fig:imperfect_completeness}
\end{figure}
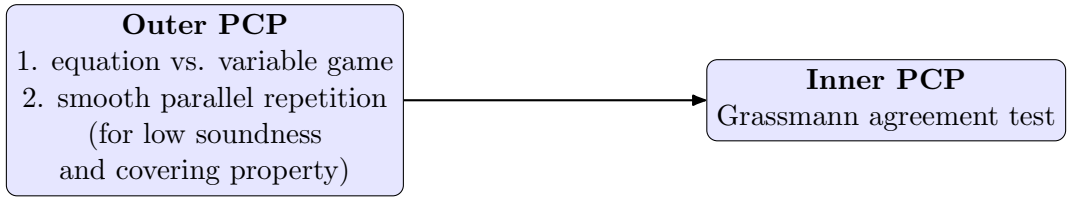

The above construction loses perfect completeness right from the start: it uses the $3$\textsf{Lin}$_{\mathbb{F}_2}$ problem, which inherently has imperfect completeness. Moreover, in some sense this seems necessary: the Grassmann encoding presented above only allows for linear checks in the encoded string $x\in\mathbb{F}_2^n$. 

Thus, to achieve perfect completeness it is necessary to come up with an encoding scheme which is more powerful than the Grassmann encoding, in the sense that it allows for non-linear checks on the encoded string $x\in\mathbb{F}_2^{E}$. By itself this is not a hard task; for example, the long-code would work. However, the key difficulty is that one additionally needs the encoding scheme to satisfy a covering property, and then build an outer PCP with a suitable advice feature.

\subsection{Our Techniques}\label{sec:techniques}
In this section we discuss the new ideas that go into our construction.
\subsubsection{A First Attempt: the Quadratic Grassmann Encoding}\label{sec:first_attempt}
We begin by outlining a natural first attempt at a reduction achieving perfect completeness. It starts with the problem of quadratic equations, $3\mathsf{QS}_{\mathbb{F}_2}$, defined as follows.
\begin{definition}\label{def:3QS}
    An instance of the $3$\textsf{QS}$_{\mathbb{F}_2}$ problem consists of a set of variables $\scrV$ and a set of equations $\scrE$, where each $e\in \scrE$ is a quadratic equation over $\mathbb{F}_2$ containing $3$ distinct variables.
\end{definition}
Using standard results, it can be proved that there is $s<1$ such that 
$\mathsf{Gap}3\mathsf{QS}_{\mathbb{F}_2}[1,s]$ is $\mathbf{NP}$-hard. Starting with an instance $(\scrV,\scrE)$ of $3$\textsf{QS}$_{\mathbb{F}_2}$, one can construct the $\beta$-smooth $k$-fold repeated game corresponding to it, $\Psi_{\beta,k} = \Psi_{\beta,k}(\scrV,\scrE)$ in the same way as in~\Cref{sec:outer_pcp_imperfect} (we ignore the advice for now). The main question is how to compose it with a Grassmann-style encoding so as to get to a $d$-to-$1$ Game. 

In the game $\Psi_{\beta,k}$, the verifier performs two types of checks: consistency between the responses of Alice and Bob, and quadratic equations in the response of Alice. The Grassmann encoding is compatible with the former type of checks, but not with the latter type of checks, in the sense that we do not know how to incorporate checking quadratic equations on the encoded message into the Grassmann encoding itself. Instead, it is natural to consider the \emph{quadratic Grassmann encoding}.

Let $x\in\mathbb{F}_2^{E}$ represent Alice's response in the game $\Psi_{\beta,k}$. In the quadratic Grassmann encoding, one considers the functions over $\mathbb{F}_2^{E\times E}$ defined by $z\mapsto\inner{z}{x\otimes x}$, and takes its Grassmann encoding $QF$ of $x$, which assigns to each $\ell$-dimensional subspace $L\subseteq \mathbb{F}_2^{E\times E}$ the linear function
\[QF[L]:L\rightarrow \bF_{2}\text{ defined by }z\mapsto\inp{z}{x\otimes x}.\]
The main benefit of the quadratic Grassmann encoding is that it indeed allows us to test quadratic equations as side conditions, as we explain next.

\paragraph{The quadratic-equation side conditions:} Given a legitimate quadratic Grassmann codeword $QF$, it is possible to incorporate with it checks of quadratic equations on the encoded message. Indeed, suppose that $e_1,\ldots,e_k$ are quadratic equations over $E$, let $s^{(1)},\ldots,s^{(k)}\in\mathbb{F}_2^{E\times E}$ be the indicator vectors of the pairs that appear in $e_1,\ldots,e_k$ respectively, and let $t_1,\ldots,t_k$ be the constant terms. Then the equation $e_i$ can be written as $\inp{s^{(i)}}{x\otimes x}=t_i$ for each $i\in \{1,\ldots,k\}$. As before, this can be folded into the encoding by defining $N_{E} = \operatorname{span}(s^{(1)},\ldots,s^{(k)})$, asking for the table to assign values to the sum $L+N_E$ (instead of to just $L$), and requiring that their projection onto $N_E$ completely agrees with the linear function $g\colon N_E\to\mathbb{F}_2$ defined by $g(s^{(i)})=t_i$. This way of incorporating quadratic side conditions will work for all of the constructions we present below, so for simplicity of presentation they will be omitted henceforth.

\vspace{1ex}

On the other hand, the quadratic Grassmann encoding has $2$ very significant drawbacks:
\begin{enumerate}[label=(\arabic*)]
    \item \textbf{Codeword testing:} Even ignoring the issue of non-standard-type soundness, the Grassmann agreement test from~\Cref{sec:inner_pcp_imperfect} is not actually a tester for the quadratic Grassmann code. Namely, if a table $G$ that assigns each $\ell$-dimensional subspace $L\subseteq \mathbb{F}_2^{E\times E}$ a linear function $G[L]\colon L\to \mathbb{F}_2$ passes it with non-negligible probability, then at best we will be able to argue that it has significant agreement with the Grassmann encoding of an arbitrary $x'\in\mathbb{F}_2^{E\times E}$, i.e., not necessarily of rank $1$. We need an additional feature that will guarantee that $x'$ has rank $1$ (or at the very least small rank) and is symmetric.
    \item \textbf{The covering property:} The covering property does not hold for the quadratic Grassmann encoding. Namely, fixing a question for Alice $E$ and looking at typical possible questions $E'$ for Bob (conditioned on $E$), the (lifting of the) spaces $\mathbb{F}_2^{E'\times E'}$ do not cover $\mathbb{F}_2^{E\times  E}$. Indeed, the number of such typical questions for Bob is roughly $\binom{3k}{2\beta k}$, and the size of Bob's space is roughly $2^{(3-2\beta )^2 k^2}$, so in total they can cover at most 
    \[
    \binom{3k}{2\beta k}2^{(3-2\beta )^2 k^2}
    \leq 2^{O(\beta k\log(1/\beta))} 2^{9k^2 - 8\beta k^{2}} = o(2^{9k^2})
    \]
    points in $\mathbb{F}_2^{E\times E}$, which is a tiny fraction.
\end{enumerate}
The first issue, namely ensuring that the encoded message has a low rank, is significantly easier to handle, and to do so we incorporate tensor tests. More precisely, instead of asking for the value of the Grassmann encoding at a subspace $L\subseteq\mathbb{F}_2^{E\times E}$, we will also take a random subspace $U\subseteq\mathbb{F}_2^{E}$ of dimension $\ell$, and then ask for the value at the subspace $L + U\otimes U$, requiring that the projection onto $U\otimes U$ is of rank at most $1$. This yields the following $4$-to-$4$ Grassmann agreement +  tensor test. Upon receiving a table $F$ assigning linear functions to spaces $L + U\otimes U$ as above, do the following:
\begin{enumerate}
    \item Sample $L_{1},L_{2}\subseteq\mathbb{F}_2^{E\times E}$ of dimension $\ell$ intersecting at dimension $\ell-1$, and $U_{1},U_{2}\subseteq\mathbb{F}_2^E$ of dimension $\ell$ intersecting at dimension $\ell-1$.
    \item Query $F[L_{1}+U_{1}\otimes U_{1}]$ and $F[L_{2}+U_{2}\otimes U_{2}]$. Accept if they agree on \[L_{1}\cap L_{2}+(U_{1}\cap U_{2})\otimes (U_{1}\cap U_{2}),\] and their respective projections onto $U_{1}\otimes U_{1}$ and $U_{2}\otimes U_{2}$ are symmetric and of rank at most $1$.
\end{enumerate}
Given the (non-standard) soundness of the basic Grassmann agreement test, one can analyze the Grassmann agreement + tensor test and show that it satisfies a similar (non-standard) notion of soundness. Also, this test is $4$-to-$4$: given a linear function on $L_{1}$, there are at most $2$ linear functions on $L_{2}$ that agree with it, and given a symmetric rank-$1$ functional on $U_{1}\otimes U_{1}$, there are at most $2$ symmetric rank-$1$ linear functions on $U_{2}\otimes U_{2}$ agreeing with it. 

Ensuring symmetry is more straightforward. To do that, define for distinct points $i,j\in E$ the matrix $M_{i,j}\in \mathbb{F}_2^{E\times E}$ which is $1$ only at locations $(i,j)$ and $(j,i)$, and set $N_{\mathsf{sym}} = \operatorname{span}(\{M_{i,j}~|~i\neq j\})$. Then, $x'\in\mathbb{F}_2^{E\times E}$ is symmetric if and only if the function $z\mapsto \inp{z}{x'}$ vanishes on $N_{\mathsf{sym}}$. Thus, instead of requesting an assignment on $L+U\otimes U$, the verifier requests an assignment on $L+U\otimes U+N_{\mathsf{sym}}$, and checks that its projection on $N_{\mathsf{sym}}$ vanishes.

The second issue, namely the fact that the covering property fails, is substantially more difficult to handle, and in the next section we begin the discussion about it.

\subsubsection{Variations on the Quadratic Grassmann Encoding}\label{sec:variations_gras}
From now on, it will be convenient for us to think of points from $\mathbb{F}_2^{E\times E}$ and $\mathbb{F}_2^{E'\times E'}$ as matrices. The covering property fails for the quadratic Grassmann encoding because if, letting Bob's question be denoted by $E'$, we have $i\in E\setminus E'$, then $i$-th row and column will always be $0$ in liftings from $\mathbb{F}_2^{E'\times E'}$ to $\mathbb{F}_2^{E\times E}$. On the other hand, typical matrices from $\mathbb{F}_2^{E\times E}$ do not contain fully-$0$ rows. This raises the question of whether we can replace the structure of $E\times E$ with a sparser structure that circumvents this atypical behaviour of liftings. Moreover, requiring a response from the provers for each pair of variables in their respective variable sets is not necessary, as the quadratic equations we wish to check only require access to $\Theta(k)$ pairs (out of the $\binom{k}{2}$ many ones that exist).

This view motivates the following general recipe: can we find a graph $H$ with vertex set $E$ and edge set $\cE$, such that we can restrict the quadratic Grassmann encoding to $\mathbb{F}_2^{\cE}$ in a meaningful way? The quadratic Grassmann encoding discussed above corresponds to choosing $H$ to be a complete graph, but there are many other options for $H$. Below are a few simple choices we tried:
\begin{enumerate}[label=(\arabic*)]
    \item Expander graphs: One can try to take $H$ to be an arbitrary sparse expander graph. It can be shown that if the degree of $H$ is small then now the covering property will hold. Unfortunately, an arbitrary expander graph does not support any form of tensor tests.
    \item Cliques in a sparse high-dimensional expander: one can take a sparse high-dimensional expander $X = (X(0),\ldots,X(k))$, where $X(0)$ is the vertex set, and for each $i$, $X(i)\subseteq\binom{X(0)}{i+1}$, and take $H$ to be the graph underlying it. Because the degree of $X$ is constant, the covering property holds. Also, we can perform local type of tensor tests: if $K\in X(k)$ is a $(k+1)$-face, then $K$ is a clique in $X$, so we can choose a subspace $U\subseteq \mathbb{F}_2^{X(0)}$ fully supported in $K$, and have access to all of the entries in $U\otimes U$. 
    
    This observation initially seems very promising, but it has one serious issue. We can only afford to perform such a local tensor test for a single clique (or perhaps $O(1)$-many cliques if we are fine with increasing the value of $d$ for which we are proving~\Cref{conj:d-to-1}). Unfortunately, these localized tensor tests do not actually guarantee that the encoded function has a small rank, and Alice and Bob can use a cheating strategy in the following way. When Alice receives a space of the type $L+U\otimes U$, she extracts from it a rank-$1$ vector, and takes any element $v$ in its support. Then it is highly likely that $v$ belongs to $K$, and because $X$ has constant degree, the neighbourhood $N(v)$ of $v$ has constant size. Thus, Alice can choose a linear function on $\mathbb{F}_2^{X(1)}$ which is rank-$1$ on $\mathbb{F}_2^{X(1)\cap (N(v)\times  N(v))}$ but otherwise arbitrary, and responds according to it. Bob proceeds in the same way, where they agree beforehand on what is the rank-$1$ function they will use, and what the arbitrary function they will use on the rest of the space.

    We remark that another attempt at implementing tensor tests via high-dimensional expanders is based on coboundary-type expansion, and we have not explored this thoroughly.
\end{enumerate}
In summary, the attempts above suggest that there are meaningful variants of the quadratic Grassmann encoding that enjoy the covering property, but they come at the cost of, at best, being able to ensure small rank locally. This type of guarantee seems insufficient, at least with an outer PCP along the lines of~\Cref{sec:outer_pcp_imperfect}. 

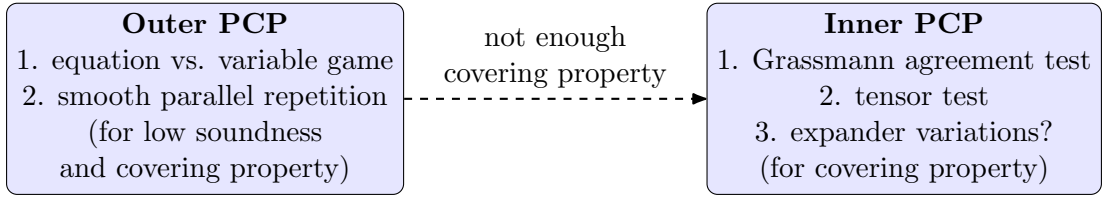
\begin{figure}[H]
\centering
\begin{tikzpicture}
\node (A) [node, align = center] at (0,0) {\textbf{Outer PCP}\\ 1. equation vs. variable game\\
2. smooth parallel repetition\\
(for low soundness\\
and covering property)};
\node (B) [node, align = center, right=4cm of A]
{\textbf{Inner PCP}\\ 1. Grassmann agreement test \\
2. tensor test\\
3. expander variations? \\
(for covering property)};

\draw [dashed arrow] (A) -- (B) node[midway, above, align = center] {not enough\\ covering property};
\end{tikzpicture}
\caption{Our failed attempt for perfect completeness}
\label{fig:failed_perfect_completeness}
\end{figure}

\subsubsection{A Second Attempt: Three-Step Construction}\label{sec:second_attempt}

So far, the described attempts to create covering property for a quadratic-Grassmann-based inner PCP are not successful. In particular, it seems that the idea of ``sparsifying'' the quadratic-Grassmann encoding conflicts with tensor test. To address this conflict, we separate the tensor test from the inner PCP and name it the ``middle PCP'' (see~\Cref{fig:perfect_completeness}). Notice that the tensor test (proposed in~\Cref{sec:first_attempt}) is a Grassmann structure over $\bF_{2}^{E}$, instead of over $\bF_{2}^{E\times E}$. Therefore, the tensor test by itself is compatible with the covering property provided by the smooth parallel repetition in the outer PCP. We only have to create more covering property for the inner PCP (which now consists only of a Grassmann agreement test) \emph{after} incorporating tensor tests into the middle PCP. This will be provided by a ``lazy parallel repetition'' on the level of the middle PCP (see~\Cref{fig:perfect_completeness}).

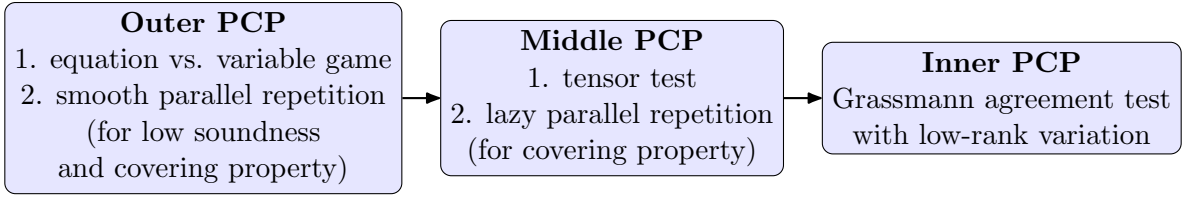
\begin{figure}[H]
\centering
\begin{tikzpicture}
\node (A) [node, align = center] at (0,0) {\textbf{Outer PCP}\\ 1. equation vs. variable game\\
2. smooth parallel repetition\\
(for low soundness\\
and covering property)};
\node (B) [node, align = center, right=0.5cm of A]
{\textbf{Middle PCP}\\ 1. tensor test \\
2. lazy parallel repetition\\
(for covering property)};
\node (C) [node, align = center, right=0.5cm of B]
{\textbf{Inner PCP}\\
Grassmann agreement test\\
with low-rank variation
};
\draw [single arrow] (A) -- (B) ;
\draw [single arrow] (B) -- (C) ;
\end{tikzpicture}
\caption{Our construction for perfect completeness}
\label{fig:perfect_completeness}
\end{figure}

\subsubsection{Our Outer PCP}\label{sec:outer_pcp_intro}

Besides using quadratic equations to replace linear ones, the main difference between our outer PCP and the one in~\Cref{sec:outer_pcp_imperfect} is that we need to establish a stronger soundness condition. This is because the tensor tests we ultimately perform will only guarantee that the decoded assignment has low rank (as opposed to rank-$1$).

Given a collection of quadratic equations $(\scrV,\scrE)$ over $\mathbb{F}_2$ and a solution 
$x\in \bF_{2}^{\scrV}$, one can consider the rank-$1$ (or 0) matrix $x\otimes x\in \bF_{2}^{\scrV\times \scrV}$ and view it as an assignment to $(\scrV,\scrE)$. This proposes a broader set of solutions given by low-rank matrices, motivating the following definition:
\begin{definition}
    Let $e$ be a quadratic equation over $\scrV$ described by $\inp{x}{Ax}=b$ for a matrix $A\in \bF_{2}^{\scrV\times \scrV}$, and let $x^{(1)},\ldots,x^{(r)}\in\mathbb{F}_2^{\scrV}$ be assignments. We say that $x^{(1)},\ldots,x^{(r)}$ jointly satisfy $e$ if
    \[
    \sum\limits_{i=1}^{r}\sum\limits_{u,v\in \scrV}A_{uv}x^{(i)}_{u}x^{(i)}_{v}=b.
    \]
\end{definition}
We remark that~\cite{KhotSaketSuper} use the same notion, which they refer to as ``superposition complexity''. With this definition in mind, the starting point of our reduction is the following result.\footnote{A more formal version of~\Cref{thm:starting_point_intro} is stated in~\Cref{thm:hg_coloring}.}
\begin{theorem}\label{thm:starting_point_intro}
    For all $r\in\mathbb{N}^{+}$, there exists $\eps>0$ such that the following holds. Given an instance $\Psi=\Psi(\scrV,\scrE)$ of $3\mathsf{QS}_{\mathbb{F}_2}$, it is $\mathbf{NP}$-hard to distinguish between the following two cases:
    \begin{enumerate}[label=(\arabic*)]
        \item \textbf{Yes case:} $\operatorname{val}(\Psi)=1$.
        \item \textbf{No case:} for all $x^{(1)},\ldots,x^{(r)}\in \bF_{2}^{\scrV}$, we have that $x^{(1)},\ldots,x^{(r)}$ jointly satisfy at most $1-\eps$ fraction of the equations in $\scrE$.
    \end{enumerate}
\end{theorem}

As in~\Cref{sec:outer_pcp_imperfect}, we can view the instances in~\Cref{thm:starting_point_intro} as an equation vs. variable game, and then apply a $\beta$-smooth $k$-fold parallel repetition to get a game $\Psi_{q,\beta,k}$. 

The soundness condition in~\Cref{thm:starting_point_intro} translates to a variant of $\Psi_{q,\beta,k}$ where we ask the provers to return $r$ vectors. In particular, upon receiving a question $(E,A)$, Alice needs to return $x^{(1)},\ldots,x^{(r)}\in \bF_{2}^{E}$ that jointly satisfy all the equations in $E$. It turns out that one can modify this game into a different game where we ask the provers to return low-rank symmetric matrices. In particular, Alice now needs to return a symmetric matrix $X\in \bF_{2}^{E\times E}$ that satisfies the equations in $E$, in the sense that for each $e$ in $E$ described by $\inp{x}{Ax}=b$, we have
\[
\sum_{u,v\in E}A_{uv}X_{uv}=b.
\]
In~\Cref{subsec:joint_satisfaction_to_bilinear}, we will show that this conversion from joint satisfaction to low-rank-matrix answers does not affect the soundness by too much. 

For the purpose of composing with the middle and inner PCPs, we view an Alice-answer $X\in \bF_{2}^{E\times E}$ as a symmetric bilinear map $f_{X}:\bF_{2}^{E}\otimes \bF_{2}^{E}\rightarrow \bF_{2}$ defined by
\[
f_{X}(x\otimes y)=\sum_{u,v\in E}X_{uv}x_{u}y_{v}\qquad\text{for all }x,y\in \bF_{2}^{E}.
\]
Similarly, Bob should return a bilinear map $g:\bF_{2}^{E'}\otimes \bF_{2}^{E'}\rightarrow \bF_{2}$.

\subsubsection{Our Middle PCP}\label{sec:middle_PCP_intro}

In the outer PCP described in~\Cref{sec:outer_pcp_intro}, the verifier (roughly speaking) carries out the following (cf.~\Cref{sec:outer_pcp_imperfect}):
\begin{enumerate}
\item Sample $(E,A)$ and send it to Alice. Sample $(E',B)$ and send it to Bob.
\item Receive from Alice a symmetric bilinear map $f:\bF_{2}^{E}\otimes \bF_{2}^{E}\rightarrow \bF_{2}$ of rank at most $r$, and from Bob a symmetric bilinear map $g:\bF_{2}^{E'}\otimes \bF_{2}^{E'}\rightarrow \bF_{2}$.
\item Accept if the restriction of $f$ to $\bF_{2}^{E'}\otimes \bF_{2}^{E'}$ is identical to $g$.
\end{enumerate}

\paragraph{Local tensor check:} the first goal of the middle PCP is to replace the global rank check ``$\operatorname{rank}(f)\leq r$'' with a local tensor check. We move to the game $\Psi_{\mathsf{mid}}$ where the verifier carries out the following:
\begin{enumerate}
\item Sample $(E,A)$ and $(E',B)$ as in the outer PCP game $\Psi_{q,\beta,k}$.
\item Sample a subspace $U\in \mathsf{Gr}(\bF_{2}^{E'},\ell_{2})$ and a subspace $U'\in \mathsf{Gr}(\bF_{2}^{E'},\ell_{2}+1)$ containing $U$. The subspace $U\subseteq \bF_{2}^{E'}$ will be regarded as a subspace of Alice's space $\bF_{2}^{E}$ by appending $0$'s in the missing coordinates.
\item Send $(E,A,U)$ to Alice and $(E',B,U')$ to Bob.
\item Receive from Alice a symmetric bilinear map $f:\bF_{2}^{E}\otimes \bF_{2}^{E}\rightarrow \bF_{2}$ such that $\operatorname{rank}(f|_{U\otimes U})\leq 1$, and from Bob a symmetric bilinear map $g:\bF_{2}^{E'}\otimes \bF_{2}^{E'}\rightarrow \bF_{2}$.
\item Accept if the restriction of $f$ to $\bF_{2}^{E'}\otimes \bF_{2}^{E'}$ is identical to $g$.
\end{enumerate}

In words, we have replaced the global rank check on Alice's response with a local rank check, that verifies that the rank of $f$ on $U\otimes U$ is at most $1$, where $U$ is a (nearly uniformly chosen) $\ell_{2}$-dimensional subspace given to Alice.

The completeness of this transformation is clear, but the soundness is much less so. Indeed, because we provide $U$ to Alice, she knows where the rank test will be performed, and therefore she can potentially use that information to cheat. Nevertheless, we show that for all $\delta>0$, if we fix $r$, $q$, $\ell_{2}$ appropriately and let $k$ be sufficiently large, then in the ``no case'' of~\Cref{thm:starting_point_intro} we have $\operatorname{val}(\Psi_{\mathsf{mid}})\leq \delta$. The intuition here is as follows:
\begin{enumerate}[label=(\arabic*)]
    \item If Alice did not use any cheating strategy, but her $f:\bF_{2}^{E}\otimes \bF_{2}^{E}\rightarrow \bF_{2}$ was of rank more than $r$, then the probability that $f|_{U\times U}$ has rank at most $1$ approaches $0$ as $\ell_{2}$ goes to infinity. 
    \item The type of cheating strategies Alice may use is limited. To see that, fix strategies $F,G$ for Alice and Bob that achieve value at least $s$, fix the question of Bob, and sample two questions for Alice that only differ on the tensor test, say $(E,A,U_1)$ and 
    $(E,A,U_2)$. Using the Cauchy-Schwarz inequality, we get that with probability at least $s^2$, Alice and Bob win both of these challenges. In that case, Alice's functions $F[E,A,U_1]$, $F[E,A,U_2]$ have to agree with Bob's function, which is the same in both scenarios (his challenge has not changed). In particular, in that scenario we must have that
    \[
    F[E,A,U_1]|_{\bF_{2}^{E'}\otimes \bF_{2}^{E'}}
    =F[E,A,U_2]|_{\bF_{2}^{E'}\otimes \bF_{2}^{E'}}.
    \]
    In words, fixing $(E,A)$, we get that Alice's responses (as a function of the tensor test $U$ she receives) must pass a sort of Grassmann agreement test. If that test had the standard notion of soundness, it would mean that Alice's responses have to basically follow an honest strategy, and a version of the argument from the first item will further guarantee that this honest strategy uses a low-rank function. 
    \item  Using the non-standard notion of soundness (which we know to hold), the discussion in~\Cref{sec:nonstnadrad_soundness} implies that after an appropriate zoom-in (and zoom-out, which we ignore for now), Alice's responses do behave like an honest strategy, and one can argue as in the first item. The implementation of the zoom-ins in the soundness analysis of $\Psi_{\mathsf{mid}}$ uses the advice $A$.
\end{enumerate} 

\paragraph{Lazy parallel repetition:} next, we wish to turn the check that 
$f|_{\mathbb{F}_2^{E'}\otimes \mathbb{F}_2^{E'}}=g$ from $\Psi_{\mathsf{mid}}$ into a $2$-to-$1$ test, and towards this end we intend to use the Grassmann encoding. A naive approach would have been to apply the Grassmann encoding directly on Alice's space $\mathbb{F}_2^{E}\otimes \mathbb{F}_2^{E} = \mathbb{F}_2^{E\times E}$ and on Bob's space $\mathbb{F}_2^{E'}\otimes \mathbb{F}_2^{E'} = \mathbb{F}_2^{E'\times E'}$, but as discussed in~\Cref{sec:first_attempt} this fails due to the lack of the covering property. Our solution for this issue is to first hide the actual challenge within a large collection of trivial challenges, in an operation that we refer to as lazy repetition (we remark that it has appeared in PCP literature previously, see~\cite{KhotSmooth}). In the terminology of~\Cref{sec:variations_gras}, this amounts to taking a graph $H$ which is a union of a large number of cliques (with no edges between them at all).

To describe the operation, fix a parameter $t$ thought of as much larger than any of the previously mentioned parameters. Consider the game $\Psi_{\mathsf{lazy}}$ where the verifier carries out the following:
\begin{enumerate}
\item Sample $c\in [t]$, and sample challenges $(E_{c},A_{c},U_{c})$, $(E'_{c},B_{c},U'_{c})$ as in $\Psi_{\mathsf{mid}}$. We refer to $c$ as the critical index.
\item For each $j\in [t]\setminus \{c\}$, sample an Alice-challenge $(E_{j},A_{j},U_{j})$ from~$\Psi_{\mathsf{mid}}$ independently. We discard the obtained $A_{j}$ and instead let $A_{j}$ be a sequence of uniformly sampled vectors
\[
a^{(1)},\dots,a^{(q')}\in \bF_{2}^{E_{j}}\otimes \bF_{2}^{E_{j}}.
\]
Let $E_{j}'=E_{j}$, $B_{j}=A_{j}$ and $U'_{j}=U_{j}$.
\item We use calligraphic fonts to denote a sequence of objects. For instance, write $\cE=(E_{1},\dots,E_{t})$ and $\cU'=(U_{1}',\dots,U_{t}')$. Similarly define sequences $\cA$, $\cB$ and $\cU$. Let $\cV(\cE)=\bigoplus_{j=1}^{t}\Bigl(\bF_{2}^{E_{j}}\otimes \bF_{2}^{E_{j}}\Bigr)$ and $\cV(\cE')=\bigoplus_{j=1}^{t}\Bigl(\bF_{2}^{E_{j}'}\otimes \bF_{2}^{E_{j}'}\Bigr)$.
\item Send $(\cE,\cA,\cU)$ to Alice and $(\cE',\cB,\cU')$ to Bob.
\item Alice must return a linear map $f:\cV(\cE)\rightarrow\bF_{2}$ such that for each $j\in [t]$, the restriction of $f$ to $\bF_{2}^{E_{j}}\otimes \bF_{2}^{E_{j}}$ is symmetric and satisfies all equations in $E_{j}$, and $\operatorname{rank}(f|_{U_{j}\otimes U_{j}})\leq 1$. Bob must return a linear map \(g:\cV(\cE')\rightarrow \bF_{2}\).
\item Accept if $f|_{\cV(\cE')}=g$.
\end{enumerate}

Lazy repetition is easily seen to preserve perfect completeness. It also does not worsen the soundness, because Alice and Bob still have to answer the questions $(E_{c},A_{c},U_{c})$ and $(E_{c}',B_{c},U_{c}')$ planted in the critical index, independent of what they do in the non-critical indices.

The main benefit of the lazy repetition is that it improves the covering property. More precisely, consider a random $x\in \cV(\cE)$ obtained by the following process:
\begin{enumerate}
    \item Sample $c$ and $\cE'$ as in $\Psi_{\mathsf{lazy}}$.
    \item Sample $y\in \cV(\cE')$, and lift to $x\in \cV(\cE)$ by appending $0$'s in missing coordinates (which only happens in the critical index).
\end{enumerate}
Then, provided that $t$ is sufficiently large compared to $k$, the distribution of $x$ is close to uniform in $\cV(\cE)$. The intuition here comes from the following basic fact:
\begin{claim}\label{claim:prob_fact_intro}
    Suppose that $(\Omega,\mathbb{P})$ is a discrete probability space, and $\Gamma\subseteq\Omega$ is an event such that $p = \mathbb{P}(\Gamma) >0$. Then for each $t$, the following two distributions are $O_{p}(1/\sqrt{t})$-close in statistical distance:
    \begin{enumerate}[label=(\arabic*)]
        \item Sample $x_1,\ldots,x_t\sim \mathbb{P}$ independently.
        \item Sample $c\in [t]$, sample $x_c\sim \mathbb{P}$ conditioned on $x_c\in \Gamma$, and for $j\neq c$ sample $x_j\sim \mathbb{P}$ independently.
    \end{enumerate}
\end{claim}

In summary, the game $\Psi_{\mathsf{lazy}}$ has perfect completeness, small soundness, local tensor tests and the covering property. It can actually be composed with an appropriate version of the Grassmann graph but this is still not straightforward. In the next section we discuss why a na{\"i}ve composition with the standard Grassmann encoding still fails, and explain how we circumvent this issue by taking an appropriate variant of the Grassmann encoding.

\subsubsection{Our Inner PCP}\label{sec:inner_PCP_intro}

A natural attempt at composing $\Psi_{\mathsf{lazy}}$ with the Grassmann encoding is to encode the functions $f\colon \cV(\cE)\to\mathbb{F}_2$ and $g\colon \cV(\cE')\to\mathbb{F}_2$ via the Grassmann encoding, and then check consistency via the Grassmann test. While this test does satisfy the same notion of non-standard soundness as discussed in~\Cref{sec:nonstnadrad_soundness}, we do not know how to implement suitable advice vectors to facilitate the zoom-in coordination. Indeed, the Grassmann encoding via $\mathsf{Gr}(\cV(\cE),\ell_{1})$ would require uniformly random advice vectors in $\cV(\cE)$, and specifying a vector 
\[x\in\cV(\cE)=\bigoplus_{j=1}^{t}\Bigl(\bF_{2}^{E_{j}}\otimes \bF_{2}^{E_{j}}\Bigr)\]
would require $|E_{j}|^{2}$ bits of information for each index $j$. However, although the advice $\cA=(A_{1},\dots,A_{t})$ available in~$\Psi_{\mathsf{lazy}}$ does provide quadratic-in-$k$ amount of information for each non-critical index $j$ via
\[
A_{j}=\bigl(a^{(j,1)},\dots,a^{(j,q')}\bigr)\in \Bigl(\bF_{2}^{E_{j}}\otimes \bF_{2}^{E_{j}}\Bigr)^{q'},
\]
the advice available for the critical index $c\in [t]$ is
\[
A_{c}=\bigl(a^{(1)},\dots,a^{(q)}\bigr)\in \Bigl(\bF_{2}^{E_{c}}\Bigr)^{q},
\]
which contains only $O_{q}(k)$ bits of information. Therefore, it is of interest to have a Grassmann-type structure on $\bF_{2}^{E_{c}}\otimes \bF_{2}^{E_{c}}\Big.$ with only $\exp(O_{q}(k))$ different zoom-in possibilities.

\paragraph{The low-rank Grassmann encoding:} the idea is to consider not all, but only a special type of subspaces in $\mathsf{Gr}(\bF_{2}^{E_{c}}\otimes \bF_{2}^{E_{c}},\ell_{1})\Big.$. We only allow subspaces that consist of low-rank vectors in $\bF_{2}^{E_{c}}\otimes \bF_{2}^{E_{c}}$. By the rank-1 decomposition, any vector in $\bF_{2}^{E_{c}}\otimes \bF_{2}^{E_{c}}$ of rank at most $s$ can be specified by $O(s|E_{c}|)$ bits of information.

In order to systematically construct linear subspaces consisting of low-rank vectors, we use the Hadamard code (the same idea has appeared in~\cite{Golowich23}). Given a positive integer $\ell$, let $\Had_{\ell}\in \bF_{2}^{2^{\ell}\times (2^{\ell}-1)}$ be the $\{0,1\}$-Hadamard matrix with the first column (the all-$0$ column) deleted. For any $\bF_{2}$-vector space $V$ and vectors $x^{(1)},\dots,x^{(2^{\ell}-1)}\in V$, let
\[
\Lspan\bigl(x^{(1)},\dots,x^{(2^{\ell}-1)}\bigr)=\left\{\sum\limits_{j=1}^{2^{\ell}-1}\Had_{\ell}(i,j) x^{(j)}\,\middle|\,i\in \{1,2,\ldots,2^{\ell}\}\right\}.
\]
In words, for each row in $\Had_{\ell}$, the set $\Lspan\bigl(x^{(1)},\dots,x^{(2^{\ell}-1)}\bigr)$ contains a vector which is the linear combination of 
$x^{(1)},\dots,x^{(2^{\ell}-1)}$ with the coefficients given by that row. Because any linear combination of rows in $\Had_{\ell}$ is a row in $\Had_{\ell}$, it follows that $\Lspan\bigl(x^{(1)},\dots,x^{(2^{\ell}-1)}\bigr)$ is always a linear subspace, and if $x^{(1)},\dots,x^{(2^{\ell}-1)}$ are linearly independent then that subspace has dimension exactly $\ell$. 

The upshot is that if $x^{(1)},\dots,x^{(2^{\ell}-1)}\in \bF_{2}^{E_{c}}\otimes \bF_{2}^{E_{c}}$ are rank-$1$ vectors, then all vectors in the space $\Lspan\bigl(x^{(1)},\dots,x^{(2^{\ell}-1)}\bigr)$ have ranks at most $2^{\ell-1}$.
For the sake of reference, we present the following definition of~\cite{Golowich23}:\footnote{\Cref{def:intro_LGr} will not be used explicitly in this paper.}

\begin{definition}[{\cite[Definition 17]{Golowich23}}]\label{def:intro_LGr}
Suppose $V$ is an $\bF_{2}$-vector space of dimension at least $2^{\ell}$. We define $\mathsf{LGr}(V,\ell)$, called \emph{the low-rank Grassmann graph}, to be the subgraph of the Grassmann graph $\mathsf{Gr}(V\otimes V,\ell)$ induced by the vertex set
\[
\left\{\Lspan\bigl(u^{(1)}\otimes v^{(1)},\ldots,u^{(2^{\ell}-1)}\otimes v^{(2^{\ell}-1)}\bigr)~\middle|~
    \substack{\text{\normalsize  $u^{(1)},\ldots,u^{(2^{\ell}-1)}\in V$
    \text{ linearly independent}}\\\text{\normalsize $v^{(1)},\ldots,v^{(2^{\ell}-1)}\in V$
    \text{ linearly independent}}}\right\}.
\]
\end{definition}

\paragraph{Composing with middle PCP:} we can now make $\Psi_{\mathsf{lazy}}$ into a $4$-to-$1$ game by using the low-rank Grassmann encoding for the critical block $\bF_{2}^{E_{c}}\otimes \bF_{2}^{E_{c}}\Big.$, and using the ordinary Grassmann encoding for the non-critical blocks $\bF_{2}^{E_{j}}\otimes \bF_{2}^{E_{j}}$ where $j\in [t]\setminus \{c\}$. We define $\Psi_{\mathsf{final}}$ to be the game where the verifier carries out the following:
\begin{enumerate}
\item Sample the critical index $c\in [t]$, and sample $\cE,\cE',\cU,\cU'$ as in $\Psi_{\mathsf{lazy}}$.
\item For each $i\in \{1,2,\dots,2^{\ell_{1}+1}-1\}$, sample
\[
x^{(i)}=x^{(i,1)}+\dots +x^{(i,t)}\in \cV(\cE'),
\]
where $x^{(i,c)}$ is a uniformly random rank-1 vector in $\bF_{2}^{E_{c}'}\otimes \bF_{2}^{E_{c}'}$, while $x^{(i,j)}$ is a uniformly random vector in $\bF_{2}^{E_{j}'}\otimes \bF_{2}^{E_{j}'}$ for each $j\in [t]\setminus \{c\}$.

\item Define the linear subspace $L'=\Lspan\bigl(x^{(1)},\dots,x^{(2^{\ell_{1}+1}-1)}\bigr)\subseteq \cV(\cE')$. If $\dim L'<\ell_{1}+1$, give up on asking questions and accept immediately. Otherwise, let $L$ be a uniformly random subspace of $L'$ with dimension $\ell_{1}$. We regard $L\subseteq \cV(\cE')$ as a subspace of $\cV(\cE)$ by appending $0$'s in the missing coordinates (which only happens in the critical index).

\item Send the challenge $(\cE,\cU,L)$ to Alice and $(\cE',\cU',L')$ to Bob.

\item Alice must return a linear map $f:\left(L+\bigoplus_{j=1}^{t}U_{j}\otimes U_{j}\right)\rightarrow \bF_{2}$ such that $\operatorname{rank}(f|_{U_{j}\otimes U_{j}})\leq 1$ for each $j\in [t]$. Bob must return a linear map $g:\left(L'+\bigoplus_{j=1}^{t}U_{j}'\otimes U_{j}'\right)\rightarrow \bF_{2}$ such that $\operatorname{rank}(g|_{U_{j}'\otimes U_{j}'})\leq 1$ for each $j\in [t]$.

\item Accept if the restriction of $g$ to $\left(L+\bigoplus_{j=1}^{t}U_{j}\otimes U_{j}\right)$ is identical to $f$.
\end{enumerate}

The game $\Psi_{\mathsf{final}}$ is our final composed PCP (modulo incorporating the side conditions checking satisfiability of the chosen quadratic equations, and the side conditions checking symmetry; this is done by adding the spaces $N_{E_{j}}$ and $N_{\mathsf{sym}}$ as discussed in~\Cref{sec:first_attempt}).

It is easy to see that in $\Psi_{\mathsf{final}}$, for any response of Alice, there are at most 4 possible answers of Bob that can make the verifier accept. This is due to the following facts:
\begin{itemize}
\item Since $\dim L'=1+\dim L$, any linear function on $L$ has two linear extensions to $L'$.
\item For each $j\in [t]\setminus \{c\}$, we have $U_{j}=U_{j}'$. For the critical index $c$, we have $\dim U_{c}'=1+\dim U_{c}$, and hence any symmetric bilinear function on $U_{c}\otimes U_{c}$ of rank at most 1 has two symmetric bilinear extensions to $U_{c}'\otimes U_{c}'$ of rank at most 1. 
\end{itemize}

The completeness of $\Psi_{\mathsf{final}}$ is clear by using the Grassmann encoding on the strategies of the provers from $\Psi_{\mathsf{lazy}}$. For soundness, we show that a strategy for $\Psi_{\mathsf{final}}$ achieving value at least $\varepsilon$ can be used to construct a strategy for $\Psi_{\mathsf{lazy}}$ achieving value at least $\Omega_{\eps}(1)$. This argument is more elaborate, and we only give a high-level overview.

First, we describe a high-level overview where we pretend there was no need for zoom-outs in~\Cref{eq:overview_Grassmann_non_standard}. We show that for noticeable fraction of the partial challenges $(\cE,\cU)$ for Alice, the table $L\mapsto F[\cE,\cU,L]$ passes the standard Grassmann test (\Cref{eq:overview_Grassmann_test}) over $\cV(\cE)$ with noticeable probability. Applying~\Cref{eq:overview_Grassmann_non_standard}, we know that Alice can use the advice $\cA$ in $\Psi_{\mathsf{lazy}}$ as the zoom-in and thus decodes a global linear function $g$ on $\cV(\cE)$ which has rank at most $1$ on $U_j\otimes U_j$ for each $j$. We argue that for noticeable fraction of the challenges $(\cE',\cU')$ of Bob (conditioned on $\cE,\cU,\cA$), the restriction $g|_{\cV(\cE')}$ is a global linear function for Bob that has noticeable agreement with his table $L'\mapsto G[\cE',\cU',L']$ under the same zoom in (constructed from his advice). Together, this shows that for a fixed successful advice, each one of Alice and Bob only has constantly many global linear functions that have noticeable agreement with their table, and there is at least one pair of choices between them which is consistent. Thus, if each one of them chooses a random such global linear function, they will produce a satisfactory answer to $\Psi_{\mathsf{lazy}}$ with constant probability.

The need for zoom-outs in~\Cref{eq:overview_Grassmann_non_standard} complicates matters considerably, as the number of ``successful'' zoom-outs could be rather large. This issue was already handled in~\cite{DKKMS2}, and we will make use of some ideas from there. However, the fact that we use the low-rank Grassmann encoding introduces further difficulties. In~\cite{DKKMS2} it is shown that if Bob's table has many zoom-outs of codimension $w$ that admit a global linear function with noticeable agreement, then there is also a zoom-out of codimension $w-1$ that admits a global linear function with noticeable agreement. While this may be possible in our case, it requires extra features from the low-rank Grassmann graph which we do not establish. Instead, we show how to stitch together global linear functions from many different Bob-questions $(\cE',\cU')$ to a function of Alice, on which we can apply the result of~\cite{DKKMS2} (since Alice's space has the standard Grassmann graph structure).

\section{Preliminaries}
In this section we set up basic notations and tools that will be used throughout the paper.

\subsection{General Notations}\label{subsec:general_notation}

In this subsection we summarize the general notational conventions used throughout the paper. Additional notation will be introduced as needed, typically within dedicated ``Notation'' environments.

\paragraph{Linear algebra.} Given a finite dimensional $\bF_{2}$-vector space $V$ and a linear subspace $W\subseteq V$, the codimension of $W$ with respect to $V$ is denoted by $\operatorname{codim}_{V}(W)=\dim V-\dim W$. Given a finite dimensional $\bF_{2}$-vector space $V$, we let $V^*$ denote the dual space of $V$, which consists of all linear functionals $f:V\rightarrow \bF_{2}$. For any finite-dimensional $\bF_{2}$-vector spaces $V_{1},V_{2}$ and any bilinear form $f:V_{1}\otimes V_{2}\rightarrow \bF_{2}$, we define the linear map $T_{f}:V_{1}\rightarrow V_{2}^*$ by letting $T_{f}(x)$ be the functional $y\mapsto f(x\otimes y)$. 

As indicated in~\Cref{sec:outer_pcp_intro,sec:middle_PCP_intro} of the proof overview, bilinear forms will play an important role in our construction and analysis of 4-to-1 games. More linear algebraic preliminaries about bilinear forms can be found in~\Cref{sec:linear_algebra}. 

\paragraph{Grassmann graph.}
Given a finite dimensional $\bF_{2}$-vector space $V$ and a positive integer $\ell\leq \dim V$, the Grassmann graph $\mathsf{Gr}(V,\ell)$ has as vertices all $\ell$-dimensional subspaces $L\subseteq V$, and two vertices are adjacent if they intersect in dimension $\ell-1$. By a slight abuse of notation,
we also use $\mathsf{Gr}(V,\ell)$ to denote the vertex set of
this graph. We denote by $\cG(V,\ell)$ the uniform distribution over the edges of $\mathsf{Gr}(V,\ell)$. For any nonnegative integer $b\leq \ell$ (such that $b\geq 2\ell-\dim V$), we use $\cG(V,\ell,b)$ to denote the uniform distribution over the set
\[
\left\{(L_{1},L_{2})\,\big|\,L_{1},L_{2}\in\mathsf{Gr}(V,\ell)\textup{ such that }\dim(L_{1}\cap L_{2})=b\right\},
\]
i.e. pairs of $\ell$-dimensional subspaces of $V$ with an intersection of dimension $b$.

\paragraph{Probability.}  
For a finite set $\Lambda$, we write $\Exs{x \in \Lambda}{\cdot}$ and $\Prs{x \in \Lambda}{\cdot}$ to denote expectation and probability, respectively, when $x$ is drawn uniformly at random from $\Lambda$.  
If $x$ is sampled according to a specific distribution $\cD$ over $\Lambda$, we write $x \sim \cD$ in place of $x \in \Lambda$ in the subscripts.

\paragraph{Assignment tables.} In a $2$-prover  $1$-round game, we often organize a prover's strategy as an \emph{assignment table} $F[\cdot]$, where for each possible question $Q$ from the verifier, the entry $F[Q]$ takes value from an alphabet of answers. It is often convenient to allow partial assignment tables, where each entry $F[Q]$ is allowed to take a special nullity value outside of the prescribed alphabet. We usually denote such a nullity value by the symbol $\nil$. Whenever a $\nil$ entry is tested for agreement with other entries or objects, we always assume the test rejects (i.e. reports a disagreement). For instance, if $L_{1}$ and $L_{2}$ are linear subspaces of a vector space $V$, the entry $F[L_{1}]$ takes value in $(L_{1})^*\cup\{\nil\}$ and $F[L_{2}]$ takes value in $(L_{2})^*\cup \{\nil\}$, then the statement
\[
\Big[\left.F[L_{1}]\right|_{L_{1}\cap L_{2}}=\left.F[L_{2}]\right|_{L_{1}\cap L_{2}}\Big]
\]
evaluates to ``false'' whenever either of $F[L_{1}]$ or $F[L_{2}]$ is $\nil$.

\paragraph{Empty sequences.} A sequence of objects $a^{(1)},\dots,a^{(q)}$ is denoted by $\bigl(a^{(1)},\dots,a^{(q)}\bigr)$. If $q=0$, the sequence is empty and is denoted by the symbol $\emp$.
\subsection{Sub-4-to-1 Games}

In order to prove the main theorem (\Cref{thm:main}), it is convenient to work with a slight relaxation of the $4$-to-$1$ requirement: 

\begin{definition}\label{def:sub4to1_constraint}
In the setting and notation of~\Cref{def:lc}, a constraint $\Phi_{e}\subseteq \Sigma_{L}\times \Sigma_{R}$ is called a sub-$d$-to-$d'$ constraint if there is a partition $\Sigma_{L}=\Sigma_{L,1}\sqcup \dots\sqcup\Sigma_{L,t}$ into sets of size $d$ and a partition $\Sigma_{R}=\Sigma_{R,1}\sqcup \dots\sqcup \Sigma_{R,t}$ into sets of size $d'$ such that
    \[
\Phi_{e}=\bigcup_{i=1}^{t}\Gamma_{L,i}\times \Gamma_{R,i},\qquad\textup{for some subsets }\Gamma_{L,i}\subseteq \Sigma_{L,i}\textup{ and }\Gamma_{R,i}\subseteq \Sigma_{R,i}.
    \]
\end{definition}

The notion of sub-$d$-to-$d'$ constraints motivates the following definition parallel to~\Cref{def:Gap4to1}.

\begin{definition}
Let $d,d'$ and $r$ be positive integers. An instance $\Psi$ of label-cover is called a sub-$d$-to-$d'$ game if all of its constraints are sub-$d$-to-$d'$. For parameters $0<s<c\leq 1$, we write $\GapSdone{d}{d'}_{r}(c,s)$ for the restriction of $\mathsf{GapLC}_{r}(c,s)$ to sub-$d$-to-$d'$ game instances.
\end{definition}

Accordingly, we have the following slight weakening of~\Cref{thm:main}.

\begin{restatable}{theorem}{submainthm}\label{thm:sub-main}
For any $\delta\in (0,1)$, there exists a positive integer $N$ such that $\GapSdone{4}{1}_{N}(1,\delta)$ is $\mathbf{NP}$-hard.
\end{restatable}

Although~\Cref{thm:sub-main} is a relaxation of~\Cref{thm:main}, we will show in~\Cref{sec:exact_4_to_1} that~\Cref{thm:sub-main} actually implies~\Cref{thm:main}. The main body of this paper (\Cref{sec:outer_PCP,sec:middle_PCP,sec:inner_PCP}) will be devoted to proving~\Cref{thm:sub-main}.

\subsection{Technical Ingredients from Prior Work} 

In this subsection, we record the main technical ingredients of this paper that are imported from prior works.

\subsubsection{Approximate Coloring of Hypergraphs}

As discussed in~\Cref{sec:outer_pcp_intro}, our construction starts from
the hardness of approximating satisfiable $3\mathsf{QS}_{\mathbb{F}_2}$
instances. However, limitations of the tensor tests used in our middle
PCP require us to start from an $\mathbf{NP}$-hard gap problem in which
the ``yes'' and ``no'' cases impose conditions on \emph{different} constraint
satisfaction instances. Specifically, the soundness condition
in~\Cref{thm:starting_point_intro} must assert that, even after the
alphabet of the $3\mathsf{QS}_{\mathbb{F}_2}$ instance is suitably
enlarged (from $\bF_{2}$ to $(\bF_{2})^{r}$, roughly speaking), the resulting
$3\mathsf{QS}_{(\mathbb{F}_2)^{r}}$ instance still has small
value.\footnote{This formulation is characteristic of promise
constraint satisfaction problems (PCSPs); see, for
example,~\cite{krokhin2022invitation}.}

Fortunately, the desired starting point
in~\Cref{thm:starting_point_intro} is essentially implicit in the result
of~\cite{DinurRegevSmyth} on the hardness of coloring $2$-colorable
$3$-uniform hypergraphs.\footnote{In~\Cref{subsec:NAE_predicate}, we will explain why $2$-coloring a $3$-uniform hypergraph can be viewed as solving a $3\mathsf{QS}_{\mathbb{F}_2}$ instance.} Although their result is stated with the
no-case condition that the hypergraph $H$ is not $c$-colorable,
a closer inspection of their argument yields the statement below.

\begin{restatable}[\cite{DinurRegevSmyth}]{theorem}{hgcolor}\label{thm:hg_coloring}
    For every fixed positive integer $c\geq 2$, there exists a constant \linebreak $\varepsilon\geq \exp(-\exp(O(c^{3})))\Big.$ such that given a $3$-uniform hypergraph $H$, it is $\mathbf{NP}$-hard to distinguish between the following cases:
    \begin{enumerate}[label=(\arabic*)]
        \item Yes case: $H$ is $2$-colorable.
        \item No case: in every $c$-coloring of $H$, at least $\varepsilon$ fraction of the edges are monochromatic.
    \end{enumerate}
\end{restatable}

Since~\Cref{thm:hg_coloring} differs slightly from the result proved
in~\cite{DinurRegevSmyth}, we provide a proof in~\Cref{sec:apx_hyp}.
Our proof is a minor adaptation of their argument, which in turn uses the PCP
theorem~\cite{FGLSS,AS,ALMSS} and parallel repetition~\cite{Raz}.

\subsubsection{Grassmann Agreement Test}\label{subsubsec:Grassmann_agreement_test}

The main analytical/combinatorial tool used in this paper is the (non-standard) soundness of the Grassmann agreement test, informally stated in~\Cref{eq:overview_Grassmann_non_standard}. 
Under the notation specified in~\Cref{subsec:general_notation}, given a table $F[\cdot]$ that assigns to each $\ell$-dimensional subspace $L\subseteq\mathbb{F}_2^n$ either a linear function $F[L]\colon L\to\mathbb{F}_2$ or the nullity element $\nil$, the passing probability of the Grassmann test is
\[
\Pru{(L_{1},L_{2})\sim \cG(\bF_{2}^{n},\ell)}{\left.F[L_{1}]\right|_{L_{1}\cap L_{2}}=\left.F[L_{2}]\right|_{L_{1}\cap L_{2}}\Big.},
\]
where the event does not hold if either $F[L_1]$ or $F[L_2]$ are $\nil$. The following result is due to~\cite{KMS23,ellis2023analogue}, but we use a slightly different formulation which we derive directly from their results.

\begin{restatable}[\cite{KMS23,ellis2023analogue}]{theorem}{grdecodingA}\label{thm:grassmann_linear_dec}
There is an absolute constant $C>0$ such that the following holds. Let $n,\ell,q$ be nonnegative integers such that $n-q-3\geq \ell\geq q+3$. Suppose $F[\cdot]$ is a table where $F[L]\in L^*\cup\{\nil\}$ for each $L\in \mathsf{Gr}(\bF_{2}^{n},\ell)$, such that 
\begin{equation}\label{eq:Grassmann_decoding_2_assumption}
\Pru{(L_{1},L_{2})\sim \cG(\bF_{2}^{n},\ell)}{\left.F[L_{1}]\right|_{L_{1}\cap L_{2}}=\left.F[L_{2}]\right|_{L_{1}\cap L_{2}}\Big.}\geq 2^{-q+5}.
\end{equation}
Then for at least $2^{-q-\ell^{2}}$ fraction of subspaces $Q\in \mathsf{Gr}(\bF_{2}^{n},q)$, there exist a subspace $W\in \mathsf{Gr}(\bF_{2}^{n},n-q)$ containing $Q$ and a linear map $f:W\rightarrow \bF_{2}$ such that
\[
\Pru{L\in \mathsf{Gr}(\bF_{2}^{n},\ell)}{F[L]=\left.f\right|_{L}\,\Big|\, Q\subseteq L\subseteq W}\geq 2^{-Cq^{2}}.
\]
\end{restatable}

A derivation of~\Cref{thm:grassmann_linear_dec} from the formulation of~\cite{ellis2023analogue} will be presented in~\Cref{sec:grassmann_dec}.

We also need the following simpler result, which concerns not a table of local linear functions, but a table where every entry takes value from the same abstract alphabet. The following result asserts that if the ``Grassmann equality test'' passes with noticeable probability, then after zoom-ins and zoom-outs the entries take a certain constant value noticeably often.

\begin{restatable}{theorem}{grdecodingB}\label{thm:grdecoding2}
There is an absolute constant $C>0$ such that the following holds. Let $n,\ell,q$ be nonnegative integers such that $n-q-3\geq \ell\geq q+3$. Suppose $\Sigma$ is a finite set and $F[\cdot]$ is a table where $F[L]\in \Sigma\cup\{\nil\}$ for each $L\in \mathsf{Gr}(\bF_{2}^{n},\ell)$, such that 
\begin{equation}\label{eq:Grassmann_equality_test}
\Pru{(L_{1},L_{2})\sim \mathcal{G}(\bF_{2}^{n},\ell)}{F[L_{1}]=F[L_{2}]\neq \nil\Big.}\geq 2^{-q+4}.
\end{equation}
Then for at least $2^{-q-\ell^{2}}$ fraction of subspaces $Q\in \mathsf{Gr}(\bF_{2}^{n},q)$, there exist an element $\sigma\in \Sigma$ and a subspace $W\in \mathsf{Gr}(\bF_{2}^{n},n-q)$ containing $Q$ such that
\[
\Pru{L\in \mathsf{Gr}(\bF_{2}^{n},\ell)}{F[L]=\sigma\,\Big|\, Q\subseteq L\subseteq W}\geq 2^{-Cq^{2}}.
\]
\end{restatable}
\begin{proof}
    Deferred to~\Cref{sec:grassmann_dec}.
\end{proof}

\subsubsection{List Decoding Bounds}
When we use~\Cref{thm:grassmann_linear_dec} we will often need to argue that given a zoom-in $Q$, the number of possible decodings $(f,W)$ is not too large. In this section we precisely define this notion and state the bounds we will use, established by~\cite{DKKMS18}.

\begin{definition}[\cite{DKKMS18}]\label{def:tau_occurring}
Let $n,\ell,q,w$ be nonnegative integers such that $n-w\geq \ell\geq q+1$. Fix parameters $\tau_{0},\tau_{1},\dots,\tau_{w}\in (0,1)$ and write $\vec{\tau}=(\tau_{0},\dots,\tau_{w})$. Fix an $\bF_{2}$-vector space $V$ of dimension $n$ and a subspace $Q\in\mathsf{Gr}(V,q)$. Suppose $F[\cdot]$ is a table where $F[L]\in L^*\cup\{\nil\}$ for each $L\in \mathsf{Gr}(V,\ell)$. For a linear subspace $W\subseteq V$ with $\operatorname{codim} W\leq w$ and a linear functional $f:W\rightarrow \bF_{2}$, we say that the pair $(f,W)$ is \emph{$\vec{\tau}$-occurring in $F[\cdot]$ modulo $Q$} if $Q\subseteq W$ and
\[
\Pru{L\in \mathsf{Gr}(W,\ell)}{F[L]=f|_{L}\,\Big|\,Q\subseteq L}\geq \tau_{i},\qquad\textup{where }i=n-\operatorname{dim} W.
\]
\end{definition}

It is not possible to give a bound on the number of $\vec{\tau}$-occurring pairs that is independent of $n$, simply because the number of zoom-outs $W$ can be very large. We have to restrict our attention to pairs $(f,W)$ that ``cannot be extended,'' as defined below.

\begin{definition}\label{def:LIST_Q}
In the setting and notation of~\Cref{def:tau_occurring}, we say that $(f,W)$ is \emph{$\vec{\tau}$-maximal in $F[\cdot]$ modulo $Q$} if 
\begin{enumerate}[label=(\arabic*)]
\item $(f,W)$ is $\vec{\tau}$-occurring in $F[\cdot]$ modulo $Q$, but
\item for any linear subspace $W'\subseteq V$ that properly contains $W$ and any extended linear functional $f':W'\rightarrow \bF_{2}, f'|_{W} = f$, the pair $(f',W')$ is not $\vec{\tau}$-occurring in $F[\cdot]$ modulo $Q$.
\end{enumerate}
The set of all $\vec{\tau}$-maximal pairs will be denoted by $\mathsf{LIST}^{\vec{\tau}}_{Q}(F)$.
\end{definition}
We use the following bound on the size of maximal pairs proved by \cite{DKKMS18}. 
\begin{lemma}[\cite{DKKMS18}]\label{lem:dkkms_list_bound}
There is an absolute constant $C\geq 40$ such that the following holds. Let $n,\ell,q,w$ be nonnegative integers such that $n-w\geq 2\ell\geq C(q+1)$. Suppose $\vec{\tau}=(\tau_{0},\dots,\tau_{w})$ is a vector such that 
\[
\tau_{0}\geq 2^{-\ell/C}\qquad\text{and}\qquad\tau_{i-1}\leq C^{-1}\tau_{i}^{C}\textup{ for each }i\in [w].\]
Then for any subspace $Q\in \mathsf{Gr}(\bF_{2}^{n},q)$ and any table $F[\cdot]$ where $F[L]\in L^{*}\cup\{\nil\}$ for each $L\in \mathsf{Gr}(\bF_{2}^{n},\ell)$ we have \[\left|\mathsf{LIST}_{Q}^{\vec{\tau}}(F)\right|\leq \exp(C(w+1)^{2}\ell).\]
\end{lemma}

The original statement of the above lemma in~\cite{DKKMS18} does not allow for $\nil$ entries in the table. In~\Cref{sec:list_decoding_bound}, we will explain how to adapt their arguments to tables with $\nil$ entries.

In a part of the analysis of our inner PCP, we are not able to use~\Cref{lem:dkkms_list_bound} as a black box, and we will instead draw from the proof of~\Cref{lem:dkkms_list_bound}. A useful component of the proof by~\cite{DKKMS18} is the notion of \emph{sunflower} of zoom-outs, defined as follows.

\begin{definition}\label{def:subspacesunflower}
    Let $V$ be a vector space over $\mathbb{F}_2$, and let $W_1,\ldots,W_m\subseteq V$ be  subspaces of co-dimension $w$. We say that $(V,W_1,\ldots,W_m)$ forms a sunflower if $\operatorname{codim}_{V}(W_{i}\cap W_{j})=2w$ for each pair of distinct indices $i,j\in [m]$.
\end{definition}
In words, a sunflower is a collection of subspaces $W_1,\ldots,W_m$ living inside a space $V$ that have as small as possible pairwise intersections. The following result from~\cite{DKKMS18} asserts that any sufficiently large collection of zoom-outs of constant codimension contains a large sunflower structure.
\begin{lemma}[{\cite[Corollary 3.14]{DKKMS18}}]\label{lem:subspacesunflower}
    Let $V$ be a finite dimensional $\bF_{2}$-vector space, and let $W_1,\dots ,W_N\subseteq V$ be a list of subspaces of codimension $w$. Let $m$ be a positive integer such that $(m\cdot 2^w)^{w+1} \le N$. Then, there exists a subspace $W\subseteq V$ containing $m$ of the $W_i$'s, say, $W_1,\dots, W_m$, such that $(W,W_1,\dots, W_m)$ forms a sunflower. 
\end{lemma}

\subsection{Proof Organization}

The proof of~\Cref{thm:sub-main} spans~\Cref{sec:outer_PCP,sec:middle_PCP,sec:inner_PCP}, during which we construct a sequence of five 2-prover 1-round games; see~\Cref{fig:proof_organization}.

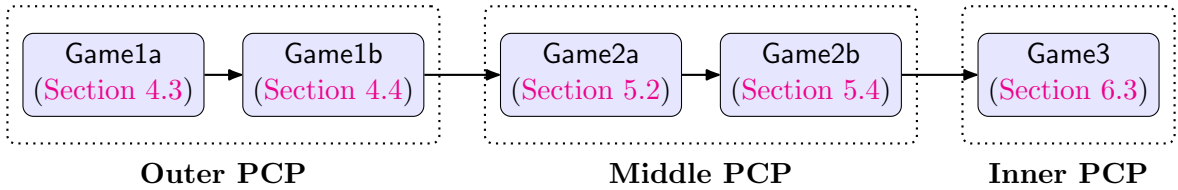
\begin{figure}[H]
\centering
\begin{tikzpicture}[dottedbox/.style={draw, dotted, thick, rounded corners=3pt}]
\node (A) [node, align = center] at (0,0) {$\GameAa$\\
(\Cref{subsec:equation_vs_variable})};
\node (B) [node, align = center, right=0.5cm of A]
{$\GameAb$\\
(\Cref{subsec:joint_satisfaction_to_bilinear})};
\node (C) [node, align = center, right=1cm of B]
{$\GameBa$\\
(\Cref{subsec:Game2a})};
\node (D) [node, align = center, right=0.5cm of C]{
$\GameBb$\\
(\Cref{subsec:Game2b})
};
\node (E) [node, align = center, right=1cm of D]{
$\GameCa$\\
(\Cref{subsec:Game3})
};
\draw [single arrow] (A) -- (B) ;
\draw [single arrow] (B) -- (C) ;
\draw [single arrow] (C) -- (D) ;
\draw [single arrow] (D) -- (E) ;

\node[
    dottedbox,
    fit = {($(A.north west)+(-2pt,0)$)($(B.south east)+(2pt,0)$)},
    inner ysep=10pt,
] (outer) {};
\node[
    dottedbox,
    fit = {($(C.north west)+(-2pt,0)$)($(D.south east)+(2pt,0)$)},
    inner ysep=10pt,
] (middle) {};
\node[
    dottedbox,
    fit = {($(E.north west)+(-2pt,0)$)($(E.south east)+(2pt,0)$)},
    inner ysep=10pt,
] (inner) {};

% Label below the dotted box
\node[below=3pt of outer, align=center] {\textbf{Outer PCP}};
\node[below=3pt of middle, align=center] {\textbf{Middle PCP}};
\node[below=3pt of inner, align=center] {\textbf{Inner PCP}};

\end{tikzpicture}
\caption{Proof organization}
\label{fig:proof_organization}
\end{figure}

For each version label $\mathsf{xx}\in\{\mathsf{1a},\mathsf{1b},\mathsf{2a},\mathsf{2b},\mathsf{3}\}$, we will prove completeness and soundness guarantees of the following form: 
\begin{enumerate}[label=(\arabic*)]
\item Completeness: if $H$ is $2$-colorable, then $\operatorname{val}(\mathsf{Gamexx}[H])=1$.
\item Soundness: if every $c$-coloring of $H$ makes at least $\varepsilon$ fraction of the edges monochromatic, then $\operatorname{val}(\mathsf{Gamexx}[H])$ is small.
\end{enumerate}

Completeness is straightforward and can be proved independently for each
game, without invoking the completeness of preceding games.
Most of the work lies in establishing soundness: for each game after
the first, the soundness proof relies on that of its predecessor
in the sequence. The final construction $\GameCa[H]$ is a sub-$4$-to-$1$ game,
whose completeness and soundness yield~\Cref{thm:sub-main}.

\subsection{Parameter Relations}

During the course of our construction of $4$-to-$1$ games, there will be many parameters on different scales (they are all constant with respect to the size of the starting hardness instance in~\Cref{thm:hg_coloring}). We summarize the relations between these constants in the following condition.

\begin{condition}\label{cond:parameters}
Let $k,q_{0},q_{1},q_{2},t,r,\ell_1,\ell_2,s$ be positive integers, let $\delta,\beta,\varepsilon\in (0,1)$, and let $\vec{\tau}=(\tau_{0},\tau_{1},\dots,\tau_{2q_{0}})$ be a vector such that for a sufficiently large absolute constant $C>0$, the conditions in~\Cref{tab:parameters} hold.
\end{condition}

\begin{table}[H]
\centering
\begin{tabular}{|Sc|Sc|Sc|Sc|}
\hline
\rowcolor{gray!30} \textbf{Parameter} & \textbf{Requirement} & \textbf{First appears in} & \textbf{Why require this}\\ \hline

$\delta\in (0,1/2)$ & & inner PCP & \\ \hline
$q_{0}\in \bN^{+}$ & $q_{0}\geq C\log_{2}(1/\delta)$ & inner PCP & inner PCP \\ \hline
$\tau_{2q_{0}}\in(0,1)$ & $\tau_{2q_{0}}\leq \exp(-q_{0}^{C})$ & inner PCP & inner PCP\\
\hline
$\tau_{2q_{0}-1},\dots,\tau_{0}\in (0,1)$ &
$\tau_{i-1}\leq C^{-1}\tau_{i}^{C},\;\forall i\in [2q_{0}]$ & inner PCP & inner PCP\\
\hline
$\ell_{1}\in \bN^{+}$ &
$\ell_{1}\geq C\log_{2}(1/\tau_{0})$ & inner PCP & inner PCP\\
\hline
$s\in \bN^{+}$ & $s=3\ell_{1}$ & inner PCP & inner PCP\\
\hline 
$q_{1}\in \bN^{+}$ & $q_{1}=2s(2^{2\ell_{1}}-2^{2\ell_{1} -q_{0}})$ & outer PCP & inner PCP\\
\hline
$q_{2}\in \bN^{+}$ & $q_{2}\geq C\ell_{1}^{2}$ & outer PCP & inner PCP\\
\hline
$\ell_{2}\in \bN^{+}$ & $\ell_{2}\geq Cq_{2}^{2}$ & middle PCP & middle PCP\\
\hline
$r\in \bN^{+}$ & $r=2q_{2}+1$ & outer PCP & middle PCP\\
\hline 
$\varepsilon\in (0,1)$ & $\varepsilon\leq \exp(-\exp(\exp(Cr)))$ & outer PCP & outer PCP\\
\hline
$k\in \bN^{+}$ & $k\geq \varepsilon^{-C}\exp(\exp(\ell_{2}))$ & outer PCP & outer PCP\\
\hline
$\beta \in (0,1)$ & $\beta=k^{-3/4}$ & outer PCP & outer PCP \\
\hline
$t\in \bN^{+}$ & $t\geq \exp(k^{C})$ & middle PCP & inner PCP \\
\hline
\end{tabular}
\caption{Parameter relations}\label{tab:parameters}
\end{table}

\Cref{cond:parameters} is meant to be a consistent setting of parameters that we can invoke in~\Cref{sec:outer_PCP,sec:middle_PCP,sec:inner_PCP}. Whenever the condition is invoked, the absolute constant $C>0$ can be different from previous invocations. More precisely, a statement of the form
\[
\Big[\textup{Assume~\Cref{cond:parameters}. Then [XX] holds.}\Big]
\]
should be interpreted as
\[
\Big[\textup{There exists an absolute constant }C>0\textup{ such that the conditions in~\Cref{tab:parameters} imply [XX].}\Big]
\]

\section{The Outer PCP}\label{sec:outer_PCP}

In this section, we formally construct and analyze our outer PCP,
following the outline in~\Cref{sec:outer_pcp_intro}. We first introduce
the necessary framework
in~\Cref{subsec:tripled_set,subsec:NAE_predicate} before turning to
the actual construction
in~\Cref{subsec:equation_vs_variable,subsec:joint_satisfaction_to_bilinear}.

\subsection{Preparation: Tripled Sets}\label{subsec:tripled_set}

As discussed in~\Cref{sec:outer_pcp_imperfect,sec:outer_pcp_intro},
Alice's challenge in the outer PCP consists of a set $E$ of variables
together with an advice sequence. The variables in $E$ are naturally
grouped into triples, each consisting of the three distinct variables
appearing in an equation of the $3\mathsf{QS}_{\mathbb{F}_2}$ instance.
We call such a set $E$ a \emph{tripled set}, formally defined as follows.

\begin{definition}
A \emph{tripled set} is a finite set \(E\) equipped with a partition into
subsets of size \(3\). We refer to the parts of this implicit partition as the
\emph{triples} of \(E\), and we write \(\mathrm{Triples}(E)\) for the collection of
its triples.
\end{definition}

As discussed in~\Cref{sec:outer_pcp_imperfect,sec:outer_pcp_intro},
our outer PCP construction applies (smooth) parallel repetition to a
$3\mathsf{QS}_{\mathbb{F}_2}$ instance. Since the $3\mathsf{QS}_{\mathbb{F}_2}$ instance provided by~\Cref{thm:hg_coloring} is
represented by a $3$-uniform hypergraph $H$, $k$-fold parallel repetition yields a
$2$-prover $1$-round game played on the vertex-disjoint union of
$k$ copies of $H$.

\begin{definition}\label{def:VHk_and_EHk}
Let $H$ be a 3-uniform hypergraph with vertex set $\mathscr{V}(H)$ and edge set $\scrE(H)$. Given a positive integer $k$, we let $H^{\sqcup k}$ denote the $k$-fold vertex-disjoint union of $H$. More explicitly, the vertex set of $H^{\sqcup k}$ is $[k]\times \mathscr{V}(H)$, and we denote $\scrV(H,k)=[k]\times \scrV(H)$; the edge set of $H^{\sqcup k}$ is $\bigcup_{i=1}^{k}\scrE_{i}(H,k)$, where
\[
\scrE_{i}(H,k)=\bigl\{\{(i,\sfv_{1}),(i,\sfv_{2}),(i,\sfv_{3})\}\,\big|\,\{\sfv_{1},\sfv_{2},\sfv_{3}\}\in \scrE(H)\bigr\}.
\]
We let $\scrE(H,k)$ be the collection of all vertex subsets  $E\subseteq [k]\times \mathscr{V}(H)$ that result from picking one edge from each $\scrE_{i}(H,k)$. That is, it consists of $E$ of the form
\begin{equation}\label{eq:def_3k_vertex_set}
E=\{v_{i,j}\mid i\in [k],\,j\in [3]\}\qquad\text{such that}\qquad \{v_{i,1},v_{i,2},v_{i,3}\}\in\scrE_{i}(H,k)\text{ for all }i\in [k],
\end{equation}
% In words, such a vertex set $E$ contains one edge from each copy of $H$. 
and we view it naturally as a tripled set, with
$\operatorname{Triples}(E)=\bigl\{\{v_{i,1},v_{i,2},v_{i,3}\}\,\big|\, i\in [k]\bigr\}$.
\end{definition}
In words, viewing the set of hyperedges in $H$ as the challenges that a verifier samples, the edges in $H^{\sqcup k}$ are naturally viewed as the challenges that the verifier chooses in the $k$-fold repeated game. 

\begin{remark}\label{rem:triple_ordered}
For a tripled set $E$ of the form \eqref{eq:def_3k_vertex_set}, the $k$ triples of $E$ are naturally ordered. For each $i\in [k]$, we will call $\{v_{i,1},v_{i,2},v_{i,3}\}$ the $i$-th triple of $E$.
\end{remark}

\subsubsection{Thinning, and Thinning with Advice}

We next introduce notation to formalize the smoothness in our application of parallel repetition. The following definition captures the structure of
the challenges in each coordinate: with probability $1-\beta$, the two
players receive identical challenges; otherwise, Alice receives the
full challenge, while Bob receives only a single variable from it.

\begin{definition}
Let \(E\) be a tripled set and let \(\beta \in [0,1]\). We write
\(E' \sim \operatorname{Thin}_{\beta}(E)\) for the random subset \(E' \subseteq E\)
obtained as follows. Independently for each triple \(T \in \operatorname{Triples}(E)\),
we include all elements of \(T\) in \(E'\) with probability \(1-\beta\), and
include a uniformly random element of \(T\) in \(E'\) with probability \(\beta\). We stress that the output \(E'\) itself is only a set, not a tripled set.
\end{definition}

In the outer PCP construction, each prover will receive \emph{two} advice sequences $A_{\inn}$ and $A_{\midd}$, the latter supplying the Grassmann-type zoom-ins needed in the middle PCP, while the former is reserved for zoom-ins in the inner PCP. The next definition formalizes the process by which the verifier samples the two advice sequences.

\begin{definition}\label{def:Joint_dist_S_A_Aprime}
Let $E$ be a tripled set. Fix nonnegative integers $q_{1},q_{2}$ and a parameter $\beta\in [0,1]$. We let $\cJ_{\outt}(E,q_{1},q_{2},\beta)$ denote the output distribution of the following process:
\begin{enumerate}
\item Sample a $\beta$-thinning $E'\sim \operatorname{Thin}_{\beta}(E)$.
\item Sample $q=q_{1}+q_{2}$ vectors $a^{(1)},\dots,a^{(q)}\in \bF_{2}^{E'}$ independently and uniformly at random. We denote $A_{\inn}=\big(a^{(1)},\dots,a^{(q_{1})}\big)$ and $A_{\midd}=\big(a^{(q_{1}+1)},\dots,a^{(q_{1}+q_{2})}\big)$.
\item Output the tuple $(E',A_{\inn},A_{\midd})$.
\end{enumerate}
\end{definition}

\begin{remark}\label{rem:ambient_1}
For a subset $E'$ of a finite set $E$, we always view $\bF_{2}^{E'}$ as a linear subspace of the ambient space $\bF_{2}^{E}$. In particular, the vectors $a^{(1)},\dots,a^{(q)}$ in~\Cref{def:Joint_dist_S_A_Aprime} are viewed also as elements of the vector space $\bF_{2}^{E}$. 
\end{remark}

In the outer PCP, Alice will see a tripled set $E$ and the two advice sequences $A_{\inn}$ and $A_{\midd}$, but not the $\beta$-thinning $E'$ which is only sent to Bob. The next proposition 
states that, even conditioned on $A_{\inn}$ and $A_{\midd}$, the distribution of $E'$ is still nicely structured. In particular, it is independent over the triples of $E$.

\begin{restatable}{proposition}{coveringouterA}\label{prop:covering_outer}
Let $E$ be a tripled set. Fix nonnegative integers $q_{1},q_{2}$ and a parameter $\beta\in [0,1)$. Fix arbitrary tuples
\[
A_{1}\in (\mathbb F_{2}^{E})^{q_{1}}
\qquad\text{and}\qquad
A_{2}\in (\mathbb F_{2}^{E})^{q_{2}}.
\]
Sample
\[
(E',A_{\inn},A_{\midd})\sim \cJ_{\outt}(E,q_{1},q_{2},\beta)
\]
conditioned on the event \(\{A_{\inn}=A_{1}\}\cap\{A_{\midd}=A_{2}\}\). Then the resulting marginal distribution of \(E'\) is independent over triples of $E$. Furthermore, each triple of $E$ is entirely included in $E'$ with conditional probability at least $1-4^{q_{1}+q_{2}}\beta$. 
\end{restatable}
\begin{proof}
    Deferred to~\Cref{subsec:covering_outer}.
\end{proof}

As emphasized throughout the proof overview, ``covering properties'' are crucial in our construction of $4$-to-$1$ games, and they will appear many times (in different forms) throughout~\Cref{sec:outer_PCP,sec:middle_PCP,sec:inner_PCP}. They are important primarily for Grassmann encodings, and a job of our outer PCP construction is to prepare a strong enough ``covering property'' for the Grassmann-based tensor tests in the middle PCP to work. On the most basic level, we need that the distribution of vectors sampled uniformly from $\bF_{2}^{E'}$ (where $E'$ is a random $\beta$-thinning), once lifted to $\bF_{2}^{E}$ as discussed in~\Cref{rem:ambient_1}, is close to uniform over $\bF_{2}^{E}$. The next proposition (essentially from~\cite{DKKMS18}) proves that this is the case if $\beta=o(|E|^{-1/2})$.

\begin{restatable}{proposition}{coveringouterB}\label{prop:covering_outer_2}
Let $E$ be a tripled set. Fix nonnegative integers $q_{1},q_{2}$ and a parameter $\beta\in [0,4^{-q_{1}-q_{2}})$. Fix an arbitrary tuple
\(
A_{1}\in (\mathbb F_{2}^{E})^{q_{1}}
\).
Sample
\[
(E',A_{\inn},A_{\midd})\sim \cJ_{\outt}(E,q_{1},q_{2},\beta)
\]
conditioned on the event \(\{A_{\inn}=A_{1}\}\). Then the resulting marginal distribution of \(A_{\midd}\) has total variation distance at most $4^{q_{1}+q_{2}}\beta\sqrt{|E|/3}$ to the uniform distribution over $(\bF_{2}^{E})^{q_{2}}$.
\end{restatable}
\begin{proof}
    Deferred to~\Cref{subsec:covering_outer}.
\end{proof}

We remark that although we include~\Cref{prop:covering_outer_2} here in~\Cref{sec:outer_PCP}, it will not be used until~\Cref{sec:middle_PCP}.

\subsection{Preparation: The NAE Predicate}\label{subsec:NAE_predicate}

The reason that 2-colorability of a 3-uniform hypergraph can be viewed as a system of quadratic equations is that the ternary not-all-equal predicate on a Boolean alphabet is a ``quadratic'' constraint: 

\begin{definition}\label{def:NAE_predicate}
We define a function $\operatorname{NAE}_{3}:\bF_{2}^{3}\rightarrow\bF_{2}$ defined by
\[
\operatorname{NAE}_{3}(x_{1},x_{2},x_{3})=x_{1}+x_{2}+x_{3}+x_{1}x_{2}+x_{2}x_{3}+x_{3}x_{1}.
\]
We note that $\operatorname{NAE}_{3}(x_{1},x_{2},x_{3})=1$ if and only if $x_{1},x_{2},x_{3}$ are not all equal.
\end{definition}

As discussed in~\Cref{sec:outer_pcp_intro}, we will use the notion of ``joint NAE-satisfaction'', which we formalize below. Here we assume that an assignment to the variables in $E$ is presented as a linear functional $f:\bF_{2}^{E}\rightarrow \bF_{2}$.

\begin{definition}\label{def:NAE_linear_functional}
Let $E$ be a tripled set and let $r$ be a positive integer. Given $r$ linear functionals $f^{(1)},\dots,f^{(r)}:\bF_{2}^{E}\rightarrow \bF_{2}$, let $a^{(1)},\dots,a^{(r)}\in \bF_{2}^{E}$ be the unique vectors such that 
\[f^{(i)}(x)=\sum_{v\in E}a^{(i)}_{v}x_{v}\qquad\text{for all }i\in [r]\text{ and }x\in \bF_{2}^{E}.\]
If we have
\[
\sum_{i=1}^{r}\mathrm{NAE}_{3}\bigl(a^{(i)}_{v_{1}},a_{v_{2}}^{(i)},a_{v_{3}}^{(i)}\bigr)=1\qquad\text{for all }\{v_{1},v_{2},v_{3}\}\in \operatorname{Triples}(E),
\]
then the linear functionals $f^{(1)},\dots,f^{(r)}$ are said to be \emph{jointly \textup{NAE}-satisfying}. In the case $r=1$, then the functional $f^{(1)}$ is simply said to be \emph{\textup{NAE}-satisfying}.
\end{definition}

We can now translate the ``2-colorability'' of a 3-uniform hypergraph $H$ into the language of NAE-satisfaction. For the sake of future reference, the following proposition is stated for the $k$-fold disjoint union $H^{\sqcup k}$.

\begin{proposition}\label{prop:2color_NAE_linear}
Suppose $H$ is a 2-colorable 3-uniform hypergraph and $k$ is a positive integer. Then there exists a linear functional $f:\bF_{2}^{\scrV(H,k)}\rightarrow \bF_{2}$ such that for any $E\in \scrE(H,k)$, the restriction $\left.f\right|_{\bF_{2}^{E}}$ is NAE-satisfying.
\end{proposition}

\begin{proof}
Since $H$ is 2-colorable, so is the $k$-fold vertex-disjoint union $H^{\sqcup k}$. Let $\varphi:\scrV(H,k)\rightarrow \bF_{2}$ be a proper 2-coloring of $H^{\sqcup k}$.  Let $f:\bF_{2}^{\scrV(H,k)}\rightarrow \bF_{2}$ be the linear functional defined by
\[
f(x)=\sum_{v\in \scrV(H,k)}\varphi(v)x_{v}\qquad\text{for all }x\in \bF_{2}^{\scrV(H,k)}.
\]
For any $E\in \scrE(H,k)$ and any $\{v_{1},v_{2},v_{3}\}\in \operatorname{Triples}(E)$, since $\{v_{1},v_{2},v_{3}\}$ is an edge of $H^{\sqcup k}$, by definition we know that $\varphi(v_{1}),\varphi(v_{2}),\varphi(v_{3})$ are not all equal. Therefore, it follows from~\Cref{def:NAE_linear_functional} that $\left.f\right|_{\bF_{2}^{E}}$ is NAE-satisfying.
\end{proof}

\subsubsection{Linear Functionals vs. Bilinear Forms}\label{subsubsec:NAE_linear_bilinear}

In later stages of the outer PCP construction, an assignment to variables in $E$ will be presented as a symmetric bilinear map $f:\mathbb{F}_2^{E}\otimes \mathbb{F}_2^{E}\rightarrow \bF_{2}$. The \textup{NAE}-satisfaction notion generalizes to bilinear maps as follows.

\begin{definition}
Let $E$ be a tripled set. Given a symmetric bilinear map $f:\bF_{2}^{E}\otimes \bF_{2}^{E}\rightarrow \bF_{2}$, let $X\in \bF_{2}^{E\times E}$ be the unique matrix such that
\[f(x\otimes y)=\sum_{u,v\in E}X_{uv}x_{u}y_{v}\qquad\text{for all }x,y\in \bF_{2}^{E}.\]
If we have
\[
X_{v_{1}v_{1}}+X_{v_{2}v_{2}}+X_{v_{3}v_{3}}+X_{v_{1}v_{2}}+X_{v_{2}v_{3}}+X_{v_{3}v_{1}}=1\qquad\text{for all }\{v_{1},v_{2},v_{3}\}\in \operatorname{Triples}(E),
\]
then the bilinear map $f$ is said to be \emph{\textup{NAE}-satisfying}.
\end{definition}

In~\Cref{subsec:joint_satisfaction_to_bilinear}, we will need to decode from a low-rank symmetric bilinear map a collection of linear functionals. This process is essentially the standard ``rank-1 decomposition'' in linear algebra, which we formalize below. 

\begin{definition}\label{def:tensor_decomposition}
Let $V$ be a finite-dimensional $\bF_{2}$-vector space and let $r$ be a positive integer. Given a symmetric bilinear map $f:V\otimes V\rightarrow \bF_{2}$ and $r$ linear functionals $f^{(1)},\dots,f^{(r)}:V\rightarrow \bF_{2}$, we say that $\bigl(f^{(1)},\dots,f^{(r)}\bigr)$ is a \emph{tensor decomposition} of $f$ if the following conditions hold:
\begin{enumerate}[label=(\arabic*)]
\item For each $i\in [r]$, we have $f^{(i)}\in \operatorname{image}(T_{f})$. In other words, there exists a vector $x^{(i)}\in V$ such that $f^{(i)}(y)=f(x^{(i)}\otimes y)$ for each $y\in V$.
\item For each $x,y\in V$, we have
\begin{equation}\label{eq:def_tensor_decomposition}
f(x\otimes y)=\sum_{i=1}^{r}f^{(i)}(x)f^{(i)}(y).
\end{equation}
\end{enumerate}
\end{definition}

The following two propositions show that the \textup{NAE}-satisfaction of symmetric bilinear maps is interchangeable with the joint-\textup{NAE}-satisfaction of its tensor decompositions.

\begin{proposition}\label{prop:NAE_quadratic_to_linear}
Let $E$ be a tripled set and let $r$ be a positive integer. Suppose $f:\bF_{2}^{E}\otimes \bF_{2}^{E}\rightarrow \bF_{2}$ is an \textup{NAE}-satisfying symmetric bilinear map. Then any tensor decomposition $\bigl(f^{(1)},\dots,f^{(r)}\bigr)$ of $f$ is jointly \textup{NAE}-satisfying.
\end{proposition}

\begin{proof}
Let $a^{(1)},\dots,a^{(r)}\in \bF_{2}^{E}$ be the unique vectors such that 
\[f^{(i)}(x)=\sum_{v\in E}a^{(i)}_{v}x_{v}\qquad\text{for all }i\in [r]\text{ and }x\in \bF_{2}^{E}.\]
Let $X\in \bF_{2}^{E\times E}$ be the matrix defined by 
\[
X_{uv}=\sum_{i=1}^{r}a^{(i)}_{u}a^{(i)}_{v}\qquad\text{for all }u,v\in E.
\]
Plugging the expansion of $f^{(i)}$ into~\eqref{eq:def_tensor_decomposition} and changing the order of summation, we get that
\[
f(x\otimes y)=\sum_{u,v\in E}X_{uv}x_{u}y_{v}\qquad\text{for all }x,y\in \bF_{2}^{E}.
\]
Since $f$ is \textup{NAE}-satisfying, for each triple $\{v_{1},v_{2},v_{3}\}\in \operatorname{Triples}(E)$ we must have
\[
X_{v_{1}v_{1}}+X_{v_{2}v_{2}}+X_{v_{3}v_{3}}+X_{v_{1}v_{2}}+X_{v_{2}v_{3}}+X_{v_{3}v_{1}}=1,
\]
and hence
\[
\sum_{i=1}^{r}\left(a_{v_{1}}^{(i)}a_{v_{1}}^{(i)}+a_{v_{2}}^{(i)}a_{v_{2}}^{(i)}+a_{v_{3}}^{(i)}a_{v_{3}}^{(i)}+a_{v_{1}}^{(i)}a_{v_{2}}^{(i)}+a_{v_{2}}^{(i)}a_{v_{3}}^{(i)}+a_{v_{3}}^{(i)}a_{v_{1}}^{(i)}\right)=1.
\]
By~\Cref{def:NAE_predicate} we thus have
\[
\sum\limits_{i=1}^{r}\mathrm{NAE}_{3}\bigl(a^{(i)}_{v_{1}},a_{v_{2}}^{(i)},a_{v_{3}}^{(i)}\bigr)=1\qquad\text{for all }\{v_{1},v_{2},v_{3}\}\in \operatorname{Triples}(E),
\]
and thus $f^{(1)},\dots,f^{(r)}$ are jointly \textup{NAE}-satisfying.
\end{proof}

\begin{proposition}\label{prop:NAE_linear_to_quadratic}
Let $E$ be a tripled set and let $r$ be a positive integer. Suppose $f:\bF_{2}^{E}\otimes \bF_{2}^{E}\rightarrow \bF_{2}$ is a symmetric bilinear map and $f^{(1)},\dots,f^{(r)}:\bF_{2}^{E}\rightarrow \bF_{2}$ are jointly \textup{NAE}-satisfying linear functionals. If \eqref{eq:def_tensor_decomposition} holds for all $x,y\in \bF_{2}^{E}$, then $f$ is \textup{NAE}-satisfying.
\end{proposition}
\begin{proof}
  The proof follows the same lines as in~\Cref{prop:NAE_quadratic_to_linear}, and we omit the details.
\end{proof}

The following straightforward corollary of~\Cref{prop:NAE_linear_to_quadratic} will be used in~\Cref{lem:Game2a_soundness}. 

\begin{corollary}\label{cor:NAE_rank1_exists}
Let $E$ be a nonempty tripled set. There exists an \textup{NAE}-satisfying symmetric bilinear map $f:\bF_{2}^{E}\otimes \bF_{2}^{E}\rightarrow \bF_{2}$ of rank 1.
\end{corollary}
\begin{proof}
There clearly exists a nonzero linear functional $f^{(1)}:\bF_{2}^{E}\rightarrow \bF_{2}$ that is \textup{NAE}-satisfying (by itself). We can then let $f:\bF_{2}^{E}\otimes \bF_{2}^{E}\rightarrow \bF_{2}$ be the rank-$1$ bilinear map defined by $f(x\otimes y)=f^{(1)}(x)f^{(1)}(y)$ for all $x,y\in \bF_{2}^{E}$, and apply~\Cref{prop:NAE_linear_to_quadratic} to show that $f$ is \textup{NAE}-satisfying.
\end{proof}

Combining~\Cref{prop:2color_NAE_linear,prop:NAE_linear_to_quadratic} yields the following ``bilinear'' version of~\Cref{prop:2color_NAE_linear}.

\begin{corollary}\label{cor:2color_NAE_quad}
Suppose $H$ is a $2$-colorable $3$-uniform hypergraph and $k$ is a positive integer. Then there exists a rank-$1$ symmetric bilinear form $f:\bF_{2}^{\scrV(H,k)}\otimes \bF_{2}^{\scrV(H,k)}\rightarrow \bF_{2}$ such that for any $E\in \scrE(H,k)$, the restriction $\left.f\right|_{\bF_{2}^{E}\otimes \bF_{2}^{E}}$ is \textup{NAE}-satisfying.
\end{corollary}
\begin{proof}
By~\Cref{prop:2color_NAE_linear}, there is a (nonzero) linear functional $f^{(1)}:\bF_{2}^{\scrV(H,k)}\rightarrow \bF_{2}$ such that $\left.f^{(1)}\right|_{\bF_{2}^{E}}$ is \textup{NAE}-satisfying for every $E\in \scrE(H,k)$. We can then let \[f:\bF_{2}^{\scrV(H,k)}\otimes \bF_{2}^{\scrV(H,k)}\rightarrow \bF_{2}\] be the rank-$1$ symmetric bilinear map defined by $f(x\otimes y)=f^{(1)}(x)f^{(1)}(y)$ for all $x,y\in \bF_{2}^{\scrV(H,k)}$, and apply~\Cref{prop:NAE_linear_to_quadratic} to show that $\left.f\right|_{\bF_{2}^{E}\otimes \bF_{2}^{E}}$ is \textup{NAE}-satisfying for each $E\in\scrE(H,k)$.
\end{proof} 

\subsection{The Equation vs.~Variable Game}\label{subsec:equation_vs_variable}

We are now ready to start the construction of our outer PCP. The first version of our outer PCP is defined as follows (cf.~\Cref{sec:outer_pcp_imperfect}). 

\begin{definition}\label{def:game1a}
Given positive integers $k,r$, nonnegative integers $q_{1},q_{2}$, a parameter $\beta\in [0,1]$ and a 3-uniform hypergraph $H$, let $\GameAa[H,k,q_1,q_2,r,\beta]$ be the $2$-prover $1$-round game where the verifier executes the following:
\begin{enumerate}
\item Sample a uniformly random tripled set $E\in\scrE(H,k)$.

\item Sample $(E',A_{\inn},A_{\midd})\sim\cJ_{\outt}(E,q_{1},q_{2},\beta)$. Send $(E,A_{\inn},A_{\midd})$ to Alice\footnote{In Alice's view, $A_{\inn}$ and $A_{\midd}$ are sequences of vectors in $\bF_{2}^{E}$, instead of in $\bF_{2}^{E'}$; it is important that Alice does not get to see $E'$.} and $(E',A_{\inn},A_{\midd})$ to Bob. 

  \item Alice must return $r$ linear functionals $f^{(1)},\dots,f^{(r)}: \bF_{2}^{E}\rightarrow \bF_{2}$ that are jointly \textup{NAE}-satisfying; Bob must return $r$ linear functionals $g^{(1)},\dots,g^{(r)}:\bF_{2}^{E'}\rightarrow \bF_{2}$. Accept if 
  \[\left.f^{(i)}\right|_{\bF_{2}^{E'}}=g^{(i)}\qquad\text{for all }i\in [r].\]
\end{enumerate}
\end{definition}
\begin{remark}\label{rem:k-fold-repetition-view}
We remark that the game $\GameAa[H,k,q_{1},q_{2},r,\beta]$ can be viewed as the $k$-fold parallel repetition of $\GameAa[H,1,q_{1},q_{2},r,\beta]$, and the game $\GameAa[H,1,0,0,r,1]$ is the equation-vs.-variable game corresponding to the result of~\Cref{thm:hg_coloring}.
\end{remark}

We now analyze the completeness and the soundness of $\GameAa$. 

\begin{lemma}[Game1a completeness]\label{lem:Game1a_completeness}
For any positive integer $k$, nonnegative integers $q_{1},q_{2}$, parameter $\beta\in [0,1]$ and $2$-colorable $3$-uniform hypergraph $H$, we have
\[
\operatorname{val}\Bigl(\GameAa[H,k,q_{1},q_{2},1,\beta]\Bigr)=1.
\]
\end{lemma}
\begin{proof}

We obtain from~\Cref{prop:2color_NAE_linear} a linear functional $f:\bF_{2}^{\scrV(H,k)}\rightarrow \bF_{2}$, whose restriction on $\bF_{2}^{E}$ is \textup{NAE}-satisfying for any $E\in \scrE(H,k)$. We can then construct a strategy for Alice in the game $\GameAa[H,k,q_{1},q_{2},1,\beta]$, as follows:
\begin{enumerate}
\item Receive the question $(E,A_{\inn},A_{\midd})$.
\item Return the functional $\left. f\right|_{\bF_{2}^{E}}$, which is guaranteed by~\Cref{prop:2color_NAE_linear} to be \textup{NAE}-satisfying.
\end{enumerate}
Similarly, the strategy for Bob is:
\begin{enumerate}
\item Receive the question $(E',A_{\inn},A_{\midd})$.
\item Return the functional $\left. f\right|_{\bF_{2}^{E'}}$.
\end{enumerate}
It is clear that Bob's answer is always consistent with Alice's answer. Therefore, if Alice and Bob operate under such strategies, the verifier in the game $\GameAa[H,k,q_{1},q_{2},1,\beta]$ accepts with probability 1.
\end{proof}

To show soundness of $\GameAa[H,k,q_{1},q_{2},r,\beta]$, we start by analyzing the soundness of the basic game $\GameAa[H,1,0,0,r,1]$, which is effectively the standard equation-versus-variable game with ``superpositions.''
\begin{lemma}\label{lem:Game1a_soundness_no_advice}
Fix a positive integer $r$ and a parameter $\varepsilon\in (0,1)$. Suppose $H$ is a 3-uniform hypergraph such that in every $2^{r}$-coloring of $H$, at least $\varepsilon$ fraction of the edges are monochromatic. Then we have
\[
\operatorname{val}\Bigl(\GameAa[H,1,0,0,r,1]\Bigr)\leq 1-\frac{\varepsilon}{3}.
\]
\end{lemma}

\begin{proof}
Let $V$ be the vertex set of $H^{\sqcup 1}$ (which is isomorphic to $H$). 

We fix an arbitrary strategy for Alice and Bob in the game $\GameAa[H,1,0,0,r,1]$, and we will prove that the verifier rejects with probability at least $\varepsilon/3$. Since there are no advice vectors in this game, the question for Alice is simply a random edge of $H^{\sqcup 1}$; and as the smoothness parameter is set to 1, the question for Bob will be a single vertex.

Our goal is to define a $2^{r}$-coloring map $\varphi:V\rightarrow\bF_{2}^{r}$. For each vertex $v\in V$, the strategy for Bob must assign an answer to the question $\{v\}$, and the answer must be $r$ linear functionals \[g_{v}^{(1)},\dots,g_{v}^{(r)}:\bF_{2}^{\{v\}}\rightarrow \bF_{2}.\] 
Aggregating Bob's answers for all $v\in V$ yields $r$ linear functionals
\[
g^{(1)},\dots,g^{(r)}:\bF_{2}^{V}\rightarrow \bF_{2}.
\]
We can then let $\varphi:V\rightarrow \bF_{2}^{r}$ be the unique assignment map such that
\begin{equation}\label{eq:game1a_soundness_g_varphi}
\left(g^{(1)}(x),g^{(2)}(x),\dots,g^{(r)}(x)\right)=\sum_{v\in V}\varphi(v)x_{v}\qquad\text{for all }x\in \bF_{2}^{V}.
\end{equation}

By the assumption of the lemma, for at least $\varepsilon$ fraction of the edges $\{v_{1},v_{2},v_{3}\}\in \scrE(H,1)$, we have $\varphi(v_{1})=\varphi(v_{2})=\varphi(v_{3})$. Fix such an edge $E=\{v_{1},v_{2},v_{3}\}$, and let $f^{(1)},\dots,f^{(r)}:\bF_{2}^{E}\rightarrow \bF_{2}$ be Alice's answer to the question $E$. By~\Cref{def:game1a}, $f^{(1)},\dots,f^{(r)}$ must be jointly NAE-satisfying. However,  $\varphi(v_{1})=\varphi(v_{2})=\varphi(v_{3})$ implies (via~\Cref{def:NAE_linear_functional} and~\Cref{eq:game1a_soundness_g_varphi}) that
\[
\left.g^{(1)}\right|_{\bF_{2}^{E}},\left.g^{(2)}\right|_{\bF_{2}^{E}},\dots,\left.g^{(r)}\right|_{\bF_{2}^{E}}
\]
are not jointly NAE-satisfying. Therefore, there exists a vertex $v\in \{v_{1},v_{2},v_{3}\}$ with
\begin{equation}\label{eq:game_1a_soundness}
\left(\left.f^{(1)}\right|_{\bF_{2}^{\{v\}}},\left.f^{(2)}\right|_{\bF_{2}^{\{v\}}},\dots,\left.f^{(r)}\right|_{\bF_{2}^{\{v\}}}\right)\neq \left(g_{v}^{(1)},g^{(2)}_{v},\dots,g^{(r)}_{v}\right).
\end{equation}
Conditioned on the event that Alice receives the question $E$, with probability $1/3$ the question for Bob is $\{v\}$, in which case their answers are not consistent due to \eqref{eq:game_1a_soundness}. Since this is true for at least $\varepsilon$-fraction of the edges $E\in \scrE(H,1)$, the verifier must reject with probability at least $\varepsilon/3$.
\end{proof}

Next, we incorporate the advice as well as constantly many repetitions, and show that the value does not increase by too much:
\begin{lemma}\label{lem:Game1a_soundness_first_rep}
Fix positive integers $k,r$, nonnegative integers $q_{1},q_{2}$ and parameters $\beta,\varepsilon\in (0,1)$ such that
\(k\geq 2^{q_{1}+q_{2}}\beta^{-1}\).
Suppose $H$ is a 3-uniform hypergraph such that in every $2^{r}$-coloring of $H$, at least $\varepsilon$ fraction of the edges are monochromatic. Then we have
\[
\operatorname{val}\Bigl(\GameAa[H,k,q_{1},q_{2},r,\beta]\Bigr)\leq 1-\frac{\varepsilon}{6}.
\]
\end{lemma}

\begin{proof}
We write $q=q_{1}+q_{2}$. As in~\Cref{rem:k-fold-repetition-view}, we view $\GameAa[H,k,q_{1},q_{2},r,\beta]$ as the $k$-fold parallel repetition of the game $\GameAa[H,1,q_{1},q_{2},r,\beta]$. Let $(E,A_{\inn},A_{\midd})$ and $(E',A_{\inn},A_{\midd})$ be the questions for Alice and Bob in the $k$-fold repetition game, respectively. We write
\[
A_{\inn}=\big(a^{(1)},\dots,a^{(q_{1})}\big)\qquad\text{and}\qquad A_{\midd}=\big(a^{(q_{1}+1)},\dots,a^{(q_{1}+q_{2})}\big).
\]

We say a triple $\{v_{1},v_{2},v_{3}\}\in \operatorname{Triples}(E)$ is ``critical'' if we have
\[a^{(i)}_{v_{1}}=a^{(i)}_{v_{2}}=a^{(i)}_{v_{3}}=0\text{ for all }i\in [q]\qquad\text{and}\qquad \left|E'\cap \{v_{1},v_{2},v_{3}\}\right|=1.\]
Let $I\subseteq [k]$ be the collection of indices $i\in [k]$ such that the $i$-th triple of $E$ is critical (see~\Cref{rem:triple_ordered}). Conditioned on the event that $I\neq \emptyset$, we pick a uniformly random $i\in I$ and note that the $i$-th ``coordinate'' of the parallel repetition is a copy of the equation vs. variable game $\GameAa[H,1,0,0,r,1]$ and is independent of the remaining coordinates. Therefore,~\Cref{lem:Game1a_soundness_no_advice} implies that conditioned on $\{I\neq \emptyset\}$, the verifier rejects with probability at least $\varepsilon/3$. Since
\[
\Pr{I\neq \emptyset}=1-\left(1-2^{-q}\beta\right)^{k}\geq 1-2^{-k\cdot 2^{-q}\beta}\geq 1-2^{-1}=\frac{1}{2},
\]
we conclude that in the $k$-fold repetition game $\GameAa[H,k,q_{1},q_{2},r,\beta]$, the verifier rejects with probability at least $\frac{1}{2}\cdot \varepsilon/3=\varepsilon/6$.
\end{proof}

Finally, we apply the parallel repetition theorem for projection games by~\cite{rao2008parallel} to get the soundness of $\GameAa$.

\begin{lemma}[Game1a soundness]\label{lem:Game1a_soundness}
Assume \Cref{cond:parameters}. Suppose $H$ is a 3-uniform hypergraph such that in every $4^{r}$-coloring of $H$, at least $\varepsilon$ fraction of the edges are monochromatic. Then we have\footnote{We have $2r$ instead of $r$ as the rank parameter here, in order to prepare for applying~\Cref{lem:Game1b_soundness_low_rank} later.}
\[
\operatorname{val}\Bigl(\GameAa[H,k,q_{1},q_{2},2r,\beta]\Bigr)\leq \exp(-k^{1/5}).
\]
\end{lemma}
\begin{proof}
Let $k_{0}=\lceil 2^{q_{1}+q_{2}}\beta^{-1}\rceil$. By~\Cref{lem:Game1a_soundness_first_rep} we know that
\[
\operatorname{val}\Bigl(\GameAa[H,k_{0},q_{1},q_{2},2r,\beta]\Bigr)\leq 1-\frac{\varepsilon}{6}.
\]
Since $\GameAa[H,k_{0},q_{1},q_{2},2r,\beta]$ is a projection game, we may apply~\cite[Theorem 1.4]{rao2008parallel} to conclude that its $\lfloor k/k_{0}\rfloor$-fold parallel repetition has value at most $\exp(-c\varepsilon^{2}\lfloor k/k_{0}\rfloor)$, for some absolute constant $c\in (0,1)$. It is then easy to see (via~\Cref{rem:k-fold-repetition-view}) that
\[
\operatorname{val}\Bigl(\GameAa[H,k,q_{1},q_{2},2r,\beta]\Bigr)\leq \exp(-c\varepsilon^{2}\lfloor k/k_{0}\rfloor)\leq \exp(-k^{1/5}),
\]
where in the last transition we use the relation $k/k_{0}\gg \varepsilon^{-2}k^{1/5}$ guaranteed by~\Cref{cond:parameters}.
\end{proof}

\subsection{From Joint Satisfaction to Bilinear Forms}\label{subsec:joint_satisfaction_to_bilinear}

As hinted in~\Cref{sec:outer_pcp_intro} and~\Cref{subsubsec:NAE_linear_bilinear}, the last goal in the outer PCP stage is to change the format of the provers' responses to symmetric bilinear forms. More precisely, we will consider the following $2$-prover $1$-round game defined below. It is identical to $\GameAa[H,k,q_{1},q_{2},r,\beta]$, except that the provers are expected to provide a single bilinear form of low-rank (instead of $r$ linear forms).
\begin{definition}\label{def:game1b}
Given positive integers $k,r$, nonnegative integers $q_{1},q_{2}$, a parameter $\beta\in [0,1)$ and a 3-uniform hypergraph $H$, let $\GameAb[H,k,q_{1},q_{2},r,\beta]$
be the $2$-prover-$1$-round game where the verifier executes the following:
\begin{enumerate}
\item Sample a uniformly random tripled set $E\in\scrE(H,k)$.

\item Sample $(E',A_{\inn},A_{\midd})\sim\cJ_{\outt}(E,q_{1},q_{2},\beta)$. Send $(E,A_{\inn},A_{\midd})$ to Alice and $(E',A_{\inn},A_{\midd})$ to Bob.  
\item Alice must return a symmetric bilinear form $f:\bF_{2}^{E}\otimes \bF_{2}^{E}\rightarrow\bF_{2}$ that is \textup{NAE}-satisfying and has rank at most $r$; Bob must return a bilinear form $g: \bF_{2}^{E'}\otimes \bF_{2}^{E'}\rightarrow \bF_{2}$. Accept if
\[
\left.f\right|_{\bF_{2}^{E'}\otimes \bF_{2}^{E'}}=g.
\]
\end{enumerate}
\end{definition}

% With this in mind, the main result of this section is as follows:
% \begin{lemma}\label{lem:outer_thirdstep}
%     Let $r,q_1,q_2,k\in\mathbb{N}$, $\eps>0$ and $\beta\in (0,1)$ be parameters as in~\Cref{cond:parameters}. There is a polynomial time reduction that given a $3$-uniform hypergraph $H$ produces $\GameAa[H,k,q_{1},q_{2},r,\beta]$ such that the following properties hold:
%     \begin{enumerate}
%         \item \textbf{Completeness:} if $H$ is $2$-colorable, then $\mathsf{val}(\GameAb[H,k,q_{1},q_{2},1,\beta])=1$.
%         \item \textbf{Soundness:} if in any $2^r$-coloring of $H$ at least $\eps$ fraction of the edges are monochromatic, then 
%         $\mathsf{val}(\GameAb[H,k,q_{1},q_{2},r,\beta])\leq 2^{-k^{1/5}}$.
%         \end{enumerate}
% \end{lemma}
% The reduction is from~\Cref{thm:hg_coloring}, and is described in~\Cref{def:game1b}. We next present the analysis.
% \textcolor{red}{Dor: working from here.}

We now analyze the completeness and the soundness of $\GameAb$.

\begin{lemma}[Game1b completeness]\label{lem:Game1b_completeness}
For any positive integer $k$, nonnegative integers $q_{1},q_{2}$, parameter $\beta\in [0,1)$ and $2$-colorable $3$-uniform hypergraph $H$, we have
\[
\operatorname{val}\Bigl(\GameAb[H,k,q_{1},q_{2},1,\beta]\Bigr)=1.
\]
\end{lemma}
\begin{proof}
We obtain from~\Cref{cor:2color_NAE_quad} a rank-1 symmetric bilinear map $f:\bF_{2}^{\scrV(H,k)}\otimes \bF_{2}^{\scrV(H,k)}\rightarrow \bF_{2}$, whose restriction on $\bF_{2}^{E}\otimes \bF_{2}^{E}$ is \textup{NAE}-satisfying for every $E\in \scrE(H,k)$. We can then construct a strategy for Alice in the game $\GameAb[H,k,q_{1},q_{2},1,\beta]$, as follows:
\begin{enumerate}
\item Receive the question $(E,A_{\inn},A_{\midd})$.
\item Return the symmetric bilinear map $\left.f\right|_{\bF_{2}^{E}\otimes \bF_{2}^{E}}$, which is guaranteed by~\Cref{cor:2color_NAE_quad} to be \textup{NAE}-satisfying. Furthermore, it must have rank at most 1 because $\operatorname{rank}(f)=1$. 
\end{enumerate}
Similarly, the strategy for Bob is:
\begin{enumerate}
\item Receive the question $(E',A_{\inn},A_{\midd})$.
\item Return the symmetric bilinear map $\left.f\right|_{\bF_{2}^{E'}\otimes \bF_{2}^{E'}}$.
\end{enumerate}
It is clear that under this set of strategies, the verifier accepts with probability 1.
\end{proof}

\begin{remark}\label{rem:Game1b_completeness_Bob}
In $\GameAb[H,k,q_{1},q_{2},1,\beta]$, although Bob's answer $g:\bF_{2}^{E'}\otimes \bF_{2}^{E'}\rightarrow \bF_{2}$ to a question $(E',A_{\inn},A_{\midd})$ is not required to be symmetric or of rank at most 1 (we only pose such requirements on Alice's answer), we note that under the strategy in the completeness proof (\Cref{lem:Game1b_completeness}), Bob's answers are indeed always symmetric and have rank at most 1. This feature will propagate to the subsequent constructions $\GameBa$, $\GameBb$, and their completeness proofs. Eventually, this observation will be used in the completeness proof of $\GameCa$ (\Cref{lem:Game3_completeness}).  
\end{remark}

We now move on to the soundness. The following lemma establishes a connection between the soundness of $\GameAa$ and $\GameAb$, and more precisely it shows how to derive a strategy for the former from a strategy for the latter. Intuitively, this comes 
from the fact that a rank $r$ bilinear form has at most $O_{r}(1)$ many decompositions into rank-$1$ bilinear forms, but some care is needed (specifically, the event that the rank of Alice's bilinear form drops in Bob's space is problematic, but thankfully it has low probability).
\begin{lemma}\label{lem:Game1b_soundness_low_rank}
For any positive integers $k,r$, nonnegative integers $q_{1},q_{2}$, parameter $\beta\in [0,4^{-q_1-q_2})$ and 3-uniform hypergraph $H$, we have
\[
\operatorname{val}\Big(\GameAa[H,k,q_{1},q_{2},2r,\beta]\Big)\geq 4^{-r^{2}}\left(\operatorname{val}\Big(\GameAb[H,k,q_{1},q_{2},r,\beta]\Big)-4^{q_{1}+q_{2}}r\beta\right).
\]
\end{lemma}
\begin{proof}
We organize the proof into the following three steps.

\paragraph{Step 1: controlling rank drop.} 
We denote
\[
s=\operatorname{val}\Big(\GameAb[H,k,q_{1},q_{2},r,\beta]\Big).
\]
Let $F[\cdot]$ and $G[\cdot]$ be the strategies of Alice and Bob, respectively, under which the verifier accepts with probability $s$ in $\GameAb[H,k,q_{1},q_{2},r,\beta]$. For each $E\in \scrE(H,k)$ and $(E',A_{\inn},A_{\midd})$ in the support of $\cJ_{\outt}(E,q_{1},q_{2},\beta)$, we have bilinear maps
\[
F[E,A_{\inn},A_{\midd}]:\bF_{2}^{E}\otimes \bF_{2}^{E}\rightarrow \bF_{2}\qquad\text{and}\qquad G[E',A_{\inn},A_{\midd}]:\bF_{2}^{E'}\otimes \bF_{2}^{E'}\rightarrow \bF_{2}.
\]

A pair of questions $(E,A_{\inn},A_{\midd})$ and $(E',A_{\inn},A_{\midd})$ for Alice and Bob is said to be ``good'' if
\[
\left.F[E,A_{\inn},A_{\midd}]\right|_{\bF_{2}^{E'}\otimes \bF_{2}^{E'}}=G[E',A_{\inn},A_{\midd}];
\]
it is said to be ``excellent'' if in addition to being good, it also satisfies
\[
\operatorname{rank}\bigl(F[E,A_{\inn},A_{\midd}]\bigr)=\operatorname{rank}\bigl(G[E',A_{\inn},A_{\midd}]\bigr).
\]

By definition, a random pair of questions is ``good'' with probability $s$. Furthermore, conditioned on each fixed Alice question $(E,A_{\inn},A_{\midd})$, since $\operatorname{rank}\bigl(F[E,A_{\inn},A_{\midd}]\bigr)\leq r$, we know from~\Cref{cor:full_rank_core} and~\Cref{prop:covering_outer} that
\begin{align*}
\Pru{E'}{\operatorname{rank}\Bigl(\left.F[E,A_{\inn},A_{\midd}]\right|_{\bF_{2}^{E'}\otimes \bF_{2}^{E'}}\Bigr)=\operatorname{rank}\Bigl(F[E,A_{\inn},A_{\midd}]\Bigr)\,\middle|\, E,A_{\inn},A_{\midd}}&\geq (1-4^{q_{1}+q_{2}}\beta)^{r}\\
&\geq 1-4^{q_{1}+q_{2}}r\beta.
\end{align*}
Therefore, the probability that a random pair of questions is ``good'' but not ``excellent'' is at most $4^{q_{1}+q_{2}}r\beta$. Consequently, a random pair of questions is ``excellent'' with probability at least $s-4^{q_{1}+q_{2}}r\beta$.

\paragraph{Step 2: deriving a strategy for Game1a.} 
Note that in $\GameAa[H,k,q_{1},q_{2},2r,\beta]$ the joint distribution of questions to Alice and Bob is the same as in $\GameAb[H,k,q_{1},q_{2},r,\beta]$. We next derive strategies for the provers in the former game from the tables $F[\cdot]$ and $G[\cdot]$. The strategy for Alice is as follows:
\begin{enumerate}
\item Receive the question $(E,A_{\inn},A_{\midd})$, and let $f=F[E,A_{\inn},A_{\midd}]$. Note that $f:\bF_{2}^{E}\otimes \bF_{2}^{E}\rightarrow \bF_{2}$ is symmetric and \textup{NAE}-satisfying.
\item Use~\Cref{lem:rank1_decomposition} to obtain a tensor decomposition $\bigl(f^{(1)},\dots,f^{(2r)}\bigr)$ of $f$. 
\item Return the linear functionals $f^{(1)},\dots,f^{(2r)}:\bF_{2}^{E}\rightarrow \bF_{2}$. They must be jointly \textup{NAE}-satisfying, due to~\Cref{prop:NAE_quadratic_to_linear}.
\end{enumerate}
Accordingly, the strategy for Bob is:
\begin{enumerate}
\item Receive the question $(E',A_{\inn},A_{\midd})$, and let $g=G[E',A_{\inn},A_{\midd}]$.
\item For each $i\in [2r]$, select a linear functional $g^{(i)}$ from the subspace $\operatorname{image}(T_{g})\subseteq (\bF_{2}^{E'})^*$ independently and uniformly at random.
\item Return the linear functionals $g^{(1)},\dots,g^{(2r)}:\bF_{2}^{E'}\rightarrow \bF_{2}$.
\end{enumerate}

\paragraph{Step 3: analysis of the strategy.} We show that the above strategy answers every ``excellent'' pair of questions consistently with probability at least $4^{-r^{2}}$, thereby concluding the desired inequality
\[
\operatorname{val}\Big(\GameAa[H,k,q_{1},q_{2},2r,\beta]\Big)\geq 4^{-r^{2}}(s-4^{q_{1}+q_{2}}r\beta).
\]
Suppose $(E,A_{\inn},A_{\midd})$ and $(E',A_{\inn},A_{\midd})$ are an ``excellent'' pair of questions for Alice and Bob, and let
\[
f=F[E,A_{\inn},A_{\midd}]\qquad\text{and}\qquad g=G[E',A_{\inn},A_{\midd}].
\]
Let $\bigl(f^{(1)},\dots,f^{(2r)}\bigr)$ be an arbitrary tensor decomposition of $f$. By~\Cref{def:tensor_decomposition}, we know that $f^{(i)}\in \operatorname{image}(T_{f})$ for each $i\in [2r]$. Furthermore, since $\operatorname{rank}(f)=\operatorname{rank}(g)$ (due to the excellence condition), it follows from~\Cref{cor:full_rank_restriction} that
\[
\left.f^{(i)}\right|_{\bF_{2}^{E'}}\in \operatorname{image}(T_{g})\qquad\text{for all }i\in [2r].
\]
Since the cardinality of $\operatorname{image}(T_{g})$ is at most $2^{\operatorname{rank}(g)}\leq 2^{r}$, if $g^{(1)},\dots,g^{(2r)}$ are independently and uniformly sampled from $\operatorname{image}(T_{g})$, we have
\[
\Pr{\left.f^{(i)}\right|_{\bF_{2}^{E'}}=g^{(i)}\text{ for all }i\in [2r]}\geq (2^{-r})^{2r}=4^{-r^2},
\]
as desired.
\end{proof}

Using~\Cref{lem:Game1b_soundness_low_rank}, we can now prove the soundness of $\GameAb$.
\begin{lemma}[Game1b soundness]\label{lem:Game1b_soundness}
Assume \Cref{cond:parameters}. Suppose $H$ is a 3-uniform hypergraph such that in every $4^{r}$-coloring of $H$, at least $\varepsilon$ fraction of the edges are monochromatic. Then we have
\[
\operatorname{val}\Bigl(\GameAb[H,k,q_{1},q_{2},r,\beta]\Bigr)\leq k^{-1/2}.
\]
\end{lemma}
\begin{proof}
Combining~\Cref{lem:Game1a_soundness,lem:Game1b_soundness_low_rank}, we have
\[
\operatorname{val}\Bigl(\GameAb[H,k,q_{1},q_{2},r,\beta]\Bigr)\leq 4^{q_{1}+q_{2}}r\beta + 4^{r^{2}}\exp(-k^{1/5})\leq k^{-1/2},
\]
where the last transition is due to~\Cref{cond:parameters}.
\end{proof}

\section{The Middle PCP}\label{sec:middle_PCP}

In this section we construct and analyze our middle PCP as discussed in~\Cref{sec:middle_PCP_intro}. \Cref{subsec:thinning_subspace,subsec:Game2a,subsec:middle_soundness} are devoted to the design and analysis of $\GameBa$, where we incorporate local tensor checks. We then perform a lazy parallel repetition to get $\GameBb$ in~\Cref{subsec:Game2b}.

\subsection{Preparation: Thinning with Subspaces}\label{subsec:thinning_subspace}

In the middle PCP, we will use a Grassmann-type structure to perform ``local tensor checks.'' Along with a question $(E,A_{\inn},A_{\midd})$ sent to Alice (as in the outer PCP), the verifier also sends a subspace $U\subseteq \bF_{2}^{E}$ and checks whether the restriction of Alice's answer $f:\bF_{2}^{E}\otimes \bF_{2}^{E}\rightarrow \bF_{2}$ to $U\otimes U$ has rank at most $1$. The following augmentation of~\Cref{def:Joint_dist_S_A_Aprime} captures the process by which the verifier samples the tensor-check subspace $U$. 

\begin{definition}\label{def:Joint_dist_S_A_Aprime_U}
Let $E$ be a tripled set. Fix a parameter $\beta\in [0,1]$ and nonnegative integers $q_{1},q_{2},\ell$ such that $\ell\geq q_{2}$. We let $\cJ_{\midd}(E,q_{1},q_{2},\beta,\ell)$ denote the output distribution of the following process:
\begin{enumerate}
\item Sample $(E',A_{\inn},A_{\midd})\sim \cJ_{\outt}(E,q_{1},q_{2},\beta)$.
\item Sample $\ell-q_{2}$ vectors $a^{(1)},\dots,a^{(\ell-q_{2})}\in \bF_{2}^{E'}$ independently and uniformly at random. Let $A'=\bigl(a^{(1)},\dots,a^{(\ell-q_{2})}\bigr)$, and obtain a subspace $U=\operatorname{span}(A_{\midd})+\operatorname{span}(A')\subseteq \bF_{2}^{E'}$.
\item Output the tuple $(E',A_{\inn},A_{\midd},U)$.
\end{enumerate}
\end{definition}
\begin{remark}
Similarly to~\Cref{rem:ambient_1}, the subspace $U\subseteq \bF_{2}^{E'}$ in~\Cref{def:Joint_dist_S_A_Aprime_U} is viewed also as a subspace of $\bF_{2}^{E}$.
\end{remark}

\begin{remark}\label{rem:Jmid_no_dependence_q_2}
Consider a sample $(E',A_{\inn},A_{\midd},U)\sim \cJ_{\midd}(E,q_{1},q_{2},\beta,\ell)$ as defined in~\Cref{def:Joint_dist_S_A_Aprime_U}. We note that the distribution of $(E',A_{\inn},U)$ does not depend on $q_{2}$. This is because after fixing $E'$ and $A_{\inn}$, the advice sequence $A_{\midd}$ is uniformly random over $(\bF_{2}^{E'})^{q_{2}}$.
\end{remark}

The next proposition (which can be derived from~\Cref{prop:covering_outer_2}) is the crucial ``covering property'' we need in the middle PCP: a random subspace of $\bF_{2}^{E'}$ (where $E'$ is a random $\beta$-thinning of $E$) lifted to $\bF_{2}^{E}$ is close to a random subspace of $\bF_{2}^{E}$ (even with conditioning on an advice). 

\begin{restatable}{proposition}{coveringmiddleA}\label{prop:covering_middle}
Let $E$ be a tripled set. Fix a parameter $\beta\in [0,1)$ and nonnegative integers $q_{1},q_{2},\ell$ such that $|E|\geq \ell\geq q_{2}$. Fix arbitrary tuples
\[
A_{1}\in (\mathbb F_{2}^{E})^{q_{1}}
\qquad\text{and}\qquad
A_{2}\in (\mathbb F_{2}^{E})^{q_{2}}\text{ such that }\dim\operatorname{span}(A_{2})=q_{2}.
\]
Sample
\[
(E',A_{\inn},A_{\midd},U)\sim \cJ_{\midd}(E,q_{1},q_{2},\beta,\ell)
\]
conditioned on the event \(\{A_{\inn}=A_{1}\}\cap\{A_{\midd}=A_{2}\}\). Then the resulting marginal distribution of \(U\) has total variation distance at most
\[
4^{q_{1}+\ell}\beta\sqrt{|E|/3}+2^{\ell-|E|}
\]
to the uniform distribution over the set \(\big\{U\in \mathsf{Gr}(\mathbb F_{2}^{E},\ell)\,\big|\, \operatorname{span}(A_{2})\subseteq U\big\}\).
\end{restatable}
\begin{proof}
    Deferred to~\Cref{subsec:covering_middle}.
\end{proof}

We also need the following simpler smoothness property of $\cJ_{\midd}$ (derived from~\Cref{prop:covering_outer}).

\begin{proposition}\label{prop:covering_middle_2}
Let $E$ be a tripled set. Fix a parameter $\beta\in [0,1)$ and nonnegative integers $q_{1},\ell$. Fix an arbitrary tuple $A\in (\bF_{2}^{E})^{q_{1}}$ and an arbitrary subspace $U_{0}\in \mathsf{Gr}(\bF_{2}^{E},\ell)$. Sample
\[
(E',A_{\inn},\emp,U)\sim\cJ_{\midd}(E,q_{1},0,\beta,\ell)
\]
conditioned on the event $\{A_{\inn}=A\}\cap\{U=U_{0}\}$. Then for any subset $S\subseteq E$ with $|S|\leq 2$, the probability that $S\subseteq E'$ is at least $1-4^{q_{1}+\ell+1}\beta$.
\end{proposition}
\begin{proof}
According to~\Cref{def:Joint_dist_S_A_Aprime_U}, when $(E',A_{\inn},A_{\midd},U)$ is sampled from $\cJ_{\midd}(E,q_{1},0,\beta,\ell)$, a vector sequence $A'\in (\bF_{2}^{E'})^{\ell}$ is generated, and then $U$ is taken to be $\operatorname{span}(A')$. Let
\[
\cB=\left\{\bigl(a^{(1)},\dots,a^{(\ell)}\bigr)\in (\bF_{2}^{E})^{\ell}\,\middle|\,\operatorname{span}\bigl(a^{(1)},\dots,a^{(\ell)}\bigr)=U_{0}\right\}
\]
be the collection of bases of $U_{0}$. The event $\{U=U_{0}\}$ is thus equal to the disjoint union of events $\bigcup_{B\in \cB}\{A'=B\}$. It suffices to prove that for each $B\in \cB$, conditioned on the event $\{A_{\inn}=A\}\cap \{A'=B\}$, we have $\Pr{S\subseteq E'}\geq 1-4^{q_{1}+\ell+1}\beta$ for any two-element subset $S\subseteq E$. It is easy to see that the distribution of $E'$ conditioned on $\{A_{\inn}=A\}\cap \{A'=B\}$ is the same as in the following process (by comparing~\Cref{def:Joint_dist_S_A_Aprime,def:Joint_dist_S_A_Aprime_U}):
\begin{itemize}
\item sample $(E',A_{\inn},A_{\midd})\sim \cJ_{\outt}(E,q_{1},\ell,\beta)$ conditioned on $\{A_{\inn}=A\}\cap\{A_{\midd}=B\}$.
\end{itemize}
The desired conclusion thus easily follows from~\Cref{prop:covering_outer}. 
\end{proof}

The following notation will be used later in~\Cref{subsubsec:inner_PCP_covering}. Here the subscript ``$\mathsf{rev}$'' stands for ``reverse,'' meaning that we reversely generate a thinning $E'$ after seeing a subspace $U$ contained in $\bF_{2}^{E'}$.

\begin{definition}\label{def:Joint_reverse}
Let $E$ be a tripled set. Fix a parameter $\beta\in [0,1)$ and a nonnegative integer $\ell$. Given any subspace $U_{0}\subseteq \bF_{2}^{E}$ with dimension at most $\ell$, we let $\cJ_{\rev}(E,U_{0},\beta,\ell)$ be the output distribution of the following process:
\begin{enumerate}
\item Sample $(E',\emp,\emp,U)\sim \cJ_{\midd}(E,0,0,\beta,\ell)$ conditioned on the event $\{U=U_{0}\}$.
\item Output $E'$.
\end{enumerate}
\end{definition}

\subsection{The Local Tensor Check Game}\label{subsec:Game2a}
We now construct the first version of our middle PCP. The 2-prover 1-round game constructed in the next definition is similar to $\GameAb$ but with one important distinction. Letting $f$ be Alice's response, instead of requiring $\mathrm{rank}(f)$ is small, we require that $\mathrm{rank}(f|_{U\otimes U})$ is small, where $U$ is (roughly) a random $\ell$-dimensional subspace in Alice's space $\bF_{2}^{E}$ which is given to Alice as part of her challenge. Additionally, as part of his challenge, we give Bob a random subspace $U'$ of dimension (roughly) $\ell+1$ that contains $U$.

\begin{definition}\label{def:Game2}
Given positive integers $k,\ell_{2}$, nonnegative integers $q_{1},q_{2}$ such that $\ell_{2}\geq q_{2}$, a parameter $\beta\in [0,1)$ and a 3-uniform hypergraph $H$, let $\GameBa[H,k,\ell_{2},q_{1},q_{2},\beta]$ be the $2$-Prover-$1$-Round game where the verifier executes the following:
\begin{enumerate}
\item Sample a uniformly random tripled set $E\in \scrE(H,k)$.
\item Sample $(E',A_{\inn},A_{\midd},U)\sim \cJ_{\midd}(E,q_{1},q_{2},\beta,\ell_{2})$. Sample a random vector $x\in \bF_{2}^{E'}$ and let $U'=U+\operatorname{span}(x)$. 
\item Send $(E,A_{\inn},U)$ to Alice and $(E',A_{\inn},U')$ to Bob.
\item Alice must return an \textup{NAE}-satisfying symmetric bilinear form $f:\bF_{2}^{E}\otimes \bF_{2}^{E}\rightarrow\bF_{2}$ such that $\mathrm{rank}(f|_{U\otimes U})\leq 1$; Bob must return a bilinear form $g: \bF_{2}^{E'}\otimes \bF_{2}^{E'}\rightarrow \bF_{2}$. Accept if
\[
\left.f\right|_{\bF_{2}^{E'}\otimes \bF_{2}^{E'}}=g.
\]
\end{enumerate}
\end{definition}

\begin{remark}
In~\Cref{def:Game2}, the verifier does not use $A_{\midd}$ after generating it. Therefore, the distribution of $(E,A_{\inn},U,U')$ does not depend on $q_{2}$. Consequently, the value of the game $\GameBa[H,k,\ell_{2},q_{1},q_{2},\beta]$ does not depend on $q_{2}$ either. The purpose of $A_{\midd}$ in~\Cref{def:Game2} is solely for the convenience of its soundness analysis in~\Cref{subsec:middle_soundness}.
\end{remark}

\begin{remark}\label{rem:middle_Uprime}
In~\Cref{def:Game2}, it is easy to see that the distribution of $(E',A_{\inn},A_{\midd},U')$ is identical to $\cJ_{\midd}(E,q_{1},q_{2},\beta,\ell_{2}+1)$.
\end{remark}

The completeness of $\GameBa$ is straightforward and is given in the following lemma.
\begin{lemma}[Game2a completeness]\label{lem:Game2a_completeness}
For any positive integers $k,\ell_{2}$, nonnegative integers $q_{1},q_{2}$ such that $k\geq \ell_{2}\geq q_{2}$, parameter $\beta \in [0,1)$ and 2-colorable 3-uniform hypergraph $H$, we have
\[
\operatorname{val}\Bigl(\GameBa[H,k,\ell_{2},q_{1},q_{2},\beta]\Bigr)=1.
\]
\end{lemma}
\begin{proof}
It is clear that the provers can ignore the subspaces $U$ or $U'$ given to them and pretend that they are in the game $\GameAb[H,k,q_{1},0,1,\beta]$. Since there is a strategy for $\GameAb[H,k,q_{1},0,1,\beta]$ with 0 rejection probability due to~\Cref{lem:Game1b_completeness}, it follows that the same strategy also passes $\GameBa[H,k,\ell_{2},q_{1},q_{2},\beta]$ with probability 1.  
\end{proof}
The soundness analysis of $\GameBa$, presented in the next subsection, is much more involved. Intuitively, if we did not supply Bob with the space $U'$, the analysis would be relatively straightforward. Indeed, in that case, while in principle Alice's response could depend on the space $U$ she receives, in actuality such strategies do not achieve a high value, as they have to be consistent with Bob (whose response does not depend on $U$ at all). Thus, in that case Alice's response can be thought of as not being dependent on $U$, and hence the tensor test on $U$ is random and faithfully checks whether Alice's response has low-rank globally. 

To make our final game $4$-to-$1$ we must supply Bob with the space $U'$ (otherwise we would not be able to compose with the inner PCP), and the soundness analysis in that case becomes more complicated. Indeed, there are genuine cheating strategies that Alice and Bob can use, and our argument below analyzes them (using the soundness of the Grassmann test) and shows that they can still be converted into good strategies for $\GameAb$. 
\subsection{Soundness Analysis}\label{subsec:middle_soundness}

In this subsection, we fix a 3-uniform hypergraph $H$, positive integers $k,\ell_{2}$, nonnegative integers $q_{1},q_{2}$ and a parameter $\beta\in (0,1)$. We consider the game $\GameBa[H,k,\ell_{2},q_{1},q_{2},\beta]$ as defined in~\Cref{def:Game2}. The following notational convention will be used throughout~\Cref{subsec:middle_soundness} (but \emph{not} in later parts of the paper).

\begin{notation}\label[ntn]{ntn:D_mid}
We use $\cD_{\midd}$ to denote the joint distribution of $(E,A_{\inn},A_{\midd},U,E',U')$ in the game $\GameBa[H,k,\ell_{2},q_{1},q_{2},\beta]$, where $(E,A_{\inn},U)$ is the question sent to Alice and $(E',A_{\inn},U')$ is the question sent to Bob. We use the symbol \(*\) to denote
marginalization over the corresponding component. For instance, sampling from
the marginal distribution of \((A_{\inn},U,E')\) under \(\cD_{\midd}\) is denoted by $(*,A_{\inn},*,U,E',*)\sim \cD_{\midd}$.
\end{notation}

\begin{notation}\label[ntn]{ntn:D_mid_condition}
We will often condition the distribution $\cD_{\midd}$ on some subset of the six components
\((E,A_{\inn},A_{\midd},U,E',U')\) and consider the induced joint distribution of the remaining
components. For instance, if $E$ and $U'$ are fixed, then the conditional joint distribution of $(U,E')$ is denoted by
\[
\Pru{(E,*,*,U,E',U')\sim \cD_{\midd}}{\cdot\,\big|\, E,U'}.
\]
\end{notation}

\begin{notation}\label{ntn:middle_Alice-table}
A table $F[\cdot]$ is called an Alice-table if for each possible Alice-question $(E,A_{\inn},U)$ under $\cD_{\midd}$, the entry $F[E,A_{\inn},U]$ is either $\nil$ or an \textup{NAE}-satisfying symmetric bilinear form \[F[E,A_{\inn},U]:\bF_{2}^{E}\otimes \bF_{2}^{E}\rightarrow \bF_{2}\qquad\text{such that}\qquad\operatorname{rank}\left(\left.F[E,A_{\inn},U]\right|_{U\otimes U}\right)\leq 1.\]
Similarly, a table $G[\cdot]$ is called a Bob-table if for each possible Bob-question $(E',A_{\inn},U')$ under $\cD_{\midd}$, the entry $G[E',A_{\inn},U']$ is either $\nil$ or a bilinear form
\[G[E',A_{\inn},U']:\bF_{2}^{E'}\otimes \bF_{2}^{E'}\rightarrow \bF_{2}.\]
\end{notation}

In the following notation, the function $\operatorname{Vrf}_{F,G}(\cdot,\cdot,\cdot,\cdot,\cdot)$ is the indicator of whether the verifier accepts under the strategy given by $F[\cdot]$ and $G[\cdot]$. Since $A_{\midd}$ is irrelevant to the verifier's acceptance, the $\operatorname{Vrf}_{F,G}(\cdot,\cdot,\cdot,\cdot,\cdot)$ only has five arguments instead of six.

\begin{notation}\label{ntn:middle_vrf}
Fix an Alice-table $F[\cdot]$ and a Bob-table $G[\cdot]$. For any tuple $(E,A_{\inn},*,U,E',U')$ with nonzero probability under $\cD_{\midd}$, we define
\[
\operatorname{Vrf}_{F,G}(E,A_{\inn},U,E',U')=\begin{cases}
1,&\text{if } \left.F[E,A_{\inn},U]\right|_{\bF_{2}^{E'}\otimes \bF_{2}^{E'}}=G[E',A_{\inn},U'],\\
0,&\text{otherwise.}
\end{cases}
\]
We will often sample \((E,A_{\inn},*,U,E',U')\) from \(\cD_{\midd}\) conditioned on some
subset of its components. For instance, if \(E\) and \(U'\) are fixed, then
we denote the corresponding conditional expectation of
\(\operatorname{Vrf}_{F,G}(E,A_{\inn},U,E',U')\) by
\[
\operatorname{Vrf}_{F,G}(E,*,*,*,U')
=
\Pru{(E,A_{\inn},*,U,E',U')\sim \cD_{\midd}}
{
\left. F[E,A_{\inn},U]\right|_{\mathbb F_{2}^{E'}\otimes \mathbb F_{2}^{E'}}
=
G[E',A_{\inn},U']
\,\middle|\, E,U'}.
\]
When none of the components is fixed, we write $\operatorname{Vrf}_{F,G}=\operatorname{Vrf}_{F,G}(*,*,*,*,*)$ for short.
\end{notation}

\subsubsection{Preparing for Grassmann Decoding}

The core of the soundness analysis of $\GameBa$ is to use the Grassmann decoding theorem (\Cref{thm:grdecoding2}). In order to make use of that theorem, we first need to convert the agreement between Alice's and Bob's tables into the condition that Alice's table \emph{alone} passes the ``Grassmann equality test'' (see \Cref{eq:Grassmann_equality_test}) with noticeable probability. This is achieved by the following lemma. 

\begin{lemma}\label{lem:preparing_middle}
Assume~\Cref{cond:parameters}. Let $F[\cdot]$ be an Alice-table and $G[\cdot]$ be a Bob-table. For any $E\in \scrE(H,k)$ and any $A_{\inn}\in (\bF_{2}^{E})^{q_{1}}$, we have
\[
\Pru{(U_{1},U_{2})\sim \cG(\bF_{2}^{E},\ell_{2})}{F[E,A_{\inn},U_{1}]=F[E,A_{\inn},U_{2}]\neq \nil\Big.}\geq \operatorname{Vrf}_{F,G}(E,A_{\inn},*,*,*)^{2}-2^{-\ell_{2}}.
\]
\end{lemma}

\begin{proof} We organize the proof into the following three steps:

\paragraph{Step 1: fixing $E'$ and $U'$.} Fix a subset $E'\subseteq E$ with size $|E'|\geq \ell_{2}+1$ and a subspace $U'\in \mathsf{Gr}(\bF_{2}^{E'},\ell_{2}+1)$ such that the tuple $(E,A_{\inn},*,*,E',U')$ has nonzero probability under $\cD_{\midd}$. By~\Cref{def:Joint_dist_S_A_Aprime_U,def:Game2}, it is easy to see that sampling $(E,A_{\inn},*,U,E',U')\sim \cD_{\midd}$ conditioned on $E$,~$A_{\inn}$,~$E'$ and $U'$ is equivalent to sampling a uniformly random subspace $U\in \mathsf{Gr}(U',\ell_{2})$. Therefore, we have
\begin{align}
\operatorname{Vrf}_{F,G}(E,A_{\inn},*,E',U')^{2}
&=\Exu{(E,A_{\inn},*,U,E',U')\sim\cD_{\midd}}{\operatorname{Vrf}_{F,G}(E,A_{\inn},U,E',U')\,\Big|\,E,A_{\inn},E',U'}^{2}\nonumber\\
&=\Exu{U_{1},U_{2}\in \mathsf{Gr}(U',\ell_{2})}{\operatorname{Vrf}_{F,G}(E,A_{\inn},U_{1},E',U')\operatorname{Vrf}_{F,G}(E,A_{\inn},U_{2},E',U')}.\label{eq:preparing_middle_UfromUprime}
\end{align}
Expanding the right-hand side of~\Cref{eq:preparing_middle_UfromUprime} by definition, we get
\begin{align}
&\quad\operatorname{Vrf}_{F,G}(E,A_{\inn},*,E',U')^{2}\nonumber\\
&=\Pru{U_{1},U_{2}\in \mathsf{Gr}(U',\ell_{2})}{\left.F[E,A_{\inn},U_{1}]\right|_{\bF_{2}^{E'}\otimes \bF_{2}^{E'}}=G[E',A_{\inn},U']=\left.F[E,A_{\inn},U_{2}]\right|_{\bF_{2}^{E'}\otimes \bF_{2}^{E'}}}\nonumber\\
&\leq \Pru{U_{1},U_{2}\in \mathsf{Gr}(U',\ell_{2})}{\left.F[E,A_{\inn},U_{1}]\right|_{\bF_{2}^{E'}\otimes \bF_{2}^{E'}}=\left.F[E,A_{\inn},U_{2}]\right|_{\bF_{2}^{E'}\otimes \bF_{2}^{E'}}}.\label{eq:preparing_middle_res_on_Eprime}
\end{align}

\paragraph{Step 2: fixing only $U'$.} For any fixed $U'\in\mathsf{Gr}(\bF_{2}^{E},\ell_{2}+1)$ and $U_{1},U_{2}\in \mathsf{Gr}(U',\ell_{2})$ such that $F[E,A_{\inn},U_{1}]\neq F[E,A_{\inn},U_{2}]$, one of the following two possibilities holds:
\begin{enumerate}[label=(\arabic*)]
\item At least one of $F[E,A_{\inn},U_{1}]$ and $F[E,A_{\inn},U_{2}]$ is $\nil$.
\item Both of $F[E,A_{\inn},U_{1}]$ and $F[E,A_{\inn},U_{2}]$ are bilinear forms $\bF_{2}^{E}\otimes \bF_{2}^{E}\rightarrow \bF_{2}$, and their difference $f=F[E,A_{\inn},U_{1}]-F[E,A_{\inn},U_{2}]$ is nonzero. In this case, there exist two unit vectors $x,y\in \bF_{2}^{E}$ (i.e. $|\textup{supp}(x)|=|\textup{supp}(y)|=1$) such that $f(x\otimes y)=1$. Therefore, for any subset $E'\subseteq E$ that contains $\textup{supp}(x)\cup \textup{supp}(y)$, we have $\left.F[E,A_{\inn},U_{1}]\right|_{\bF_{2}^{E'}\otimes \bF_{2}^{E'}}\neq \left.F[E,A_{\inn},U_{2}]\right|_{\bF_{2}^{E'}\otimes \bF_{2}^{E'}}$.
\end{enumerate}
In either of the above two cases, it follows from~\Cref{prop:covering_middle_2} and~\Cref{rem:middle_Uprime} that (note that $U',U_{1}$ and $U_{2}$ remain fixed in the following inequality)
\[
\Pru{(E,A_{\inn},*,*,E',U')\sim \cD_{\midd}}{\left.F[E,A_{\inn},U_{1}]\right|_{\bF_{2}^{E'}\otimes \bF_{2}^{E'}}=\left.F[E,A_{\inn},U_{2}]\right|_{\bF_{2}^{E'}\otimes \bF_{2}^{E'}}\,\Big|\,E,A_{\inn},U'}\leq 4^{q_{1}+\ell_{2}+2}\beta.
\]
Therefore, taking expectation of~\Cref{eq:preparing_middle_res_on_Eprime} over the conditional distribution of $E'$ yields (note that $U_{1}$ and $U_{2}$ are no longer fixed in the following inequality)
\begin{align}
&\quad\Exu{(E,A_{\inn},*,*,E',U')\sim \cD_{\midd}}{\operatorname{Vrf}_{F,G}(E,A_{\inn},*,E',U')^{2}\,\Big|\, E,A_{\inn},U'}\nonumber\\
&\leq \Pru{U_{1},U_{2}\in \mathsf{Gr}(U',\ell_{2})}{F[E,A_{\inn},U_{1}]=F[E,A_{\inn},U_{2}]\neq \nil\Big.}+4^{q_{1}+\ell_{2}+2}\beta.\label{eq:preparing_middle_smoothness_2}
\end{align}

\paragraph{Step 3: both $E'$ and $U'$ unfixed.}  

By \Cref{prop:covering_middle} and~\Cref{rem:middle_Uprime}, we know that conditioned on the fixed $E$ and $A_{\inn}$, the distribution of $U'$ in the sample $(E,A_{\inn},*,*,*,U')\sim \cD_{\midd}$ is \(\bigl(4^{q_{1}+\ell_{2}+1}\beta\sqrt{k}+2^{\ell_{2}+1-3k}\bigr)\)-close to uniform over $\mathsf{Gr}(\bF_{2}^{E},\ell_{2}+1)$. Therefore, we have
\begin{align*}
&\quad\operatorname{Vrf}_{F,G}(E,A_{\inn},*,*,*)^{2}\nonumber\\
&= \Exu{(E,A_{\inn},*,*,E',U')\sim \cD_{\midd}}{\operatorname{Vrf}_{F,G}(E,A_{\inn},*,E',U')\,\Big|\,E,A_{\inn}}^{2}\tag{by definition}\nonumber\\
&\leq 
\Exu{(E,A_{\inn},*,*,E',U')\sim \cD_{\midd}}{\operatorname{Vrf}_{F,G}(E,A_{\inn},*,E',U')^{2}\,\Big|\,E,A_{\inn}}\tag{by Cauchy-Schwarz}\nonumber\\
&\leq \Exu{U'\in \mathsf{Gr}(\bF_{2}^{E},\ell_{2}+1)}{\Exu{(E,A_{\inn},*,*,E',U')\sim \cD_{\midd}}{\operatorname{Vrf}_{F,G}(E,A_{\inn},*,E',U')^{2}\,\Big|\,E,A_{\inn},U'}}+\\
&\qquad 4^{q_{1}+\ell_{2}+1}\beta\sqrt{k}+2^{\ell_{2}+1-3k}.
\end{align*}
We then apply~\Cref{eq:preparing_middle_smoothness_2} to the above and get
\begin{align}
\operatorname{Vrf}_{F,G}(E,A_{\inn},*,*,*)^{2}
&\leq \Exu{U'\in \mathsf{Gr}(\bF_{2}^{E},\ell_{2}+1)}{\Pru{U_{1},U_{2}\in \mathsf{Gr}(U',\ell_{2})}{F[E,A_{\inn},U_{1}]=F[E,A_{\inn},U_{2}]\neq \nil\Big.}}+\nonumber\\
&\qquad 4^{q_{1}+\ell_{2}+1}\beta\left(4+\sqrt{k}\right)+2^{\ell_{2}+1-3k}.\label{eq:preparing_middle_smoothness}
\end{align}

For any fixed $U'\in \mathsf{Gr}(\bF_{2}^{E},\ell_{2}+1)$, two independent samples $U_{1},U_{2}\in\mathsf{Gr}(U',\ell_{2})$ are equal to each other with probability $(2^{\ell_{2}+1}-1)^{-1}$. Furthermore, if $U'\in \mathsf{Gr}(\bF_{2}^{E},\ell_{2}+1)$ is sampled uniformly at random and $U_{1},U_{2}\in\mathsf{Gr}(U',\ell_{2})$ are sampled uniformly conditioned on $U_{1}\neq U_{2}$, the distribution of $(U_{1},U_{2})$ is uniform over the edges of the Grassmann graph $\mathsf{Gr}(\bF_{2}^{E},\ell_{2})$. Therefore, \Cref{eq:preparing_middle_smoothness} implies
\begin{align*}
\operatorname{Vrf}_{F,G}(E,A_{\inn},*,*,*)^{2}
&\leq \Pru{(U_{1},U_{2})\sim \cG(\bF_{2}^{E},\ell_{2})}{F[E,A_{\inn},U_{1}]=F[E,A_{\inn},U_{2}]\neq \nil \Big.}+\frac{1}{2^{\ell_{2}+1}-1}\\
& +4^{q_{1}+\ell_{2}+1}\beta\left(4+\sqrt{k}\right)+2^{\ell_{2}+1-3k}\\
&\leq \Pru{(U_{1},U_{2})\sim \cG(\bF_{2}^{E},\ell_{2})}{F[E,A_{\inn},U_{1}]=F[E,A_{\inn},U_{2}]\neq \nil\Big.}+2^{-\ell_{2}},
\end{align*}
where we used~\Cref{cond:parameters} in the last transition.
\end{proof}

\subsubsection{Invoking the Grassmann Decoding Theorem}

We are now ready to apply the Grassmann decoding theorem to Alice's table. The advice $A_{\midd}$, which played no real role in the definition of $\GameBa$, helps us state the following lemma.

\begin{lemma}\label{lem:invoked_middle}
Assume~\Cref{cond:parameters}. Let $F[\cdot]$ be an Alice-table and $G[\cdot]$ be a Bob-table. For any $E\in \scrE(H,k)$ and any $A_{\inn}\in (\bF_{2}^{E})^{q_{1}}$ such that $\operatorname{Vrf}_{F,G}(E,A_{\inn},*,*,*)\geq 2^{-q_{2}/3}$, the following holds: with probability at least $4^{-\ell_{2}^{2}}$ over $(E,A_{\inn},A,*,*,*)\sim \cD_{\midd}$ conditioned on $E$ and $A_{\inn}$, there exists an \textup{NAE}-satisfying symmetric bilinear map $f_{A}:\bF_{2}^{E}\otimes \bF_{2}^{E}\rightarrow \bF_{2}$ of rank at most $2q_{2}+1$ satisfying
\[
\Pru{(E,A_{\inn},A,U,*,*)\sim \cD_{\midd}}{F[E,A_{\inn},U]=f_{A}\,\Big|\,E,A_{\inn},A}\geq 2^{-\ell_{2}^{2}}.
\]
\end{lemma}

\begin{proof}
We organize the proof into the following three steps.

\paragraph{Step 1: Grassmann decoding.} Since $\operatorname{Vrf}_{F,G}(E,A_{\inn},*,*,*)\geq 2^{-q_{2}/3}$,~\Cref{lem:preparing_middle} implies
\[
\Pru{(U_{1},U_{2})\sim \cG(\bF_{2}^{E},\ell_{2})}{F[E,A_{\inn},U_{1}]=F[E,A_{\inn},U_{2}]\neq \nil\Big.}\geq 2^{-2q_{2}/3}-2^{-\ell_{2}}\geq 2^{-q_{2}+4},
\]
where we used~\Cref{cond:parameters} in the last transition. Applying~\Cref{thm:grdecoding2} to the table $F[E,A_{\inn},\cdot]$ yields\footnote{We apply~\Cref{thm:grdecoding2} with $\Sigma$ equal to the collection of \textup{NAE}-satisfying symmetric bilinear maps from $\bF_{2}^{E}\otimes \bF_{2}^{E}$ to $\bF_{2}$.} a collection of subspaces $\cQ\subseteq \mathsf{Gr}(\bF_{2}^{E},q_{2})$ such that:
\begin{enumerate}[label=(\arabic*)]
\item $|\cQ|\geq 2^{-q_{2}-\ell_{2}^{2}}\left|\mathsf{Gr}(\bF_{2}^{E},q_{2})\right|$;
\item for each $Q\in \cQ$, there exists a subspace $W_{Q}\in \mathsf{Gr}(\bF_{2}^{E},3k-q_{2})$ and an \textup{NAE}-satisfying symmetric bilinear map $h_{Q}:\bF_{2}^{E}\otimes \bF_{2}^{E}\rightarrow \bF_{2}$ satisfying
\begin{equation}\label{eq:invoking_middle_zoomout_1}
\Pru{U\in \mathsf{Gr}(\bF_{2}^{E},\ell_{2})}{F[E,A_{\inn},U]=h_{Q}\,\Big|\,Q\subseteq U\subseteq W_{Q}}\geq 2^{-Cq_{2}^{2}},
\end{equation}
where $C$ is an absolute constant. For any $U\in \mathsf{Gr}(\bF_{2}^{E},\ell_{2})$ such that $F[E,A_{\inn},U]=h_{Q}$, by definition (see~\Cref{ntn:middle_Alice-table}) we must have
\(
\operatorname{rank}(h_{Q}|_{U\otimes U})\leq 1\). It thus follows from~\Cref{eq:invoking_middle_zoomout_1} that
\[
\Pru{U\in \mathsf{Gr}(\bF_{2}^{E},\ell_{2})}{\operatorname{rank}(h_{Q}|_{U\otimes U})\leq 1\,\big|\,Q\subseteq U\subseteq W_{Q}}\geq 2^{-Cq_{2}^{2}}>2^{4-(\ell_{2}-q_{2})/2},
\]
where we used~\Cref{cond:parameters} in the last transition. By~\Cref{lem:rk1bound}, we must have $\operatorname{rank}(h_{Q})\leq 2q_{2}+1$.
\end{enumerate}

For each $Q\in \cQ$, since $W_{Q}$ has codimension $q_{2}$ with respect to $\bF_{2}^{E}$, we have
\begin{equation}\label{eq:invoking_middle_zoomout_2}
\Pru{U\in \mathsf{Gr}(\bF_{2}^{E},\ell_{2})}{U\subseteq W_{Q}\,\big|\,Q\subseteq U}=\qbinom{3k-2q_{2}}{\ell_{2}-q_{2}}\qbinom{3k-q_{2}}{\ell_{2}-q_{2}}^{-1}\geq 2^{-\ell_{2}(q_{2}+1)}.
\end{equation}
Combining \Cref{eq:invoking_middle_zoomout_1} and~\Cref{eq:invoking_middle_zoomout_2} yields
\begin{equation}\label{eq:invoking_middle_zoomin}
\Pru{U\in \mathsf{Gr}(\bF_{2}^{E},\ell_{2})}{F[E,A_{\inn},U]=h_{Q}\,\Big|\,Q\subseteq U}\geq 2^{-Cq_{2}^{2}-\ell_{2}(q_{2}+1)}\geq 2^{1-\ell_{2}^{2}},
\end{equation}
where we used~\Cref{cond:parameters} in the last transition.

\paragraph{Step 2: good $Q$ translates to good $A$.} Now for any $A\in (\bF_{2}^{E})^{q_{2}}$, we let 
\[f_{A}=\begin{cases}
h_{\operatorname{span}(A)},&\text{if }\dim\operatorname{span}(A)=q_{2}\text{ and }\operatorname{span}(A)\in \cQ,\\
0,&\text{otherwise}.
\end{cases}\]
If $\dim\operatorname{span}(A)=q_{2}$ and $\operatorname{span}(A)\in \cQ$, then by~\Cref{eq:invoking_middle_zoomin} we have
\begin{equation}\label{eq:invoking_middle_zoomin_2}
\Pru{U\in \mathsf{Gr}(\bF_{2}^{E},\ell_{2})}{F[E,A_{\inn},U]=f_{A}\,\Big|\,\operatorname{span}(A)\subseteq U}\geq 2^{1-\ell_{2}^{2}}.
\end{equation}
By~\Cref{prop:covering_middle}, we know that conditioned on the fixed $E$, $A_{\inn}$ and $A_{\midd}=A$, the distribution of $U$ in the sample $(E,A_{\inn},A,U,*,*)\sim \cD_{\midd}$ is \(\bigl(4^{q_{1}+q_{2}+\ell_{2}}\beta\sqrt{k}+2^{\ell_{2}-3k}\bigr)\)-close to uniform over the set \(\big\{U\in \mathsf{Gr}(\mathbb F_{2}^{E},\ell_{2})\,\big|\, \operatorname{span}(A)\subseteq U\big\}\). Therefore,~\Cref{eq:invoking_middle_zoomin_2} implies
\[
\Pru{(E,A_{\inn},A,U,*,*)\sim \cD_{\midd}}{F[E,A_{\inn},U]=f_{A}\,\Big|\,E,A_{\inn},A}\geq 2^{1-\ell_{2}^{2}}-4^{q_{1}+q_{2}+\ell_{2}}\beta\sqrt{k}-2^{\ell_{2}-3k}\geq 2^{-\ell_{2}^{2}},
\]
where we used~\Cref{cond:parameters} in the last transition.

\paragraph{Step 3: abundance of good $A$'s.} It remains to show that
\[
\Pru{(E,A_{\inn},A,*,*,*)\sim \cD_{\midd}}{\dim\operatorname{span}(A)=q_{2}\text{ and }\operatorname{span}(A)\in \cQ\,\Big|\,E,A_{\inn}}\geq 4^{-\ell_{2}^{2}}.
\]
By~\Cref{prop:covering_outer_2}, it suffices to show
\[
\Pru{A\in (\bF_{2}^{E})^{q_{2}}}{\dim\operatorname{span}(A)=q_{2}\text{ and }\operatorname{span}(A)\in \cQ\Big.}\geq 4^{-\ell_{2}^{2}}+4^{q_{1}+q_{2}}\beta\sqrt{k}.
\]
But the left-hand side of the above display is clearly equal to
\[
\Pru{A\in (\bF_{2}^{E})^{q_{2}}}{\dim \operatorname{span}(A)=q_{2}\Big.}\cdot\frac{|\cQ|}{\left|\mathsf{Gr}(\bF_{2}^{E},q_{2})\right|}\geq \frac{1}{4}\cdot 2^{-q_{2}-\ell_{2}^{2}}\geq 4^{-\ell_{2}^{2}}+4^{q_{1}+q_{2}}\beta\sqrt{k},
\]
where we used~\Cref{cond:parameters} in the last transition.
\end{proof}

\subsubsection{Deriving Strategy for the Outer PCP}
We now use the structural results regarding Alice's strategy to complete the soundness analysis of $\GameBa$. Before we do so, fixing an Alice-table and a Bob-table, we prune the Alice-table and only keep the responses on questions which have noticeable success probability. This will ensure that whatever functions we decode from Alice's table have noticeable consistency with Bob's table.

\begin{definition}\label{def:middle_pruning}
Let $\eta\in (0,1)$ be a threshold parameter. Given an Alice-table $F[\cdot]$ and a Bob-table $G[\cdot]$ for $\GameBa[H,k,\ell_{2},q_{1},q_{2},\beta]$, consider the Alice-table $F'[\cdot]$ defined by
\[
F'[E,A_{\inn},U]=\begin{cases}
F[E,A_{\inn},U],&\text{if }\operatorname{Vrf}_{F,G}(E,A_{\inn},U,*,*)\geq \eta,\\
\texttt{nil},&\text{otherwise}.
\end{cases}
\]
The pair of tables $(F',G)$ is called the \emph{$\eta$-pruning} of $(F,G)$.
\end{definition}

The following straightforward proposition shows that $\eta$-pruning drops the success probability of Alice and Bob by at most $\eta$. 
\begin{proposition}\label{prop:middle_pruning}
Let $\eta\in (0,1)$ be a threshold parameter. Given an Alice-table $F[\cdot]$ and a Bob-table $G[\cdot]$ for $\GameBa[H,k,\ell_{2},q_{1},q_{2},\beta]$, let $(F',G)$ be the $\eta$-pruning of $(F,G)$. Then
\[
\operatorname{Vrf}_{F',G}\geq \operatorname{Vrf}_{F,G}-\eta.\]
\end{proposition}
\begin{proof}
For any $E\in \scrE(H,k)$, any $A_{\inn}\in (\bF_{2}^{E})^{q_{1}}$ and any subspace $U\subseteq \bF_{2}^{E}$ with $\dim U\leq \ell_{2}$, from~\Cref{def:middle_pruning} it is easy to see that
\[
\operatorname{Vrf}_{F',G}(E,A_{\inn},U,*,*)\geq \operatorname{Vrf}_{F,G}(E,A_{\inn},U,*,*)-\eta.
\]
Taking expectation over $(E,A_{\inn},*,U,*,*)\sim \cD_{\midd}$ then yields the desired conclusion.
\end{proof}

We are now ready to state and prove the soundness of $\GameBa$.

\begin{lemma}[Game2a soundness]\label{lem:Game2a_soundness}
Assume~\Cref{cond:parameters}. Suppose that in every $4^{r}$-coloring of $H$, at least $\varepsilon$ fraction of the edges are monochromatic. Then we have
\[
\operatorname{val}\Bigl(\GameBa[H,k,\ell_{2},q_{1},q_{2},\beta]\Bigr)\leq 2^{-q_{2}/4}.
\]
\end{lemma}

\begin{proof}
Assume on the contrary that the value of $\GameBa[H,k,\ell_{2},q_{1},q_{2},\beta]$ is larger than $2^{-q_{2}/4}$. Then by definition, there are an Alice-table $F[\cdot]$ and a Bob-table $G[\cdot]$ for $\GameBa[H,k,\ell_{2},q_{1},q_{2},\beta]$ such that \(\operatorname{Vrf}_{F,G}>2^{-q_{2}/4}\). Let $(F',G)$ be the $(2^{1-q_{2}/3})$-pruning of $(F,G)$. By~\Cref{prop:middle_pruning} we have
\[
\operatorname{Vrf}_{F',G}\geq 2^{-q_{2}/4}-2^{1-q_{2}/3}\geq 2^{1-q_{2}/3},
\]
where we used $q_{2}\geq 24$ (as guaranteed by~\Cref{cond:parameters}).

\paragraph{Step 1: strategy for Game1b.} For each $E\in \scrE(H,k)$ and each possible sample $(E',A_{\inn},A_{\midd})$ from $\cJ_{\outt}(E,q_{1},q_{2},\beta)$, let $\cF[E,A_{\inn},A_{\midd}]$ be the collection of all \textup{NAE}-satisfying symmetric bilinear maps $f:\bF_{2}^{E}\otimes \bF_{2}^{E}\rightarrow \bF_{2}$ such that $\operatorname{rank}(f)\leq 2q_{2}+1$ and
\[
\Pru{(E,A_{\inn},A_{\midd},U,*,*)\sim \cD_{\midd}}{F'[E,A_{\inn},U]=f\,\Big|\,E,A_{\inn},A_{\midd}}\geq 2^{-\ell_{2}^{2}}.
\]
We now derive strategies for the provers in $\GameAb[H,k,q_{1},q_{2},r,\beta]$ from the tables $F'[\cdot]$ and $G[\cdot]$. The strategy for Alice is as follows:
\begin{enumerate}
\item Receive the question $(E,A_{\inn},A_{\midd})$. 
\item If $\cF[E,A_{\inn},A_{\midd}]$ is empty, return an arbitrary NAE-satisfying symmetric bilinear map $f:\bF_{2}^{E}\otimes \bF_{2}^{E}\rightarrow \bF_{2}$ of rank $1$ (see~\Cref{cor:NAE_rank1_exists}). 
\item If $\cF[E,A_{\inn},A_{\midd}]$ is nonempty, return an arbitrary NAE-satisfying symmetric bilinear map $f\in \cF[E,A_{\inn},A_{\midd}]$. By the definition of $\cF[E,A_{\inn},A_{\midd}]$, we have $\operatorname{rank}(f)\leq 2q_{2}+1=r$, where the last transition is due to~\Cref{cond:parameters}.
\end{enumerate}
The strategy for Bob is:
\begin{enumerate}
\item Receive the question $(E',A_{\inn},A_{\midd})$.
\item Conditioned on $E'$, $A_{\inn}$ and $A_{\midd}$, sample $(*,A_{\inn},A_{\midd},*,E',U')\sim \cD_{\midd}$.
\item Return the bilinear map $G[E',A_{\inn},U']$ (note that by definition, the table $G[\cdot]$ does not contain $\nil$ entries).
\end{enumerate}

\paragraph{Step 2: analysis of the strategy.} Since
\[
\Exu{(E,A_{\inn},*,*,*,*)\sim \cD_{\midd}}{\operatorname{Vrf}_{F',G}(E,A_{\inn},*,*,*)}=\operatorname{Vrf}_{F',G}\geq2^{1-q_{2}/3},
\]
we have\footnote{Here we use the fact that $\operatorname{Vrf}_{F',G}(E,A_{\inn},*,*,*)$ always lies in $[0,1]$.}
\begin{equation}\label{eq:middle_soundness_1}
\Pru{(E,A_{\inn},*,*,*,*)\sim \cD_{\midd}}{\operatorname{Vrf}_{F',G}(E,A_{\inn},*,*,*)\geq 2^{-q_{2}/3}}\geq 2^{-q_{2}/3}.
\end{equation}
For any $E\in \scrE(H,k)$ and $A_{\inn}\in (\bF_{2}^{E})^{q_{1}}$ such that $\operatorname{Vrf}_{F',G}(E,A_{\inn},*,*,*)\geq 2^{-q_{2}/3}$, we know from~\Cref{lem:invoked_middle} that
\[
\Pru{(E,A_{\inn},A_{\midd},*,*,*)\sim \cD_{\midd}}{\cF[E,A_{\inn},A_{\midd}]\neq \emptyset\,\Big|\,E,A_{\inn}}\geq 4^{-\ell_{2}^{2}}.
\]
Combining with~\Cref{eq:middle_soundness_1} yields
\begin{equation}\label{eq:middle_soundness_2}
\Pru{(E,A_{\inn},A_{\midd},*,*,*)\sim \cD_{\midd}}{\cF[E,A_{\inn},A_{\midd}]\neq \emptyset\Big.}\geq 2^{-2\ell_{2}^{2}-q_{2}/3}.
\end{equation}

Given a tuple $(E,A_{\inn},A_{\midd})$ and an element $f\in \cF[E,A_{\inn},A_{\midd}]$, we have
\begin{align*}
&\quad\Pru{(E,A_{\inn},A_{\midd},*,E',U')\sim \cD_{\midd}}{G[E',A_{\inn},U']=\left. f\right|_{\bF_{2}^{E'}\otimes \bF_{2}^{E'}}\,\middle|\, E,A_{\inn},A_{\midd}}\\
&\geq \Pru{(E,A_{\inn},A_{\midd},U,E',U')\sim \cD_{\midd}}{F'[E,A_{\inn},U]=f\text{ and }\operatorname{Vrf}_{F',G}(E,A_{\inn},U,E',U')=1\,\Big|\,E,A_{\inn},A_{\midd}}\\
&=\Exu{(E,A_{\inn},A_{\midd},U,*,*)\sim \cD_{\midd}}{\ind{F'[E,A_{\inn},U]=f}\cdot \operatorname{Vrf}_{F',G}(E,A_{\inn},U,*,*)\,\Big|\,E,A_{\inn},A_{\midd}}\\
&\geq 2^{1-q_{2}/3}\cdot \Pru{(E,A_{\inn},A_{\midd},U,*,*)\sim \cD_{\midd}}{F'[E,A_{\inn},U]=f\,\Big|\,E,A_{\inn},A_{\midd}}\tag{because $(F',G)$ is the $(2^{1-q_{2}/3})$-pruning of $(F,G)$}\\
&\geq 2^{1-\ell_{2}^{2}-q_{2}/3}.\tag{since $f\in \cF[E,A_{\inn},A_{\midd}]$}
\end{align*}
Importantly, conditioned on $E$, $A_{\inn}$, $A_{\midd}$ and $E'$, the distribution of $U'$ in a conditional sample $(E,A_{\inn},A_{\midd},*,E',U')\sim \cD_{\midd}$ does not depend on $E$. Therefore, the above display implies that under the strategies presented in Step 1 of the proof, for any tuple $(E,A_{\inn},A_{\midd})$ with any element $f\in \cF[E,A_{\inn},A_{\midd}]$, the following expectation is at least $2^{1-\ell_{2}^{2}-q_{2}/3}$:
\[
\Exu{(E,A_{\inn},A_{\midd},*,E',*)\sim \cD_{\midd}}{\Pru{\text{Bob}}{\text{Bob returns }\left.f\right|_{\bF_{2}^{E'}\otimes \bF_{2}^{E'}}\text{ on question }(E',A_{\inn},A_{\midd})}\,\middle|\,E,A_{\inn},A_{\midd}}.
\]
Combining with~\Cref{eq:middle_soundness_2}, this implies that\footnote{Here we use $\textup{Alice}(E,A_{\inn},A_{\midd})$ to denote Alice's answer to the question $(E,A_{\inn},A_{\midd})$, and $\textup{Bob}(E',A_{\inn},A_{\midd})$ denotes Bob's answer to $(E',A_{\inn},A_{\midd})$.} 
\[\Exu{(E,A_{\inn},A_{\midd},*,E',*)\sim \cD_{\midd}}{\Pru{\textup{Bob}}{\left.\textup{Alice}(E,A_{\inn},A_{\midd})\right|_{\bF_{2}^{E'}\otimes \bF_{2}^{E'}}=\textup{Bob}(E',A_{\inn},A_{\midd})}}\]
is at least
\[2^{-2\ell_{2}^{2}-q_{2}/3}\cdot 2^{1-\ell_{2}^{2}-q_{2}/3}> k^{-1/2},\]
where we used~\Cref{cond:parameters} in the last transition. Since the distribution of $(E,A_{\inn},A_{\midd},E')$ in $\GameAb[H,k,q_{1},q_{2},r,\beta]$ is identical to the distribution of the sample $(E,A_{\inn},A_{\midd},*,E',*)\sim \cD_{\midd}$, it follows that our prescribed strategies pass $\GameAb[H,k,q_{1},q_{2},r,\beta]$ with probability 
\[
\operatorname{val}\Bigl(\GameAb[H,k,q_{1},q_{2},r,\beta]\Bigr)>k^{-1/2},
\]
contradicting~\Cref{lem:Game1b_soundness}.
\end{proof}

\subsection{Lazy Parallel Repetition}\label{subsec:Game2b}
In this subsection, we apply a ``lazy parallel repetition'' to $\GameBa$ to get $\GameBb$. More precisely, we consider the following $2$-prover $1$-round game, in which Alice and Bob get $t$-blocks of challenges, one of which is a true challenge from $\GameBa$, whereas on each one of the rest of the blocks they receive the same randomly sampled challenge of Alice from $\GameBa$. Additionally, on each block the players receive a list of advice vectors.

\begin{definition}\label{def:Game2b}
Given positive integers $k,\ell_{2},t$, nonnegative integers $q_{0},q_{1}$, a parameter $\beta\in [0,1)$ and a $3$-uniform hypergraph $H$, let $\GameBb[H,k,\ell_{2},t,q_{0},q_{1},\beta]$ be a $2$-Prover $1$-Round game where the verifier executes the following:
\begin{enumerate}
\item For each $j\in [t]$, sample a uniformly random tripled set $E_{j}\in \scrE(H,k)$.
\item Sample a uniformly random index $c\in [t]$. Sample $(E_{c}',A_{c},\emp,U_{c})\sim \cJ_{\midd}(E_{c},q_{1},0,\beta,\ell_{2})$. Sample a random vector $x\in \bF_{2}^{E'_{c}}$ and let $U'_{c}=U_{c}+\operatorname{span}(x)$.
\item For each $j\in [t]\setminus \{c\}$, do the following:
\begin{itemize}
\item Sample a uniformly random advice sequence $A_{j}=\bigl(a^{(j,1)},\dots,a^{(j,q_{0})}\bigr)\in \bigl(\bF_{2}^{E_{j}}\otimes \bF_{2}^{E_{j}}\bigr)^{q_{0}}$, i.e. $a^{(j,i)}$ is a uniformly random vector in $\bF_{2}^{E_{j}}\otimes \bF_{2}^{E_{j}}$ for each $i\in [q_{0}]$. 
\item Sample $(E_{j}^{\mathrm{ghost}},\emp,\emp,U_{j})\sim \cJ_{\midd}(E_{j},0,0,\beta,\ell_{2})$.
\item Let $E_{j}'=E_{j}$ and $U_{j}'=U_{j}$.
\end{itemize}
\item Denote sequences
\begin{alignat*}{3}
\cE&=(E_{1},\dots,E_{t}),\qquad& \cU&=(U_{1},\dots,U_{t}),\qquad& \cA&=(A_{1},\dots,A_{t}),\qquad\text{and}\\
\cE'&=(E_{1}',\dots,E_{t}'),\qquad&\cU'&=(U_{1}',\dots,U_{t}').
\end{alignat*}
\item Send $(\cE,\cA,\cU)$ to Alice and $(c,\cE',\cA,\cU')$ to Bob.
\item Alice must return NAE-satisfying symmetric bilinear forms 
\[f_{j}:\bF_{2}^{E_{j}}\otimes\bF_2^{ E_{j}}\rightarrow \bF_{2}\text{ such that }\operatorname{rank}(f_{j}|_{U_{j}\otimes U_{j}})\leq 1\qquad\text{for each }j\in [t];\]
Bob must return bilinear forms 
\[
g_{j}:\bF_{2}^{E'_{j}}\otimes\bF_{2}^{E'_{j}}\rightarrow \bF_{2}\qquad\text{for each }j\in [t]. 
\]
Accept if $g_{j}$ is the restriction of $f_{j}$ to $\bF_{2}^{E'_{j}}\otimes\bF_{2}^{E'_{j}}$ for each $j\in [t]$.
\end{enumerate}
\end{definition}
If we ignore the lists $A_j$, the game we have above is exactly the lazy parallel repetition version of $\GameBa$. Lazy parallel repetition is well-known to preserve completeness and soundness (but not improve them), as well as to improve the ``smoothness'' of a game. We include the lists $A_j$ to enable the advice feature in the inner PCP, and we show that they do not hurt the soundness of the construction (they have no effect on other features of the construction). 

The completeness of $\GameBb$ is again straightforward:

\begin{lemma}[Game2b completeness]\label{lem:Game2b_completeness} 
For any positive integers $k,\ell_{2},t$ such that $k\geq \ell_{2}$, nonnegative integers $q_{0},q_{1}$, parameter $\beta\in [0,1)$ and 2-colorable 3-uniform hypergraph $H$, we have
\[
\operatorname{val}\Bigl(\GameBb[H,k,\ell_{2},t,q_{0},q_{1},\beta]\Bigr)=1.
\]
\end{lemma}
\begin{proof}
    We fix strategies from~\Cref{lem:Game2a_completeness} that achieve value $1$. Alice ignores the lists $A_j$ and answers each block $j\in [t]$ separately according to her strategy from $\GameBa$. Bob answers according to Alice's strategy for $\GameBa$ on the noncritical blocks $j\in [t]\setminus \{c\}$, and according to Bob's strategy for $\GameBa$ on the critical block $j=c$.
\end{proof}

The following lemma establishes the soundness of $\GameBb$.

\begin{lemma}[Game2b soundness]\label{lem:Game2b_soundness}
Assume~\Cref{cond:parameters}. Suppose that in every $4^{r}$-coloring of $H$, at least $\varepsilon$ fraction of the edges are monochromatic. Then we have
\[
\operatorname{val}\Bigl(\GameBb[H,k,\ell_{2},t,q_{0},q_{1},\beta]\Bigr)\leq 2^{-q_{2}/4}.
\]
\end{lemma}

\begin{proof}
It suffices to show that
\[
\operatorname{val}\Bigl(\GameBb[H,k,\ell_{2},t,q_{0},q_{1},\beta]\Bigr)\leq \operatorname{val}\Bigl(\GameBa[H,k,\ell_{2},q_{1},0,\beta]\Bigr),
\]
after which the conclusion follows from~\Cref{lem:Game2a_soundness} (combined with~\Cref{rem:Jmid_no_dependence_q_2}). 

Assume Alice and Bob have a strategy for $\GameBb[H,k,\ell_{2},t,q_{0},q_{1},\beta]$, and we next show how to derive a strategy for $\GameBa[H,k,\ell_{2},q_{1},0,\beta]$. In the latter game, Alice and Bob can first agree upon the following ``shared data'' using public randomness:
\begin{itemize}
\item a uniformly random critical index $c\in [t]$, and
\item a uniformly random tuple $(E_{j},A_{j},U_{j})$ for each non-critical index $j\in [t]\setminus \{c\}$. More precisely, sample $E_{j}$ from $\scrE(H,k)$ uniformly, sample $A_{j}$ from $\bigl(\bF_{2}^{E_{j}}\otimes \bF_{2}^{E_{j}}\bigr)^{q_{0}}$ uniformly, and sample $U_{j}$ by drawing $(E_{j}^{\mathrm{ghost}},\emp,\emp,U_{j})\sim \cJ_{\midd}(E_{j},0,0,\beta,\ell_{2})$.
\end{itemize}
Given a question $(E,A,U)$ in $\GameBa[H,k,\ell_{2},q_{1},0,\beta]$, Alice can construct $(\cE,\cA,\cU)$ by planting $(E,A,U)$ in the critical block $c$ and using the shared data for the non-critical blocks. Similarly, given a question $(E',A',U')$, Bob can construct $(\cE',\cA',\cU')$ by planting $(E',A',U')$ in the critical block $c$ and using the shared data for the non-critical blocks. The provers can then pretend that they are in $\GameBb[H,k,\ell_{2},t,q_{0},q_{1},\beta]$, and use their strategy for $\GameBb[H,k,\ell_{2},t,q_{0},q_{1},\beta]$ to answer the verifier in $\GameBa[H,k,\ell_{2},q_{1},0,\beta]$. Whenever the strategy for $\GameBb[H,k,\ell_{2},t,q_{0},q_{1},\beta]$ provides consistent answers to $(\cE,\cA,\cU)$ and $(c,\cE',\cA',\cU')$, the consistency in the critical block implies that the derived answers to the questions $(E,A,U)$ and $(E',A',U')$ in $\GameBa[H,k,\ell_{2},q_{1},0,\beta]$ are also consistent. Therefore, the verifier in $\GameBa[H,k,\ell_{2},q_{1},0,\beta]$ accepts with probability at least the acceptance probability in $\GameBb[H,k,\ell_{2},t,q_{0},q_{1},\beta]$.
\end{proof}

\section{The Inner PCP}\label{sec:inner_PCP}

In this section we construct and analyze our inner PCP as discussed in~\Cref{sec:inner_PCP_intro}. We first prepare some necessary framework in~\Cref{subsec:tripled_packet,subsec:low_rank_Grassmann}. The final sub-$4$-to-$1$ game, denoted by $\GameCa$, will then be defined in~\Cref{subsec:Game3}. The proof of the soundness of $\GameCa$ is rather long, and will be completed in~\Cref{subsec:inner_soundness}.

\subsection{Preparation: Tripled Packets}\label{subsec:tripled_packet}

In the inner PCP, we will add another layer of Grassmann-type structure onto $\GameBb$, which is already somewhat complicated. To make future notations more concise, in this subsection we set up some notational framework for the objects featured in $\GameBb$ (\Cref{def:Game2b}). 

Recall that in~\Cref{def:Game2b}, for each index $j\in [t]$, Alice needs to return a bilinear map $f:\bF_{2}^{E_{j}}\otimes \bF_{2}^{E_{j}}\rightarrow \mathbb{F}_2$ that is both symmetric and \text{NAE}-satisfying. The following definition combines the two conditions, exploiting the fact that both are ``linear'' conditions in $f$.

\begin{definition}
Let $E$ be a tripled set. For any $v\in E$, we let $e_{v}\in \bF_{2}^{E}$ denote the standard basis vector supported on $v$; that is, for $u\in E$ let $(e_{v})_{u}=1$ if and only if $u=v$. 
\begin{enumerate}[label=(\arabic*)]
\item We define
$\cN_{1}(E)\subseteq \bF_{2}^{E}\otimes \bF_{2}^{E}$ to be the linear subspace spanned by all vectors of the form
\[e_{v_{1}}\otimes e_{v_{1}}+e_{v_{2}}\otimes e_{v_{2}}+e_{v_{3}}\otimes e_{v_{3}}+e_{v_{1}}\otimes e_{v_{2}}+e_{v_{2}}\otimes e_{v_{3}}+e_{v_{3}}\otimes e_{v_{1}},\]
where $\{v_{1},v_{2},v_{3}\}$ ranges in $\operatorname{Triples}(E)$.\footnote{We assume that for each triple $\{v_{1},v_{2},v_{3}\}\in \operatorname{Triples}(E)$, there is a fixed ordering of the elements $v_{1},v_{2},v_{3}$.}
\item We let $\cN_{2}(E)\subseteq \bF_{2}^{E}\otimes \bF_{2}^{E}$ be the linear subspace spanned by all anti-symmetric vectors, i.e.
\[
\cN_{2}(E)=\operatorname{span}\left\{x\otimes y+y\otimes x\,\middle|\,x,y\in\bF_{2}^{E}\right\}.
\]
\item Let $\cN(E)=\cN_{1}(E)+\cN_{2}(E)$ be the sum of the two subspaces.
\item A linear map $f:\cN(E)\rightarrow \bF_{2}$ is said to be \emph{\textup{SymNAE}-satisfying} if we have:
\begin{itemize}
\item 
\(
f(e_{v_{1}}\otimes e_{v_{1}}+e_{v_{2}}\otimes e_{v_{2}}+e_{v_{3}}\otimes e_{v_{3}}+e_{v_{1}}\otimes e_{v_{2}}+e_{v_{2}}\otimes e_{v_{3}}+e_{v_{3}}\otimes e_{v_{1}})=1
\)
for all triples $\{v_{1},v_{2},v_{3}\}\in \operatorname{Triples}(E)$; and
\item $f(x\otimes y+y\otimes x)=0$ for all $x,y\in \bF_{2}^{E}$.
\end{itemize}
\end{enumerate}
\end{definition}
Intuitively, the space $\cN_{1}(E)$ from the first item above is the subspace of $\bF_{2}^{E}\otimes \bF_{2}^{E}$ on which the value of a bilinear form $f\colon \bF_{2}^{E}\otimes \bF_{2}^{E}\to\bF_2$ is fully determined if it is \textup{NAE}-satisfying. The space $\cN_{2}(E)$ from the second item above is the subspace of $\bF_{2}^{E}\otimes \bF_{2}^{E}$ on which the value of a bilinear form $f\colon \bF_{2}^{E}\otimes \bF_{2}^{E}\to\bF_2$ is fully determined if it is symmetric. Thus, overall, whether $f$ is both NAE-satisfying and symmetric is fully captured by its value on $\cN(E)$, and this is the definition of \textup{SymNAE}-satisfying. 

The following proposition asserts that there is always a linear map on $\cN(E)$ which is \textup{SymNAE}-satisfying.

\begin{proposition}\label{prop:SymNAE_uniqueness}
For any tripled set $E$, there is a unique linear map $f:\cN(E)\rightarrow \bF_{2}$ that is \textup{SymNAE}-satisfying.
\end{proposition}
\begin{proof}
    This is clear because $\cN_1(E)\cap\cN_2(E) = \{0\}$.
\end{proof}

In order to aggregate the spaces $\bF_{2}^{E_{j}'}\otimes \bF_{2}^{E_{j}'}$ in all blocks, we introduce the following notation. Here we use $\cE'$ instead of $\cE$ to emphasize that the sets $E_{1}',\dots,E_{t}'$ are not necessarily tripled sets.

\begin{definition}
For a sequence of (not necessarily tripled) finite sets $\cE'=(E_{1}',\dots,E_{t}')$, we let $\cV(\cE')$ denote the $\bF_{2}$-vector space
\[
\cV(\cE')=\bigoplus_{j=1}^{t}\left(\bF_{2}^{E_{j}'}\otimes \bF_{2}^{E_{j}'}\right).
\]
\end{definition}

\begin{remark}
Let $H$ be a 3-uniform hypergraph. For $E_{1},\dots,E_{t}\in \scrE(H,k)$ and $\cE=(E_{1},\dots,E_{t})$, the space $\cV(\cE)$ can be regarded as a subspace of the huge ambient space
\[
\left(\bF_{2}^{\scrV(H,k)}\otimes \bF_{2}^{\scrV(H,k)}\right)^{t}.
\]
\end{remark}

Next, we define subspaces of $\cV(\cE')$ that correspond to the tensor checks we impose on Alice or Bob in $\GameBb$.

\begin{definition}
Fix a sequence of (not necessarily tripled) finite sets $\cE'=(E_{1}',\dots,E_{t}')$ and a sequence of linear subspaces $\cU=(U_{1},\dots,U_{t})$ with $U_{j}\subseteq  \bF_{2}^{E_{j}'}$ for each $j\in [t]$. We let $\bigoplus \cU\subseteq \cV(\cE')$ be the linear subspace defined by
\[
\toplus \cU=\bigoplus_{j=1}^{t}\Bigl(U_{j}\otimes U_{j}\Bigr)\subseteq \bigoplus_{j=1}^{t}\Bigl(\bF_{2}^{E_{j}'}\otimes \bF_{2}^{E_{j}'}\Bigr)=\cV(\cE').
\]
A linear map $f:\toplus \cU\rightarrow \bF_{2}$ is said to be \emph{\textup{SR}-satisfying}\footnote{In this terminology, ``S'' stands for symmetric, and ``R'' stands for rank.} if for each $j\in [t]$, the restriction $f|_{U_{j}\otimes U_{j}}$ is a symmetric bilinear map of rank at most 1.
\end{definition}

Combining the side conditions, symmetry and tensor checks gives rise to the space $\cN(\cE,\cU)$ defined as follows.
\begin{definition}
A \emph{tripled packet} is a pair $(\cE,\cU)$, where $\cE=(E_{1},\dots,E_{t})$ is a sequence of tripled sets, and $\cU=(U_{1},\dots,U_{t})$ is a sequence of linear subspaces with $U_{j}\subseteq \bF_{2}^{E_{j}}$ for each $j\in [t]$. We let $\cN(\cE,\cU)\subseteq \cV(\cE)$ be the linear subspace defined by
\[
\cN(\cE,\cU)=\bigoplus_{j=1}^{t}\Bigl(\cN(E_{j})+U_{j}\otimes U_{j}\Bigr)\subseteq \bigoplus_{j=1}^{t}\Bigl(\bF_{2}^{E_{j}}\otimes \bF_{2}^{E_{j}}\Bigr)=\cV(\cE).
\]
\end{definition}

Thus, whether $f:\cV(\cE)\rightarrow \bF_{2}$ satisfies all side conditions, symmetry and tensor tests amounts to checking the following condition on $f|_{\cN(\cE,\cU)}$.

\begin{definition}\label{def:SNR_satisfying}
Given a tripled packet $(\cE,\cU)$ where $\cE=(E_{1},\dots,E_{t})$ and $\cU=(U_{1},\dots,U_{t})$, a linear map $f:\cN(\cE,\cU)\rightarrow \bF_{2}$ is said to be \emph{\textup{SNR}-satisfying}\footnote{In this terminology, ``S'' stands for symmetric, ``N'' stands for \textup{NAE}, and ``R'' stands for rank.} if we have:
\begin{itemize}
\item $\left.f\right|_{\cN(E_{j})}$ is \textup{SymNAE}-satisfying for each $j\in [t]$.
\item $\operatorname{rank}(f|_{U_{j}\otimes U_{j}})\leq 1$ for each $j\in [t]$.
\end{itemize}
\end{definition}

\subsubsection{Thinning of Tripled Packets}
Recall that in $\GameBb$, we plant the true challenge in one randomly chosen coordinate, and it plays a crucial role in our argument. We often refer to it as the \emph{critical index/coordinate/block}, and make use of the following notations.

\begin{definition}\label{def:critical_descendant}
Let $(E_{1},\dots,E_{t})$ be a sequence of tripled sets. Fix an index $c\in [t]$ and a parameter $\beta\in [0,1]$. We write $\cE'\sim \operatorname{Thin}_{c,\beta}(\cE)$ for the random sequence $\cE'=(E_{1}',\dots,E_{t}')$ obtained by setting $E_{j}'=E_{j}$ for all $j\in [t]\setminus\{c\}$ and sampling $E_{c}'\sim\operatorname{Thin}_{\beta}(E_{c})$. A sequence $\cE'$ in the support of the distribution $\operatorname{Thin}_{c,1/2}(\cE)$ is called a \emph{$c$-critical descendant} of $\cE$.
\end{definition}

The following definition captures the distribution of $(\cU,\cE',\cU')$ conditioned on $\cE$ in~\Cref{def:Game2b}. Here we do not take into consideration the advice list $\cA$, because it will not be used to construct our final sub-$4$-to-$1$ game (it is crucial however for the soundness analysis).

\begin{definition}\label{def:Joint_inner}
Let $\cE=(E_{1},\dots,E_{t})$ be a sequence of tripled sets. Fix a parameter $\beta\in [0,1]$ and a nonnegative integer $\ell$. We let $\cJ_{\inn}(\cE,\beta,\ell)$ denote the output distribution of the following process:
\begin{enumerate}
\item Sample a uniformly random index $c\in [t]$. Sample $(E_{c}',\emp,\emp,U_{c})\sim \cJ_{\midd}(E_{c},0,0,\beta,\ell)$. Sample a random vector $x\in \bF_{2}^{E_{c}'}$ and let $U_{c}'=U_{c}+\operatorname{span}(x)$.
\item For each $j\in [t]\setminus\{c\}$, sample $(E_{j}^{\mathrm{ghost}},\emp,\emp,U_{j})\sim \cJ_{\midd}(E_{j},0,0,\beta,\ell)$. Let $E_{j}'=E_{j}$ and $U_{j}'=U_{j}$.
\item Let $\cE'=(E_{1}',\dots,E_{t}')$, $\cU=(U_{1},\dots,U_{t})$ and $\cU'=(U_{1}',\dots,U_{t}')$.
\item Output the tuple $(c,\cU,\cE',\cU')$.
\end{enumerate}
\end{definition}

We will make use of the following simple observation:
\begin{proposition}\label{prop:Joint_inner_independence}
In the setting and notation of~\Cref{def:Joint_inner}, the random subspaces $U_{1},\dots,U_{t}$ are mutually independent. Furthermore, they are jointly independent of the random index $c$.
\end{proposition}
\begin{proof}
    Regardless of whether $j=c$ or $j\neq c$, the subspace $U_{j}$ is spanned by $\ell$ independent vectors supported on a random $\beta$-thinning of $E_{j}$.
\end{proof}

\subsection{Preparation: Low-Rank Grassmann}\label{subsec:low_rank_Grassmann}
As discussed in~\Cref{sec:inner_PCP_intro}, in the inner PCP we want to use linear subspaces of $\bF_{2}^{E_{c}}\otimes \bF_{2}^{E_{c}}$ that consist only of low-rank vectors. As stated in~\Cref{def:intro_LGr}, these linear subspaces will form a subset $\mathsf{LGr}(\bF_{2}^{E_{c}},\ell)$ of $\mathsf{Gr}(\bF_{2}^{E_{c}}\otimes \bF_{2}^{E_{c}},\ell)$, and the subgraph it induces from the Grassmann graph is called the \emph{low-rank Grassmann graph}. In this subsection, we give formal definitions of the low-rank Grassmann and prove a number of elementary facts about it. We will primarily be interested in the basic ``statistical'' properties of $\mathsf{LGr}(\bF_{2}^{E_{c}},\ell)$, such as its zoom-out probabilities being close to those of the full Grassmann graph:
\begin{equation}\label{eq:LGR_zoom_out_informal}
\Pru{L\in \mathsf{LGr}(V,\ell)}{L\subseteq W}\approx \Pru{L\in \mathsf{Gr}(V\otimes V,\ell)}{L\subseteq W},
\end{equation}
for any vector space $V$ and subspace $W\subseteq V\otimes V$ with small codimension. More sophisticated analysis of the low-rank Grassmann graph, such as its spectral expansion (see~\cite{Golowich23}), will not be used in this paper.

\subsubsection{Construction from Hadamard Code}
The definition of the low-rank Grassmann will use linear combinations of vectors given by the Hadamard matrix, defined as follows.
\begin{definition}
Let $\ell$ be a positive integer. For each $i\in \{0,1,\dots,2^{\ell}-1\}$ and $b\in \{0,1,\dots,\ell-1\}$, define $\Bit_{\ell}(i,b)\in \bF_{2}$ to be the mod-$2$ remainder of the integer $\lfloor i\cdot2^{-b}\rfloor$. Equivalently, the sequence 
\[\bigl(\Bit_{\ell}(i,\ell-1),\Bit_{\ell}(i,\ell-2),\dots,\Bit_{\ell}(i,0)\bigr)\]
is the length-$\ell$ binary representation of $i$, with the most significant bit listed first. 

For each integer $i\in \{0,1,\dots,2^{\ell}-1\}$ and vector $u\in \bF_{2}^{\ell}$, define
\[
\Had_{\ell}(u,i)=\sum_{b=1}^{\ell}u_{b}\Bit_{\ell}(i,\ell-b)\in \bF_{2},
\]
to be the ``inner product'' of $u$ with the binary representation of $i$ (viewed as a vector in $\bF_{2}^{\ell}$).
\end{definition}

We next define the span of a list of vectors according to the Hadamard matrix.
\begin{definition}\label{def:Lspan}
Let $\ell,s$ be positive integers, and let $q=2s(2^{\ell}-1)$. Suppose $E'$ is a finite set, and $A\in (\bF_{2}^{E'})^{q}$ is a sequence of vectors
\[
A=\bigl(a^{(1)},\dots,a^{(q)}\bigr)
.\]
For each $u\in \bF_{2}^{\ell}$, define a vector
\begin{equation}\label{eq:def_of_Lvec}
\Lvec(A,\ell,u)=\sum_{i=0}^{2^{\ell}-2}\left(\Had_{\ell}(u,2^{\ell}-i-1)\sum_{j=1}^{s}a^{(2si+2j-1)}\otimes a^{(2si+2j)}\right)\in \bF_{2}^{E'}\otimes \bF_{2}^{E'}.
\end{equation}
Since $u\mapsto \Had_{\ell}(u,2^{\ell}-i-1)$ is a linear map for each $i$, the map $u\mapsto \Lvec(A,\ell,u)$ from $\bF_{2}^{\ell}$ to $\bF_{2}^{E'}\otimes \bF_{2}^{E'}$ is also linear. For any linear subspace $U\subseteq \bF_{2}^{\ell}$, we denote the image of the latter linear map restricted to $U$ by
\[
\Lspan(A,\ell,U)=\left\{\Lvec(A,\ell,u)\,\middle|\,u\in U\right\}\subseteq \bF_{2}^{E'}\otimes \bF_{2}^{E'}.
\]
\end{definition}

Let us make a few remarks about~\Cref{eq:def_of_Lvec}. To get some intuition it is useful to consider a vector defined along the lines of~\Cref{eq:def_of_Lvec}, except that we replace the sum of rank $1$ vectors with arbitrary vectors
$x^{(i)}\in \bF_{2}^{E'}\otimes \bF_{2}^{E'}$, getting
\begin{equation}\label{eq:simple_span_H}
\sum_{i=0}^{2^{\ell}-2}\Had_{\ell}(u,2^{\ell}-i-1)x^{(i)}.
\end{equation}
Ranging over all $u\in\mathbb{F}_2^{\ell}$, the resulting collection of vectors is easily seen to be a linear subspace of $\bF_{2}^{E'}\otimes \bF_{2}^{E'}$ of dimension at most $\ell$. Moreover, if the vectors $x^{(i)}$ are taken uniformly at random, then, conditioned on this subspace having dimension $\ell$, it is uniformly distributed among all $\ell$-dimensional subspaces of $\bF_{2}^{E'}\otimes \bF_{2}^{E'}$.

It follows that the space $\Lspan(A,\ell,U)$ is also a subspace of $\bF_{2}^{E'}\otimes \bF_{2}^{E'}$ with at most $\ell$ dimensions. However, it is not distributed uniformly when the vectors in the list $A$ are taken uniformly, and this fact is crucial for us. Indeed, it is easy to see that each vector in $\Lspan(A,\ell,U)$ has rank at most $2^{\ell-1} s$; this is the source of the term ``low-rank'' in the name. In the following remark, we argue that if the list $A$ is generic, then the rank of each nonzero vector is exactly $2^{\ell-1}s$ (which will be convenient for us for technical reasons).

\begin{remark}\label{rem:saturated}
In the setting and notation of~\Cref{def:Lspan}, suppose  that the odd-indexed vectors $a^{(1)},a^{(3)},\dots,a^{(q-1)}$ are linearly  independent, and that the even-indexed vectors $a^{(2)},a^{(4)},\dots,a^{(q)}$ are also linearly independent. Then for any $u\in \bF_{2}^{\ell}\setminus\{0\}$, since $\Had_{\ell}(u,i)=1$ for exactly $2^{\ell-1}$ of the $i$'s in $\{1,2,\dots,2^{\ell}-1\}$, it follows from~\Cref{eq:def_of_Lvec} that
\[
\operatorname{rank}\left(\Lvec(A,\ell,u)\right)=2^{\ell-1}s.
\]
Furthermore, it is easy to see that $\dim(\Lspan(A,\ell,\bF_{2}^{\ell}))=\ell$.
\end{remark}

\subsubsection{Mixed-Type Advice Sequences}
Our inner PCP will include a tensor product of the low-rank Grassmann graph (on the critical block) with the standard Grassmann graph (on non-critical blocks). In this section, we present a few notations that will be useful in presenting and working with the resulting product space.
\begin{definition}\label{def:mixed_vectors}
Given a sequence of finite sets $\cE'=(E_{1}',\dots,E_{t}')$, a subset of indices $J\subseteq [t]$ and positive integers $\ell,s$, we let $\scrA_{J}(\cE',\ell,s)$ be the collection of advice sequences $(A_{1},\dots,A_{t})$ such that:
\begin{enumerate}[label=(\arabic*)]
\item For each $j\in [t]\setminus J$, the component $A_{j}$ is a sequence of $\ell$ vectors $A_{j}=(a^{(j,1)},\dots,a^{(j,\ell)})$, where \[a^{(j,i)}\in \bF_{2}^{E_{j}'}\otimes \bF_{2}^{E_{j}'}\qquad\text{for all }i\in [\ell].\]
\item For each $c\in J$, the component $A_{c}$ is a sequence of $q=2s(2^{\ell}-1)$ vectors $A_{c}=(a^{(c,1)},\dots,a^{(c,q)})$, where \[a^{(c,i)}\in \bF_{2}^{E_{c}'}\qquad\text{for all }i\in [q].\] 
If both sequences $\bigl(a^{(c,1)},a^{(c,3)},\dots,a^{(c,q-1)}\bigr)$ and $\bigl(a^{(c,2)},a^{(c,4)},\dots,a^{(c,q)}\bigr)$ are linearly independent, we say that $A_{c}$ is \emph{saturated}.
\end{enumerate}
We let $\scrA_{J}^*(\cE',\ell,s)$ denote the sub-collection of advice sequences $(A_{1},\dots,A_{t})\in \scrA_{J}(\cE',\ell,s)$ such that $A_{c}$ is saturated for each $c\in J$.
\end{definition}

\begin{remark}
A few remarks are in order:
\begin{enumerate}[label=(\arabic*)]
    \item The set of coordinates in $J$ in~\Cref{def:mixed_vectors}  should be thought of as the potentially critical coordinates, hence the collection there is specified by vectors from the base space $\bF_2^{E_j'}$. The coordinates outside $J$ are not critical, and the vectors there are completely arbitrary vectors from $\bF_2^{E_j'}\otimes \bF_2^{E_j'}$. 
    \item An advice sequence in $\scrA_{J}(\cE',\ell,s)$ should be thought of as the analog of a spanning set of a subspace. In other words, there may be linear dependencies between vectors in a member from $\scrA_{J}(\cE',\ell,s)$. In our setting of parameters this will be rare (analogously to the fact that a random set of $o(n)$ independently chosen vectors in $\mathbb{F}_2^{n}$ are linearly independent with high probability), and $\scrA_{J}^{*}(\cE',\ell,s)\subseteq \scrA_{J}(\cE',\ell,s)$ is the subcollection that corresponds to tuples with no linear dependencies.
\end{enumerate}    
\end{remark}

The following definition associates vectors from $\cV(\cE')$ with each advice sequence. In words, on each potentially critical coordinate we take the low-rank Hadamard span as per~\Cref{def:Lspan}, and on non-critical coordinates we take the span for arbitrary vectors as per~\Cref{eq:simple_span_H}.

\begin{definition}\label{def:Mvec_Mspan}
Fix a sequence of finite sets $\cE'=(E_{1}',\dots,E_{t}')$, a proper subset of indices $J\subsetneq [t]$ and positive integers $\ell,s$. Given an advice sequence $\cA=(A_{1},\dots,A_{t})\in \scrA_{J}(\cE',\ell,s)$, write $A_{j}=\bigl(a^{(j,1)},\dots,a^{(j,\ell)}\bigr)$ for each $j\in [t]\setminus J$. For each $u\in \bF_{2}^{\ell}$, define\footnote{We suppress the parameters $\ell,s$ in the notation $\operatorname{Mvec}_{J}(\cA,u)$; indeed, the values of $\ell$ and $s$ can be inferred from $\cA$, as long as $J\neq \emptyset$ and $J\neq [t]$.}
\begin{equation}\label{eq:def_of_Mvec}
\operatorname{Mvec}_{J}(\cA,u)=\sum_{c\in J}\operatorname{Lvec}(A_{c},\ell,u)+\sum_{j\in [t]\setminus J}\left(\sum_{b=1}^{\ell}u_{b}\cdot a^{(j,b)}\right)\in \cV(\cE').
\end{equation}
For any linear subspace $U\subseteq \bF_{2}^{\ell}$, we denote
\begin{equation}\label{eq:def_of_Mspan}
\operatorname{Mspan}_{J}(\cA,U)=\left\{\operatorname{Mvec}_{J}(A,u)\,\middle|\,u\in U\right\}\subseteq \cV(\cE').
\end{equation}
\end{definition}

\begin{notation}
In the remainder of the paper, we will often use the notation
\(\operatorname{Mspan}_{J}(\cA,U)\) (defined in~\Cref{eq:def_of_Mspan})
with \(U=\bF_{2}^{\ell}\). When \(J\) is clear from context, we abbreviate
\[
\sigma(\cA)
:=
\operatorname{Mspan}_{J}\bigl(\cA,\bF_{2}^{\ell}\bigr).
\]
We also write
\[
\rho(\cA)
:=
\bigl(
\operatorname{Mvec}_{J}(\cA,e_{1}),
\dots,
\operatorname{Mvec}_{J}(\cA,e_{\ell})
\bigr)\in (\cV(\cE'))^{\ell},
\]
where \(e_{1},\dots,e_{\ell}\) denote the standard basis vectors of
\(\bF_{2}^{\ell}\). Thus, $\sigma(\cA)$ is the linear span of the vectors in the sequence $\rho(\cA)$, and we write $\sigma(\cA)=\operatorname{span}(\rho(\cA))$.
\end{notation}

Given an advice sequence $\cA\in \scrA_{J}(\cE',\ell,s)$, the subspace $\sigma(\cA)$ is ``mixed-type'' in the sense that its projection on a critical block $c\in J$ consists of low-rank vectors in $\bF_{2}^{E_{c}'}\otimes \bF_{2}^{E_{c}'}$, while its projection on a non-critical block $j\in [t]\setminus J$ consists of generic vectors in $\bF_{2}^{E_{j}'}\otimes \bF_{2}^{E_{j}'}$.

Next, we will want to facilitate notions analogous to subspace containment for mixed-type subspaces. Towards this end, it is useful to note that if we considered~\Cref{eq:simple_span_H} and only looked at vectors $u$ with $u_{\ell}=0$, then we would get an $(\ell-1)$-dimensional space which is a subspace of the larger $\ell$-dimensional space resulting from taking all $u$. Moreover, this subspace would not depend on what is $x^{(2^{\ell}-2)}$ (as its coefficient will always be $0$). Generalizing this intuition, we get the following analogous notion of sub-advice.

\begin{definition}\label{def:sub-advice}
Fix a sequence of finite sets $\cE'=(E_{1}',\dots,E_{t}')$, a proper subset of indices $J\subsetneq [t]$, positive integers $\ell,s$ and a nonnegative integer $q_{0}\leq \ell$. Given an advice sequence $\cA=(A_{1},\dots,A_{t})\in \scrA_{J}(\cE',\ell,s)$, we define a \emph{sub-advice} $\cA_{[q_{0}]}=(A_{1}',\dots,A_{t}')\in\scrA_{J}(\cE',q_{0},2^{\ell-q_{0}}s)$ as follows:
\begin{enumerate}[label=(\arabic*)]
\item For each $j\in [t]\setminus J$, write $A_{j}=\bigl(a^{(j,1)},\dots,a^{(j,\ell)}\bigr)$ and let $A'_{j}=\bigl(a^{(j,1)},\dots,a^{(j,q_{0})}\bigr)$.
\item For each $c\in J$, write $A_{c}=\bigl(a^{(c,1)},\dots,a^{(c,q)}\bigr)$, where $q=2s(2^{\ell}-1)$. Denote $q_{1}=2s(2^{\ell}-2^{\ell-q_{0}})$ and let $A_{c}'=\bigl(a^{(c,1)},\dots,a^{(c,q_{1})}\bigr)$.
\end{enumerate}
\end{definition}

The following proposition formalizes the intuition that taking a sub-advice corresponds to spanning a subspace of the original space.
\begin{proposition}
In the setting and notation of~\Cref{def:sub-advice}, we have\footnote{Here, we identify $\bF_2^{q_0}$ with the subspace
$
\bF_2^{q_0}\times\{0\}^{\ell-q_0}\subseteq\bF_2^\ell;
$
that is, a $q_0$-bit string is extended to an $\ell$-bit string by appending $\ell-q_0$ zeros.}
\[
\operatorname{Mvec}_{J}(\cA,u)=\operatorname{Mvec}_{J}(\cA_{[q_{0}]},u)\qquad\text{for all }u\in \bF_{2}^{q_{0}},
\]
Consequently, by~\Cref{eq:def_of_Mspan} we have \[\operatorname{Mspan}_{J}(\cA,\bF_{2}^{q_{0}})=\operatorname{Mspan}_{J}(\cA_{[q_{0}]},\bF_{2}^{q_{0}}).\]
\end{proposition}

\begin{proof}
Denote $s_{1}=2^{\ell-q_{0}}s$. As in~\Cref{def:sub-advice}, let $q=2s(2^{\ell}-1)$ and $q_{1}=2s(2^{\ell}-2^{\ell-q_{0}})$.

Let $i$ be an integer ranging in $\{0,1,\dots,2^{\ell}-2\}$. For each $u\in \bF_{2}^{q_{0}}$ considered as an $\ell$-bit string, since its last $(\ell-q_{0})$ bits are all 0's, the value of \[\Had_{\ell}(u,2^{\ell}-i-1)=\sum_{b=1}^{q_{0}}u_{b}(1-\Bit_{\ell}(i,\ell-b))\] does not depend on the last $(\ell-q_{0})$ bits of the binary representation of $i$. Furthermore, if the first $q_{0}$ bits of $i$ are all 1's, then $\Had_{\ell}(u,2^{\ell}-i-1)=0$. Therefore, if we define $\pi(i)=2^{\ell-q_{0}}\lfloor i/2^{\ell-q_{0}}\rfloor$ for each $i\in \{0,1,\dots,2^{\ell}-2^{\ell-q_{0}}-1\}$, for each $c\in J$ and $u\in \bF_{2}^{q_{0}}$ we can deduce from~\Cref{eq:def_of_Lvec} that
\begin{align*}
\Lvec(A_{c},\ell,u)&=\sum_{i=0}^{2^{\ell}-2^{\ell-q_{0}}-1}\left(\Had_{\ell}(u,2^{\ell}-\pi(i)-1)\sum_{j=1}^{s}a^{(c,2si+2j-1)}\otimes a^{(c,2si+2j)}\right)\\
&=\sum_{i=0}^{2^{q_{0}}-2}\left(\Had_{\ell}(u,2^{\ell}-2^{\ell-q_{0}}i-1)\sum_{j=1}^{s_{1}}a^{(c,2s_{1}i+2j-1)}\otimes a^{(c,2s_{1}i+2j)}\right)\\
&=\sum_{i=0}^{2^{q_{0}}-2}\left(\Had_{q_{0}}(u,2^{q_{0}}-i-1)\sum_{j=1}^{s_{1}}a^{(c,2s_{1}i+2j-1)}\otimes a^{(c,2s_{1}i+2j)}\right)=\Lvec(A_{c}',q_{0},u).
\end{align*}
It then follows from~\Cref{eq:def_of_Mvec} that
\[
\operatorname{Mvec}_{J}(\cA,u)=\sum_{c\in J}\Lvec(A_{c}',q_{0},u)+\sum_{j\in [t]\setminus J}\left(\sum_{b=1}^{q_{0}}u_{b}\cdot a^{(j,b)}\right)=\operatorname{Mvec}_{J}(\cA_{[q_{0}]},u),
\]
for any $u\in \bF_{2}^{q_{0}}$.
\end{proof}

\begin{notation}
Throughout the remainder of the paper, the set \(J\subseteq[t]\) of
critical indices will always satisfy \(|J|\in \{1,2\}\). When
\(J=\{c_1,c_2\}\) with \(c_1\neq c_2\), we suppress the braces and write
\(c_1,c_2\) in place of \(\{c_1,c_2\}\) whenever \(J\) appears as a
subscript. For example, we write
\[
\scrA_{c_1,c_2}(\cE',\ell,s)
\qquad\text{instead of}\qquad
\scrA_{\{c_1,c_2\}}(\cE',\ell,s).
\]
Similarly, when \(J=\{c\}\), we write, for instance,
\[
\operatorname{Mspan}_{c}(\cA,U)
\qquad\text{instead of}\qquad
\operatorname{Mspan}_{\{c\}}(\cA,U).
\]
\end{notation}

\subsubsection{Zoom-out Probabilities} 

We next prove a few elementary properties of the ``zoom-out probabilities'' associated with mixed-type advice sequences.
Specifically, given a linear subspace $W\subseteq\cV(\cE')$ of small
codimension, we aim to show that the probability that
$\sigma(\cA)\subseteq W$, where $\cA$ is sampled at random
from $\scrA_{J}(\cE',\ell,s)$, is close to the probability that a
random subspace $L\in\mathsf{Gr}(\cV(\cE'),\ell)$ is contained in $W$.

For technical reasons, we also need a ``zoomed-in'' version of this
comparison. Specifically, we consider a random advice sequence
$\cA\in\scrA_{J}(\cE',\ell,s)$ conditioned on
$\sigma(\cA_{[q_{0}]})=Q$, where $Q\subseteq\cV(\cE')$ is a fixed
$q_{0}$-dimensional subspace. This should be analogous to sampling
a random subspace $L\in\mathsf{Gr}(\cV(\cE'),\ell)$ conditioned on
$Q\subseteq L$. Our main goal is to argue that
\[
    \Pru{\cA\in\scrA_{J}(\cE',\ell,s)}
    {\sigma(\cA)\subseteq W\,\middle|\,\sigma(\cA_{[q_{0}]})=Q}
    \approx
    \Pru{L\in\mathsf{Gr}(\cV(\cE'),\ell)}
    {L\subseteq W\,\middle|\,Q\subseteq L}.
\]

The first step is to show that conditioned on a fixed sub-advice $\cA_{[q_{0}]}$, the vectors in $\sigma(\cA)\setminus \sigma(\cA_{[q_{0}]})$ still behave randomly with respect to any test function $f:\cV(\cE')\rightarrow \bF_{2}$.

\begin{proposition}\label{prop:mixed_zoom_out_1}
Fix a sequence of non-empty finite sets $\cE'=(E_{1}',\dots,E_{t}')$, a subset of indices $J\subsetneq [t]$, positive integers $\ell,s$ and a nonnegative integer $q_{0}< \ell$. Let $f:\cV(\cE')\rightarrow \bF_{2}$ be an arbitrary linear functional that is not identically zero. For any fixed advice sequence $\cA'\in \scrA_{J}(\cE',q_{0},2^{\ell-q_{0}}s)$ and any vector $u\in \bF_{2}^{\ell}\setminus \bF_{2}^{q_{0}}$, we have
\[
\frac{1}{2}-2^{-s-1}\leq \Pru{\cA\in \scrA_{J}(\cE',\ell,s)}{f\bigl(\operatorname{Mvec}_{J}(\cA,u)\bigr)=1\,\middle|\,\cA_{[q_{0}]}=\cA'}\leq \frac{1}{2}+2^{-s-1}.
\]
\end{proposition}
\begin{proof}
As usual, we denote $\cA=(A_{1},\dots,A_{t})$ and $A_{j}=\bigl(a^{(j,1)},a^{(j,2)},\dots\bigr)$ for each $j\in [t]$. Since $f$ is not identically zero, there exists an index $j\in [t]$ such that the restriction of $f$ to $\bF_{2}^{E_{j}'}\otimes \bF_{2}^{E_{j}'}$ is nonzero. Depending on whether this index lies in $J$, we divide into two cases:

\textbf{Case 1:} the restriction of $f$ to $\bF_{2}^{E_{j}'}\otimes \bF_{2}^{E_{j}'}$ is nonzero for some $j\in [t]\setminus J$. In this case, if we let $b\in \{q_{0}+1,\dots,\ell\}$ be an index such that $u_{b}=1$, even after conditioning on all vectors in $\cA$ except for $a^{(j,b)}$, the value of $f\bigl(\operatorname{Mvec}_{J}(\cA,u)\bigr)$ (see \Cref{eq:def_of_Mvec}) is still an unbiased Bernoulli random variable purely due to the randomness of the vector $a^{(j,b)}\in \bF_{2}^{E_{j}'}\otimes \bF_{2}^{E_{j}'}$. Since $a^{(j,b)}$ is not included in the sub-advice $\cA_{[q_{0}]}$, we conclude that
\[
\Pru{\cA\in \scrA_{J}(\cE',\ell,s)}{f\bigl(\operatorname{Mvec}_{J}(\cA,u)\bigr)=1\,\middle|\,\cA_{[q_{0}]}=\cA'}=\frac{1}{2}.
\]

\textbf{Case 2:} the restriction of $f$ to $\bF_{2}^{E_{c}'}\otimes \bF_{2}^{E_{c}'}$ is nonzero for some $c\in J$. As in case 1, we let $b\in \{q_{0}+1,\dots,\ell\}$ be an index such that $u_{b}\neq 0$. Now let $i=2^{\ell}-2^{\ell-b}-1$, whose $\ell$-bit binary representation has a 0 in the $b$-th position and 1's in all other positions. Even after conditioning on all vectors in $\cA$ except for the $2s$ vectors (none of which is included in the sub-advice $\cA_{[q_{0}]}$)
\[
a^{(c,2si+1)},a^{(c,2si+2)},\dots,a^{(c,2si+2s)},
\]
the value of $f\bigl(\operatorname{Mvec}_{J}(\cA,u)\bigr)$ (see~\Cref{eq:def_of_Mvec,eq:def_of_Lvec}) is still the Bernoulli random variable
\[
X=\sum_{j=1}^{s}f\bigl(a^{(c,2si+2j-1)}\otimes a^{(c,2si+2j)}\bigr)
\]
or its negation. By independence among the vectors in $\cA$, it is easy to see that $\Pr{X=1}$ lies in the interval $[\frac{1}{2}-2^{-s-1},\frac{1}{2}+2^{-s-1}]$, and hence we have
\[
\frac{1}{2}-2^{-s-1}\leq \Pru{\cA\in \scrA_{J}(\cE',\ell,s)}{f\bigl(\operatorname{Mvec}_{J}(\cA,u)\bigr)=1\,\middle|\,\cA_{[q_{0}]}=\cA'}\leq \frac{1}{2}+2^{-s-1}.\qedhere
\]
\end{proof}

Recall that if $W\in\mathsf{Gr}(\bF_{2}^{n},n-w)$, $L'\in \mathsf{Gr}(W,\ell-1)$, and we choose an $\ell$-dimensional subspace $L\subseteq \mathbb{F}_2^{n}$ uniformly at random containing $L'$,
then $L\subseteq W$ with probability $2^{-w}+o(1)$. The following proposition establishes an analogous fact, with the notion of sub-advice replacing the containment $L'\subseteq L$.

\begin{proposition}\label{prop:mixed_zoom_out_2}
Fix a sequence of non-empty finite sets $\cE'=(E_{1}',\dots,E_{t}')$, a subset of indices $J\subsetneq [t]$, and positive integers $\ell,s$. Let $W\subseteq \cV(\cE')$ be a linear subspace with codimension $w$. For any fixed advice sequence $\cA'\in \scrA_{J}(\cE',\ell-1,2s)$ such that $\sigma(\cA')\subseteq W$, we have
\[
2^{-w}-2^{-s}\leq \Pru{\cA\in \scrA_{J}(\cE',\ell,s)}{\sigma(\cA)\subseteq W\,\middle|\,\cA_{[\ell-1]}=\cA'}\leq 2^{-w}+2^{-s}.
\]
\end{proposition}
\begin{proof}
We fix an arbitrary vector $u\in \bF_{2}^{\ell}\setminus \bF_{2}^{\ell-1}$. Since $\sigma(\cA')\subseteq W$ and \[\sigma(\cA)=\sigma(\cA_{[\ell-1]})+\operatorname{span}(\operatorname{Mvec}_{J}(\cA,u)),\]
the desired inequality is equivalent to
\begin{equation}\label{eq:mixed_zoom_out_2_desired}
2^{-w}-2^{-s}\leq\Pru{\cA\in \scrA_{J}(\cE',\ell,s)}{\operatorname{Mvec}_{J}(\cA,u)\in W\,\middle|\,\cA_{[\ell-1]}=\cA'}\leq 2^{-w}+2^{-s}.
\end{equation}

Let $\cF$ be the collection of linear functionals $f:\cV(\cE')\rightarrow \bF_{2}$ such that $W\subseteq \ker(f)$, and denote
\begin{equation}\label{eq:mixed_zoom_out_doublecount}
p=\Exu{f\in \cF}{\Pru{\cA\in \scrA_J(\cE',\ell,s)}{f\bigl(\operatorname{Mvec}_{J}(\cA,u)\bigr)=1\,\middle|\,\cA_{[\ell-1]}=\cA'}}.
\end{equation}
Note that $|\cF|=2^{w}$, and for each $f\in \cF$ that is not identically zero we may apply~\Cref{prop:mixed_zoom_out_1} to estimate the probability expression on the right-hand side of~\Cref{eq:mixed_zoom_out_doublecount}, which yields
\begin{equation}\label{eq:mixed_zoom_out_doublecount_1}
\frac{2^{w}-1}{2^{w}}\left(\frac{1}{2}-2^{-s-1}\right)\leq p\leq \frac{2^{w}-1}{2^{w}}\left(\frac{1}{2}+2^{-s-1}\right).
\end{equation}
On the other hand, by exchanging the expectation and the probability on the right-hand side of~\Cref{eq:mixed_zoom_out_doublecount} (i.e. double counting), we have
\begin{equation}\label{eq:mixed_zoom_out_doublecount_2}
p=\frac{1}{2}\left(1-\Pru{\cA\in \scrA_{J}(\cE',\ell,s)}{\operatorname{Mvec}_{J}(\cA,u)\in W\,\middle|\,\cA_{[\ell-1]}=\cA'}\right).
\end{equation}
Comparing~\Cref{eq:mixed_zoom_out_doublecount_1,eq:mixed_zoom_out_doublecount_2} immediately yields the desired~\Cref{eq:mixed_zoom_out_2_desired}.
\end{proof}

Given $\cA'\in \scrA_{J}(\cE',\ell-1,2s)$ such that $\sigma(\cA')\subseteq W$,~\Cref{prop:mixed_zoom_out_2} estimates the probability that a random extension $\cA\in \scrA_{J}(\cE',\ell,s)$ of $\cA'$ still ``stays in'' $W$, in the sense that $\sigma(\cA)\subseteq W$. Roughly speaking, this means that ``random extension'' by $1$ dimension from $\sigma(\cA_{[\ell-1]})$ to $\sigma(\cA)$ behaves nicely with respect to zoom-outs. This allows us to iterate the process and show that a random extension by $\ell-q_{0}$ dimensions from $\sigma(\cA_{[q_{0}]})$ to $\sigma(\cA)$ also behaves as expected with respect to zoom-outs. 

\begin{corollary}\label{cor:mixed_zoom_out_advice}
Fix a sequence of non-empty finite sets $\cE'=(E_{1}',\dots,E_{t}')$, a subset of indices $J\subsetneq [t]$, and positive integers $\ell,s$. Let $W\subseteq \cV(\cE')$ be a linear subspace with codimension $w$, and suppose that $s\geq \ell+w$. Then for any nonnegative integer $q_{0}<\ell$ and any fixed advice sequence $\cA'\in \scrA_{J}(\cE',q_{0},2^{\ell-q_{0}}s)$ such that $\sigma(\cA')\subseteq W$, we have
\[
2^{-w(\ell-q_{0})-1}\leq \Pru{\cA\in \scrA_{J}(\cE',\ell,s)}{\sigma(\cA)\subseteq W\,\middle|\,\cA_{[q_{0}]}=\cA'}\leq 2^{-w(\ell-q_{0})+1}.
\]
\end{corollary}
\begin{proof}
For each $i\in \{q_{0},\dots,\ell\}$, let $\Gamma_{i}$ denote the event $\left\{\sigma(\cA_{[i]})\subseteq W\right\}$.
\[
\Pru{\cA\in \scrA_{J}(\cE',\ell,s)}{\sigma(\cA)\subseteq W\,\middle|\,\cA_{[q_{0}]}=\cA'}=\prod_{i=q_{0}}^{\ell-1}\Pr{\Gamma_{i+1}\,\middle|\,\left\{\cA_{[q_{0}]}=\cA'\right\}\cap\Gamma_{i}}.
\]
Since we can apply~\Cref{prop:mixed_zoom_out_2} to every factor on the right-hand side, the conclusion easily follows. Note that the condition $s\geq \ell+w$ guarantees
\[
(2^{-w}-2^{-s})^{\ell-q_{0}}\geq 2^{-w(\ell-q_{0})}\left(1-(\ell-q_{0})2^{-s+w}\right)\geq 2^{-w(\ell-q_{0})}(1-2^{\ell-1-s+w})\geq 2^{-w(\ell-q_{0})-1},
\]
and similarly for the upper bound.
\end{proof}

\subsubsection{Mixed-Type Grassmann Graph}

So far, we have mainly worked with uniformly random mixed-type advice sequences $\cA\in\scrA_{J}(\cE',\ell,s)$, the vectors in which may have linear dependencies (though unlikely). It will be useful later in the analysis of our inner PCP to consider those $\cA$ with no linear dependency among their vectors, spanning actual $\ell$-dimensional subspaces $\sigma(\cA)\subseteq \cV(\cE')$. 

This naturally leads us to the sub-collection $\cA^*_{J}(\cE',\ell,s)$ defined in~\Cref{def:mixed_vectors}. For sequences $\cA=(A_{1},\dots,A_{t})\in \cA^{*}_{J}(\cE',\ell,s)$, the lists $A_{c}$ in the critical blocks $c\in J$ are always saturated, and by~\Cref{rem:saturated}, it follows that $\dim \sigma(\cA)=\ell$ (as long as $J$ is nonempty). Therefore, these subspaces $\sigma(\cA)$ lie in the domain $\mathsf{Gr}(\cV(\cE'),\ell)$, and they induce a subgraph of the Grassmann graph $\mathsf{Gr}(\cV(\cE'),\ell)$. We call this subgraph a ``mixed-type Grassmann graph,'' as defined below.

\begin{definition}\label{def:mixed_grassmann_domain}
Fix a sequence of finite sets $\cE'=(E_{1}',\dots,E_{t}')$, a nonempty subset of indices $J\subsetneq  [t]$ and positive integers $\ell,s$. Define $\mathsf{MGr}_{J}(\cE',\ell,s)$ to be the following collection of linear subspaces:
\[
\mathsf{MGr}_{J}(\cE',\ell,s)=\left\{\sigma(\cA)\,\middle|\,\cA\in \scrA_{J}^*(\cE',\ell,s)\right\}\subseteq \mathsf{Gr}(\cV(\cE'),\ell).
\]
\end{definition}

An important fact (whose proof we omit) is that for different subspaces $L\in \mathsf{MGr}_{J}(\cE',\ell,s)$, the number of $\cA\in \scrA_{J}^{*}(\cE',\ell,s)$ such that $\sigma(\cA)=L$ is the same.\footnote{To prove that, one needs to use the fact that for any $\bF_{2}$-vector space $V$ and any $x\in V\otimes V$ of rank $s$, the number of ways to write $x=a^{(1)}\otimes a^{(2)}+\dots + a^{(2s-1)}\otimes a^{(2s)}$, where $a^{(i)}\in V$ for each $i\in [2s]$, depends only on $s$.} Therefore, if we sample a uniformly random $\cA\in \scrA_{J}^*(\cE',\ell,s)$, the span $\sigma(\cA)$ is uniformly random over $\mathsf{MGr}_{J}(\cE',\ell,s)$. The next proposition asserts that even if one chose $\cA$ from $\scrA_{J}(\cE',\ell,s)$, the resulting distribution of $\sigma(\cA)$ is still close to uniform over $\mathsf{MGr}_{J}(\cE',\ell,s)$, since $\left|\cA^{*}_J(\cE',\ell,s)\right|=(1-o(1))\cdot \left|\cA_J(\cE',\ell,s)\right|$.

% With the view of~\Cref{def:mixed_grassmann_domain}, for a given space $Q = \sigma(\cA')\subseteq \cV(\cE')$ one can consider two possible subsets of vertices of $\mathsf{MGr}_{J}(\cE',\ell,s)$ as analogs of the zoom-ins of $Q$:
% \begin{enumerate}
%     \item \textbf{Containment:} one can consider all $L\in \mathsf{MGr}_{J}(\cE',\ell,s)$ that contain $Q$, and consider the uniform distribution on them.
%     \item \textbf{Sub-advice:} one can consider the uniform distribution over all advice sequences $\cA$ that have $\cA'$ as sub-advice, and then consider the induced distribution on subspaces $\sigma(\cA)$.
% \end{enumerate}
% The following proposition asserts that these two distributions are very close in statistical distance. 

\begin{proposition}\label{prop:mix_advice_vs_MGr}
Fix a sequence of non-empty finite sets $\cE'=(E_{1}',\dots,E_{t}')$, a nonempty subset of indices $J\subsetneq [t]$, positive integers $\ell,s$ and a nonnegative integer $q_{0}<\ell$. Assume that $|E_{c}'|\geq 2^{\ell+2}s$ for each $c\in J$. For any fixed advice sequence $\cA'\in \scrA_{J}^*(\cE',q_{0},2^{\ell-q_{0}}s)$, if we sample $\cA\in\scrA_{J}(\cE',\ell,s)$ conditioned on $\sigma(\cA_{[q_{0}]})=\sigma(\cA')$, the resulting distribution of $\sigma(\cA)$ has total variation distance at most
\[\sum_{c\in J}2^{-|E_{c}'|/2}\]
to the uniform distribution over the set \(\big\{L\in \mathsf{MGr}_{J}(\cE',\ell,s)\,\big|\, \sigma(\cA')\subseteq L\big\}\).
\end{proposition}

\begin{proof}
It is not hard to see that if we sample $\cA\in \scrA_{J}^*(\cE',\ell,s)$ conditioned on $\sigma(\cA_{[q_{0}]})=\sigma(\cA')$, the resulting distribution of $\sigma(\cA)$ is exactly the uniform distribution over $\big\{L\in \mathsf{MGr}_{J}(\cE',\ell,s)\,\big|\, \sigma(\cA')\subseteq L\big\}
$. 
On the other hand, writing $q=2s(2^{\ell}-1)$ and $q_{1}=2s(2^{\ell}-2^{\ell-q_{0}})$, for any $\cA'\in \scrA_{J}^{*}(\cE',q_{0},2^{\ell-q_{0}}s)$ we have
\begin{align*}
&\quad\Pru{\cA\in \scrA_{J}(\cE',\ell,s)}{\cA\in \scrA^*_{J}(\cE',\ell,s)\,\middle|\,\cA_{[q_{0}]}=\cA'}\\
&=\Pr{\bigl(a^{(c,1)},\dots,a^{(c,q)}\bigr)\text{ is saturated for each }c\in J\,\middle|\,\text{fixed saturated sequences $\bigl(a^{(c,1)},\dots,a^{(c,q_{1})}\bigr)$}}\tag{where $a^{(c,i)}$ (for $c\in J$ and $i\in [q]$) are sampled independently and uniformly from $\bF_{2}^{E_{c}'}$}\\
&\geq \prod_{c\in J}\left(1-2^{(2^{\ell}-1)s-|E_{c}'|}\right)^{2}\geq 1-\sum_{c\in J}2^{-|E_{c}'|/2}.\qedhere
\end{align*}
\end{proof}

The next proposition asserts that for a valid zoom-in and zoom-out pair $Q,W$ (i.e., $Q\subseteq W$), among the subspaces in $\mathsf{MGr}_{J}(\cE',\ell,s)$ containing $Q$, the fraction of them contained in $W$ behaves similarly to the standard Grassmann domain.
\begin{proposition}\label{prop:MGr_zoom_out}
Fix a sequence of non-empty finite sets $\cE'=(E_{1}',\dots,E_{t}')$, a subset of indices $J\subsetneq [t]$ such that $1\leq |J|\leq 2$, positive integers $\ell,s$ and nonnegative integers $q_{0},w$. Assume
\[
s\geq \ell+w,\qquad q_{0}<\ell,\qquad\text{and}\qquad |E_{c}'|\geq 2^{\ell+2}s\text{ for each }c\in J
.\]
Then for any fixed subspaces 
\[Q\in \mathsf{MGr}_{J}(\cE',q_{0},2^{\ell-q_{0}}s)\qquad\text{and}\qquad W\in \mathsf{Gr}\bigl(\cV(\cE'),\dim \cV(\cE')-w\bigr)\] such that $Q\subseteq W$, we have
\[
2^{-w(\ell-q_{0})-2}\leq \Pru{L\in \mathsf{MGr}_{J}(\cE',\ell,s)}{L\subseteq W\mid Q\subseteq L}\leq 2^{-w(\ell-q_{0})+2}.
\]
\end{proposition}
\begin{proof}
Pick $\cA'$ uniformly at random from the set
\[\left\{\cA'\in\scrA_{J}(\cE',q_{0},2^{\ell-q_{0}}s)\,\Big|\,\sigma(\cA')=Q\right\}.\]
It follows from~\Cref{prop:mix_advice_vs_MGr} that
\[
\left|\Pru{L\in \mathsf{MGr}_{J}(\cE',\ell,s)}{L\subseteq W\mid Q\subseteq L}-\Exu{\cA'}{\Pru{\cA\in \scrA_{J}(\cE',\ell,s)}{\sigma(\cA)\subseteq W\,\middle|\,\cA_{[q_{0}]}=\cA'}}\right|\leq \sum_{c\in J}2^{-|E_{c}'|/2}\leq 2^{-w\ell-2}.
\]
Combining the above with~\Cref{cor:mixed_zoom_out_advice} yields the desired conclusion.
\end{proof}

\subsubsection{Extension of Agreement}

The next somewhat technical result is needed (only) in~\Cref{subsubsec:Bob_existence_decoding}, and to explain it we first give the analog in the simpler Grassmann setting. Suppose we have a table $F$ that assigns to each subspace in $\mathsf{Gr}(\mathbb{F}_2^n,\ell+1)$ a linear function $F[L]$ (or $\nil$), a subspace $W\subseteq \bF_{2}^{n}$ of constant codimension, and a linear function $f\colon W\to \mathbb{F}_2$ that largely agrees with $F$ in the sense that
\[
\Pru{L\in\mathsf{Gr}(\mathbb{F}_2^{n},\ell+1)}{f|_{L\cap W} = F[L]|_{L\cap W}\Big.}\geq \delta.
\]
We wish to extend $f$ into a linear function $g:\bF_{2}^{n}\rightarrow \bF_{2}$ that agrees with $F$ not only on the intersection with $W$. Then a random extension would suffice: for any $L$ such that  
$f|_{L\cap W} = F[L]|_{L\cap W}$ we will have that $g|_{L} = F[L]$ with constant probability. The following lemma establishes an analogous statement in the context of mixed-type subspaces. 

\begin{lemma}\label{lem:MGr_agreement_extension}
Fix a sequence of non-empty finite sets $\cE'=(E_{1}',\dots,E_{t}')$, a critical index $c\in [t]$, positive integers $\ell,s$, nonnegative integers $q_{0},w$ and a parameter $\delta\in (0,1)$. Assume
\[
s\geq 2\ell\qquad\text{and}\qquad\qquad\ell>2q_{0}+2w.
\]
Let $W\subseteq \cV(\cE')$ be a linear subspace with codimension $w$, and $Q\subseteq W$ a linear subspace such that $\sigma(\cA)=Q$ for some $\cA\in \scrA_{c}(\cE',q_{0},2^{\ell+1-q_{0}}s)$. Suppose $F[\cdot]$ is a table where $F[L]\in L^*\cup\{\nil\}$ for each linear subspace $L\subseteq \cV(\cE')$, and there exists a linear function $f:W\rightarrow \bF_{2}$ such that
\begin{equation}\label{eq:MGr_extension_condition}
\Pru{\cA\in \scrA_{c}(\cE',\ell+1,s)}{\left.F[\sigma(\cA)]\right|_{\sigma(\cA_{[\ell]})}=\left.f\right|_{\sigma(\cA_{[\ell]})}\,\Big|\,\sigma(\cA_{[q_{0}]})=Q\textup{ and }\sigma(\cA_{[\ell]})\subseteq W}\geq \delta.
\end{equation}
Then there exists a subspace $W'\subseteq \cV(\cE')$ containing $W$ and a linear function $g:W'\rightarrow \bF_{2}$ such that
\begin{equation}\label{eq:MGr_extension_goal}
\Pru{\cA\in \scrA_{c}(\cE',\ell+1,s)}{F[\sigma(\cA)]=\left.g\right|_{\sigma(\cA)}\,\Big|\,\sigma(\cA_{[q_{0}]})=Q\textup{ and }\sigma(\cA)\subseteq W'}\geq \frac{\delta}{5}.
\end{equation}
\end{lemma}

\begin{proof}
We let $\eta$ be the maximum possible value of the left-hand side of~\Cref{eq:MGr_extension_goal} when $W'$ ranges in the set of linear subspaces containing $W$ and $g$ ranges in $(W')^*$. The goal is to show that $5\eta\geq \delta$. Throughout the proof, we let $L,L'\subseteq \cV(\cE')$ be random subspaces sampled by the following process:
\begin{enumerate}
\item Sample $\cA\in \scrA_{c}(\cE',\ell+1,s)$ conditioned on $\sigma(\cA_{[q_{0}]})=Q$.
\item Let $L=\sigma(\cA_{[\ell]})$ and $L'=\sigma(\cA)$.
\end{enumerate}
The most important fact about the joint distribution of $(L,L')$ that we are going to use is the following: conditioned on $L'$, for any linear subspace such that $Q\subseteq L_{0}\subseteq L'$ and $\dim L_{0}=\dim L'-1$, the probability that $L=L_{0}$ is exactly $(2^{\ell+1-q_{0}}-1)^{-1}$. 

Let $p=\Pr{L'\subseteq W}$ and denote
\begin{align*}
\gamma_{0}&=\Pr{F[L']=f|_{L'}\,\Big|\,L'\subseteq W},\\
\gamma_{1}&=\Pr{\left.F[L']\right|_{Q}=f|_{Q}\textup{ but }F[L']\neq f|_{L'}\,\Big|\,L'\subseteq W}.
\end{align*}
By the definitions of $\gamma_{0}$ and $\eta$, we immediately have $\eta\geq \gamma_{0}$. Furthermore, by picking a random $g:W\rightarrow \bF_{2}$ conditioned on $g|_{Q}=f|_{Q}$, we have
\[
\Exu{g}{\Pr{F[L']=g|_{L'}\textup{ and }L'\subseteq W\Big.}}\geq 2^{-(\ell+1-q_{0})}\cdot (\gamma_{0}p+\gamma_{1}p),
\]
and hence
\begin{equation}\label{eq:agreement_extension_randomize}
\eta\geq \Exu{g}{\Pr{F[L']=g|_{L'}\,\Big|\,L'\subseteq W}}\geq 2^{-(\ell+1-q_{0})}\cdot (\gamma_{0}+\gamma_{1}).
\end{equation}

We denote by $\cW$ the collection of linear subspaces $W'\subseteq \cV(\cE')$ such that $\dim W'=1+\dim W$ and $W\subseteq W'$. For each $W'\in \cW$, we let 
\begin{align*}
p_{W'}&=\Pr{L'\subseteq W'\textup{ but }L'\not\subseteq W},\textup{ and}\\
\qquad \gamma_{W'}&=\begin{cases}
\Pr{\left.F[L']\right|_{L'\cap W}=f|_{L'\cap W}\,\Big|\,L'\subseteq W'\textup{ but }L'\not\subseteq W},&\text{if }p_{W'}>0\\
0,&\text{if }p_{W'}=0.
\end{cases}
\end{align*}
For each $W'\in \cW$, there are exactly two linear extensions of $f$ to $W'$. If we choose $g:W'\rightarrow \bF_{2}$ to be a uniformly random one of the two, then it is easy to see that
\[
\Exu{g}{\Pr{F[L']=g|_{L'}\textup{ and }L'\subseteq W'\Big.}}=\gamma_{0}p+\frac{1}{2}\gamma_{W'}p_{W'},
\]
and hence
\begin{equation}\label{eq:agreement_extension}
\eta\geq \Exu{g}{\Pr{F[L']=g|_{L'}\,\Big|\,L'\subseteq W'}}=\frac{\gamma_{0}p+\gamma_{W'}p_{W'}/2}{p+p_{W'}}.
\end{equation}

It is not hard to see that
\begin{equation}\label{eq:agreement_extension_LinW}
\Pr{L\subseteq W}=p+\frac{1}{2^{\ell+1-q_{0}}-1}\sum_{W\in \cW}p_{W'},
\end{equation}
and furthermore,
\begin{align*}
&\quad\Pr{\left.F[L']\right|_{L}=f|_{L}\textup{ and }L\subseteq W\Big.}\\
&=\gamma_{0}p+\frac{1}{2^{\ell+1-q_{0}}-1}\left(\gamma_{1}p+\sum_{W'\in\cW}\gamma_{W'}p_{W'}\right)\\
&\leq \gamma_{0}p+2\cdot 2^{-(\ell+1-q_{0})}\cdot (\gamma_{0}+\gamma_{1})p+\frac{2}{2^{\ell+1-q_{0}}-1}\sum_{W'\in \cW}\left(\gamma_{0}p+\frac{1}{2}\gamma_{W'}p_{W'}\right)\\
&\leq 3\eta\cdot p+2\eta\cdot\frac{1}{2^{\ell+1-q_{0}}-1}\sum_{W'\in \cW}(p+p_{W'})\tag{using $\gamma_{0}\leq \eta$ and~\Cref{eq:agreement_extension_randomize,eq:agreement_extension}}\\
&\leq 5\eta\cdot \Pr{L\subseteq W}.\tag{using~\Cref{eq:agreement_extension_LinW} and $|\cW|=2^{w}-1\leq 2^{\ell+1-q_{0}}-1$}
\end{align*}
This gives the desired inequality
\[
5\eta\geq \Pr{\left.F[L']\right|_{L}=f|_{L}\,\Big|\,L\subseteq W}\geq \delta,
\]
where we used~\Cref{eq:MGr_extension_condition} in the last transition.
\end{proof}

\subsection{The 4-to-1 Game Construction}\label{subsec:Game3}

In this subsection, we use the framework prepared in~\Cref{subsec:tripled_packet,subsec:low_rank_Grassmann} to construct our final sub-$4$-to-$1$ game. Before defining the sub-$4$-to-$1$ game, we first rewrite the description of $\GameBb$ (defined in~\Cref{subsec:Game2b}) using the notations introduced in~\Cref{subsec:tripled_packet,subsec:low_rank_Grassmann}.

\begin{definition}
Let $H$ be a 3-uniform hypergraph. Given positive integers $k$ and $t$, we let $\scrE(H,k,t)$ denote the collection of all sequences $\cE=(E_{1},\dots,E_{t})$ such that $E_{j}\in \scrE(H,k)$ for all $j\in [t]$.
\end{definition}

We can now rewrite the definition of $\GameBb[H,k,\ell_{2},t,q_{0},q_{1},\beta]$ as follows.

\begin{definition}\label{def:Game2c}
Given positive integers $k,\ell_{2},t,s_{1}$, a nonnegative integer $q_{0}$, a parameter $\beta\in [0,1)$ and a 3-uniform hypergraph $H$, let $\GameBc[H,k,\ell_{2},t,s_{1},q_{0},\beta]$ be a 2-prover 1-round game where the verifier executes the following:
\begin{enumerate}
\item Sample a uniformly random sequence $\cE=(E_{1},\dots,E_{t})\in \scrE(H,k,t)$.
\item Sample $(c,\cU,\cE',\cU')\sim \cJ_{\inn}(\cE,\beta,\ell_{2})$.
\item Sample a uniformly random sequence $\cA=(A_{1},\dots,A_{t})\in \scrA_{c}(\cE',q_{0},s_{1})$.\footnote{If $q_{0}=0$, we take $\cA=\emp$ to be an empty sequence.}
\item Send $(\cE,\cA,\cU)$ to Alice and $(c,\cE',\cA,\cU')$ to Bob.
\item Alice must return a linear map $f:\cV(\cE)\rightarrow\bF_{2}$ such that $\left.f\right|_{\cN(\cE,\cU)}$ is \textup{SNR}-satisfying; Bob must return a linear map $g:\cV(\cE')\rightarrow \bF_{2}$. Accept if $\left.f\right|_{\cV(\cE')}=g$.
\end{enumerate}
\end{definition}

It is easy to observe that~\Cref{def:Game2c} defines the same game as substituting $q_{1}=2s_{1}(2^{q_{0}}-1)$ in~\Cref{def:Game2b} does. 

\begin{proposition}\label{prop:Game2b=Game2c}
For any positive integers $k,\ell_{2},t,s_{1}$, nonnegative integer $q_{0}$, parameter $\beta\in [0,1)$ and $3$-uniform hypergraph $H$, we have
\[\operatorname{val}\Bigl(\GameBb[H,k,\ell_{2},t,q_{0},2s_{1}(2^{q_{0}}-1),\beta]\Bigr)=\operatorname{val}\Bigl(\GameBc[H,k,\ell_{2},t,s_{1},q_{0},\beta]\Bigr).\]
\end{proposition}

\begin{proof}
It is easy to see the equivalence between $\mathsf{Game2b}$ and $\mathsf{Game2c}$ by comparing~\Cref{def:Game2b,def:Game2c}. 
\end{proof}

We are now ready to define our sub-$4$-to-$1$ game based on $\GameBc$. The main idea here is that instead of providing the functions $f,g$ fully as in the game in $\GameBc$, the provers are expected to provide an encoding of them based on the mixed-type Grassmann graph. More specifically, each player receives a linear subspace of their ambient space, and is requested to respond by a linear function on it. In order to make the game sub-$4$-to-$1$, the subspace Alice receives is contained in the subspace Bob receives and is of $1$ less dimension, and additionally we incorporate inside it the local tensor tests as well as the side-condition equations.
\begin{definition}\label{def:game3a}
Given positive integers $k,\ell_{1},\ell_{2},t,s$, a parameter $\beta\in[0,1)$ and a 3-uniform hypergraph $H$, let $\GameCa[H,k,\ell_{1},\ell_{2},t,s,\beta]$ be a 2-prover 1-round game where the verifier executes the following:
\begin{enumerate}
\item Sample a uniformly random sequence $\cE=(E_{1},\dots,E_{t})\in \scrE(H,k,t)$.
\item Sample $(c,\cU,\cE',\cU')\sim \cJ_{\inn}(\cE,\beta,\ell_{2})$.
\item Sample a uniformly random sequence $\cA=(A_{1},\dots,A_{t})\in \scrA_{c}(\cE',2\ell_{1},s)$. Let \[L=\sigma(\cA_{[\ell_{1}]})\subseteq \cV(\cE')\qquad\text{and}\qquad L'=\sigma(\cA_{[\ell_{1}+1]})\subseteq \cV(\cE').\]
\item Send $(\cE,\cU,L+\cN(\cE,\cU))$ to Alice and $(c,\cE',\cU',L'+\toplus \cU')$ to Bob. 
\item Alice must return a linear map $f:(L+\cN(\cE,\cU))\rightarrow \bF_{2}$ such that $\left.f\right|_{\cN(\cE,\cU)}$ is \textup{SNR}-satisfying; Bob must return a linear map $g:(L'+\toplus\cU')\rightarrow \bF_{2}$ such that $g|_{\bigoplus \cU'}$ is \textup{SR}-satisfying. Accept if $g|_{(L+\bigoplus \cU)}=f|_{(L+\bigoplus\cU)}$.
\end{enumerate}
\end{definition}

The game defined in~\Cref{def:game3a} is clearly a sub-$4$-to-$1$ game, in the sense that for any response of Alice there are at most $4$ responses for Bob which would make the verifier accept. We will make a more formal argument at the end of this subsection.

The completeness of $\GameCa$ is straightforward and is given in the following lemma.
\begin{lemma}[Game3 completeness]\label{lem:Game3_completeness}
For any positive integers $k,\ell_{1},\ell_{2},t,s$, a parameter $\beta\in [0,1)$ and a 2-colorable 3-uniform hypergraph $H$, we have
\[
\operatorname{val}\Bigl(\GameCa[H,k,\ell_{1},\ell_{2},t,s,\beta]\Bigr)=1.
\]
\end{lemma}
\begin{proof}
By~\Cref{lem:Game2b_completeness} and~\Cref{prop:Game2b=Game2c}, there is a set of strategies for Alice and Bob that achieve value 1 in $\GameBc[H,k,\ell_{2},t,s_{1},0,\beta]$ (here we may take $s_{1}=1$, but it is irrelevant anyway). Note that since the second-to-last parameter in $\GameBc[H,k,\ell_{2},t,s_{1},0,\beta]$ is set to be zero, there is no advice in this game. Suppose the strategies are represented by a table $F[\cdot]$ for Alice and a table $G[\cdot]$ for Bob, where $F[\cE,\cU]:\cV(\cE)\rightarrow \bF_{2}$ is a linear map for each Alice-question $(\cE,\cU)$, and $G[c,\cE',\cU']:\cV(\cE')\rightarrow \bF_{2}$ is a linear map for each Bob-question $(c,\cE',\cU')$. 

We can then construct a strategy for Alice in the game $\GameCa[H,k,\ell_{1},\ell_{2},t,s,\beta]$ as follows:
\begin{enumerate}
\item Receive the question $(\cE,\cU,L+\cN(\cE,\cU))$.
\item Let $f$ be the restriction of $F[\cE,\cU]$ to $L+\cN(\cE,\cU)$, and return $f$. Note that by~\Cref{def:Game2c}, the restriction of $F[\cE,\cU]$ to $\cN(\cE,\cU)$ must be \textup{SNR}-satisfying.
\end{enumerate}
Similarly, the strategy for Bob is:
\begin{enumerate}
\item Receive the question $(c,\cE',\cU',L'+ \toplus \cU')$.
\item Let $g$ be the restriction of $G[c,\cE',\cU']$ to $L'+\toplus \cU'$, and return $g$. Although the definition of $\GameBc$ itself does not require the restriction $G[c,\cE',\cU']|_{\bigoplus \cU'}$ to be \textup{SR}-satisfying, the way we constructed $G[\cdot]$ does lead to \textup{SR}-satisfying restrictions (see~\Cref{rem:Game1b_completeness_Bob}). 
\end{enumerate}
Since we always have $F[\cE,\cU]|_{\cV(\cE')}=G[c,\cE',\cU']$, it is clear that we also have $f|_{(L+\bigoplus\cU)}=g|_{(L+\bigoplus\cU)}$. Moreover, we choose $F,G$ to be the strategies induced from the proper $2$-coloring of the original graph through the series of reductions. The strategies are symmetric and have rank at most 1 restricted to each block. Therefore, under this set of strategies, the verifier of $\GameCa[H,k,\ell_{1},\ell_{2},t,s,\beta]$ accepts with probability 1.
\end{proof}

The soundness of $\GameCa$ is stated in the following lemma.

\begin{lemma}[Game3  soundness]\label{lem:innerpcp_soundness}
    Assume~\Cref{cond:parameters}. Suppose that in every $4^{r}$-coloring of $H$, at least $\varepsilon$ fraction of the edges are monochromatic. Then we have
\[
\operatorname{val}\Bigl(\GameCa[H,k,\ell_{1},\ell_{2},t,s,\beta]\Bigr)\leq \delta. 
\]
\end{lemma}

The proof of~\Cref{lem:innerpcp_soundness} is significantly more involved, and will be completed in~\Cref{subsec:inner_soundness}. We first assume this soundness lemma and prove the main theorem (\Cref{thm:sub-main}, restated below).

\submainthm*

\begin{proof}[Proof assuming~\Cref{lem:innerpcp_soundness}]
Without loss of generality, assume $\delta<1/2$. Using the parameters in~\Cref{cond:parameters}, we show a polynomial-time reduction from the gap problem in~\Cref{thm:hg_coloring} (with the number of colors $c=4^{r}$) to $\GapSdone{4}{1}_{N}(1,\delta)$, for $N=2^{\ell_{2}t+\ell_{1}+2}$.

The reduction runs as follows: given a 3-uniform hypergraph $H$, we construct the 2-prover 1-round game $\GameCa[H,k,\ell_{1},\ell_{2},t,s,\beta]$, which is viewed as a label-cover instance.\footnote{For the standard translation between 2-prover 1-round games and label-cover instances, see the discussion at the beginning of~\Cref{sec:outer_pcp_imperfect}.} By the definition of $\GameCa$ (\Cref{def:game3a}), this label-cover instance can be constructed in $\mathrm{poly}(n)$ time, where $n$ is the size of $H$. The completeness and soundness of the reduction are already given by~\Cref{lem:Game3_completeness,lem:innerpcp_soundness}, respectively. It only remains to show that $\GameCa[H,k,\ell_{1},\ell_{2},t,s,\beta]$ is always a sub-$4$-to-$1$ game.

As in~\Cref{def:game3a}, consider a question $L+\cN(\cE,\cU)$ sent to Alice and a question $L'+\toplus \cU'$ sent to Bob. 

Given any possible response 
$g:(L'+\toplus \cU')\rightarrow \bF_{2}$
of Bob, if Alice's response $f:(L+\cN(\cE,\cU))\rightarrow \bF_{2}$ makes the verifier accept, the restriction $f|_{L+\bigoplus \cU}$ must be the fixed function $g|_{L+\bigoplus \cU}$. Furthermore, by~\Cref{prop:SymNAE_uniqueness} and~\Cref{def:SNR_satisfying}, any linear functional on $\toplus \cU$ has at most one \textup{SNR}-satisfying extension to $\cN(\cE,\cU)$. Therefore, there exists at most one extension of $g|_{L+\bigoplus \cU}$ to $(L+\cN(\cE,\cU))$ that has an \textup{SNR}-satisfying restriction on $\cN(\cE,\cU)$. 

Conversely, given any possible response $f:(L+\cN(\cE,\cU))\rightarrow \bF_{2}$ of Alice, we claim that there are at most 4 possible responses $g:(L'+\toplus \cU')\rightarrow \bF_{2}$
of Bob that make the verifier accept. Indeed, in order to pass the verifier's test, the restriction $g|_{L+\bigoplus \cU}$ must be the fixed function $f|_{L+\bigoplus \cU}$. Since $\cU'=(U_{1}',\dots,U_{t})$ differs from $\cU=(U_{1},\dots,U_{t})$ in only one critical block $c$, for which $\dim U_{c}'$ is larger than $\dim U_{c}$ by at most $1$, we know that $f|_{\bigoplus \cU}$ has at most 2 extensions to $\bigoplus \cU'$ that are \textup{SR}-satisfying. Since $\dim L'$ is larger than $\dim L$ by at most 1, we know that $f|_{L}$ has at most 2 linear extensions to $L'$. Therefore, $f|_{L+\bigoplus \cU}$ has at most 4 extensions to $L'+\toplus \cU'$ that have an \textup{SR}-satisfying restriction on $\toplus \cU'$. 
\end{proof}

\subsection{Soundness Analysis}\label{subsec:inner_soundness}

The goal of this subsection is to prove~\Cref{lem:innerpcp_soundness}. Throughout this subsection, we fix a 3-uniform hypergraph $H$, positive integers $k,\ell_{1},\ell_{2},t,s$, a nonnegative integer $q_{0}$ and a parameter $\beta\in (0,1)$. We consider the game $\GameCa[H,k,\ell_{1},\ell_{2},t,s,\beta]$ as defined in~\Cref{def:game3a}. 

At a high-level, the proof of~\Cref{lem:innerpcp_soundness} proceeds in four steps:
\begin{enumerate}
    \item \textbf{Constructing tables:} recall that for each question $(\cE,\cU,L+\cN(\cE,\cU))$ to Alice, her response is a linear function $L+\cN(\cE,\cU)\to\mathbb{F}_2$ which, in particular, restricts to a linear function $L\rightarrow \bF_{2}$. For each fixed pair $(\cE,\cU)$, these
    responses uniquely determine a table $F[\cE,\cU,\cdot]$ that
    assigns to each linear subspace $L\subseteq\cV(\cE)$ either
    an element of $L^*$ or $\nil$. Similarly, for each fixed pair
    $(\cE',\cU')$, we record Bob's responses in a table
    $G[\cE',\cU',\cdot]$.
    
    \item \textbf{Decoding strategy for Alice:} our first goal is to use the consistency between $F[\cdot]$
    and $G[\cdot]$, together with the Cauchy--Schwarz inequality
    as in~\Cref{lem:preparing_middle}, to show that, on average,
    $F[\cE,\cU,\cdot]$ passes the Grassmann agreement test with
    noticeable probability. Because of the critical coordinate,
    the distribution of $L$ induced by the construction is not
    exactly uniform over $\mathsf{Gr}(\cV(\cE),\ell_{1})$.
    Nevertheless, lazy parallel repetition ensures that this
    distribution is close to uniform; we refer to this fact as
    the \emph{covering property of mixed-type subspaces}.

    Applying the Grassmann decoding theorem
    (\Cref{thm:grassmann_linear_dec}) then yields a family $\cQ$
    of constant-dimensional zoom-ins in $\cV(\cE)$ of noticeable
    relative size. For each $Q\in\cQ$, there is a zoom-out
    $W\subseteq\cV(\cE)$ of constant codimension and a global
    linear function $f_{\cE,\cU,Q,W}\colon W\to\bF_{2}$ such that
    $F[\cE,\cU,L]=f_{\cE,\cU,Q,W}|_L$ for a noticeable fraction
    of the $\ell_{1}$-dimensional subspaces $L$ satisfying
    $Q\subseteq L\subseteq W$.

    Alice uses this observation to construct her strategy
    in $\GameBc$. More precisely, she chooses $Q$ using her advice, and chooses $W$ for which a global linear function as above exists. Strictly speaking, she has to choose $W$ to be
    inclusion-maximal among the subspaces with this property; this is to ensure (via~\Cref{lem:dkkms_list_bound}) that for a fixed $Q$, the list of possible answers Alice may give is small.
    
    \item \textbf{Decoding strategy for Bob:} for fixed $\cE'$ and $\cU'$, the distribution of $L'$ induced by the construction is not close to uniform over $\mathsf{Gr}(\cV(\cE'),\ell_{1}+1)$, but only close to uniform over the mixed-type Grassmann domain $\mathsf{MGr}(\cE',\ell_{1}+1,2^{\ell_{1}-1}s)$. Unlike in Alice's case, for this mixed-type Grassmann domain we have neither a decoding theorem like~\Cref{thm:grassmann_linear_dec} nor a list-decoding bound like~\Cref{lem:dkkms_list_bound}. To argue that Bob can still construct a strategy for $\GameBc$ from his table $G[\cE',\cU',\cdot]$, we utilize the consistency between $F[\cdot]$ and $G[\cdot]$ to show that Alice's decoding is often also good for Bob.

    \item \textbf{Consistency of decodings:}
    the preceding step only established that Alice's decoding has noticeable agreement with Bob's table $G[\cdot]$. It does not guarantee that Bob can actually decode from $G[\cdot]$ a global function
    compatible with Alice's decoding. This is because there may be a very large number of global linear functions that have noticeable agreement with $G[\cE',\cU',\cdot]$, and Bob may not be able to identify the one that corresponds to Alice's decoding:
    the preceding step guarantees only that such a function exists.
    
    The usual approach is to bound the number of global linear
    functions with noticeable agreement with
    $G[\cE',\cU',\cdot]$. If this were an ordinary Grassmann
    table, \Cref{lem:dkkms_list_bound} would provide the desired
    bound. However, our $G[\cE',\cU',\cdot]$ is defined on a mixed-type
    Grassmann domain, for which we do not have an analogous
    list-decoding bound, although we expect one to hold.
    
    We work around this issue by exploiting the consistency between $F[\cdot]$ and $G[\cdot]$ a second time. Without bounding the number of possible decodings Bob may choose from, we let Bob choose an arbitrary global function from his potentially unbounded list, and then ``transfer'' his decoding back to Alice's side. Roughly speaking, we argue that whichever global functions Bob chooses for the pairs $(\cE',\cU')$, a noticeable fraction of them will ``appear'' in Alice's list (of possible decodings from $F[\cdot]$) for some $(\cE,\cU)$. Since Alice's lists
    have bounded size by~\Cref{lem:dkkms_list_bound}, she can
    select the correct candidates (i.e. those that are transferred from Bob's side) with probability
    bounded away from zero.
    
    The downside is that it is harder to ``transfer'' decodings back from Bob to Alice (than we did from Alice to Bob). The primary reason is that Bob's decoded global functions are defined on $\cV(\cE')$, which is smaller than Alice's space $\cV(\cE)$ (by quite a lot). A single decoded function
    $g\colon\cV(\cE')\to\bF_{2}$ is therefore insufficient
    for such a transfer. Instead, we use the fact that Bob
    has chosen a decoding for \emph{every} suitable pair
    $(\cE',\cU')$. We show that sufficiently many of these
    decodings, corresponding to different pairs
    $(\cE',\cU')$, can be patched together into a linear
    function $\cV(\cE)\to\bF_{2}$ and \emph{jointly}
    transferred to Alice's table $F[\cE,\cU,\cdot]$.
\end{enumerate}

\subsubsection{Notation Setup}
We begin by setting up some notational conventions that will be used throughout~\Cref{subsec:inner_soundness}. We will use the following  abuse of notation:
\begin{notation}
Recall that in~\Cref{def:game3a}, the sequence $\cE'=(E_{1}',\dots,E_{t}')$ is a $c$-critical descendant (see~\Cref{def:critical_descendant}) of $\cE$. It is not necessarily the case that one can tell the critical coordinate $c$ from $\cE'$ itself: when $|E_{c}'|$ happens to be $3k$ (although this happens with only a very small probability), the sequence $\cE'$ is equal to $\cE$. However, in the rest of the proof, we assume that the critical coordinate is implicit in the notation $\cE'$. For instance, in $\GameBc$, the question sent to Bob can be written as $(\cE',\cA,\cU')$.
\end{notation}

The next two notational definitions below are parallels of~\Cref{ntn:D_mid,ntn:D_mid_condition} in~\Cref{subsec:middle_soundness}.

\begin{notation}
We use $\cD_{\inn}$ to denote the joint distribution of $(\cE,\cU,\cE',\cU',\cA)$ sampled by the verifier in $\GameCa[H,k,\ell_{1},\ell_{2},t,s,\beta]$. We also need a variant of $\cD_{\inn}$ where the five components are expanded into seven components: denote $\cD_{\inn}^{+}$ to be the joint distribution of $(\cE,\cU,L,Q,\cE',\cU',L')$, where $(\cE,\cU,\cE',\cU',\cA)$ is sampled from $\cD_{\inn}$ and 
\[Q:=\sigma(\cA_{[q_{0}]}),\qquad L:=\sigma(\cA_{[\ell_{1}]}),\qquad L'=\sigma(\cA_{[\ell_{1}+1]}).\]
We use the symbol $*$ to denote marginalization over the corresponding component. For instance, sampling from the marginal distribution of $(\cU,\cE')$ under $\cD_{\inn}$ is denoted by $(*,\cU,\cE',*,*)\sim \cD_{\inn}$.
\end{notation}

\begin{notation}
We will often condition the distribution $\cD_{\inn}$ on some subset of the five components $(\cE,\cU,\cE',\cU',\cA)$ and consider the induced joint distribution of the remaining components. For instance, if $\cE$ and $\cU'$ are fixed, then the conditional joint distribution of $(\cU,\cE')$ is denoted by
\[
\Pru{(\cE,\cU,\cE',\cU',*)\sim \cD_{\inn}}{\cdot\,\big|\,\cE,\cU'}.
\]
Similar conditioning notation will be used for the seven-component joint distribution $\cD_{\inn}^{+}$.
\end{notation}

The notions of \emph{Alice-table} and \emph{Bob-table} to be defined in~\Cref{ntn:inner_Alice_table} are slightly more nuanced than their counterparts in~\Cref{ntn:middle_Alice-table}. In particular, the entries of an Alice-table do not exactly correspond to the possible questions Alice may receive.

\begin{notation}\label{ntn:inner_unfolded_Alice_question}
For any possible sample $(\cE,\cU,*,*,*)\sim \cD_{\inn}$, the pair $(\cE,\cU)$ is called an \emph{Alice-packet}. Similarly, a pair $(\cE',\cU')$ is called a \emph{Bob-packet} if $(*,*,\cE',\cU',*)$ has nonzero probability under $\cD_{\inn}$. For any possible sample $(\cE,\cU,L,*,*,*,*)\sim \cD_{\inn}^{+}$, the tuple $(\cE,\cU,L)$ is called an \emph{unfolded-Alice-question}. Similarly, a tuple $(\cE',\cU',L')$ is called an \emph{unfolded-Bob-question} if $(*,*,*,*,\cE',\cU',L')$ has nonzero probability under $\cD_{\inn}^{+}$.
\end{notation}

\begin{notation}\label{ntn:inner_Alice_table}
A table $F[\cdot]$ is called an Alice-table if for each Alice-packet $(\cE,\cU)$ and each linear subspace $L\subseteq \cV(\cE)$, the entry $F[\cE,\cU,L]$ is either $\nil$ or a linear functional $F[\cE,\cU,L]:L\rightarrow \bF_{2}$. Similarly, a table $G[\cdot]$ is called a Bob-table if for each Bob-packet $(\cE',\cU')$ and each linear subspace $L'\subseteq \cV(\cE')$, the entry $G[\cE',\cU',L']$ is either $\nil$ or a linear functional $G[\cE',\cU',L']:L'\rightarrow \bF_{2}$.
\end{notation}

The following definition is a parallel of~\Cref{ntn:middle_vrf}.
\begin{notation}\label{ntn:inner_Vrf}
Fix an Alice-table $F[\cdot]$ and a Bob-table $G[\cdot]$. For any tuple $(\cE,\cU,L,*,\cE',\cU',L')$ with nonzero probability under $\cD_{\inn}^{+}$, we define
\[
\operatorname{Vrf}_{F,G}(\cE,\cU,L,\cE',\cU',L')=\begin{cases}
1,&\text{if }\left.G[\cE',\cU',L']\right|_{L}=F[\cE,\cU,L],\\
0,&\text{otherwise}.
\end{cases}
\]
We will often sample $(\cE,\cU,L,*,\cE',\cU',L')$ from $\cD_{\inn}^{+}$ conditioned on some subset of its components. For instance, if $\cE$ and $\cU'$ are fixed, then the conditional expectation of $\operatorname{Vrf}_{F,G}(\cE,\cU,L,\cE',\cU',L')$ is denoted by
\[
\operatorname{Vrf}_{F,G}(\cE,*,*,*,\cU',*)=\Pru{(\cE,\cU,L,*,\cE',\cU',L')\sim \cD_{\inn}^{+}}{\left.G[\cE',\cU',L']\right|_{L}=F[\cE,\cU,L]\,\Big|\,\cE,\cU'}.
\]
When none of the components is fixed, we write $\operatorname{Vrf}_{F,G}=\operatorname{Vrf}_{F,G}(*,*,*,*,*,*)$ for short.
\end{notation}

\subsubsection{Covering Property of Mixed-Type Subspaces}\label{subsubsec:inner_PCP_covering}
Note that in~\Cref{def:game3a}, the subspace $L$ sent to Alice is not uniformly chosen from her space $\cV(\cE)$, and instead it is chosen from $\cV(\cE')$ with a random critical coordinate $c$. For the test we perform to make sense from Alice's perspective (as well as for the advice feature we arranged to work) it is crucial that the distribution of $L$ is statistically close to being uniform over subspaces of her ambient space. The following definition and lemma make this assertion more precise, and in fact they show that the two distributions are close even conditioning on a small number of fixed vectors being contained in $L$. 

\begin{definition}\label{def:TVD_B}
For a fixed Alice-packet $(\cE,\cU)$ and a sequence of vectors 
\[B=(a^{(1)},\dots,a^{(q_{0})})\in (\cV(\cE))^{q_{0}},\] 
let $\operatorname{TVD}_{\cE,\cU}(B)$ denote the total variation distance between
\begin{enumerate}[label=(\arabic*)]
\item the distribution of $\rho(\cA)$ when $(\cE,\cU,*,*,\cA)\sim \cD_{\inn}$ is sampled conditioned on the fixed Alice-packet $(\cE,\cU)$ and the event $\left\{\rho(\cA_{[q_{0}]})=B\right\}$ (if this event happens with zero probability, let $\operatorname{TVD}_{\cE,\cU}(B)=1$); and
\item the uniform distribution over the set
\[\left\{\bigl(a^{(1)},\dots,a^{(q_{0})},b^{(q_{0}+1)},\dots,b^{(2\ell_{1})}\bigr)\,\middle|\,b^{(q_{0}+1)},\dots,b^{(2\ell_{1})}\in \cV(\cE)\right\}\subseteq (\cV(\cE))^{2\ell_{1}}.\]
\end{enumerate}
\end{definition}

\begin{lemma}\label{lem:inner_covering_advice}
Assume~\Cref{cond:parameters}. For any fixed Alice-packet $(\cE,\cU)$, we have
\[
\Exu{(\cE,\cU,*,*,\cA)\sim \cD_{\inn}}{\operatorname{TVD}_{\cE,\cU}\bigl(\rho(\cA_{[q_{0}]})\bigr)\,\Big|\,\cE,\cU}\leq t^{-1/3}.
\]
\end{lemma}

\begin{proof}
Towards the goal of applying~\Cref{prop:covering_inner_pre}, we set up some notation. Write 
\[\cE=(E_{1},\dots,E_{t})\qquad\textup{and}\qquad\cU=(U_{1},\dots,U_{t}).
\] For each $j\in [t]$, define sets $\scrA_{j},\scrB_{j}$ and a map $\pi_{j}:\scrA_{j}\rightarrow \scrB_{j}$ by
\begin{align*}
\scrA_{j}&=\left\{\bigl(a^{(1)},\dots,a^{(2\ell_{1})}\bigr)\,\middle|\,a^{(i)}\in \bF_{2}^{E_{j}}\otimes \bF_{2}^{E_{j}}\text{ for each }i\in [2\ell_{1}]\right\}=\left(\bF_{2}^{E_{j}}\otimes \bF_{2}^{E_{j}}\right)^{2\ell_{1}},\\
\scrB_{j}&=\left(\bF_{2}^{E_{j}}\otimes \bF_{2}^{E_{j}}\right)^{q_{0}},\qquad\text{and}\qquad\pi_{j}: \bigl(a^{(1)},\dots,a^{(2\ell_{1})}\bigr)\mapsto\bigl(a^{(1)},\dots,a^{(q_{0})}\bigr).
\end{align*}
For each $j\in [t]$, let $\mu_{j}$ be the output distribution of the following process:
\begin{enumerate}
\item Sample $E_{j}'\sim \cJ_{\rev}(E_{j},U_{j},\beta,\ell_{2})$.
\item Pick $q=2(2^{2\ell_{1}}-1)s$ vectors $a^{(j,1)},\dots,a^{(j,q)}$ from $\bF_{2}^{E_{j}'}$ independently and uniformly at random. Let $A=\big(a^{(j,1)},\dots,a^{(j,q)}\bigr)$.
\item Output the sequence
\[
\bigl(\Lvec(A,2\ell_{1},e_{1}),\dots,\Lvec(A,2\ell_{1},e_{2\ell_{1}})\bigr)\in \scrA_{j},
\]
where $e_{1},\dots,e_{2\ell_{1}}\in \bF_{2}^{2\ell_{1}}$ are the standard basis vectors.
\end{enumerate}
We can then define a distribution $\mu$ over $\prod_{j=1}^{t}\scrA_{j}$ as in the statement of~\Cref{prop:covering_inner_pre}; it is easy to see (using~\Cref{prop:Joint_inner_independence}) that $\mu$ is identical to the distribution of $\rho(\cA)$ when $(\cE,\cU,*,*,\cA)\sim \cD_{\inn}$ is sampled conditioned on $(\cE,\cU)$.

Now we can apply~\Cref{prop:covering_inner_pre} to $\scrA_{j},\scrB_{j}$, $\pi_{j}$ and $\mu_{j}$ as defined above. The conclusion of~\Cref{prop:covering_inner_pre} says that 
\[
\Exu{(\cE,\cU,*,*,\cA)\sim \cD_{\inn}}{\operatorname{TVD}_{\cE,\cU}\bigl(\rho(\cA_{[q_{0}]})\bigr)\,\Big|\,\cE,\cU}\leq \sqrt{\frac{2^{18k^{2}\ell_{1}}}{4t}}\leq t^{-1/3},
\]
where we used~\Cref{cond:parameters} in the last transition.
\end{proof}

While~\Cref{def:TVD_B} deals with the vector list $\rho(\cA)$, we eventually need to study Alice's space $L$, which is spanned by the first $\ell_{1}$ vectors in $\rho(\cA)$. The next definition is a variation of~\Cref{def:TVD_B} that directly addresses the distribution of $L$ (the reason we started with~\Cref{def:TVD_B} was purely to facilitate the application of~\Cref{prop:covering_inner_pre} in~\Cref{lem:inner_covering_advice}). 

\begin{definition}\label{def:TVD_Q}
For a fixed Alice-packet $(\cE,\cU)$, a linear subspace $Q\subseteq \cV(\cE)$ and a linear subspace $W\subseteq \cV(\cE)$ containing $Q$, let $\operatorname{TVD}_{\cE,\cU,W}(Q)$ be the total variation distance between
\begin{enumerate}[label=(\arabic*)]
\item the distribution of $\sigma(\cA_{[\ell_{1}]})$ when $(\cE,\cU,*,*,\cA)\sim \cD_{\inn}$ is sampled conditioned on the fixed Alice-packet $(\cE,\cU)$ and the event $\left\{\sigma(\cA_{[q_{0}]})=Q\right\}\cap\left\{\sigma(\cA_{[\ell_{1}]})\subseteq W\right\}$ (if this event happens with zero probability, let $\operatorname{TVD}_{\cE,\cU,W}(Q)=1$); and
\item the uniform distribution over the set $\bigl\{L\in \mathsf{Gr}(\cV(\cE),\ell_{1})\,\big|\,Q\subseteq L\subseteq W\bigr\}$.
\end{enumerate}
If $W=\cV(\cE)$, we omit the subscript $W$ and write $\operatorname{TVD}_{\cE,\cU}(Q)=\operatorname{TVD}_{\cE,\cU,\cV(\cE)}(Q)$.
\end{definition}

The following three lemmas about~\Cref{def:TVD_Q} are easy to prove.

\begin{lemma}\label{lem:inner_covering_zoomout}
Assume~\Cref{cond:parameters}. Fix an Alice-packet $(\cE,\cU)$, a linear subspace $Q\in \mathsf{Gr}(\cV(\cE),q_{0})$ and a linear subspace $W\subseteq \cV(\cE)$ containing $Q$. Let $w=\operatorname{codim} W$, and suppose $\dim W\geq \ell_{1}$. We always have
\[
\operatorname{TVD}_{\cE,\cU,W}(Q)\leq 2^{\ell_{1}(w+1)}\cdot\operatorname{TVD}_{\cE,\cU}(Q).
\]
\end{lemma}
\begin{proof}
This is simply due to
\[
\Pru{L\in \mathsf{\Gr}(\cV(\cE),\ell_{1})}{L\subseteq W\,\big|\, Q\subseteq L}\geq 2^{-\ell_{1}(w+1)}.\qedhere
\]
\end{proof}

\begin{lemma}\label{lem:inner_covering_independence}
Fix an Alice-packet $(\cE,\cU)$ and a linear subspace $Q\in \mathsf{Gr}(\cV(\cE),q_{0})$. We have
\[
\operatorname{TVD}_{\cE,\cU}(Q)\leq \Exu{(\cE,\cU,*,*,\cA)\sim \cD_{\inn}}{\operatorname{TVD}_{\cE,\cU}\bigl(\rho(\cA_{[q_{0}]})\bigr)\,\Big|\,\cE,\cU\text{ and }\left\{\sigma(\cA_{[q_{0}]})=Q\right\}}
+2^{\ell_{1}-9k^{2}t}.\]
\end{lemma}

\begin{proof}
Use the following fact to pass from the uniform distribution over $(\cV(\cE))^{\ell_{1}}$ to the uniform distribution over $\mathsf{Gr}(\cV(\cE),\ell_{1})$: for any sequence of vectors $B=(a^{(1)},\dots,a^{(q_{0})})\in (\cV(\cE))^{q_{0}}$ such that $\operatorname{span}(B)=Q$, if $b^{(q_{0}+1)},\dots,b^{(\ell_{1})}\in \cV(\cE)$ are sampled independently and uniformly at random, the probability that $a^{(1)},\dots,a^{(q_{0})},b^{(q_{0}+1)},\dots,b^{(\ell_{1})}$ are not linearly independent is at most $2^{\ell_{1}-9k^{2}t}$.
\end{proof}

\begin{lemma}\label{lem:inner_covering_main}
Assume~\Cref{cond:parameters}. For any fixed Alice-packet $(\cE,\cU)$, we have
\[
\Exu{(\cE,\cU,*,*,\cA)\sim \cD_{\inn}}{\operatorname{TVD}_{\cE,\cU}\bigl(\sigma(\cA_{[q_{0}]})\bigr)\,\Big|\,\cE,\cU}\leq t^{-1/4}.
\]
\end{lemma}

\begin{proof}
Using~\Cref{lem:inner_covering_independence} followed by~\Cref{lem:inner_covering_advice}, we have
\begin{align*}
&\quad\Exu{(\cE,\cU,*,*,\cA)\sim \cD_{\inn}}{\operatorname{TVD}_{\cE,\cU}\bigl(\sigma(\cA_{[q_{0}]})\bigr)\,\Big|\,\cE,\cU}\\
&\leq \Exu{(\cE,\cU,*,*,\cA)\sim \cD_{\inn}}{\operatorname{TVD}_{\cE,\cU}\bigl(\rho(\cA_{[q_{0}]})\bigr)\,\Big|\,\cE,\cU}
+2^{\ell_{1}-9k^{2}t}+\Pru{(\cE,\cU,*,*,\cA)\sim \cD_{\inn}}{\dim \sigma(\cA_{[q_{0}]})<q_{0}\,\Big|\,\cE,\cU}\\
&\leq t^{-1/3}+2^{\ell_{1}-9k^{2}t}+2^{q_{0}-9k^{2}(t-1)}\leq t^{-1/4},
\end{align*}
where we used~\Cref{cond:parameters} in the last transition.
\end{proof}

To summarize the discussion, we define the notion of a smooth zoom-in, which is a subspace $Q$ such that conditioned on $(\cE,\cU,Q)$, the distribution of $L$ in $(\cE,\cU,L,Q,*,*,*)\sim \cD_{\inn}^{+}$ is close to the uniform distribution over the set $\bigl\{L\in \mathsf{Gr}(\cV(\cE),\ell_{1})\,\big|\,Q\subseteq L\bigr\}$.
\begin{definition}\label{def:smooth_Q}
For a fixed Alice-packet $(\cE,\cU)$, a linear subspace $Q\subseteq \cV(\cE)$ is said to be \emph{smooth} with respect to $(\cE,\cU)$ if $\dim Q=q_{0}$ and $\operatorname{TVD}_{\cE,\cU}(Q)\leq t^{-1/8}$. 
\end{definition}

The following corollary shows that $Q$ is smooth with probability close to $1$.
\begin{corollary}\label{cor:smooth_Q_prob}
Assume~\Cref{cond:parameters}. For any fixed Alice-packet $(\cE,\cU)$, we have
\[
\Pru{(\cE,\cU,*,Q,*,*,*)\sim \cD_{\inn}^{+}}{Q\textup{ is smooth with respect to }(\cE,\cU)\,\Big|\,\cE,\cU}\geq 1-t^{-1/8}.
\]
\end{corollary}
\begin{proof}
    Immediate from~\Cref{lem:inner_covering_main} and Markov's inequality.
\end{proof}

Finally, we record a covering property without advice, meaning that the distribution of $L$ in $(\cE,\cU,L,*,*,*,*)\sim \cD_{\inn}^{+}$ (without conditioning on $Q$) is close to uniform over $\mathsf{Gr}(\cV(\cE),\ell_{1})$.

\begin{proposition}\label{prop:covering_inner_no_advice}
Assume~\Cref{cond:parameters}. Fix a positive integer $\ell\leq 2\ell_{1}$. For any fixed Alice-packet $(\cE,\cU)$, if we sample $(\cE,\cU,*,*,\cA)\sim \cD_{\inn}$ conditioned on $(\cE,\cU)$, the resulting distribution of $\sigma(\cA_{[\ell]})$ has total variation distance at most $t^{-1/4}$ to the uniform distribution over $\mathsf{Gr}(\cV(\cE),\ell)$.
\end{proposition}

\begin{proof}
The proof is entirely analogous to the (combined) proofs of~\Cref{lem:inner_covering_advice,lem:inner_covering_independence}, and is hence omitted. The main difference is that when one invokes~\Cref{prop:covering_inner_pre}, this time one can take $\scrB_{j}$ to be a singleton set for each $j\in [t]$.  
\end{proof}

\subsubsection{Alice's Existence of Decoding}
The following lemma, analogous to~\Cref{lem:preparing_middle}, shows that Alice's table passes the Grassmann agreement test with noticeable probability.
\begin{lemma}\label{lem:inner_prepare_decode}
Assume~\Cref{cond:parameters}. Let $F[\cdot]$ be an Alice-table and $G[\cdot]$ be a Bob-table. For any Alice-packet $(\cE,\cU)$, we have
\[
\Pru{(L_{1},L_{2})\sim \cG(\cV(\cE),\ell_{1})}{\left.F[\cE,\cU,L_{1}]\right|_{L_{1}\cap L_{2}}=\left.F[\cE,\cU,L_{2}]\right|_{L_{1}\cap L_{2}}\Big.}\geq \operatorname{Vrf}_{F,G}(\cE,\cU,*,*,*,*)^{2}-2^{-\ell_{1}}.
\]
\end{lemma}

\begin{proof}
We organize the proof into the following two steps.

\paragraph{Step 1: fixing $(\cE',\cU',L'
)$.}
Fix a Bob-packet $(\cE',\cU')$ and a subspace $L'\in \mathsf{Gr}(\cV(\cE'),\ell_{1}+1)$ such that 
\[
\Pru{(\cE,\cU,\cE',\cU',\cA)\sim \cD_{\inn}}{\sigma(\cA_{[\ell_{1}+1]})=L'\,\Big|\,\cE,\cU,\cE',\cU'}>0.
\]
By~\Cref{def:game3a}, it is clear that if we sample $(\cE,\cU,\cE',\cU',\cA)\sim \cD_{\inn}$ conditioned on $\cE,\cU,\cE',\cU'$ and the event $\left\{\sigma(\cA_{[\ell_{1}+1]})=L'\right\}$, the resulting distribution of $\sigma(\cA_{[\ell_{1}]})$ is uniform over the set $\mathsf{Gr}(L',\ell_{1})$. Therefore, we have
\begin{align}
&\quad \operatorname{Vrf}_{F,G}(\cE,\cU,*,\cE',\cU',L')^{2}\nonumber\\
&=\Exu{(\cE,\cU,\cE',\cU',\cA)\sim \cD_{\inn}}{\operatorname{Vrf}_{F,G}(\cE,\cU,*,\cE',\cU',L')\,\Big|\,\cE,\cU,\cE',\cU'\text{ and }\{\sigma(\cA_{[\ell_{1}+1]})=L'\}}^{2}\nonumber\\
&=\Exu{L_{1},L_{2}\in \mathsf{Gr}(L',\ell_{1})}{\operatorname{Vrf}_{F,G}(\cE,\cU,L_{1},\cE',\cU',L')\operatorname{Vrf}_{F,G}(\cE,\cU,L_{2},\cE',\cU',L')}.\label{eq:preparing_inner_LfromLprime}
\end{align}
Expanding the right-hand side of~\Cref{eq:preparing_inner_LfromLprime}, we get
\begin{align}
&\quad \operatorname{Vrf}_{F,G}(\cE,\cU,*,\cE',\cU',L')^{2}\nonumber\\
&=\Pru{L_{1},L_{2}\in \mathsf{Gr}(L',\ell_{1})}{F[\cE,\cU,L_{1}]=\left.G[\cE',\cU',L']\right|_{L_{1}}\text{ and }F[\cE,\cU,L_{2}]=\left.G[\cE',\cU',L']\right|_{L_{2}}}\nonumber\\
&\leq \Pru{L_{1},L_{2}\in \mathsf{Gr}(L',\ell_{1})}{\left.F[\cE,\cU,L_{1}]\right|_{L_{1}\cap L_{2}}=\left.F[\cE,\cU,L_{2}]\right|_{L_{1}\cap L_{2}}\Big.}.\label{eq:preparing_inner_L1L2_fixed_EU}
\end{align}

\paragraph{Step 2: $(\cE',\cU',L')$ unfixed.} By~\Cref{prop:covering_inner_no_advice}, we know that conditioned on $(\cE,\cU)$, the distribution of $\sigma(\cA_{[\ell_{1}+1]})$ in the sample $(\cE,\cU,*,*,\cA)\sim \cD_{\inn}$ is $t^{-1/4}$-close to uniform over $\mathsf{Gr}(\cV(\cE),\ell_{1}+1)$. Therefore, we have 
\begin{align*}
&\quad\operatorname{Vrf}_{F,G}(\cE,\cU,*,*,*,*)^{2}\nonumber\\
&= \Exu{(\cE,\cU,\cE',\cU',\cA)\sim \cD_{\inn}}{\operatorname{Vrf}_{F,G}\left(\cE,\cU,*,\cE',\cU',\sigma(\cA_{[\ell_{1}+1]})\right)\,\Big|\,\cE,\cU}^{2}\tag{by definition}\nonumber\\
&\leq 
\Exu{(\cE,\cU,\cE',\cU',\cA)\sim \cD_{\inn}}{\operatorname{Vrf}_{F,G}\left(\cE,\cU,*,\cE',\cU',\sigma(\cA_{[\ell_{1}+1]})\right)^{2}\,\Big|\,\cE,\cU}\tag{by Cauchy-Schwarz}\nonumber\\
&\leq \Exu{L'\in \mathsf{Gr}(\cV(\cE),\ell_{1}+1)}{\Exu{(\cE,\cU,\cE',\cU',\cA)\sim \cD_{\inn}}{\operatorname{Vrf}_{F,G}(\cE,\cU,*,\cE',\cU',L')^{2}\,\Big|\,\cE,\cU\text{ and }\{\sigma(\cA_{[\ell_{1}+1]})=L'\}}}+t^{-1/4}.
\end{align*}
We can now apply~\Cref{eq:preparing_inner_L1L2_fixed_EU} to the above and get
\begin{align*}
&\quad\operatorname{Vrf}_{F,G}(\cE,\cU,*,*,*,*)^2\\
&\leq \Exu{L'\in \mathsf{Gr}(\cV(\cE),\ell_{1}+1)}{\Pru{L_{1},L_{2}\in \mathsf{Gr}(L',\ell_{1})}{\left.F[\cE,\cU,L_{1}]\right|_{L_{1}\cap L_{2}}=\left.F[\cE,\cU,L_{2}]\right|_{L_{1}\cap L_{2}}\Big.}}+t^{-1/4}\\
&\leq \Pru{(L_{1},L_{2})\sim \cG(\cV(\cE),\ell_{1})}{\left.F[\cE,\cU,L_{1}]\right|_{L_{1}\cap L_{2}}=\left.F[\cE,\cU,L_{2}]\right|_{L_{1}\cap L_{2}}\Big.}+\frac{1}{2^{\ell_{1}+1}-1}+t^{-1/4}\\
&\leq \Pru{(L_{1},L_{2})\sim \cG(\cV(\cE),\ell_{1})}{\left.F[\cE,\cU,L_{1}]\right|_{L_{1}\cap L_{2}}=\left.F[\cE,\cU,L_{2}]\right|_{L_{1}\cap L_{2}}\Big.}+2^{-\ell_{1}},\tag{using~\Cref{cond:parameters}}
\end{align*}
as desired.
\end{proof}

When we eventually prove the soundness of $\GameCa$, we will need to ``decode'' a strategy for $\GameBc$ based on tables $F[\cdot]$ and $G[\cdot]$. The following definition captures the type of tuples $(\cE,\cU,Q)$ for which Alice can successfully decode an answer in $\GameBc$.

\begin{definition}\label{def:answerable_Alice_question}
Let $F[\cdot]$ be an Alice table, and let $(\cE,\cU)$ be an Alice-packet. For any linear subspace $Q\subseteq \cV(\cE)$ with $\dim Q\leq q_{0}$, we say that the tuple $(\cE,\cU,Q)$ is \emph{answerable} with respect to $F[\cdot]$ if $\dim Q=q_{0}$ and the following statement holds: there exists a linear subspace $W\in \mathsf{Gr}(\cV(\cE),9k^{2}t-q_{0})$ containing $Q$ and a linear map $f:W\rightarrow \bF_{2}$ such that
\[
\Pru{L\in \mathsf{Gr}(\cV(\cE),\ell_{1})}{F[\cE,\cU,L]=\left.f\right|_{L}\,\Big|\,Q\subseteq L\subseteq W}\geq \exp(-q_{0}^{3}).
\]
\end{definition}

Combining~\Cref{lem:inner_prepare_decode} and the Grassmann decoding theorem (\Cref{thm:grassmann_linear_dec}) shows that with noticeable probability, Alice can decode an answer in $\GameBc$ from the table $F[\cdot]$.

\begin{lemma}\label{lem:answerable_Alice_question}
Assume~\Cref{cond:parameters}. Let $F[\cdot]$ be an Alice-table and let $G[\cdot]$ be a Bob-table. For any Alice-packet $(\cE,\cU)$ such that $\operatorname{Vrf}_{F,G}(\cE,\cU,*,*,*,*)\geq 2^{-q_{0}/3}$, we have
\[
\Pru{(\cE,\cU,*,Q,*,*,*)\sim \cD_{\inn}^{+}}{(\cE,\cU,Q)\textup{ is answerable with respect to }F[\cdot]\,\Big|\,\cE,\cU}\geq 4^{-\ell_{1}^{2}}.
\]
\end{lemma}

\begin{proof}
Since $\operatorname{Vrf}_{F,G}(\cE,\cU,*,*,*,*)\geq 2^{-q_{0}/3}$,~\Cref{lem:inner_prepare_decode} implies
\[
\Pru{(L_{1},L_{2})\sim \cG(\cV(\cE),\ell_{1})}{\left.F[\cE,\cU,L_{1}]\right|_{L_{1}\cap L_{2}}=\left.F[\cE,\cU,L_{2}]\right|_{L_{1}\cap L_{2}}}\geq 2^{-2q_{0}/3}-2^{-\ell_{1}}\geq 2^{-q_{0}+5},
\]
where we used~\Cref{cond:parameters} in the last transition. Applying~\Cref{thm:grassmann_linear_dec} to the table $F[\cE,\cU,L]$, where $L$ ranges in $\mathsf{Gr}(\cV(\cE),\ell_{1})$, yields a collection of subspaces $\cQ\subseteq \mathsf{Gr}(\cV(\cE),q_{0})$ such that:
\begin{enumerate}[label=(\arabic*)]
\item $|\cQ|\geq 2^{-q_{0}-\ell_{1}^{2}}\left|\mathsf{Gr}(\cV(\cE),q_0)\right|$;
\item for each $Q\in \cQ$, there exists a linear subspace $W_{Q}\in  \mathsf{Gr}(\cV(\cE),9k^{2}t-q_{0})$ and a linear map $f_{Q}:W_{Q}\rightarrow \bF_{2}$ such that
\[
\Pru{L\in \mathsf{Gr}(\cV(\cE),\ell_{1})}{F[\cE,\cU,L]=\left.f_{Q}\right|_{L}\,\Big|\,Q\subseteq L\subseteq W_{Q}}\geq \exp(-Cq_{0}^{2})\geq \exp(-q_{0}^{3}),
\]
where $C$ is an absolute constant and the last transition is due to~\Cref{cond:parameters}. By~\Cref{def:answerable_Alice_question}, this means that for each $Q\in \cQ$, the tuple $(\cE,\cU,Q)$ is answerable with respect to $F[\cdot]$.
\end{enumerate}

Therefore, it now suffices to show that
\[
\Pru{(\cE,\cU,*,*,\cA)\sim \cD_{\inn}}{\sigma(\cA_{[q_{0}]})\in \cQ\,\Big|\,\cE,\cU}\geq 4^{-\ell_{1}^{2}}.
\]
But it follows immediately from~\Cref{prop:covering_inner_no_advice} that
\[
\Pru{(\cE,\cU,*,*,\cA)\sim \cD_{\inn}}{\sigma(\cA_{[q_{0}]})\in \cQ\,\Big|\,\cE,\cU}\geq \Pru{Q\in\mathsf{Gr}(\cV(\cE),q_{0})}{Q\in \cQ}-t^{-1/4}\geq 2^{-q_{0}-\ell_{1}^{2}}-t^{-1/4}\geq 4^{-\ell_{1}^{2}},
\]
where we used~\Cref{cond:parameters} in the last transition.
\end{proof}

%In conclusion, we have shown that if the $F,G$ are Alice and Bob tables that agree with noticeable probability, then with noticeable probability Alice can use her zoom-in $Q$ to come up with a global linear function $f$ agreeing. We next show a similar assertion to Bob.

\subsubsection{Bob's Existence of Decoding}\label{subsubsec:Bob_existence_decoding}

We next develop tools for showing that Bob can also use the table $G[\cdot]$ to obtain answers in $\GameBc$. As mentioned earlier, the Bob-table $G[\cdot]$ uses mixed-type Grassmann graphs (instead of standard Grassmann graphs) to encode answers, and we do not have a decoding theorem for mixed-type Grassmann graphs. As a workaround, we will argue that Bob can still decode by ``transferring'' Alice's decoding to Bob. We stress that Bob will still decode from his table $G[\cdot]$, not Alice's table $F[\cdot]$; the ``transfer'' happens only in the sense that we use the properties satisfied by $F[\cdot]$ to deduce a similar property of $G[\cdot]$.

\begin{definition}\label{def:pre_answerable_Bob_question}
Let $G[\cdot]$ be a Bob table, and let $(\cE',\cU')$ be a Bob-packet with critical index $c\in [t]$. For any linear subspace $Q\subseteq \cV(\cE')$ with $\dim Q\leq q_{0}$, we say that the tuple $(\cE',\cU',Q)$ is \emph{pre-answerable} with respect to $G[\cdot]$ if the following statements hold:
\begin{enumerate}[label=(\arabic*)]
\item The tuple $(*,*,*,Q,\cE',\cU',*)$ has nonzero probability under $\cD_{\inn}^{+}$. In other words, $Q=\sigma(\cA)$ for some $\cA\in\scrA_{c}(\cE',q_{0},2^{2\ell_{1}-q_{0}}s)$.
\item There exists a linear subspace $W\subseteq \cV(\cE')$ containing $Q$ with $\operatorname{codim} W\leq q_{0}$ and a linear map $f:W\rightarrow \bF_{2}$ such that
\[
\Pru{\cA\in \scrA_{c}(\cE',\ell_{1}+1,s')}{G[\cE',\cU',\sigma(\cA)]=f|_{\sigma(\cA)}\,\Big|\,\sigma(\cA_{[q_{0}]})=Q\textup{ and }\sigma(\cA)\subseteq W}\geq \exp(-q_{0}^{4}),
\]
where $s'=2^{\ell_{1}-1}s$.
\end{enumerate}
\end{definition}

The next definition formalizes the type of ``pruning'' we need for Alice's decoding to be able to ``transfer'' to Bob's side.

\begin{definition}\label{def:transferability}
Let $\eta\in (0,1)$ be a threshold parameter. An Alice-table $F[\cdot]$ and a Bob-table $G[\cdot]$ are said to be \emph{$\eta$-transferable} (with each other) if the following statements hold:
\begin{enumerate}[label=(\arabic*)]
\item For any unfolded-Alice-question $(\cE,\cU,L)$ such that $F[\cE,\cU,L]\neq \nil$, we have \[\operatorname{Vrf}_{F,G}(\cE,\cU,L,*,*,*)\geq \eta.\]
\item For any unfolded-Bob-question $(\cE',\cU',L')$ such that $G[\cE',\cU',L']\neq \nil$, we have
\[\operatorname{Vrf}_{F,G}(*,*,*,\cE',\cU',L')\geq \eta.\]
\end{enumerate}
\end{definition}

The next lemma shows how to transfer Alice's decoding to Bob's side, if the tables $F[\cdot]$ and $G[\cdot]$ are transferable with each other.

\begin{lemma}\label{lem:Bob_pre_answerable}
Assume~\Cref{cond:parameters}. Let $F[\cdot]$ be an Alice-table, and let $G[\cdot]$ be a Bob-table that is $2^{-q_{0}}$-transferable with $F[\cdot]$. Fix an Alice-packet $(\cE,\cU)$ and a linear subspace $Q\in  \mathsf{Gr}(\cV(\cE),q_{0})$ such that
\begin{enumerate}[label=(\arabic*)]
\item $Q$ is smooth with respect to $(\cE,\cU)$ (as per~\Cref{def:smooth_Q}), and
\item $(\cE,\cU,Q)$ is answerable with respect to $F[\cdot]$ (as per~\Cref{def:answerable_Alice_question}). 
\end{enumerate}
Then
\begin{equation}\label{eq:Bob_pre_answerable_prob}
\Pru{(\cE,\cU,*,Q,\cE',\cU',*)\sim \cD_{\inn}^{+}}{(\cE',\cU',Q)\textup{ is pre-answerable with respect to }G[\cdot]\,\Big|\,\cE,\cU,Q}\geq 4^{-\ell_{1}^{2}}.
\end{equation}
\end{lemma}

\begin{proof}

We organize the proof into the following three steps.

\paragraph{Step 1: transference.} Since $(\cE,\cU,Q)$ is answerable with respect to $F[\cdot]$, by~\Cref{def:answerable_Alice_question} there exists a linear subspace $W\in \mathsf{Gr}(\cV(\cE),9k^{2}t-q_{0})$ and a linear map $f:W\rightarrow \bF_{2}$ such that
\[
\Pru{L\in \mathsf{Gr}(\cV(\cE),\ell_{1})}{F[\cE,\cU,L]=\left.f\right|_{L}\,\Big|\,Q\subseteq L\subseteq W}\geq \exp(-q_{0}^{3}).
\]
We then have
\begin{align*}
&\quad\Pru{(\cE,\cU,L,Q,*,*,*)\sim \cD_{\inn}^{+}}{F[\cE,\cU,L]=\left.f\right|_{L}\,\Big|\,\cE,\cU,Q\text{ and }\{L\subseteq W\}}\\
&\geq \Pru{L\in \mathsf{Gr}(\cV(\cE),\ell_{1})}{F[\cE,\cU,L]=f|_{L}\,\Big|\, Q\subseteq L\subseteq W}-\operatorname{TVD}_{\cE,\cU,W}(Q)\tag{by~\Cref{def:TVD_Q}}\\
&\geq \exp(-q_{0}^{3})-2^{\ell_{1}(q_{0}+1)}\operatorname{TVD}_{\cE,\cU}(Q)\tag{by~\Cref{lem:inner_covering_zoomout}}\\
&\geq \exp(-q_{0}^{3})-2^{\ell_{1}(q_{0}+1)}t^{-1/8}\geq \exp(-2q_{0}^{3})\tag{by~\Cref{def:smooth_Q} and~\Cref{cond:parameters}}.
\end{align*}
By the definition of $\operatorname{Vrf}_{F,G}(\cE,\cU,L,\cE',\cU',L')$ in~\Cref{ntn:inner_Vrf}, we have
\begin{align*}
&\quad\Pru{(\cE,\cU,L,Q,\cE',\cU',L')\sim \cD_{\inn}^{+}}{\left.G[\cE',\cU',L']\right|_{L}=f|_{L}\,\Big|\,\cE,\cU,Q\textup{ and }\{L\subseteq W\}}\\
&\geq \Pru{(\cE,\cU,L,Q,\cE',\cU',L')\sim\cD_{\inn}^{+}}{F[\cE,\cU,L]=f|_{L}\textup{ and }\operatorname{Vrf}_{F,G}(\cE,\cU,L,\cE',\cU',L')=1\,\Big|\,\cE,\cU,Q\textup{ and }\{L\subseteq W\}}\\
&\geq 2^{-q_{0}}\cdot\Pru{(\cE,\cU,L,Q,*,*,*)\sim \cD_{\inn}^{+}}{F[\cE,\cU,L]=f|_{L}\,\Big|\,\cE,\cU,Q\textup{ and }\{L\subseteq W\}}\tag{since $F[\cdot]$ and $G[\cdot]$ are $2^{-q_{0}}$ transferable}\\
&\geq 2^{-q_{0}}\exp(-2q_{0}^{3})\geq \exp(-3q_{0}^{3}).
\end{align*}

\paragraph{Step 2: locating good Bob-packets.} 
For any Bob-packet $(\cE',\cU')$ such that $(\cE,\cU,*,Q,\cE',\cU',*)$ has nonzero probability under $\cD_{\inn}^{+}$, we say it is ``good'' if
\begin{equation}\label{eq:good_Bob_packet_def}
\Pru{(\cE,\cU,L,Q,\cE',\cU',L')\sim \cD_{\inn}^{+}}{\left.G[\cE',\cU',L']\right|_{L}=f|_{L}\,\Big|\,\cE,\cU,Q,\cE',\cU'\textup{ and }\{L\subseteq W\}}\geq \exp(-3q_{0}^{3})/2.
\end{equation}
By the conclusion from Step 1 of the proof, it follows that
\begin{equation}\label{eq:good_Bob_packet_prob}
\Pru{(\cE,\cU,L,Q,\cE',\cU',*)\sim \cD_{\inn}^{+}}{(\cE',\cU')\textup{ is good}\,\Big|\,\cE,\cU,Q\textup{ and }\{L\subseteq W\}}\geq \exp(-3q_{0}^{3})/2.
\end{equation}
Since $Q$ is smooth with respect to $(\cE,\cU)$, by~\Cref{def:TVD_Q,def:smooth_Q} we have
\begin{align*}
\Pru{(\cE,\cU,L,Q,*,*,*)\sim \cD_{\inn}^{+}}{L\subseteq W\,\Big|\,\cE,\cU,Q}&\geq \Pru{L\in \mathsf{Gr}(\cV(\cE),\ell_{1})}{L\subseteq W\,\big|\,Q\subseteq L}-\operatorname{TVD}_{\cE,\cU}(Q)\\
&\geq 2^{-\ell_{1}(q_{0}+1)}-t^{-1/8}\geq 2^{-\ell_{1}^{2}}, 
\end{align*}
where we used~\Cref{cond:parameters} in the last transition. Consequently,~\Cref{eq:good_Bob_packet_prob} implies
\begin{equation}\label{eq:good_Bob_packet_unconditional_prob}
\Pru{(\cE,\cU,*,Q,\cE',\cU',*)\sim \cD_{\inn}^{+}}{(\cE',\cU')\textup{ is good}\,\Big|\,\cE,\cU,Q}\geq 2^{-\ell_{1}^{2}}\cdot\exp(-3q_{0}^{3})/2\geq 4^{-\ell_{1}^{2}},
\end{equation}
where we used~\Cref{cond:parameters} in the last transition.

\paragraph{Step 3: extension.} By~\Cref{def:game3a}, the condition~\Cref{eq:good_Bob_packet_def} is equivalent to
\[
\Pru{\cA\in \scrA_{c}(\cE',\ell_{1}+1,s')}{\left.G[\cE',\cU',\sigma(\cA)]\right|_{\sigma(\cA_{[\ell_{1}]})}=f|_{\sigma(\cA_{[\ell_{1}]})}\,\Big|\,\sigma(\cA_{[q_{0}]})=Q\textup{ and }\sigma(\cA_{[\ell_{1}]})\subseteq W}\geq \exp(-3q_{0}^{3})/2,
\]
where $c$ is the critical index of $\cE'$ and $s'=2^{\ell_{1}-1}s$. By~\Cref{lem:MGr_agreement_extension}, this implies
\[
\Pru{\cA\in \scrA_{c}(\cE',\ell_{1}+1,s')}{G[\cE',\cU',\sigma(\cA)]=g|_{\sigma(\cA)}\,\Big|\,\sigma(\cA_{[q_{0}]})=Q\textup{ and }\sigma(\cA)\subseteq W'}\geq \exp(-3q_{0}^{3})/10\geq \exp(-q_{0}^{4}),
\]
for some linear subspace $W'\subseteq \cV(\cE')$ containing $W\cap \cV(\cE')$ and some linear function $g:W'\rightarrow \bF_{2}$. Since the codimension of $W'$ in $\cV(\cE')$ is at most
\[
\dim \cV(\cE')-\dim( W\cap \cV(\cE'))\leq \dim \cV(\cE)-\dim W= q_{0},
\]
this means that $(\cE',\cU',Q)$ is pre-answerable with respect to $G[\cdot]$ whenever $(\cE',\cU')$ is good. Therefore,~\Cref{eq:good_Bob_packet_unconditional_prob} implies the desired~\Cref{eq:Bob_pre_answerable_prob}.
\end{proof}

For technical reasons, we will later use the following slightly stricter version of~\Cref{def:pre_answerable_Bob_question}. We upgrade the notion of ``pre-answerable'' to ``answerable'' if the same type of decoding as in~\Cref{def:pre_answerable_Bob_question} can be done on the mixed-type Grassmann domain (in~\Cref{def:pre_answerable_Bob_question}, with a very small probability $\sigma(\cA)$ may fall outside of the mixed-type Grassmann domain).

\begin{definition}\label{def:bob_answerable}
Let $G[\cdot]$ be a Bob table, and let $(\cE',\cU')$ be a Bob-packet with critical index $c\in [t]$. For any linear subspace $Q\subseteq\cV(\cE')$ with $\dim Q\leq q_{0}$, we say that the tuple $(\cE',\cU',Q)$ is \emph{answerable} with respect to $G[\cdot]$ if $Q\in \mathsf{MGr}_{c}(\cE',q_{0},2^{2\ell_{1}-q_{0}}s)$ and the following statement holds: there exist a linear subspace $W\subseteq \cV(\cE')$ containing $Q$ with $\operatorname{codim} W\leq q_{0}$ and a linear map $f:W\rightarrow \bF_{2}$ such that
\[
\Pru{L'\in \mathsf{MGr}(\cE',\ell_{1}+1,s')}{G[\cE',\cU',L']=f|_{L'}\,\Big|\,Q\subseteq L'\subseteq W}\geq \exp(-q_{0}^{5}),
\]
where $s'=2^{\ell_{1}-1}s$. More specifically, we also say that $(\cE',\cU',Q)$ is \emph{answerable with respect to $G[\cdot]$ using $(f,W)$},  and define $\cS_{G}[\cE',\cU',Q]$ to be a list of such pairs $(f,W)$.
\end{definition}

The following proposition shows that as long as $Q$ itself is in the mixed-type Grassmann domain $\mathsf{MGr}(\cE',q_{0},2^{2\ell_{1}-q_{0}}s)$, the pre-answerability of $(\cE',\cU',Q)$ can be automatically upgraded to answerability.

\begin{proposition}\label{prop:pre-answerable_to_answerable}
Assume~\Cref{cond:parameters}. Let $F[\cdot]$ be an Alice-table, and let $G[\cdot]$ be a Bob-table that is $2^{-q_{0}}$-transferable with $F[\cdot]$. Suppose $(\cE',\cU')$ is a Bob-packet with critical index $c$. Then for any $Q\in \mathsf{MGr}_{c}(\cE',q_{0},2^{2\ell_{1}-q_{0}}s)$ such that $(\cE',\cU',Q)$ is pre-answerable with respect to $G[\cdot]$, the same tuple $(\cE',\cU',Q)$ is also answerable with respect to $G[\cdot]$.
\end{proposition}
\begin{proof}
By~\Cref{prop:mix_advice_vs_MGr},\footnote{When applying~\Cref{prop:mix_advice_vs_MGr}, note that we always have $|E_{c}'|\geq k$.} the total variation distance between the following two distributions of $L'$ is at most, say, $k^{-2}$:
\begin{enumerate}[label=(\arabic*)]
\item sample $L'\in \mathsf{MGr}_{c}(\cE',\ell_{1}+1,2^{\ell_{1}-1}s)$ conditioned on $Q\subseteq L'$; and
\item sample $\cA\in \scrA_{c}(\cE',\ell_{1}+1,2^{\ell_{1}-1}s)$ conditioned on $\sigma(\cA_{[q_{0}]})=Q$, and let $L'=\sigma(\cA)$.
\end{enumerate}
We know from~\Cref{prop:MGr_zoom_out} that even if we further condition both distributions on $\{L'\subseteq W\}$, for any linear subspace $W\subseteq \cV(\cE')$ of codimension at most $q_{0}$, the resulting distributions are still $k^{-1}$-close in total variation distance. By definition, this allows one to convert pre-answerability to answerability.
\end{proof}

We can now conclude the following counterpart of~\Cref{lem:answerable_Alice_question}, showing that answerability on Bob's side also happens with noticeable probability.

\begin{lemma}\label{lem:Bob_answerable}
Assume~\Cref{cond:parameters}. Let $F[\cdot]$ be an Alice-table, and let $G[\cdot]$ be a Bob-table that is $2^{-q_{0}}$-transferable with $F[\cdot]$. Then for any Alice-packet $(\cE,\cU)$ such that $\operatorname{Vrf}_{F,G}(\cE,\cU,*,*,*,*)\geq 2^{-q_{0}/3}$, we have
\[
\Pru{(\cE,\cU,*,Q,\cE',\cU',*)\sim \cD_{\inn}^{+}}{(\cE',\cU',Q)\textup{ is answerable with respect to }G[\cdot]\,\Big|\,\cE,\cU}\geq \exp(-6\ell_{1}^{2}).
\]
\end{lemma}

\begin{proof}
First of all, combining~\Cref{cor:smooth_Q_prob} with~\Cref{lem:answerable_Alice_question,lem:Bob_pre_answerable} yields
\begin{align}
\Pru{(\cE,\cU,*,Q,\cE',\cU',*)\sim \cD_{\inn}^{+}}{(\cE',\cU',Q)\textup{ is pre-answerable with respect to }G[\cdot]\,\Big|\,\cE,\cU}&\geq (4^{-\ell_{1}^{2}}-t^{-1/8})\cdot4^{-\ell_{1}^{2}}\nonumber\\
&\geq \exp(-5\ell_{1}^{2}),\label{eq:Bob_answerable_prelim}
\end{align}
where we used~\Cref{cond:parameters} in the last transition.

Using~\Cref{cond:parameters}, it is not hard to see (cf.~\Cref{prop:mix_advice_vs_MGr})
\begin{equation}\label{eq:Bob_answerable_2}
\Pru{(\cE,\cU,*,Q,\cE',*,*)\sim \cD_{\inn}^{+}}{Q\in \mathsf{MGr}(\cE',q_{0},2^{2\ell_{1}-q_{0}}s)\,\Big|\,\cE,\cU}\geq 1-k^{-1}.
\end{equation}
For any Bob-packet $(\cE',\cU')$ and any $Q\in \mathsf{MGr}(\cE',q_{0},2^{2\ell_{1}-q_{0}}s)$, the pre-answerability of $(\cE',\cU',Q)$ can be converted to answerability, due to~\Cref{prop:pre-answerable_to_answerable}. Therefore, combining~\Cref{eq:Bob_answerable_prelim,eq:Bob_answerable_2} yields the desired conclusion.
\end{proof}

\begin{corollary}\label{cor:Bob_answerable}
Assume~\Cref{cond:parameters}. Let $F[\cdot]$ be an Alice-table, and let $G[\cdot]$ be a Bob-table that is $2^{-q_{0}}$-transferable with $F[\cdot]$. If $\Vrf_{F,G}\geq 2^{-q_{0}/4}$, then
\[
\Pru{(*,*,*,Q,\cE',\cU',*)\sim \cD_{\inn}^{+}}{(\cE',\cU',Q)\textup{ is answerable with respect to }G[\cdot]\Big.}\geq \exp(-7\ell_{1}^{2}).
\]
\end{corollary}
\begin{proof}
This follows immediately from~\Cref{lem:Bob_answerable} because the condition $\Vrf_{F,G}\geq 2^{-q_{0}/4}\geq 2^{1-q_{0}/3}$ implies
\[
\Pru{(\cE,\cU,*,*,*,*,*)\sim \cD_{\inn}^{+}}{\Vrf_{F,G}(\cE,\cU,*,*,*,*)\geq 2^{-q_{0}/3}}\geq 2^{-q_{0}/3}.\qedhere
\]
\end{proof}

\subsubsection{Sending Bob's Decoding back to Alice}

As mentioned in the introductory text of~\Cref{subsec:inner_soundness}, the next main goal is to argue that Bob's decoded functions (those in the list $\cS_{G}[\cE',\cU',Q]$) can be ``transferred'' back to Alice's side. 

For intuition, we invite the reader to imagine that each answerable tuple $(\cE',\cU',Q)$ on Bob's side may ``send'' one of the elements of $\cS_{G}[\cE',\cU',Q]$ to each of the Alice-side tuples $(\cE,\cU,Q)$ such that $(\cE,\cU,*,Q,\cE',\cU',*)$ has nonzero probability under $\cD_{\inn}^{+}$. The tuples $(\cE,\cU,Q)$ will then check whether the messages they receive have noticeable agreement with the Alice-table $F[\cE,\cU,\cdot]$. 

We will define two levels of compatibility between these received ``messages'' and $F[\cE,\cU,\cdot]$. The first level is called \emph{$F$-friendliness}, defined as follows.

\begin{definition}\label{def:friendly_decoding}
Let $F[\cdot]$ be an Alice-table. For any possible sample $(\cE,\cU,*,Q,\cE',*,*)$ under $\cD_{\inn}^{+}$, any linear subspace $W\subseteq \cV(\cE')$ containing $Q$, and any linear map $f:W\rightarrow \bF_{2}$, we say that $(\cE,\cU,Q)$ is \emph{$F$-friendly to} the message $(\cE',Q',f,W)$ if the following holds:
\begin{enumerate}[label=(\arabic*)]
\item $Q\in \mathsf{MGr}(\cE',q_{0},2^{2\ell_{1}-q_{0}}s)$, and $\operatorname{codim}_{\cV(\cE')}W\leq q_{0}$;
\item $Q'\subseteq Q$ is a subspace of codimension $\dim Q-\dim Q'=b\in \{0,1\}$; and
\item the following probability is at least $\exp(-2q_{0}^{5})$:
\begin{equation}\label{eq:sending_back_Q+L}
\Pru{L\in\mathsf{MGr}(\cE',\ell_{1},2^{\ell_{1}}s)}{F[\cE,\cU,L]=f|_{L}\,\middle|\,Q'\subseteq L\subseteq W\textup{ and }Q+L\in \mathsf{MGr}(\cE',\ell_{1}+b,2^{\ell_{1}-b}s)}.
\end{equation}
\end{enumerate}
We use $\cR_{F}[\cE,\cU,Q]$ to denote the set of messages that $(\cE,\cU,Q)$ is $F$-friendly to.
\end{definition}

The notion of $F$-friendliness does not require too much: it only inspects the entries of $F[\cE,\cU,\cdot]$ at the mixed-type-Grassmann locations, which (in some sense) are directly ``connected'' to Bob's table $G[\cE',\cU',\cdot]$. The next level of compatibility, defined as follows, is much more stringent. It checks (almost) the full Grassmann table $F[\cE,\cU,\cdot]$, only with a small-codimensional zoom-out $W_{0}$, and is satisfied only if there exists some global well-agreed function $f_{0}:W_{0}\rightarrow \bF_{2}$ that is compatible with the incoming message.

\begin{definition}\label{def:attached_decoding}
In the setting and notation of~\Cref{def:friendly_decoding}, we say that $(\cE,\cU,Q)$ is \emph{$F$-attached to} the message $(\cE',Q',f,W)$ if in addition to the listed three conditions, the following fourth condition is also satisfied.
\begin{enumerate}[label=(\arabic*), start = 4]
\item There exist:
\begin{itemize}
\item a linear subspace $W_{0}\subseteq \cV(\cE)$ such that $\operatorname{codim}_{\cV(\cE)}(W_{0})\leq 2q_{0}$ and $W\subseteq W_{0}$; and 
\item a linear map $f_{0}:W_{0}\rightarrow \bF_{2}$ such that $f_{0}|_{W}=f$ and
\[
    \Pru{L\in \mathsf{Gr}(\cV(\cE),\ell_{1})}{F[\cE,\cU,L]=f_{0}|_{L}\,\Big|\,Q'\subseteq L\subseteq W_{0}}\geq \exp(-q_{0}^9).
\]
\end{itemize}
\end{enumerate}
\end{definition}

We first show below that for each message $(\cE',Q',f,W)$ sent by $(\cE',\cU',Q)$, a noticeable fraction of the receivers $(\cE,\cU,Q)$ will be $F$-friendly to the message. In~\Cref{subsubsec:consistency}, we will argue that any $(\cE,\cU,Q)$ that receives sufficiently many $F$-friendly messages would be $F$-attached to one of them.

\begin{lemma}\label{lem:Bob_back_to_Alice}
Assume~\Cref{cond:parameters}. Let $F[\cdot]$ be an Alice-table, and let $G[\cdot]$ be a Bob-table that is $2^{-q_{0}}$-transferable with $F[\cdot]$. Let $(\cE',\cU')$ be a Bob-packet and $Q\in \mathsf{MGr}(\cE',q_{0},2^{2\ell_{1}-q_{0}}s)$. Suppose that $(\cE',\cU',Q)$ is answerable with respect to $G[\cdot]$ using $(f,W)$. Then
\[
\Pru{(\cE,\cU,*,Q,\cE',\cU',*)\sim \cD_{\inn}^{+}}{(\cE',Q',f,W)\in \cR_{F}[\cE,\cU,Q]\textup{ for some }Q'\,\Big|\,\cE',\cU',Q}\geq \exp(-2q_{0}^{5}).
\]
\end{lemma}

\begin{proof}
By the definition of answerability and transferability, we have
\begin{equation}\label{eq:sending_back_1}
\Exu{L'\in\mathsf{MGr}(\cE',\ell_{1}+1,2^{\ell_{1}-1}s)}{\Pru{(\cE,\cU,L,*,\cE',\cU',L')\sim \cD_{\inn}^{+}}{F[\cE,\cU,L]=f|_{L}\,\Big|\,\cE',\cU',L'}\,\middle|\,Q\subseteq L'\subseteq W}\geq \gamma,
\end{equation}
where $\gamma=2^{-q_{0}}\exp(-q_{0}^{5})$.

Note that if we sample $(\cE,\cU,L,*,\cE',\cU',L')\sim \cD_{\inn}^{+}$ conditioned on $(\cE',\cU')$, then $(\cE,\cU)$ is distributed independently of $(L,L')$. Furthermore, once $L'$ is also fixed and has dimension $\ell_{1}+1$, the subspace $L$ is distributed uniformly over $\mathsf{Gr}(L',\ell_{1})$. Therefore,~\Cref{eq:sending_back_1} implies that with probability at least $\gamma/2$ over $(\cE,\cU)$, we have
\begin{equation}\label{eq:sending_back_2}
\Exu{L'\in\mathsf{MGr}(\cE',\ell_{1}+1,2^{\ell_{1}-1}s)}{\Pru{L\in \mathsf{Gr}(L',\ell_{1})}{F[\cE,\cU,L]=f|_{L}\Big.}\,\middle|\,Q\subseteq L'\subseteq W}\geq \gamma /2.
\end{equation}
Let $\cQ$ be the distribution of $L\cap Q$ in the following sampling process:
\begin{enumerate}
\item Sample $L'$ uniformly at random from $\mathsf{MGr}(\cE',\ell_{1}+1,2^{\ell_{1}-1}s)$ conditioned on $Q\subseteq L'\subseteq W$.
\item Sample $L$ uniformly at random from $\mathsf{Gr}(L',\ell_{1})$.
\end{enumerate}
We can equivalently describe the above process as follows:
\begin{enumerate}
\item Sample $Q'\sim \cQ$. 
\item If $\dim Q'=\dim Q-1$, sample $L\in \mathsf{MGr}(\cE',\ell_{1},2^{\ell_{1}}s)$ conditioned on $Q'\subseteq L\subseteq W$ and $Q+L\in \mathsf{MGr}(\cE',\ell_{1}+1,2^{\ell_{1}-1}s)$ (note that this always leads to $L\cap Q=Q'$).
\item If $\dim Q'=\dim Q$ (which means $Q'=Q$), sample $L'$ uniformly at random from $\mathsf{MGr}(\cE',\ell_{1}+1,2^{\ell_{1}-1}s)$ conditioned on $Q\subseteq L'\subseteq W$, and then sample a uniformly random hyperplane $L\in\mathsf{Gr}(L',\ell_{1})$ containing $Q$.
\end{enumerate}
Note that in the case $Q'=Q$, we can replace the third step above by directly sampling a uniformly random $L\in \mathsf{MGr}(\cE',\ell_{1},2^{\ell_{1}}s)$ conditioned on $Q\subseteq L\subseteq W$, at the cost of incurring a total variation distance of at most $2^{-s/2}$.\footnote{To bound the total variation distance, use~\Cref{prop:mixed_zoom_out_2} and the fact that $\operatorname{codim}_{\cV(\cE')}W\leq q_{0}\ll s.$} Therefore, we can unify the cases $\dim Q'=\dim Q-1$ and $\dim Q'=\dim Q$, by always sampling $L\in \mathsf{MGr}(\cE',\ell_{1},2^{\ell_{1}}s)$ conditioned on $Q'\subseteq L\subseteq W$ and $Q+L\in \mathsf{MGr}(\cE',\ell_{1}+b,2^{\ell_{1}-b}s)$, where $b=\dim Q-\dim Q'$. This corresponds exactly to~\Cref{eq:sending_back_Q+L}. Consequently,~\Cref{eq:sending_back_2} implies
\[
\Exu{Q'\sim \cQ}{\textup{the value of the probability in~\Cref{eq:sending_back_Q+L}}\Big.}\geq \gamma/2-2^{-s/2}\geq \exp(-2q_{0}^{5}),
\]
which implies $(\cE',Q',f,W)\in \cR_{F}[\cE,\cU,Q]$ for some subspace $Q'$ that has nonzero probability under $\cQ$.

To summarize, we have proven that (see the text before~\Cref{eq:sending_back_2})
\[
\Pru{(\cE,\cU,*,*,\cE',\cU',*)\sim \cD_{\inn}^{+}}{(\cE',Q',f,W)\in \cR_{F}[\cE,\cU,Q]\textup{ for some }Q'\,\Big|\,\cE',\cU'}\geq \gamma/2\geq \exp(-2q_{0}^{5}).
\]
This is the same as the desired conclusion, because when we sample $(\cE,\cU,*,Q,\cE',\cU',*)\sim \cD_{\inn}^{+}$ conditioned on $(\cE',\cU')$, the pair $(\cE,\cU)$ is distributed independently of $Q$.
\end{proof}

\subsubsection{Consistency of Decodings}\label{subsubsec:consistency}

As mentioned earlier, we now show that any $(\cE,\cU,Q)$ that receives sufficiently many $F$-friendly messages \emph{with different critical coordinates} would be $F$-attached to one of them. 

Note that the conclusion of the following lemma is just an expanded way to say that $(\cE,\cU,Q)$ is $F$-attached to $(\cE_{i}',Q',f_{i},W_{i})$ for some $i\in [N]$. 

\begin{lemma}\label{lem:sunflower_transfer}
    Assume~\Cref{cond:parameters}. Let $F[\cdot]$ be an Alice-table. Fix an Alice-packet $(\cE,\cU)$ and a linear subspace $Q\in \mathsf{Gr}(\cV(\cE),q_{0})$, and a subspace $Q'\subseteq Q$ of codimension $\dim Q-\dim Q'=b\in \{0,1\}$. Suppose we have the following objects:
    \begin{enumerate}[label=(\arabic*)]
    \item pairwise distinct indices $c_{1},\dots,c_{N}\in [t]$, where $N\geq \exp(k^{8})$;
    \item for each $i\in [N]$, a $c_{i}$-critical descendant $\cE_{i}'=(E_{i,1}',\dots,E_{i,t}')$ of $\cE$, such that 
    \[
|E_{1,c_{1}}'|=|E_{2,c_{2}}'|=\dots=|E_{N,c_{N}}'|;
    \]
    \item for each $i\in [N]$, a pair $(f_{i},W_{i})$ such that $(\cE,\cU,Q)$ is $F$-friendly to $(\cE'_{i},Q',f_{i},W_{i})$.
    \end{enumerate}
    Then there exist a linear subspace $W\subseteq \cV(\cE)$ with $\codim_{\cV(\cE)}(W)\leq 2q_0$ and a linear map $f:W\rightarrow \bF_{2}$ such that 
    \[
    \Pru{L\in \mathsf{Gr}(\cV(\cE),\ell_{1})}{F[\cE,\cU,L]=f|_{L}\,\Big|\,Q'\subseteq L\subseteq W}\geq \exp(-q_{0}^9).
    \]
    Furthermore, for some $i\in [N]$ we have $W_{i}\subseteq W$ and $f|_{W_{i}}=f_{i}$.
\end{lemma}

Before starting the formal proof, we give an overview of the proof of~\Cref{lem:sunflower_transfer}, which is inspired by the proof of~\Cref{lem:dkkms_list_bound} in~\cite{DKKMS18,Minzer2022}. For simplicity of presentation we ignore the issue of zoom-ins. 

We know from the $F$-friendliness conditions that each $f_{i}:W_{i}\rightarrow \bF_{2}$ agrees well with the table $F[\cE,\cU,\cdot]$ on the mixed-type-Grassmann locations, i.e. mixed-type subspaces $L$ that are contained in $\cV(\cE')$. In order to patch these locally agreed functions together to some $f:W\rightarrow \bF_{2}$, we use an $\ell_{1}$-vs.-$\ell_{1}/10$-dimensional Grassmann agreement test. In each $W_{i}$, we plant a ``local linearity test'' by sampling $L_{1},L_{2}\in \mathsf{MGr}(\cE',\ell_{1},2^{\ell_{1}}s)$ conditioned on $L_{1},L_{2}\subseteq W_{i}$ and $\dim (L_{1}\cap L_{2})=\lfloor \ell_{1}/10\rfloor$, and test if  
\[
F[\cE,\cU,L_{1}]|_{L_{1}\cap L_{2}}=F[\cE,\cU,L_{2}]|_{L_{1}\cap L_{2}}.
\]
Each of these local linearity tests passes with noticeable probability, due to $F$-friendliness to $(\cE',Q',f_{i},W_{i})$.

If the subspaces $W_{i}\subseteq \cV(\cE)$ satisfy a sunflower condition (as per~\Cref{def:subspacesunflower}), these local linearity tests, when aggregated together, will be noticeably similar to a ``global linearity test'' that samples $L_{1},L_{2}\in \mathsf{Gr}(\cV(\cE),\ell_{1})$ conditioned on $L_{1},L_{2}\subseteq W$ and $\dim (L_{1}\cap L_{2})=\lfloor \ell_{1}/10\rfloor$. Then, due to~\Cref{cor:list_decoding_linearity_test_nil}, we can decode from the table $F[\cE,\cU,\cdot]$ some global linear function $f:W\rightarrow \bF_{2}$ that has good global agreement with $F[\cE,\cU,\cdot]$. 

Of course, the subspaces $W_{i}\subseteq \cV(\cE)$ do not necessarily satisfy the sunflower condition. Nevertheless, using~\Cref{lem:subspacesunflower} we can find a sub-collection of these subspaces that form a sunflower, which is also good enough.

After being able to decode a global linear function $f:W\rightarrow \bF_{2}$, we still have to show that it agrees with $f_{i}$ for some $i\in [N]$. We consider the list of all global linear functions $f:W\rightarrow \bF_{2}$ that has good agreement with $F[\cE,\cU,\cdot]$. A list-decoding argument will show that the size of the list is not large, and nearly all of the passing probability of the ``global linearity test'' above can be attributed to this list. This will imply that some function in the list does agree with some $f_{i}$.

\begin{proof}[Proof of~\Cref{lem:sunflower_transfer}]
Without loss of generality, we assume $c_{i}=i$ for each $i\in [N]$, and let
\begin{equation}\label{eq:consistency_same_deficit}
|E_{1,1}'|=|E_{2,2}'|=\dots=|E_{N,N}'|=k'.
\end{equation}
For each pair of distinct indices $i,j\in [N]$, we use $\cE_{i,j}'$ to denote the ``common thinning'' of $\cE_{i}'$ and $\cE_{j}'$, obtained by replacing the $i$-th block of $\cE=(E_{1},\dots,E_{t})$ with $E_{i,i}'$ and replacing the $j$-th block with $E_{j,j}'$. We then have $\cV(\cE_{i}')\cap \cV(\cE_{j}')=\cV(\cE_{i,j}')$. The proof is divided into the following six steps. 

\paragraph{Step 1: sunflower argument.} 
    
By~\Cref{def:friendly_decoding}, we have $\operatorname{codim}_{\cV(\cE_{i}')}W_{i}\leq q_{0}$ for all $i\in [N]$. By the pigeonhole principle, there exists an integer $w'\in \{0,1,\dots,q_{0}\}$ such that $\Big.\operatorname{codim}_{\cV(\cE_{i}')}W_{i}=w'$ for at least $N_{1}:=\lceil N/(q_{0}+1)\rceil$ of the indices $i\in [N]$. Without loss of generality assume they are the first $N_{1}$ indices in $[N]$. Due to~\Cref{eq:consistency_same_deficit}, for all $i\in [N_{1}]$ we have
\begin{equation}\label{eq:consistency_codim_Wi}
\codim_{\cV(\cE)}W_{i}=\operatorname{codim}_{\cV(\cE'_{i})}W_{i}+\operatorname{codim}_{\cV(\cE)}\cV(\cE_{i}')=w'+9k^{2}-(k')^{2}
\end{equation}

Let $M=\lfloor \exp(k^{3})\rfloor$ and let $m\leq M$ be a positive integer to be determined later. We apply~\Cref{lem:subspacesunflower} to the subspaces $W_{1},\dots,W_{N_{1}}\subseteq \cV(\cE)$, obtaining a subspace $W\subseteq \cV(\cE)$ and indices $i_{1},\dots,i_{M}\in [N_{1}]$ such that $(W,W_{i_{1}},\dots,W_{i_{M}})$ forms a sunflower (see~\Cref{def:subspacesunflower}; only the first $m$ petals will be actually used). Without loss of generality, we assume the indices $i_{1},\dots,i_{m}$ are $1,2,\dots,m$. Note that by the definition of sunflower, we have
$W_{1},\dots,W_{m}\subseteq W$, and
\[
\operatorname{codim}_{W}W_{1}=\dots=\operatorname{codim}_{W}W_{m}=w_{0}
\]
for some integer $w_{0}\geq 0$. Denoting $w=\operatorname{codim}_{\cV(\cE)}W$, we deduce from~\Cref{eq:consistency_codim_Wi} that
\begin{equation}\label{eq:consistency_codim_identity_1}
w'+9k^{2}-(k')^{2}=w_{0}+w.
\end{equation}
Furthermore, the definition of sunflower implies that for any pair of distinct indices $i,j\in [m]$, we have $\operatorname{codim}_{W} W_{i}\cap W_{j}=2w_{0}$. Since $W_{i}\cap W_{j}\subseteq \cV(\cE_{i,j}')$, it follows that
\begin{equation}\label{eq:consistency_codim_identity_2}
2w_{0}+w=\operatorname{codim}_{\cV(\cE)}W_{i}\cap W_{j}={}\operatorname{codim}_{\cV(\cE)}\cV(\cE_{i,j}')+\operatorname{codim}_{\cV(\cE_{i,j}')}W_{i}\cap W_{j}\geq 2(9k^{2}-(k')^{2}). 
\end{equation}
Comparing~\Cref{eq:consistency_codim_identity_1,eq:consistency_codim_identity_2} yields $w\leq 2w'$, and hence  $\operatorname{codim}_{\cV(\cE)}W\leq 2w'\leq 2q_{0}$.

\paragraph{Step 2: probability of lying in a petal.}

Write $d=\lfloor \ell_{1}/10\rfloor$ and $D=2\ell_{1}-d$, and $h=D-\dim Q'=D-q_{0}+b$. Throughout the rest of the proof, we suppose $L_{+}\in \mathsf{Gr}(\cV(\cE),D)$ is a random subspace conditioned on $Q'\subseteq L_{+}$, and $L_{1},L_{2}\in \mathsf{Gr}(L_{+},\ell_{1})$ are then sampled independently conditioned on $Q'\subseteq L_{1},L_{2}$. Note that with probability at least $1-2^{-d+q_{0}}$, we have $\dim (L_{1}\cap L_{2})=d$.

For each $i\in [m]$, we let $\Gamma_{i}$ be the event 
\[\left\{Q+L_{+}\in \mathsf{MGr}(\cE'_{i},D+b,2^{d-b}s)\right\}\cap\left\{L_{+}\subseteq W_{i}\right\}.\]
Due to~\Cref{eq:consistency_same_deficit}, we know that the probability  
\[\Pr{Q+L_{+}\in\mathsf{MGr}(\cE_{i}',D+b,2^{d-b}s)\,\middle|\,L_{+}\subseteq \cV(\cE_{i}')}\]
is the same for each $i\in [m]$.\footnote{This uses the fact that $Q\in \mathsf{MGr}(\cE_{i}',q_{0},2^{2\ell_{1}-q_{0}}s)$ for each $i\in [m]$. This fact holds because otherwise $(\cE,\cU,Q)$ could not be $F$-friendly to $(\cE_{i}',Q',f_{i},W_{i})$.} We use $\gamma$ to denote this probability, and we trivially have $\gamma\geq \exp(-50\ell_{1}k^{2}).$ Note that
\[
\Pr{\Gamma_{i}}=\Pr{L_{+}\subseteq \cV(\cE_{i}')}\cdot \gamma\cdot \Pru{L\in \mathsf{MGr}(\cE_{i}',D+b,2^{d-b}s)}{L\subseteq W_{i}\,\big|\,Q\subseteq L}.
\]
The first and third factors on the right-hand side are ``zoom-out probabilities,'' and are both easy to estimate (for the third, use~\Cref{prop:MGr_zoom_out}), leading to the bounds 
\begin{equation}\label{eq:consistency_Gamma_prob}
2^{-(w_{0}+w)h-3}\cdot \gamma\leq \Pr{\Gamma_{i}}\leq 2^{-(w_{0}+w)h+3}\cdot \gamma.
\end{equation}

Furthermore, it is easy to see that for any pair of distinct indices $i,j\in [m]$, we have (the projections of $L_{+}$ to the $i$-th block and the $j$-th block are almost independent)
\begin{equation}\label{eq:consistency_intersection_MGr}
\gamma^{2}/2\leq \Pr{Q+L_{+}\in\mathsf{MGr}(\cE_{i}',D+b,2^{d-b}s)\cap \mathsf{MGr}(\cE_{j}',D+b,2^{d-b}s)\,\middle|\,L_{+}\subseteq \cV(\cE_{i,j}')}\leq 2\gamma^{2}.
\end{equation}
By~\Cref{def:mixed_grassmann_domain}, the intersection of mixed-type Grassmann domains in~\Cref{eq:consistency_intersection_MGr} is the same as $\mathsf{MGr}_{i,j}(\cE_{i,j}',D+b,2^{d-b}s)$. We write
\[
\Pr{Q+L_{+}\in\mathsf{MGr}_{i,j}(\cE_{i,j}',D+b,2^{d-b}s)\,\middle|\,L_{+}\subseteq \cV(\cE_{i,j}')}=\gamma_{i,j}\in \left[\frac{\gamma^{2}}{2},2\gamma^{2}\right].
\]
Then we have
\[
\Pr{\Gamma_{i}\cap \Gamma_{j}}=\Pr{L_{+}\subseteq \cV(\cE_{i,j}')}\cdot\gamma_{i,j}\cdot \Pru{L\in \mathsf{MGr}(\cE_{i,j}',D+b,2^{d-b}s)}{L\subseteq W_{i}\cap W_{j}\,\big|\,Q\subseteq L}. 
\]
The first and third factors on the right-hand side are easy to estimate (for the third, use~\Cref{prop:MGr_zoom_out}), leading to the bounds (recall~\Cref{eq:consistency_codim_identity_2})
\begin{equation}\label{eq:consistency_Gamma_intersect}
2^{-(2w_{0}+w)h-4}\cdot\gamma^{2}\leq \Pr{\Gamma_{i}\cap \Gamma_{j}}\leq 2^{-(2w_{0}+w)h+4}\cdot \gamma^{2}.
\end{equation}

\paragraph{Step 3: passing probability in a petal.} For each $i\in [m]$, let $\Lambda_{i}$ be the event
\[
\Gamma_{i}\cap\left\{F[\cE,\cU,L_{1}]=\left.f_{i}\right|_{L_{1}}\right\}\cap \left\{F[\cE,\cU,L_{2}]=\left.f_{i}\right|_{L_{2}}\right\}\cap\left\{\dim (L_{1}\cap L_{2})=d\right\}.
\]
We note that conditioned on $\Gamma_{i}$, the distribution of $Q+L_{+}$ is uniformly random over the set
\[\left\{L\in \mathsf{MGr}(\cE_{i}',D+b,2^{d-b}s)\,\middle|\,Q\subseteq L\subseteq  W_{i}\right\}.\] By~\Cref{prop:MGr_zoom_out}, it follows that the distribution of $Q+L_{1}$ is $16$-almost uniformly\footnote{Here we say a distribution over a finite set $S$ is $C$-almost uniform if the probability of every element in the set lies in $[C^{-1}|S|^{-1},C|S|^{-1}]$.} random over the set \[\left\{L\in\mathsf{MGr}(\cE_{i}',\ell_{1}+b,2^{\ell_{1}-b}s)\,\middle|\, Q\subseteq L \subseteq W_{i}\right\}.\] 
Using the assumption that $(\cE,\cU,Q)$ is $F$-friendly to $(\cE_{i}',Q',f_{i},W_{i})$, this implies
\[
\Pr{F[\cE,\cU,L_{1}]=\left.f_{i}\right|_{L_{1}}\,\Big|\,\Gamma_{i}}\geq \exp(-2q_{0}^{5})/16.
\]
By Cauchy-Schwarz, we know that
\begin{align*}
&\quad\Pr{F[\cE,\cU,L_{1}]=\left.f_{i}\right|_{L_{1}}\textup{ and }F[\cE,\cU,L_{2}]=\left.f_{i}\right|_{L_{2}}\,\Big|\,\Gamma_{i}}\\
&=\Exu{L_{+}}{\Pru{L_{1}}{F[\cE,\cU,L_{1}]=\left.f_{i}\right|_{L_{1}}\,\Big|\,L_{+}}^{2}\,\middle|\,\Gamma_{i}}\\
&\geq \Pr{F[\cE,\cU,L_{1}]=\left.f_{i}\right|_{L_{1}}\,\Big|\,\Gamma_{i}}^{2}\geq \exp(-4q_{0}^{5})/256.
\end{align*}
Therefore, 
\begin{equation}\label{eq:consistency_Lambda_prob}
\Pr{\Lambda_{i}\mid \Gamma_{i}}\geq \exp(-4q_{0}^{5})/256-2^{-d+q_{0}}\geq \exp(-q_{0}^{6}),
\end{equation}
where we used~\Cref{cond:parameters} in the last transition.

\paragraph{Step 4: inclusion-exclusion.} 

For each $i\in [m]$, let $\Lambda_{i}^*$ be the event 
\[\Lambda_{i}\setminus\left(\bigcup\nolimits_{j\in [m],\;j\neq i}\Gamma_{j}\right).\]
Combining~\Cref{eq:consistency_Gamma_prob,eq:consistency_Gamma_intersect,eq:consistency_Lambda_prob}, we get
\[
\Pr{\Lambda_{i}^{*}}\geq \exp(-q_{0}^{6})\cdot2^{-(w_{0}+w)h-3}\gamma-(m-1)\cdot 2^{-(2w_{0}+w)h+4}\gamma^{2}. 
\]
Since the events $\Lambda_{1}^*,\dots,\Lambda_{m}^*$ are pairwise disjoint, we have
\[
\Pr{\bigcup\nolimits_{i=1}^{m}\Lambda_{i}^*}\geq 2^{-wh+4}\left(\exp(-q_{0}^{6}-10)\cdot m\cdot 2^{-w_{0}h}\gamma -m(m-1)2^{-2w_{0}h}\gamma^{2}\right).
\]
We now choose $m=\left\lceil \exp(-q_{0}^{6}-11)\cdot 2^{w_{0}h}\gamma^{-1}\right\rceil\leq \exp(k^{3})$. Plugging into the above equation yields
\begin{equation}\label{eq:consistency_union_Lambda}
\Pr{\bigcup\nolimits_{i=1}^{m}\Lambda_{i}^*}\geq 2^{-wh+4}\exp(-q_{0}^{7}).
\end{equation}
For future reference, we also record the estimate (using~\Cref{eq:consistency_Gamma_prob,eq:consistency_Gamma_intersect,eq:consistency_Lambda_prob})
\begin{equation}\label{eq:consistency_Lambda_cond_Gamma}
\Pr{\Lambda_{i}^*|\Gamma_{i}}\geq \exp(-q_{0}^{6})-(m-1)\cdot 2^{-w_{0}h+7}\gamma \geq \exp(-q_{0}^{6})/2.
\end{equation}

Note that
\[
\bigcup\nolimits_{i=1}^{m}\Lambda_{i}^*\subseteq \bigcup\nolimits_{i=1}^{m}\left\{L_{+}\subseteq W_{i}\right\}\subseteq \left\{L_{+}\subseteq W\right\}.
\]
Since $\Pr{L_{+}\subseteq W}\leq 2^{-wh+1}$,~\Cref{eq:consistency_union_Lambda} implies that
\begin{equation}\label{eq:consistency_union_lambda_W}
\Pr{\bigcup\nolimits_{i=1}^{m}\Lambda_{i}^*\,\middle|\,L_{+}\subseteq W}\geq \exp(-q_{0}^{7}).
\end{equation}
In particular, by the definition of the events $\Lambda_{i}$, this means
\[
\Pr{\left.F[\cE,\cU,L_{1}]\right|_{L_{1}\cap L_{2}}=\left.F[\cE,\cU,L_{2}]\right|_{L_{1}\cap L_{2}}\,\Big|\,\dim(L_{1}\cap L_{2})=d\text{ and }L_{+}\subseteq W}\geq \exp(-q_{0}^{7}). 
\]
In the above probability expression, the distribution of the pair $(L_{1},L_{2})$ is identical to sampling from $\cG(W,\ell_{1},d)$ conditioned on $Q'\subseteq L_{1},L_{2}$ (see~\Cref{subsec:general_notation} for the notation $\cG(W,\ell_{1},d)$). Therefore, due to~\Cref{cor:list_decoding_linearity_test_nil}, there exists a linear function $g:W\rightarrow \bF_{2}$ such that
\begin{equation}\label{eq:consistency_global_existence}
\Pru{L\in\mathsf{Gr}(\cV(\cE),\ell_{1})}{F[\cE,\cU,L]=g|_{L}\,\Big|\,Q'\subseteq L\subseteq W}\geq \exp(-q_{0}^{9}).
\end{equation}

\paragraph{Step 5: list-decoding.} Let $g_{1},\dots,g_{p}$ be the list of all linear functions $g:W\rightarrow \bF_{2}$ satisfying~\eqref{eq:consistency_global_existence}, where $p$ is a positive integer. By~\Cref{lem:list_decoding_bound_baby}, we know that $p\leq 2\exp(q_{0}^{9})$. We consider a ``pruned'' Grassmann table $F'[\cdot]$ defined as follows: for each $L\in \mathsf{Gr}(W,\ell_{1})$,
\begin{itemize}
\item if $F[\cE,\cU,L]=g_{j}|_{L}$ for some $j\in [p]$, then let $F'[L]=\nil$;
\item if $F[\cE,\cU,L]\neq g_{j}|_{L}$ for any $j\in [p]$, then let $F'[L]=F[\cE,\cU,L]$.
\end{itemize}

We claim that
\begin{equation}\label{eq:consistency_no_agreement}
\Pr{\left.F'[L_{1}]\right|_{L_{1}\cap L_{2}}=\left.F'[L_{2}]\right|_{L_{1}\cap L_{2}}\,\Big|\,\dim(L_{1}\cap L_{2})=d\text{ and }L_{+}\subseteq W}\leq  \exp(-q_{0}^{8}), 
\end{equation}
Suppose~\Cref{eq:consistency_no_agreement} does not hold. Then by~\Cref{cor:list_decoding_linearity_test_nil}, there would exist another linear function $g_{0}:W\rightarrow \bF_{2}$ (it cannot be among $g_{1},\dots,g_{p}$ because otherwise the left-hand side of the following would be zero) satisfying
\[
\Pru{L\in\mathsf{Gr}(\cV(\cE),\ell_{1})}{F'[L]=g_{0}|_{L}\,\Big|\,Q'\subseteq L\subseteq W}\geq \exp(-q_{0}^{9}),
\]
and since $F'[\cdot]$ is a ``pruning'' of $F[\cE,\cU,\cdot]$, this $g_{0}$ would also satisfy~\Cref{eq:consistency_global_existence}. That would make $g_{0}$ appear in the list $(g_{1},\dots,g_{p})$, which is a contradiction. Combining~\Cref{eq:consistency_union_lambda_W,eq:consistency_no_agreement} yields
\[
\Pr{\Bigl\{F'[L_{1}]\neq F[\cE,\cU,L_{1}]\textup{ or }F'[L_{2}]\neq F[\cE,\cU,L_{2}]\Big\}\cap \bigcup\nolimits_{i=1}^{m}\Lambda_{i}^*\,\middle|\,L_{+}\subseteq W}\geq \exp(-q_{0}^{7})-\exp(-q_{0}^{8}),
\]
and by the symmetry between $L_{1}$ and $L_{2}$, this implies
\[
\Pr{\Bigl\{F'[L_{1}]\neq F[\cE,\cU,L_{1}]\Big\}\cap \bigcup\nolimits_{i=1}^{m}\Lambda_{i}^*\,\middle|\,L_{+}\subseteq W}\geq \frac{1}{2}\left(\exp(-q_{0}^{7})-\exp(-q_{0}^{8})\right)\geq \exp(-q_{0}^{8}).
\]
By the definition of $F'$, this means
\begin{equation}\label{eq:consistency_list_decoding}
\Pr{\left(\bigcup\nolimits_{j=1}^{p}\Bigl\{F[\cE,\cU,L_{1}]= g_{j}|_{L_{1}}\Bigr\}\right)\cap\left(\bigcup\nolimits_{i=1}^{m}\Lambda_{i}^*\right)\,\middle|\,L_{+}\subseteq W}\geq \exp(-q_{0}^{8}).
\end{equation}
\paragraph{Step 6: consistency.}

By~\Cref{eq:consistency_list_decoding}, there exists an index $j\in [p]$ such that
\[
\Pr{\Bigl\{F[\cE,\cU,L_{1}]=g_{j}|_{L_{1}}\Bigr\}\cap \bigcup\nolimits_{i=1}^{m}\Lambda_{i}^*\,\middle|\,L_{+}\subseteq W}\geq \exp(-q_{0}^{8})/p\geq \exp(-3q_{0}^{9}).
\]
Without loss of generality assume this special index $j$ is $1$. Since the events $\Lambda_{1}^*,\dots,\Lambda_{m}^*$ are pairwise disjoint and all contained in the event $\{L_{+}\subseteq W\}$, there exists an index $i\in [m]$ such that
\[
\Pr{F[\cE,\cU,L_{1}]=g_{1}|_{L_{1}}\,\Big|\,\Lambda_{i}^*}\geq \exp(-3q_{0}^9).
\]
Without loss of generality assume this special index $i$ is $1$. By the definitions of $\Lambda_{1}^*$ and $\Lambda_{1}$, this is equivalent to
\[
\Pr{f_{1}|_{L_{1}}=g_{1}|_{L_{1}}\,\Big|\,\Lambda_{1}^*}\geq \exp(-3q_{0}^{9}).
\]
Combining with~\Cref{eq:consistency_Lambda_cond_Gamma} yields
\begin{equation}\label{eq:consistency_f_g_Gamma}
\Pr{f_{1}|_{L_{1}}=g_{1}|_{L_{1}}\,\Big|\,\Gamma_{1}}\geq \exp(-4q_{0}^{9}).
\end{equation}

It remains to show that $f_{1}=g_{1}|_{W_{1}}$. Assume on the contrary that $f_{1}\neq g_{1}|_{W_{1}}$. We let $K=\ker(f_{1}-g_{1}|_{W_{1}})\subsetneq W_{1}$ be the hyperplane in $W_{1}$ on which $f_{1}$ agrees with $g_{1}|_{W_{1}}$. For any $L\in \mathsf{MGr}(\cE_{1}',D+b,2^{d-b}s)$ such that $Q\subseteq L\subseteq W_{1}$ but $L\not\subseteq K$, it is not hard to see (from how we sampled $L_{1}$ and $L_{+}$) that
\[
\Pr{L_{1}\subseteq K\,\Big|\, Q+L_{+}=L}\leq 2^{-(\ell_{1}-q_{0})}.
\]
Furthermore, we know from~\Cref{prop:MGr_zoom_out} (and from the definition of $\Gamma_{1}$) that
\[
\Pr{Q+L_{+}\subseteq K\,\Big|\,\Gamma_{1}}\leq 2^{-(D-q_{0})+4}.
\]
Therefore, by union bound we have
\[
\Pr{L_{1}\subseteq K\,\Big|\,\Gamma_{1}}\leq 2^{-(\ell_{1}-q_{0})}+2^{-(D-q_{0})+4}\leq 2^{-\ell_{1}+q_{0}+1}.
\]
This directly contradicts~\Cref{eq:consistency_f_g_Gamma}, and hence concludes the proof of this lemma.
\end{proof}

\subsubsection{Deriving Strategy for the Middle PCP}

We are now ready to prove the soundness of $\GameCa[H,k,\ell_{1},\ell_{2},t,s,\beta]$.

\begin{proof}[Proof of~\Cref{lem:innerpcp_soundness}]
Assume on the contrary that the value of $\GameCa[H,k,\ell_{1},\ell_{2},t,s,\beta]$ is larger than $\delta$. The goal is to obtain a strategy for 
\[\GameBc[H,k,\ell_{2},t,2^{2\ell_{1}-q_{0}}s,q_{0},\beta]\]
achieving a value larger than $2^{-q_{2}/4}$, which would contradict~\Cref{lem:Game2b_soundness} and~\Cref{prop:Game2b=Game2c}.

We first represent Alice's and Bob's strategies in $\GameCa$ by ``folded tables.'' Recall that in $\GameCa[H,k,\ell_{1},\ell_{2},t,s,\beta]$, the verifier samples $(\cE,\cU,L,*,\cE',\cU',L')\sim \cD_{\inn}^{+}$. We write 
\[\widetilde{F}[\cE,\cU,L+\cN(\cE,\cU)]:(L+\cN(\cE,\cU))\rightarrow \bF_{2}\] for Alice's answer to the question $(\cE,\cU,L+\cN(\cE,\cU))$, and write 
\[\widetilde{G}[\cE',\cU',L'+\toplus \cU']:(L'+\toplus \cU')\rightarrow \bF_{2}\] for Bob's answer to the question $(\cE',\cU',L'+\toplus \cU')$. By assumption, there exist such folded tables $\widetilde{F}[\cdot]$ and $\widetilde{G}[\cdot]$ such that
\begin{equation}\label{eq:inner_soundness_assumption}
\Pru{(\cE,\cU,L,*,\cE',\cU',L')\sim \cD_{\inn}^{+}}{\left.\widetilde{F}[\cE,\cU,L+\cN(\cE,\cU)]\right|_{L+\bigoplus \cU}=\left.\widetilde{G}[\cE',\cU',L'+\toplus \cU']\right|_{L+\bigoplus \cU}}>\delta.
\end{equation}

\paragraph{Step 1: unfolded tables.} We define an (unfolded) Alice-table (see~\Cref{ntn:inner_Alice_table}) $F^*[\cdot]$ as follows. For any Alice-packet $(\cE,\cU)$ and any linear subspace $L\subseteq \cV(\cE)$, if $(\cE,\cU,L)$ is an unfolded-Alice-question (see~\Cref{ntn:inner_unfolded_Alice_question}), let
\[
F^*[\cE,\cU,L]=\left.\widetilde{F}[\cE,\cU,L+\cN(\cE,\cU)]\right|_{L}.
\]
Otherwise, we let $F^*[\cE,\cU,L]=\nil$.

Similarly, we define an (unfolded) Bob-table $G^*[\cdot]$ as follows. For any Bob-packet $(\cE',\cU')$ and any linear subspace $L'\subseteq \cV(\cE')$, if $(\cE',\cU',L')$ is an unfolded-Bob-question, let
\[
G^*[\cE',\cU',L']=\left.\widetilde{G}[\cE',\cU',L'+\toplus \cU']\right|_{L'}.
\]
Otherwise, we let $G^*[\cE',\cU',L']=\nil$.

We run the following iterative ``pruning'' procedure to obtain our final Alice-table $F[\cdot]$ and Bob-table $G[\cdot]$ used for decoding. 
    \begin{enumerate}
        \item Initialize with $F:= F^*$, $G:=G^*$.
        \item If there exists an unfolded-Alice-question $(\cE,\cU,L)$ such that $F[\cE,\cU,L]\neq \nil$ but 
        \begin{align*}
            \Vrf_{F,G}(\cE,\cU,L,*,*,*) < 2^{-q_0}, 
        \end{align*}
        update $F[\cE,\cU,L] = \nil$. Similarly, if there exists an unfolded-Bob-question $(\cE',\cU',L')$ such that $G[\cE',\cU',L']\neq \nil$ but
        \begin{align*}
            \Vrf_{F,G}(*,*,*,\cE',\cU',L') < 2^{-q_0}, 
        \end{align*}
        update $G[\cE',\cU',L'] := \nil$. 
        \item Repeat the preceding operations until for every unfolded-Alice-question $(\cE,\cU,L)$ such that $F[\cE,\cU,L]\neq \nil$, we have $\Vrf_{F,G}(\cE,\cU,L,*,*,*) \geq 2^{-q_0}$, and for every unfolded-Bob-question $(\cE',\cU',L')$ such that $G[\cE',\cU',L']\neq \nil$, we have $\Vrf_{F,G}(*,*,*,\cE',\cU',L') \geq 2^{-q_0}$. 
    \end{enumerate}
    It is obvious that the procedure must terminate within a finite number of operations. Furthermore, by definitions, after termination the tables $F[\cdot ], G[\cdot ]$ are $2^{-q_0}$-transferable with each other (see~\Cref{def:transferability}).
    
    We know from~\Cref{eq:inner_soundness_assumption} that $\operatorname{Vrf}_{F^*,G^*}>\delta$.
    During the pruning procedure, the agreement probability can decrease by at most $2^{-q_0+1}$, since when each unfolded question of Alice or Bob is pruned, its conditional contribution to the agreement is at most $2^{-q_0}$. As a result, we get
    \begin{align}\label{ineq:agreement_folded_tables}
         \Vrf_{F,G}\geq \Vrf_{F^*,G^*}- 2^{-q_0+1} \geq \delta/2,
    \end{align}
    under~\Cref{cond:parameters}. 

\paragraph{Step 2: strategy for Game2c.} Let $\vec{\tau}=(\tau_{0},\dots,\tau_{2q_{0}})$ be as in~\Cref{cond:parameters}. We now derive strategies for the provers in $\GameBc[H,k,\ell_{2},t,2^{2\ell_{1}-q_{0}}s,q_{0},\beta]$ from the tables $F[\cdot]$ and $G[\cdot]$. The strategy for Alice is as follows.
\begin{enumerate}
\item Receive the question $(\cE,\cA,\cU)$. Let $Q=\sigma(\cA_{[q_{0}]})$. If $\dim Q<q_{0}$, return an arbitrary answer and abort.
\item Pick a random subspace $Q'\subseteq Q$ with $\dim Q-\dim Q'\in \{0,1\}$. If $\mathsf{LIST}_{Q'}^{\vec{\tau}}(F)=\emptyset$ (see~\Cref{def:LIST_Q}), return an arbitrary answer and abort. 
\item Pick a random element $(f,W)\in \mathsf{LIST}_{Q'}^{\vec{\tau}}(F)$, and randomly extend $f$ to a linear function $\widetilde{f}:\cV(\cE)\rightarrow \bF_{2}$. If $\widetilde{f}|_{\cN(\cE,\cU)}$ is not \textup{SNR}-satisfying, return an arbitrary answer and abort. Otherwise, return $\widetilde{f}$.
\end{enumerate}
Before describing Bob's strategy, we first prepare a table of ``distinguished answers'' for Bob. For each Bob-packet $(\cE',\cU')$ and each subspace $Q\subseteq \cV(\cE')$ of dimension at most $q_{0}$, if $(\cE',\cU',Q)$ is answerable with respect to $G[\cdot]$ (see~\Cref{def:bob_answerable}), pick a distinguished element of the list $\cS_{G}[\cE',\cU',Q]$, and denote this element by $\cS^*_{G}[\cE',\cU',Q]$. Now the strategy for Bob is:
\begin{enumerate}
\item Receive the question $(\cE',\cA,\cU')$. Let $Q=\sigma(\cA_{[q_{0}]})$. 
\item If $(\cE',\cU',Q)$ is not answerable (see~\Cref{def:bob_answerable}) with respect to $G[\cdot]$, return an arbitrary answer. Otherwise, let $(f,W)=\cS_{G}^*[\cE',\cU',Q]$,  randomly extend $f$ to a linear function $\widetilde{f}:\cV(\cE')\rightarrow \bF_{2}$, and return $\widetilde{f}$.
\end{enumerate}

\paragraph{Step 3: SNR-satisfaction.} In this step, the goal is to show that in the strategy for Alice constructed in Step 2, for any $(f,W)$ in Alice's list $\mathsf{LIST}_{Q'}^{\vec{\tau}}(F)$, there exists an extension \Big.$\widetilde{f}:\cV(\cE)\rightarrow \bF_{2}$ such that $\widetilde{f}|_{\cN(\cE,\cU)}$ is \textup{SNR}-satisfying. Suppose such an extension does not exist for a pair $(f,W)\in \mathsf{LIST}_{Q'}^{\vec{\tau}}(F)$. It follows that $f|_{\cN(\cE,\cU)\cap W}$ does not have an \textup{SNR}-satisfying extension to $\cN(\cE,\cU)$. Then for any linear subspace $\widetilde{L}\subseteq \cV(\cE)$ such that
\begin{equation}\label{eq:inner_soundness_SNR_1}
\Pru{L\in \mathsf{Gr}(\cV(\cE),\ell_{1})}{L+\cN(\cE,\cU)=\widetilde{L}\,\middle|\,Q'\subseteq L\subseteq W}>0,
\end{equation}
we must have
\begin{equation}\label{eq:inner_soundness_SNR_2}
\left.\widetilde{F}[\cE,\cU,\widetilde{L}]\right|_{W\cap \widetilde{L}}\neq \left.f\right|_{W\cap \widetilde{L}}.
\end{equation}
We claim that this implies
\begin{equation}\label{eq:inner_soundness_SNR_3}
\Pru{L\in \mathsf{Gr}(\cV(\cE),\ell_{1})}{F[\cE,\cU,L]=f|_{L}\,\middle|\,Q'\subseteq L\subseteq W\text{ and }L+\cN(\cE,\cU)=\widetilde{L}}\leq 2^{-(\ell_{1}-q_{0})+3}.
\end{equation}
Then taking expectation over $\widetilde{L}$ yields
\[
\Pru{L\in \mathsf{Gr}(\cV(\cE),\ell_{1})}{F[\cE,\cU,L]=f|_{L}\,\middle|\,Q'\subseteq L\subseteq W}\leq 2^{-(\ell_{1}-q_{0})+3},
\]
contradicting the condition that $(f,W)\in\mathsf{LIST}_{Q'}^{\vec{\tau}}(F)$.

We now show that~\Cref{eq:inner_soundness_SNR_1,eq:inner_soundness_SNR_2} together imply~\Cref{eq:inner_soundness_SNR_3}. Indeed,~\Cref{eq:inner_soundness_SNR_1} implies that
\[
\dim\left(W\cap \widetilde{L}\middle/(W\cap \cN(\cE,\cU)+Q')\right)\leq \ell_{1}-\dim Q',
\]
which in turn (combined with~\Cref{eq:inner_soundness_SNR_1}) implies 
\begin{equation}\label{eq:inner_soundness_SNR_4}
\Pru{L\in\mathsf{Gr}(\cV(\cE),\ell_{1})}{L+\cN(\cE,\cU)=\widetilde{L}\,\middle|\,Q'\subseteq L\subseteq W\cap \widetilde{L}}\geq \prod_{i=1}^{\infty}(1-2^{-i})\geq \frac{1}{4}.
\end{equation}
If~\Cref{eq:inner_soundness_SNR_3} does not hold, then due to~\Cref{eq:inner_soundness_SNR_4} we have
\[
\Pru{L\in \mathsf{Gr}(\cV(\cE),\ell_{1})}{F[\cE,\cU,L]=f|_{L}\text{ and }L+\cN(\cE,\cU)=\widetilde{L}\,\middle|\, Q'\subseteq L\subseteq W\cap \widetilde{L}}\geq 2^{-(\ell_{1}-q_{0})+1},
\]
which easily contradicts~\Cref{eq:inner_soundness_SNR_2}.

\paragraph{Step 4: good and excellent tuples.} Fix an arbitrary Bob-packet $(\cE',\cU')$ and a subspace $Q\subseteq \cV(\cE')$ such that $(\cE',\cU',Q)$ is answerable with respect to $G[\cdot]$. Recall from Step 2 that a distinguished element $(f,W)=\cS_{G}^*[\cE',\cU',Q]$ is chosen from the list $\cS_{G}[\cE',\cU',Q]$. Now imagine that for any Alice-packet $(\cE,\cU)$ such that $(\cE,\cU,*,Q,\cE',\cU',*)$ has nonzero probability under $\cD_{\inn}^{+}$,
\[(\cE',\cU',Q)\text{ \emph{sends} the message }(\cE',Q',f,W)\text{ to }(\cE,\cU,Q)\]
for any $Q'\subseteq Q$ with $\dim Q-\dim Q'\in\{0,1\}$ (so there are $2^{q_{0}}$ messages sent to the same $(\cE,\cU,Q)$, because there are $2^{q_{0}}$ possibilities of $Q'$). 
We say that $(\cE,\cU,Q,\cE',\cU')$ is a \emph{good tuple} if
\[
(\cE',\cU',Q)\textup{ sends }(\cE,\cU,Q)\text{ a message that }(\cE,\cU,Q)\text{ is }F\text{-friendly to}.
\]
We say that $(\cE,\cU,Q,\cE',\cU')$ is an \emph{excellent tuple} if
\[
(\cE',\cU',Q)\textup{ sends }(\cE,\cU,Q)\text{ a message that }(\cE,\cU,Q)\text{ is }F\text{-attached to}.
\]

We claim that for any possible sample $(\cE,\cU,\cE',\cU',\cA)\sim \cD_{\inn}$, if $(\cE,\cU,Q,\cE',\cU')$ is an excellent tuple, where $Q=\sigma(\cA_{[q_{0}]})$, then the strategies in Step 2 succeed with probability at least $\exp(-\ell_{1}^{2})$ on the question pair $((\cE,\cA,\cU),(\cE',\cA,\cU'))$. Indeed, if $(\cE,\cU,Q,\cE',\cU')$ is an excellent tuple, then by~\Cref{def:attached_decoding}, for some $Q'\subseteq Q$ with $\dim Q-\dim Q'\in \{0,1\}$, there exists a pair $(f_{0},W_{0})\in \mathsf{LIST}_{Q'}^{\vec{\tau}}(F)$ such that $W\subseteq W_{0}$ and $f_{0}|_{W}=f$. Then
\begin{itemize}
\item with probability (at least) $2^{-q_{0}}$, Alice chooses the correct $Q'$;
\item with probability (at least) $\left|\mathsf{LIST}_{Q'}^{\vec{\tau}}(F)\right|^{-1}$, Alice chooses the correct $(f_{0},W_{0})$;
\item since $\operatorname{codim}_{\cV(\cE)}W_{0}\leq 2q_{0}$, with probability at least $2^{-2q_{0}}$, Alice chooses an extension of $f_{0}$ to $\cV(\cE)$ that has an \textup{SNR}-satisfying restriction on $\cN(\cE,\cU)$ (that such an extension exists is argued in Step 3 of this proof);
\item since $\operatorname{codim}_{\cV(\cE')}W\leq q_{0}$, with probability at least $2^{-q_{0}}$, Bob chooses an extension of $f$ to $\cV(\cE')$ that is consistent with Alice's choice. 
\end{itemize}
In total, the probability that the verifier in $\GameBc$ accepts on the question pair $((\cE,\cA,\cU),(\cE',\cA,\cU'))$ is at least
\[
2^{-4q_{0}}\left|\mathsf{LIST}_{Q'}^{\vec{\tau}}(F)\right|^{-1}\geq \exp(-\ell_{1}^{2}),
\]
where the size of $\mathsf{LIST}_{Q'}^{\vec{\tau}}(F)$ is bounded by~\Cref{lem:dkkms_list_bound}. 

It remains to show that
\begin{equation}\label{eq:inner_soundness_excellence_goal}
\Pru{(\cE,\cU,\cE',\cU',\cA)\sim \cD_{\inn}}{\Bigl(\cE,\cU,\sigma(\cA_{[q_{0}]}),\cE',\cU'\Bigr)\text{ is excellent}}\geq \exp(-9\ell_{1}^{2}).
\end{equation}
Once that is shown, it would follow that the verifier in $\GameBc[H,k,\ell_{2},t,2^{2\ell_{1}-q_{0}}s,q_{0},\beta]$ accepts with probability at least
\[
\exp(-9\ell_{1}^{2})\cdot\exp(-\ell_{1}^{2})> 2^{-q_{2}/4},
\]
and we would be done.

\paragraph{Step 5: finding attached messages.}

Combining~\Cref{cor:Bob_answerable} and~\Cref{lem:Bob_back_to_Alice}, and the definition of good tuples, we know that
\[
\Pru{(\cE,\cU,*,Q,\cE',\cU',*)\sim \cD_{\inn}^{+}}{(\cE,\cU,Q,\cE',\cU')\text{ is good}}\geq \exp(-8\ell_{1}^{2}).
\]
To show~\Cref{eq:inner_soundness_excellence_goal}, it suffices to show that
\[
\Pru{(\cE,\cU,*,Q,\cE',\cU',*)\sim \cD_{\inn}^{+}}{(\cE,\cU,Q,\cE',\cU')\text{ is good but not excellent}}\leq t^{-1/2}.
\]
In fact, in the remainder of the proof, we fix an arbitrary Alice-packet $(\cE,\cU)$ and show that
\begin{equation}\label{eq:inner_soundness_good_not_excellent}
\Pru{(\cE,\cU,*,Q,\cE',\cU',*)\sim \cD_{\inn}^{+}}{(\cE,\cU,Q,\cE',\cU')\text{ is good but not excellent}\,\Big|\, \cE,\cU}\leq t^{-1/2}.
\end{equation}

For any $(\cE,\cU,Q)$ and any integer $k'\in \{k,k+1,\dots,3k\}$, we let $\cC(\cE,\cU,Q,k')\subseteq [t]$ denote the set of indices $c\in [t]$ such that there exists a $c$-critical Bob-packet $(\cE',\cU')$ with:
\begin{itemize}
\item $|E_{c}'|=k'$, where $E_{c}'$ is the critical block of $\cE'=(E_{1}',\dots,E_{t}')$; and
\item $(\cE,\cU,Q,\cE',\cU')$ is good but not excellent.
\end{itemize}
We also define $\cC(\cE,\cU,Q)=\bigcup_{k'\in \{k,\dots,3k\}}\cC(\cE,\cU,Q,k')$.

Roughly speaking, $\cC(\cE,\cU,Q)$ is the set of critical coordinates for which $(\cE,\cU,Q)$ receives some message that it is $F$-friendly to but not $F$-attached to. This set cannot be too large, because otherwise,~\Cref{lem:sunflower_transfer} could help find one of these messages that $(\cE,\cU,Q)$ is $F$-attached to, leading to a contradiction. More formally, for an arbitrary $(\cE,\cU,Q)$ and thinning size $k'\in \{k,\dots,3k\}$, suppose that $\left|\cC(\cE,\cU,Q,k')\right|\geq 2^{q_{0}}\exp(k^{8})$. Then by the pigeonhole principle, for some $Q'$ and on at least $\exp(k^{8})$ distinct critical indices, $(\cE,\cU,Q)$ receives messages of the form $(\cE',Q',f,W)$ with the same thinning size $k'$, and it is $F$-friendly to but not $F$-attached to them. This contradicts~\Cref{lem:sunflower_transfer}. Therefore, we conclude that
\begin{equation}\label{eq:inner_soundness_cC_bound}
\left|\cC(\cE,\cU,Q)\right|\leq \sum_{k'=k}^{3k}\left|\cC(\cE,\cU,Q,k')\right|\leq (2k+1)\cdot 2^{q_{0}}\exp(k^{8})\leq \exp(2k^{8}).
\end{equation}

\paragraph{Step 6: finishing the proof.} We are now ready to prove~\Cref{eq:inner_soundness_good_not_excellent}. When we sample $(\cE,\cU,*,Q,\cE',\cU',*)\sim \cD_{\mathsf{in}}^+$ conditioned on $(\cE,\cU)$, let $c\in [t]$ be the critical coordinate that we use as an auxiliary random variable in the sampling procedure. If we condition only on $(\cE,\cU)$,~\Cref{prop:Joint_inner_independence} shows that the critical coordinate $c$ is uniformly distributed over $[t]$, and~\Cref{eq:inner_soundness_cC_bound} could be directly used to bound the probability of being good but not excellent. In reality (\Cref{eq:inner_soundness_good_not_excellent}), we also need to condition on $Q$. The intuition is that conditioning on $Q$ does not bias the distribution of $c$ by too much. Formally, we do the following calculation.

By definitions, we have
    \begin{align*}
        &\quad\Pru{(\cE,\cU,*,Q,\cE',\cU',*)\sim \cD_{\mathsf{in}}^+}{\text{$(\cE,\cU,Q,\cE',\cU')$ is good but not excellent}\,\Big|\,\cE,\cU} \\
        & \leq \Pru{(\cE,\cU,*,*,\cA)\sim \cD_{\mathsf{in}}}{c \in \cC\Bigl(\cE,\cU,\sigma(\cA_{[q_0]})\Bigr)\,\Big|\,\cE,\cU} \tag{seeing $Q$ as $\sigma (\cA_{[q_0]})$}\\
        &=  \sum_{j=1}^{t}\Pru{(\cE,\cU,*,*,\cA)\sim \cD_{\mathsf{in}}}{ c=j \text{ and }j\in  \cC\Bigl(\cE,\cU,\sigma(\cA_{[q_0]})\Bigr)\,\Big|\,\cE,\cU}\\
        &= \frac{1}{t}\sum_{j=1}^{t}\Pru{(\cE,\cU,*,*,\cA)\sim \cD_{\mathsf{in}}}{ j\in \cC\Bigl(\cE,\cU,\sigma(\cA_{[q_0]})\Bigr)\,\Big|\,\cE,\cU\text{ and }\{c=j\}}  \tag{the marginal of $c$ is uniform}\\
        & = \frac{1}{t}\sum_{B \in \mathcal{V}(\mathcal{E})^{q_0}}\sum_{j=1}^{t} \Pru{(\cE,\cU,*,*,\cA)\sim \cD_{\mathsf{in}}}{ j\in \cC(\cE,\cU,\operatorname{span}(B))\text{ and }\rho (\cA_{[q_0]}) = B \,\Big|\,\cE,\cU\text{ and }\{c=j\}}  \\
        &\leq \frac{\exp(2k^{8})}{t} \sum_{B \in \mathcal{V}(\mathcal{E})^q_0}  \max_{j\in [t]}\Pru{(\cE,\cU,*,*,\cA)\sim \cD_{\mathsf{in}}}{ \rho (\cA_{[q_0]}) = B \,\Big|\,\cE,\cU\text{ and }\{c=j\}}. \tag{using~\Cref{eq:inner_soundness_cC_bound}}
    \end{align*}
    It suffices to prove the following upper bound $$\sum_{B \in \mathcal{V}(\mathcal{E})^{q_0}}  \max_{j\in [t]}\Pru{(\cE,\cU,*,*,\cA)\sim \cD_{\mathsf{in}} }{ \rho (\cA_{[q_0]}) = B \,\Big|\,\cE,\cU\text{ and }\{c=j\} }\leq 2^{9q_0k^2}. $$ 
    Let $N:= 2^{9q_0k^2}$. It is easy to see that the number of valid $B$ is at most $N^{t}$. Also, for every fixed $j\in [t]$ and $B\in \cV(\cE)^{q_0}$, we have $$\Pru{(\cE,\cU,*,*,\cA)\sim \cD_{\mathsf{in}} }{ \rho (\cA_{[q_0]}) = B \,\Big|\,\cE,\cU\text{ and }\{c=j\} } \leq N^{-t+1}$$
    since the projections of $\rho (\cA_{[q_0]})$ on the non-critical blocks are independent and uniformly random. Combining the two facts, we obtain the desired estimate. 
\end{proof}

\section{Acknowledgments}
We thank Mark Braverman, Irit Dinur, Guy Kindler, Subhash Khot, Ran Raz, Muli Safra for several discussions and collaborations that led to this work.
\bibliographystyle{alpha}
\bibliography{reference}
\appendix

\section{Basic Linear Algebra}\label{sec:linear_algebra}

In this appendix, we collect some basic linear algebraic facts about bilinear maps. Recall from~\Cref{subsec:general_notation} that given an $\bF_{2}$-vector space $V$ and a bilinear map $f:V\otimes V\rightarrow \bF_{2}$, we use $T_{f}$ to denote the associated linear map from $V$ to $V^*$. The rank of $f$ is defined to be the rank of $T_{f}$. The first proposition below analyzes how the rank of $f$ changes after being restricted to a subspace $U\otimes U\subseteq V\otimes V$.

\begin{proposition}\label{prop:rank_inequality}
Let $V$ be a finite dimensional $\bF_{2}$-vector space and let $U\subseteq V$ be a linear subspace. For any bilinear form $f:V\otimes V\rightarrow \bF_{2}$, we have 
\[\mathrm{rank}(f)- 2 \operatorname{codim} U\leq \mathrm{rank}(f|_{U\otimes U})\leq \mathrm{rank}(f).\]
\end{proposition}
\begin{proof}
Let $\iota:U\hookrightarrow V$ be the embedding map, and let $g=f|_{U\otimes U}$. Then we have the commutative diagram
\[
\begin{tikzcd}
U \arrow[r, "\iota"]
  \arrow[rrr, "T_g", bend right=30]
& V \arrow[r, "T_f"]
& V^* \arrow[r, "\iota^\dagger"]
& U^*.
\end{tikzcd}
\]
Both composing with $\iota$ on the right and composing with $\iota^{\dagger}$ on the left reduce the rank by at most $\operatorname{codim} U$, so we have
\[
\operatorname{rank}(T_{f})-2\operatorname{codim}U\leq \operatorname{rank}(T_{g})\leq \operatorname{rank}(T_{f}),
\]
as desired.
\end{proof}

The next two corollaries of~\Cref{prop:rank_inequality} are both used in~\Cref{lem:Game1b_soundness_low_rank}.

\begin{corollary}\label{cor:full_rank_restriction}
Let $V$ be a finite dimensional $\bF_{2}$-vector space and let $U\subseteq V$ be a linear subspace. Suppose $f:V\otimes V\rightarrow \bF_{2}$ is a bilinear form such that $\mathrm{rank}(f|_{U\otimes U})=\mathrm{rank}(f)$. Then for any $v\in V$, there exists $u\in U$ such that
\[
f(v\otimes x)=f(u\otimes x)\qquad\text{for all }x\in V.
\]
In other words, letting $g=f|_{U\otimes U}$, we have $h|_{U}\in \operatorname{image}(T_{g})$ for each $h\in \operatorname{image}(T_{f})$.
\end{corollary}
\begin{proof}
Let $\iota:U\hookrightarrow V$ be the embedding map. Since $\mathrm{rank}(f|_{U\otimes U})=\mathrm{rank}(f)$, it follows from the proof of~\Cref{prop:rank_inequality} that
\[\operatorname{image}( T_{f}\circ \iota)=\operatorname{image}(T_{f}),\]
which is exactly the desired conclusion.
\end{proof}

\begin{corollary}\label{cor:full_rank_core}
For any positive integers $n,r$ and any symmetric bilinear form $f:\bF_{2}^{n}\otimes \bF_{2}^{n}\rightarrow \bF_{2}$ with rank $r$, there exists an $r$-element subset $S\subseteq [n]$ such that the restriction of $f$ to $\bF_{2}^{S}\otimes \bF_{2}^{S}$ has rank $r$.
\end{corollary}

\begin{proof}
Since $\dim\ker(T_{f})=n-r$, there exists an $r$-element subset $S\subseteq [n]$ such that
\begin{equation}\label{eq:full_rank_restriction}
\bF_{2}^{S}+\ker(T_{f})=\bF_{2}^{n}.
\end{equation}
Let $\iota:\bF_{2}^{S}\rightarrow \bF_{2}^{n}$ be the embedding map. To show that the restriction of $f$ to $\bF_{2}^{S}\otimes \bF_{2}^{S}$ has rank $r$, it suffices to show that
\[
\operatorname{rank}(\iota^{\dagger}\circ T_{f}\circ \iota)=\operatorname{rank}(T_{f}).
\]
The symmetry of $f$ implies
\[
(T_{f}\circ \iota)^{\dagger}=\iota^{\dagger}\circ T_{f},
\]
so it suffices to show that $\operatorname{rank}(T_{f}\circ \iota)=\operatorname{rank}(T_{f})$. But this is immediate due to \Cref{eq:full_rank_restriction}.
\end{proof}

The next lemma is a standard decomposition of bilinear forms into symmetric rank-1 forms. It is used also in~\Cref{lem:Game1b_soundness_low_rank}.

\begin{lemma}\label{lem:rank1_decomposition}
Let $V$ be a finite dimensional $\bF_{2}$-vector space, and let $r$ be a positive integer. For any symmetric bilinear form $f:V\otimes V\rightarrow \bF_{2}$ with rank $r$, there exist a nonnegative integer $k\leq 3r/2$ and $k$ vectors $v_{1},\dots,v_{k}\in V$ such that
\[
f(x\otimes y)=\sum_{i=1}^{k}f(x\otimes v_{i})f(y\otimes v_{i})\qquad\text{for all } x,y\in V.
\]
In other words, any symmetric bilinear form of rank $r$ has a tensor decomposition (see~\Cref{def:tensor_decomposition}) into at most $3r/2$ linear functionals. 
\end{lemma}

\begin{proof}
We proceed by induction on $r$. The conclusion clearly holds for $r=0$. In the rest of the proof, assume $r\geq 1$ and the conclusion is true for all smaller values of $r$. Consider the following two cases.

\textbf{Case 1:} there exists a vector $v\in V$ such that $f(v\otimes v)=1$. Define a subspace
\[
U=\{u\in V\mid f(u\otimes v)=0\}\subseteq V.
\]
It is clear that $\mathrm{ker}(T_{f})\subseteq U$, and since $v\notin U$ we also have $\dim U\leq \dim V-1$. Therefore,
\begin{align*}
\mathrm{rank}(f|_{U\otimes U})\leq \dim U-\dim \ker(T_{f})\leq \dim V-\dim \ker(T_{f})-1=\operatorname{rank}(f)-1.
\end{align*}
Define a linear map $\pi:V\rightarrow V$ by letting $\pi(x)=x-f(x\otimes v)v$ for all $x\in V$. It is clear that
\[
f(\pi(x)\otimes v)=f(x\otimes v)-f(x\otimes v)^{2}=0\qquad\text{for all }x\in V,
\]
and hence the image of $\pi$ is contained in $U$. Furthermore, for each $x,y\in V$ we have
\begin{align}
f\big(\pi(x)\otimes \pi(y)\big)&=f(x\otimes y)-2f(x\otimes v)f(y\otimes v)+f(x\otimes v)f(y\otimes v)f(v\otimes v)\nonumber\\
&=f(x\otimes y)-f(x\otimes v)f(y\otimes v).\label{eq:3r_over_2_case_1a}
\end{align}
Applying the induction hypothesis to the bilinear form $f\big(\pi(\cdot)\otimes \pi(\cdot)\big)$ (which has rank at most $r-1$ because $\mathrm{image}(\pi)\subseteq U$ and $\mathrm{rank}(f|_{U\otimes U})\leq r-1$), we get $k$ vectors $v_{1},\dots,v_{k}\in V$ with $k\leq 3(r-1)/2$ such that
\begin{equation}\label{eq:3r_over_2_case_1b}
f\big(\pi(x)\otimes \pi(y)\big)=\sum_{i=1}^{k}f(x\otimes v_{i})f(y\otimes v_{i})\qquad\text{for all } x,y\in V.
\end{equation}
Combining \Cref{eq:3r_over_2_case_1a,eq:3r_over_2_case_1b}, we get
\[
f(x\otimes y)=f(x\otimes v)f(y\otimes v)+\sum_{i=1}^{k}f(x\otimes v_{i})f(y\otimes v_{i})\qquad\text{for all } x,y\in V,
\]
as desired (noting that $k+1\leq 3r/2$).

\textbf{Case 2:} $f(v\otimes v)=0$ holds for all $v\in V$. Since $\mathrm{rank}(f)\geq 1$, there exist $v,w\in V$ such that $f(v\otimes w)=1$. It follows from the case assumption that $v$ and $w$ are linearly independent. Define a subspace
\[
U=\{u\in V\mid f(u\otimes v)=f(u\otimes w)=0\}\subseteq V.
\]
The rest of the argument is entirely analogous to Case 1. Since $U\cap \mathrm{span}\{v,w\}=\{0\}$ we have $\dim U\leq \dim V-2$. Due to $\ker(T_{f})\subseteq U$, we then have
$\mathrm{rank}(f|_{U\otimes U})\leq r-2$. Define a linear map $\pi:V\rightarrow V$ by letting  
\[
\pi(x)=x-f(x\otimes v)w-f(x\otimes w)v\qquad\text{for all }x\in V.
\]
Similarly to Case 1, it is easy to verify that
\[
f(\pi(x)\otimes v)=f(\pi(x)\otimes w)=0\qquad\text{for all }x\in V,
\]
and hence the image of $\pi$ is contained in $U$. Furthermore, for each $x,y\in V$ one gets
\begin{align*}
&\quad f\big(\pi(x)\otimes \pi(y)\big)\\
&=f(x\otimes y)-f(x\otimes v)f(y\otimes w)-f(x\otimes w)f(y\otimes v)\\
&=f(x\otimes y)-f(x\otimes v)f(y\otimes v)-f(x\otimes w)f(y\otimes w)-f(x\otimes (v+w))f(y\otimes (v+w)).
\end{align*}
As in Case 1, we then apply the induction hypothesis to decompose $f\big(\pi(x)\otimes \pi(y)\big)$ into the sum of at most $3(r-2)/2$ terms of the form $f(x\otimes v_{i})f(y\otimes v_{i})$. Thus $f(x\otimes y)$ can be written as the sum of at most $3(r-2)/2+3=3r/2$ such terms, as desired.
\end{proof}

The next lemma is the key to the success of the ``local tensor tests'' in our middle PCP, and it is used in~\Cref{lem:invoked_middle}. In words, it asserts that if a bilinear form $f$ has rank at least $2+2\codim W$, then (1) restricting it to $W\otimes W$ drops the rank by at most $2\codim W$, and (2) with high probability, restricting it to a random subspace $U\subseteq W$ of sufficiently large dimension does not drop the rank to at most $1$.

\begin{lemma}\label{lem:rk1bound}
Let $V$ be a finite dimensional $\bF_{2}$-vector space, and let $Q,W\subseteq V$ be linear subspaces such that $Q\subseteq W$. Denote $q=\dim Q$. Let $\ell$ be a positive integer such that $q+2\leq \ell\leq \dim W$. Then for any bilinear form $f:V\otimes V\rightarrow \bF_{2}$ such that $\operatorname{rank}(f)\geq 2+2\operatorname{codim} W$, we have
\[
\Pru{U\in\mathsf{Gr}(W,\ell)}{\operatorname{rank}(f|_{U\otimes U})\leq 1\,\big|\,Q\subseteq U}\leq 2^{4-(\ell-q)/2}.
\]
\end{lemma}
\begin{proof}
Let $g=f|_{W\otimes W}$. By~\Cref{prop:rank_inequality}, the assumption that $\operatorname{rank}(f)\geq 2+2\operatorname{codim} W$ implies $\operatorname{rank}(g)\geq 2$. It suffices to work with the restricted map $g$ in the rest of the proof.

Let $\ell'=\lfloor (\ell-q)/2\rfloor\geq 1$. To pick a uniformly random $U\in \mathsf{Gr}(W,\ell)$ that contains $Q$, we first pick uniformly random subspaces $U_{1},U_{2}\in\mathsf{Gr}(W,\ell')$ independently, and then pick a uniformly random $U\in \mathsf{Gr}(W,\ell)$ conditioned on $Q,U_{1},U_{2}\subseteq U$. Since we always have \[\operatorname{rank}(f|_{U\otimes U})=\operatorname{rank}(g|_{U\otimes U})\geq \operatorname{rank}(g|_{U_{1}\otimes U_{2}}),\]
it suffices to show that
\[
\Pru{U_{1},U_{2}\in\mathsf{Gr}(W,\ell')}{\operatorname{rank}(g|_{U_{1}\otimes U_{2}})\leq 1\big.}\leq 2^{3-\ell'}.
\]

As in the proof of~\Cref{prop:rank_inequality}, we let $\iota_{U_{1}}:U_{1}\hookrightarrow W$ and $\iota_{U_{2}}:U_{2}\hookrightarrow W$ be the two canonically defined embedding maps. Then we have
\[
\operatorname{rank}(g|_{U_{1}\otimes U_{2}})=\operatorname{rank}\big(\iota_{U_{2}}^{\dagger}\circ T_{g}\circ \iota_{U_{1}}\big).
\]
Note that
\[
\operatorname{rank}(T_{g}\circ \iota_{U_{1}})=\dim (U_{1}+\ker(T_{g}))-\dim \ker(T_{g}).
\]
Therefore, denoting $w=\dim W$ and $r=\operatorname{rank}(g)=w-\dim \ker(T_{g})$, we can calculate
\begin{align*}
&\quad \Pru{U_{1}\in \mathsf{Gr}(W,\ell')}{\operatorname{rank}(T_{g}\circ \iota_{U_{1}})\leq 1\big.}\\
&=\Pru{U_{1}\in \mathsf{Gr}(W,\ell')}{\dim (U_{1}+\ker(T_{g}))-\dim \ker(T_{g})\leq 1\big.}\\
&\leq \sum_{K\in \mathsf{Gr}(W,w-r+1),\;K\supseteq \ker(T_{g})}\Pru{U_{1}\in \mathsf{Gr}(W,\ell')}{U_{1}\subseteq K\big.}\\
&\leq (2^{r}-1) \cdot\qbinom{w-r+1}{\ell'} \qbinom{w}{\ell'}^{-1}\\
&=(2^{r}-1)\cdot\prod_{i=0}^{\ell'-1}\frac{2^{w-r-i+1}-1}{2^{w-i}-1}\\
&\leq 2^{r}\cdot (2^{-r+1})^{\ell'}\leq 2^{-\ell'+2}.\tag{using $r\geq 2$ and $\ell'\geq 1$}
\end{align*}
Similarly, since
\[
\operatorname{rank}\big(\iota_{U_{2}}^{\dagger}\circ T_{g}\circ \iota_{U_{1}}\big)=\dim\big(U_{2}+\ker\big(\iota_{U_{1}}^{\dagger}\circ T_{g}^{\dagger}\big)\big)-\dim \ker\big(\iota_{U_{1}}^{\dagger}\circ T_{g}^{\dagger}\big),
\]
an analogous calculation shows
\[
\Pru{U_{2}\in \mathsf{Gr}(W,\ell')}{\operatorname{rank}\big(\iota_{U_{2}}^{\dagger}\circ T_{g}\circ \iota_{U_{1}}\big)\leq 1}\leq 2^{-\ell'+2}\qquad\text{whenever}\qquad \dim \ker\big(\iota_{U_{1}}^{\dagger}\circ T_{g}^{\dagger}\big)\leq \dim W-2.
\]
But $\dim \ker\big(\iota_{U_{1}}^{\dagger}\circ T_{g}^{\dagger}\big)\leq \dim W-2$ happens exactly when $\operatorname{rank}(T_{g}\circ \iota_{U_{1}})\geq 2$, the probability of which has already been shown to be at least $1-2^{-\ell'+2}$. Therefore, we conclude that
\[
\Pru{U_{1},U_{2}\in \mathsf{Gr}(W,\ell')}{\operatorname{rank}\big(\iota_{U_{2}}^{\dagger}\circ T_{g}\circ \iota_{U_{1}}\big)\leq 1}\leq 2\cdot 2^{-\ell'+2}=2^{-\ell'+3},
\]
as desired.
\end{proof}

\section{From Sub-$4$-to-$1$ to Exact $4$-to-$1$ Games}\label{sec:exact_4_to_1}

In this appendix, we show that the hardness of approximating sub-$4$-to-$1$ games implies the exact $4$-to-$1$ games theorem.

\begin{proof}[Proof of~\Cref{thm:main} assuming~\Cref{thm:sub-main}]
Fix an arbitrary parameter $\varepsilon\in (0,1)$ and a positive integer $r$. Write $k=\lceil 4/\varepsilon\rceil$ and $r'=kr$. We will prove~\Cref{thm:main} by constructing a polynomial-time reduction from $\GapSdone{4}{1}_{r}(1,\varepsilon/2)$ to $\Gapdone{4}{1}_{r'}(1,\varepsilon)$. 

Given a sub-$4$-to-$1$ game instance $\Psi$ with underlying graph $(L\sqcup R,E)$, alphabets $\Sigma_{L},\Sigma_{R}$ and constraints $\{\Phi_{e}\subseteq \Sigma_{L}\times \Sigma_{R}\}_{e\in E}$, we construct a $4$-to-$1$ game instance $\Psi'$ as follows:
\begin{enumerate}
\item Let $\Psi'$ have the same constraint graph $(L\sqcup R,E)$ as $\Psi$, except that each edge in $E$ is replaced by $k$ parallel copies.
\item The left-alphabet of $\Psi'$ is $\Sigma_{L}'=\Sigma_{L}\times (\bZ/k\bZ)$, and the right-alphabet of $\Psi'$ is $\Sigma_{R}'=\Sigma_{R}\times (\bZ/k\bZ)$.
\item For each edge $e=(u,v)\in E$, according to~\Cref{def:sub4to1_constraint}, there exists a partition $\Sigma_{L}=\Sigma_{L,1}\sqcup \dots\sqcup \Sigma_{L,t}$ into sets of size $4$ and a partition $\Sigma_{R}=\Sigma_{R,1}\sqcup \dots\sqcup \Sigma_{R,t}$ into singleton sets such that
\begin{equation}\label{eq:sub4to1_Phi_e}
\Phi_{e}=\bigcup_{i=1}^{t}\Gamma_{L,i}\times \Gamma_{R,i},\qquad\textup{for some subsets }\Gamma_{L,i}\subseteq \Sigma_{L,i}\textup{ and }\Gamma_{R,i}\subseteq \Sigma_{R,i}.
\end{equation}
For each $j,\ell\in \bZ/k\bZ$ and $i\in [t]$, we let
\begin{align}
\Sigma^{(\ell)}_{L,i,j}&=\bigl(\Gamma_{L,i}\times \{j\}\bigr)\cup\bigl((\Sigma_{L,i}\setminus \Gamma_{L,i})\times \{j+\ell\}\bigr)\subseteq \Sigma_{L}\times (\bZ/k\bZ)=\Sigma_{L}',\label{eq:sub4to1_SigmaL}\\
\Sigma^{(\ell)}_{R,i,j}&=\bigl(\Gamma_{R,i}\times \{j\}\bigr)\cup\bigl((\Sigma_{R,i}\setminus \Gamma_{R,i})\times \{j-\ell\}\bigr)\subseteq \Sigma_{R}\times (\bZ/k\bZ)=\Sigma_{R}'.\label{eq:sub4to1_SigmaR}
\end{align}
We then define a constraint $\Phi_{e,\ell}'\subseteq \Sigma_{L}'\times \Sigma_{R}'$
\begin{equation}\label{eq:sub4to1_Phi_e_ell}
\Phi_{e,\ell}'=\bigcup_{i=1}^{t}\bigcup_{j\in \bZ/k\bZ}\Sigma_{L,i,j}^{(\ell)}\times \Sigma_{R,i,j}^{(\ell)}.
\end{equation}
Since for each $i\in [t]$, the sets $\Sigma_{L,i,1}^{(\ell)},\dots,\Sigma_{L,i,k}^{(\ell)}$ have size $|\Sigma_{L,i}|=4$ and form a partition of $\Sigma_{L,i}\times (\bZ/k\bZ)$, and similarly the sets $\Big.\Sigma_{R,i,1}^{(\ell)},\dots,\Sigma_{R,i,k}^{(\ell)}$ have size $|\Sigma_{R,i}|=1$ and form a partition of $\Sigma_{R,i}\times (\bZ/k\bZ)$, it is clear that $\Phi_{e,\ell}'$ is a $4$-to-$1$ constraint.

\item For each edge $e\in E$, we put in $\Psi'$ the $k$ constraints $\Phi_{e,1}',\dots,\Phi_{e,k}^{'}$. The total number of constraints in $\Psi'$ is therefore equal to $k|E|$.
\end{enumerate}

It is easy to see that $\Psi'$ can be constructed from $\Psi$ in polynomial time. We then show the completeness and soundness of the reduction.

\paragraph{Completeness:} if $\operatorname{val}(\Psi)=1$, there exist assignment maps $A_{L}:L\rightarrow \Sigma_{L}$ and $A_{R}:R\rightarrow \Sigma_{R}$ that satisfy all constraints $\Phi_{e}$, for $e\in E$. Then it is clear that the map $A_{L}':L\rightarrow \Sigma_{L}'$ defined by $A_{L}'(u)=(A_{L}(u),1)$ and the map $A_{R}':R\rightarrow \Sigma_{R}'$ defined by $A_{R}'(v)=(A_{R}(v),1)$ jointly satisfy all constraints $\Phi_{e,\ell}'$, for $e\in E$ and $\ell\in \bZ/k\bZ$. Consequently, we have $\operatorname{val}(\Psi')=1$.

\paragraph{Soundness:} suppose we have assignment maps $A_{L}':L\rightarrow \Sigma_{L}'$ and $A_{R}':R\rightarrow \Sigma_{R}'$ given by
\[
A'_{L}(u)=(A_{L}(u),K_{L}(u))\qquad\text{and}\qquad A_{R}'(v)=(A_{R}(v),K_{R}(v)).
\]
For any edge $e=(u,v)\in E$, we claim that if $(A_{L}(u),A_{R}(v))\not\in \Phi_{e}$, then at most 2 of the $k$ constraints $\Phi_{e,1}',\dots,\Phi_{e,k}'$ are satisfied by $(A_{L}',A_{R}')$. This means
\[
\operatorname{val}_{\Psi'}(A_{L}',A_{R}')\leq \operatorname{val}_{\Psi}(A_{L},A_{R})+2/k\leq \operatorname{val}_{\Psi}(A_{L},A_{R})+\varepsilon/2,
\]
and by the arbitrariness of $(A_{L}',A_{R}')$ it follows that
\[
\operatorname{val}(\Psi')\leq \operatorname{val}(\Psi)+\varepsilon/2,
\]
as desired.

We now fix an edge $e=(u,v)\in E$ such that $(A_{L}(u),A_{R}(v))\notin \Phi_{e}$, and show that at most 2 of the $k$ constraints $\Phi_{e,1}',\dots,\Phi_{e,k}'$ are satisfied by $(A_{L}',A_{R}')$. Let the subsets \[\Gamma_{L,i},\Sigma_{L,i}\subseteq \Sigma_{L},\qquad \Gamma_{R,i},\Sigma_{R,i}\subseteq \Sigma_{R},\qquad \Sigma_{L,i,j}^{(\ell)}\subseteq \Sigma_{L}',\qquad\text{and}\qquad\Sigma_{R,i,j}^{(\ell)}\subseteq \Sigma_{R}'\] be as defined in~\Cref{eq:sub4to1_Phi_e,eq:sub4to1_SigmaL,eq:sub4to1_SigmaR}. We divide into the following five cases.

Case 1: $(A_{L}(u),A_{R}(v))\notin \Sigma_{L,i}\times \Sigma_{R,i}$ for any $i\in [t]$. In this case, it is easy to see from~\Cref{eq:sub4to1_SigmaL,eq:sub4to1_SigmaR,eq:sub4to1_Phi_e_ell} that whichever values $K_{L}(u)$ and $K_{R}(v)$ take, we always have
\[
(A_{L}'(u),A'_{R}(v))\notin \bigcup_{\ell=1}^{k}\Phi_{e,\ell}'.
\]
Therefore, none of the constraints $\Phi_{e,1}',\dots,\Phi_{e,k}'$ are satisfied by $(A_{L}',A_{R}')$.

Case 2: $(A_{L}(u),A_{R}(v))\in (\Sigma_{L,i}\setminus\Gamma_{L,i})\times \Gamma_{R,i}$ for some $i\in [t]$. In this case, $(A_{L}'(u),A_{R}'(v))\in \Phi_{e,\ell}'\Big.$ if and only if $K_{L}(u)=K_{R}(v)+\ell$, and hence exactly one of the constraints $\Phi_{e,1}',\dots,\Phi_{e,k}'$ is satisfied by $(A_{L}',A_{R}')$.

Case 3: $(A_{L}(u),A_{R}(v))\in \Gamma_{L,i}\times (\Sigma_{R,i}\setminus \Gamma_{R,i})\Big.$ for some $i\in [t]$. This case is symmetric to the previous case, and exactly one of the constraints $\Phi_{e,1}',\dots,\Phi_{e,k}'$ is satisfied by $(A_{L}',A_{R}')$.

Case 4: $(A_{L}(u),A_{R}(v))\in (\Sigma_{L,i}\setminus\Gamma_{L,i})\times (\Sigma_{R,i}\setminus \Gamma_{R,i})\Big.$ for some $i\in [t]$. In this case, we have $(A_{L}'(u),A_{R}'(v))\in \Phi_{e,\ell}'$ if and only if $K_{L}(u)-\ell=K_{R}(v)+\ell$, and hence at most two of the constraints $\Phi_{e,1}',\dots,\Phi_{e,k}'\Big.$ are satisfied by $(A_{L}',A_{R}')$.

Case 5: $(A_{L}(u),A_{R}(v))\in \Gamma_{L,i}\times \Gamma_{R,i}$ for some $i\in [t]$. This case cannot happen because if it does then $(A_{L},A_{R})$ satisfies $\Phi_{e}$, violating the assumption.
\end{proof}

\section{Approximate Coloring of Hypergraphs}\label{sec:apx_hyp}

In this appendix, we present a proof of the ``starting point'' theorem (restated below) used in our reduction to $4$-to-$1$ games. The proof of this theorem is essentially due to~\cite{DGKR,DinurRegevSmyth}, and we claim no originality here. The inclusion of the proof here is purely for the sake of expositional completeness. 

\hgcolor*

To prove~\Cref{thm:hg_coloring}, we need to use the PCP theorem of~\cite{ALMSS,AS} and the parallel repetition theorem of~\cite{Raz}. More specifically, we apply the parallel repetition theorem to the equation-vs.-variable game associated to $3\mathsf{SAT}5$ instances.\footnote{By $3\mathsf{SAT}5$, we mean $3\mathsf{SAT}$ instances where each variable is contained in exactly $5$ clauses; see~\cite[Proposition 2.1.2]{Feige}.} This step is standard in the hardness of approximation literature, and we state the result (\Cref{thm:bi-regular_LLC}) below without proof.

\begin{definition}
Let $d_{0}$ and $d_{1}$ be positive integers. A \emph{$(d_{0},d_{1})$-biregular projection label cover instance} $\Psi=(V_{0},V_{1},E,\Sigma_{0},\Sigma_{1},(\pi_{e})_{e\in E})$ consists of the following data:
\begin{enumerate}[label=(\arabic*)]
\item two vertex sets $V_{0}$ and $V_{1}$; 
\item a set $E$ of edges between $V_{0}$ and $V_{1}$ such that in the bipartite graph $(V_{0}\sqcup V_{1},E)$, every vertex in $V_{0}$ has degree $d_{0}$ and every vertex in $V_{1}$ has degree $d_{1}$;
\item two finite alphabets $\Sigma_{0},\Sigma_{1}$; and
\item a collection of constraints $(\pi_{e})_{e\in E}$, where each $\pi_{e}$ is a map from $\Sigma_{0}$ to $\Sigma_{1}$.
\end{enumerate}
The value of $\Psi$, denoted by $\mathrm{val}(\Psi)$, is the maximum value of
\[
\frac{1}{|E|}\sum_{(v_{0},v_{1})\in E}\ind{\pi_{(v_{0},v_{1})}\big(A_{0}(v_{0})\big)=A_{1}(v_{1})}
\]
achievable by assignment maps $A_{0}:V_{0}\rightarrow \Sigma_{0}$ and $A_{1}:V_{1}\rightarrow\Sigma_{1}$.
\end{definition}

\begin{theorem}[\cite{ALMSS,AS,Raz}]\label{thm:bi-regular_LLC}
There exists an absolute constant $\gamma\in (0,1)$ such that for every fixed positive integer $u$, given a $(3^{u},5^{u})$-biregular projection label cover instance $\Psi=(V_{0},V_{1},E,\Sigma_{0},\Sigma_{1},(\pi_{e})_{e\in E})$ such that 
\(|\Sigma_{0}|=7^{u}\) and \(|\Sigma_{1}|=2^{u}\), 
it is $\mathbf{NP}$-hard to distinguish between the case $\mathrm{val}(\Psi)=1$ and the case $\mathrm{val}(\Psi)\leq 2^{-\gamma u}$.
\end{theorem}

Assuming~\Cref{thm:bi-regular_LLC}, we now present a proof of~\Cref{thm:hg_coloring}.

\begin{proof}[Proof of~\Cref{thm:hg_coloring}] We proceed by a reduction from~\Cref{thm:bi-regular_LLC}. Let $\gamma \in (0,1)$ be as in~\Cref{thm:bi-regular_LLC}, and let $u=\lceil40\gamma ^{-2}c\rceil$.
We fix a $(3^{u},5^{u})$-biregular projection label cover instance \(\Psi=(V_{0},V_{1},E,\Sigma_{0},\Sigma_{1},(\pi_{e})_{e\in E})\) with $|\Sigma_{0}|=7^{u}$ and $|\Sigma_{1}|=2^{u}$; we will show how to map it to a 3-uniform hypergraph $H$ while linking the value of $\Psi$ to the colorability of $H$.

The proof is divided into 8 steps. We first construct the hypergraph $H$ in Steps 1 and 2. The completeness of the reduction is proved in Step 3, while Steps 4 to 8 will be devoted to proving the soundness of the reduction. Except for Step 6, the proof is almost identical to the one in~\cite{DinurRegevSmyth}.\footnote{The arguments in Steps 1, 5 and 8 are due to~\cite{DGKR}.} Throughout the proof, we let $\ell=2c^{2}$.

\paragraph{Step 1: multi-layered label cover.} Following the ideas of~\cite{DGKR,DinurRegevSmyth}, the first step is to perform a ``multi-layered'' parallel repetition on $\Psi$. For each $i\in [\ell]$, let $\cV_{i}$ and $R_{i}$ be the cartesian products
\[
\cV_{i}=V_{1}^{i-1}\times V_{0}^{\ell-i}\qquad\text{and}\qquad R_{i}=\Sigma_{1}^{i-1}\times \Sigma_{0}^{\ell-i}.\]
The $(\ell-1)$-tuples $\bfv=(v_{1},\dots,v_{\ell-1})\in \cV_{i}$ will be the ``blown-up vertices'' in the $i$-th layer, and $R_{i}$ can be viewed as their alphabets. For each $i,j\in [\ell]$ such that $i<j$, let $\cE_{i,j}$ be the collection of pairs
\[
(\mathbf{v},\mathbf{v}')=\big((v_{1},\dots,v_{\ell-1}),(v_{1}',\dots,v_{\ell-1}')\big)\in \cV_{i}\times \cV_{j}
\]
such that:
\begin{enumerate}[label=(\arabic*)]
\item $v_{s}=v_{s}'$ for all $s\in \{1,2,\dots,i-1\}\cup\{j,j+1,\dots,\ell-1\}$; and
\item $(v_{s},v'_{s})\in E$ for all $s\in \{i,i+1,\dots,j-1\}$.
\end{enumerate}
Given such a pair $(\bfv,\bfv')\in \cE_{i,j}$, we define a projection map $\pi_{(\bfv,\bfv')}:R_{i}\rightarrow R_{j}$ by
\[
\pi_{(\bfv,\bfv')}(\sigma_{1},\dots,\sigma_{\ell-1})=\left(\sigma_{1},\dots,\sigma_{i-1},\pi_{(v_{i},v_{i}')}(\sigma_{i}),\dots,\pi_{(v_{j-1},v_{j-1}')}(\sigma_{j-1}),\sigma_{j},\dots,\sigma_{\ell-1}\right),
\]
for all $\sigma_{1},\dots,\sigma_{i-1}\in \Sigma_{1}$ and $\sigma_{i},\dots,\sigma_{\ell-1}\in \Sigma_{0}$.

For each $i,j\in [\ell]$ such that $i<j$, we have
\begin{equation}\label{eq:pairs_in_multi_layer}
\left|\cE_{i,j}\right|=|\cV_{j}|\cdot (5^{u})^{j-i}=(5^{u})^{j-i}\cdot|V_{1}|^{j-1}|V_{0}|^{\ell-j} =5^{u(j-i)}\cdot (5/3)^{u(\ell-j)}\cdot |V_{1}|^{\ell-1},
\end{equation}
where we used the $(3^{u},5^{u})$-biregularity of $\Psi$ and its consequence $|V_{0}|/|V_{1}|=5^{u}/3^{u}$.

\paragraph{Step 2: construction of $H$.} Following~\cite{DinurRegevSmyth}, we construct the hypergraph $H$ by composing the multi-layered label cover from Step 1 with an inner gadget based on the Kneser graph. 

For each $i\in [\ell]$, we let $r_{i}=\big\lceil (|R_{i}|-c)/2\big\rceil$. The vertex set of $H$ is the disjoint union of cartesian products
\(
\bigsqcup_{i=1}^{\ell}\big(\cV_{i}\times \binom{R_{i}}{r_{i}}\big)
\), and the hyperedges of $H$ are the collection of unordered triples
\(
\big\{(\bfv,S_{1}),(\bfv,S_{2}),(\bfv',S_{3})\big\}
\)
such that:
\begin{enumerate}[label=(\arabic*)]
\item $\bfv\in \cV_{i}$ and $\bfv'\in \cV_{j}$ for some $i,j\in [\ell]$ such that $i<j$;
\item $(\bfv,\bfv')\in \cE_{i,j}$;
\item $S_{1},S_{2}\in \binom{R_{i}}{r_{i}}$ and $S_{3}\in \binom{R_{j}}{r_{j}}$; and
\item $S_{1}\cap S_{2}=\emptyset$ and $\pi_{(\bfv,\bfv')}(R_{i}\setminus(S_{1}\cup S_{2}))\subseteq S_{3}$.
\end{enumerate}
Using~\Cref{eq:pairs_in_multi_layer}, we can upper-bound the number of edges in $H$ by
\begin{equation}\label{eq:edges_in_multi_layer}
\sum_{i,j\in [\ell],\;i<j}\left|\cE_{i,j}\right|\cdot 2^{2|R_{i}|}\cdot 2^{|R_{j}|}\leq 8^{7^{u(\ell-1)}}\ell^{2}\cdot\max_{i,j\in [\ell],\;i<j}\left|\cE_{i,j}\right|\leq  5^{u(\ell-1)}8^{7^{u(\ell-1)}}\ell^{2}\cdot |V_{1}|^{\ell-1}.
\end{equation}

It is clear that $H$ can be constructed from $\Psi$ in polynomial time. It suffices to prove the following statements:
\begin{enumerate}[label=(\arabic*)]
\item Completeness: if $\mathrm{val}(\Psi)=1$, then $H$ is 2-colorable.
\item Soundness: if under some $c$-coloring of $H$, less than $\exp(-\exp(\Omega(c^{3})))$ fraction of edges of $H$ are monochromatic, then $\mathrm{val}(\Psi)>2^{-\gamma u}$.
\end{enumerate}
The first statement will be proved in Step 3, and the second will be proved in Steps 4 to 8.

\paragraph{Step 3: completeness.} In this step, we assume $\mathrm{val}(\Psi)=1$ and show that $H$ is 2-colorable. 

Suppose $A_{0}:V_{0}\rightarrow\Sigma_{0}$ and $A_{1}:V_{1}\rightarrow\Sigma_{1}$ are assignments that together satisfy every constraint of $\Psi$. Then for any $i\in [\ell]$, we can define a lifted assignment map $\widetilde{A}_{i}:\cV_{i}\rightarrow R_{i}$ by letting 
\[
\widetilde{A}_{i}(\bfv)=\big(A_{1}(v_{1}),\dots,A_{1}(v_{i-1}),A_{0}(v_{i}),\dots,A_{0}(v_{\ell-1})\big),\qquad\text{for all }\bfv=(v_{1},\dots,v_{\ell-1})\in \cV_{i}.
\]
For any vertex $(\bfv,S)\in \cV_{i}\times \binom{R_{i}}{r_{i}}$, we color it red if $\widetilde{A}_{i}(\bfv)\in S$, and otherwise color it blue. 

Now, for any hyperedge \(
\big\{(\bfv,S_{1}),(\bfv,S_{2}),(\bfv',S_{3})\big\}
\) included in $H$ (see Step 2), we have:
\begin{enumerate}[label=(\arabic*)]
\item The vertices $(\bfv,S_{1})$ and $(\bfv,S_{2})$ cannot be both red, because $S_{1}\cap S_{2}=\emptyset$.
\item Let $i,j\in [\ell]$ be the indices such that $\bfv\in \cV_{i}$ and $\bfv'\in\cV_{j}$. If $(\bfv,S_{1})$ and $(\bfv,S_{2})$ are both colored blue, then $\widetilde{A}_{i}(\bfv)\in R_{i}\setminus(S_{1}\cup S_{2})$ and hence \[\pi_{(\bfv,\bfv')}\big(\widetilde{A}_{i}(\bfv)\big)\in S_{3},\qquad\text{ due to }\qquad \pi_{(\bfv,\bfv')}(R_{i}\setminus(S_{1}\cup S_{2}))\subseteq S_{3}.\]
By the definitions of $\pi_{(\bfv,\bfv')}$ and $\widetilde{A}_{i},\widetilde{A}_{j}$, we have $\pi_{(\bfv,\bfv')}\big(\widetilde{A}_{i}(\bfv)\big)=\widetilde{A}_{j}(\bfv')$. Therefore, if $(\bfv,S_{1})$ and $(\bfv,S_{2})$ are both colored blue, then $(\bfv',S_{3})$ must be colored red.
\end{enumerate}
In conclusion, the hyperedge \(
\big\{(\bfv,S_{1}),(\bfv,S_{2}),(\bfv',S_{3})\big\}
\) cannot be monochromatic under the proposed coloring. This means that $H$ is 2-colorable as long as $\mathrm{val}(\Psi)=1$.

\paragraph{Step 4: finding the leading color.} In the rest of the proof, we fix a $c$-coloring of $H$, denoted by the map
\[
f:\textstyle\bigsqcup_{i=1}^{\ell}\big(\cV_{i}\times \textstyle\binom{R_{i}}{r_{i}}\big)\rightarrow [c].
\]
We suppose that at most $\varepsilon$ fraction of edges of $H$ are monochromatic under $f$, for some $\varepsilon\in (0,1)$ to be specified later.

For each $i\in [\ell]$ and each $\bfv\in \cV_{i}$, we apply~\cite[Lemma 2.2]{DinurRegevSmyth} to the vertices in the set $\left\{(\bfv,S)\,\middle|\,S\in \binom{R_{i}}{r_{i}}\right\}$ to find a color $c_{\bfv}\in [c]$ such that:
\begin{enumerate}[label=(\arabic*)]
\item The fraction of sets $S\in \binom{R_{i}}{r_{i}}$ such that $f(\bfv,S)=c_{\bfv}$ is at least $|R_{i}|^{-c-1}.$
\item There exist disjoint sets $S_{1},S_{2}\in \binom{R_{i}}{r_{i}}$ such that $f(\bfv,S_{1})=f(\bfv,S_{2})=c_{\bfv}$. The set $R_{i}\setminus(S_{1}\cup S_{2})$ will be denoted by $B(\bfv)$ and used later in the proof.
\end{enumerate}
For each $i\in [\ell]$, let $c^{(i)}\in[c]$ be a color that maximizes the number of tuples $\bfv\in\cV_i$ with $c_{\bfv}=c^{(i)}$. Since $\ell=2c^2$, by the pigeonhole principle there exists a color $c^*\in[c]$ such that $c^{(i)}=c^*$ for at least $2c$ distinct indices $i\in[\ell]$. Pick exactly $2c$ indices among them, and let $I$ be the set of these $2c$ indices.

\paragraph{Step 5: locating two critical layers.} Observe that the bipartite graph $(\cV_{i}\sqcup \cV_{i+1},\cE_{i,i+1})$ is bi-regular for all $i\in [\ell-1]$, due to the bi-regularity of $\Psi$. Consider the following random walk:
\begin{enumerate}
\item Start at a uniformly random element $\bfv^{(1)}\in \cV_{1}$.
\item For each $i\in [\ell-1]$, let $\bfv^{(i+1)}$ be a uniformly random element of $\cV_{i+1}$ conditioned on $(\bfv^{(i)},\bfv^{(i+1)})\in \cE_{i,i+1}$.
\end{enumerate}
Due to the bi-regularity of the bipartite graphs $(\cV_{i}\sqcup \cV_{i+1},\cE_{i,i+1})$, this random walk produces a random sequence $(\bfv^{(1)},\dots,\bfv^{(\ell)})\in \cV_{1}\times\cdots\times \cV_{\ell}$ that is marginally uniform over each $\cV_{i}$. Furthermore, by the definition of the edge sets $\cE_{i,j}$, it is easy to see that for any $i,j\in [\ell]$ such that $i<j$, the distribution of the pair $(\bfv^{(i)},\bfv^{(j)})$ is uniform over $\cE_{i,j}$.

For each $i\in [\ell]$, let $\Gamma_{i}$ be the event that $c_{\bfv^{(i)}}=c^{(i)}$. By the marginal uniformity and the choice of the $c^{(i)}$'s we have $\Pr{\Gamma_{i}}\geq c^{-1}$. By the inclusion-exclusion principle, we have
\begin{align*}
1&\geq \Pr{\bigcup\nolimits_{i\in I}\Gamma_{i}}\geq \sum_{i\in I}\Pr{\Gamma_{i}}-\sum_{i,j\in I,\;i<j}\Pr{\Gamma_{i}\cap \Gamma_{j}}\\
&\geq |I|\cdot c^{-1}-\binom{|I|}{2}\cdot \max_{i,j\in I,\;i<j} \Pr{\Gamma_{i}\cap \Gamma_{j}}.
\end{align*}
Therefore, there exist $i,j\in I$ with $i<j$ such that \[\Pr{\Gamma_{i}\cap\Gamma_{j}}\geq \binom{|I|}{2}^{-1}\cdot(c^{-1}|I|-1)=\binom{2c}{2}^{-1}\geq \frac{1}{2c^{2}}.\]
By the uniformity of $(\bfv^{(i)},\bfv^{(j)})$ over $\cE_{i,j}$, this means that
\[
\Pru{(\bfv,\bfv')\in \cE_{i,j}}{c_{\bfv}=c^{(i)}\text{ and }c_{\bfv'}=c^{(j)}}\geq \frac{1}{2c^{2}}.
\]
Since $i,j\in I$, by the definition of $I$, the above is equivalent to
\begin{equation}\label{eq:two_critical_layers}
\Pru{(\bfv,\bfv')\in \cE_{i,j}}{c_{\bfv}=c_{\bfv'}=c^*}\geq \frac{1}{2c^{2}}.
\end{equation}
In the remainder of the proof, we fix the indices $i$ and $j$ identified above (satisfying \Cref{eq:two_critical_layers}).

\paragraph{Step 6: discarding monochromatic edges. }We let $\cE^*$ be the collection of all pairs $(\bfv,\bfv')\in \cE_{i,j}$ such that:
\begin{enumerate}[label=(\arabic*)]
\item $c_{\bfv}=c_{\bfv'}=c^*$; and
\item no hyperedge of $H$ of the form \(
\big\{(\bfv,S_{1}),(\bfv,S_{2}),(\bfv',S_{3})\big\}
\) is colored monochromatically by $f$.
\end{enumerate}
By \Cref{eq:edges_in_multi_layer}, the total number of monochromatic edges of $H$ under the coloring $f$ is at most
\[
\varepsilon\cdot 5^{u(\ell-1)}8^{7^{u(\ell-1)}}\ell^{2}\cdot |V_{1}|^{\ell-1}\leq \varepsilon\cdot 5^{u(\ell-1)}8^{7^{u(\ell-1)}}\ell^{2}\cdot \left|\cE_{i,j}\right|\leq \frac{1}{4c^{2}}\left|\cE_{i,j}\right|,
\]
where we chose a sufficiently small $\varepsilon=\exp(-\exp(O(c^{3})))$ that justifies the last transition. It then follows from \Cref{eq:two_critical_layers} and the definition of $\cE^*$ that
\begin{equation}\label{eq:multi_layer_E_star}
|\cE^{*}|\geq \frac{1}{2c^{2}}\left|\cE_{i,j}\right|- \frac{1}{4c^{2}}\left|\cE_{i,j}\right|= \frac{1}{4c^{2}}\left|\cE_{i,j}\right|.
\end{equation}

\paragraph{Step 7: strategy for the multi-layered label cover.} The goal of this step is to decode from $f$ a ``good'' assignment map for the two critical layers (identified in Step 5) of the $\ell$-layered label cover instance (constructed in Step 1). 

For each $\bfv'\in \cV_{j}$, let $N^*(\bfv')$ be the collection of tuples $\bfv\in \cV_{i}$ such that $(\bfv,\bfv')\in \cE^*$. 

In this paragraph, we define a deterministic map $\widetilde{A}_{j}:\cV_{j}\rightarrow R_{j}$ using the coloring $f$. Fix a tuple $\bfv'\in \cV_{j}$. If $c_{\bfv'}\neq c^*$, we assign $\widetilde{A}_{j}(\bfv')$ arbitrarily. Otherwise, for each $\bfv\in N^*(\bfv')$, we denote
\[
B(\bfv,\bfv'):= \pi_{(\bfv,\bfv')}(B(\bfv)).
\]
Recall from Step 4 that, the set $R_{i}\setminus B(\bfv)\subseteq R_{i}$ can be decomposed as the disjoint union of two sets $S_{1},S_{2}\in\binom{R_{i}}{r_{i}}$ such that $f(\bfv,S_{1})=f(\bfv,S_{2})=c^*$. For each $S\in \binom{R_{j}}{r_{j}}$ such that $B(\bfv,\bfv')\subseteq S$, by the construction in Step 2, the hyperedge
\[
\big\{(\bfv,S_{1}),(\bfv,S_{2}),(\bfv',S)\big\}
\]
is included in $H$. By the definition of $\cE^*$, this hyperedge cannot be monochromatic under the coloring $f$, and therefore $f(\bfv',S)\neq c^*$. However, recall also from Step 4 that the fraction of $S\in \binom{R_{j}}{r_{j}}$ such that $f(\bfv',S)=c^*$ is at least $|R_{j}|^{-c-1}$. Therefore, at least an $|R_{j}|^{-c-1}$ fraction of $S\in \binom{R_{j}}{r_{j}}$ contain none of the sets $B(\bfv,\bfv')$, where $\bfv$ ranges in $N^*(\bfv')$. We can then apply~\Cref{lem:base_game_combinatorics} to find\footnote{Note that for any $\bfv\in N^*(\bfv')$, we have $|B(\bfv,\bfv')|=\left|\pi_{(\bfv,\bfv')}(B(\bfv))\right|\leq |B(\bfv)|\leq |R_{i}|-2r_{i}\leq c$.} an element of $R_{j}$ that is contained in $B(\bfv,\bfv')$ for at least 
\[
\eta=\frac{6^{-c}}{c(c+1)\ln |R_{j}|}
\]
fraction of the tuples $\bfv\in N^*(\bfv')$. Define $\widetilde{A}_{j}(\bfv')$ to be such an element of $R_{j}$.

We now define a random assignment map $\widetilde{\cA}_{i}:\cV_{i}\rightarrow R_{i}$ by picking each $\widetilde{\cA}_{i}(\bfv)$ independently and uniformly at random from $B(\bfv)$. By the definition of $B(\bfv)$ (in Step 4), we have $|B(\bfv)|\leq |R_{i}|-2r_{i}\leq c$. Therefore, we have
\begin{align*}
&\quad\Exs{\widetilde{\cA}_{i}}{\Pru{(\bfv,\bfv')\in \cE^*}{\pi_{(\bfv,\bfv')}\big(\widetilde{\cA}_{i}(\bfv)\big)=\widetilde{A}_{j}(\bfv')}}\\
&=
\Exu{(\bfv,\bfv')\in \cE^*}{\Prs{\widetilde{\cA}_{i}}{\pi_{(\bfv,\bfv')}\big(\widetilde{\cA}_{i}(\bfv)\big)=\widetilde{A}_{j}(\bfv')}}\\
&\geq \eta\cdot \Exu{(\bfv,\bfv')\in \cE^*}{ \Prs{\widetilde{\cA}_{i}}{\pi_{(\bfv,\bfv')}\big(\widetilde{\cA}_{i}(\bfv)\big)=\widetilde{A}_{j}(\bfv')}\,\middle|\, \widetilde{A}_{j}(\bfv')\in B(\bfv,\bfv')}\tag{by the choice of $\widetilde{A}_{j}(\bfv')$}\\
&\geq \eta \cdot \Exu{(\bfv,\bfv')\in \cE^*}{\frac{1}{|B(\bfv)|}}\geq \frac{\eta}{c}\tag{by the definition of $B(\bfv,\bfv')$}.
\end{align*}
Therefore, there exists a deterministic instantiation $\widetilde{A}_{i}:\cV_{i}\rightarrow R_{i}$ of $\widetilde{\cA}_{i}$ such that
\[
\Pru{(\bfv,\bfv')\in \cE^*}{\pi_{(\bfv,\bfv')}\big(\widetilde{A}_{i}(\bfv)\big)=\widetilde{A}_{j}(\bfv')}\geq \frac{\eta}{c}.
\]
By~\Cref{eq:multi_layer_E_star}, it follows that
\begin{equation}\label{eq:multi_layer_strategy}
\Pru{(\bfv,\bfv')\in \cE_{i,j}}{\pi_{(\bfv,\bfv')}\big(\widetilde{A}_{i}(\bfv)\big)=\widetilde{A}_{j}(\bfv')}\geq \frac{\eta}{4c^{3}}.
\end{equation}
In the rest of the proof, we fix $\widetilde{A}_{i}$ to be an instantiation of $\widetilde{\cA}_{i}$ that satisfies~\Cref{eq:multi_layer_strategy}.

\paragraph{Step 8: strategy for the original label cover.} 

We define a random pair of (correlated) embedding maps $\varphi_{0}:V_{0}\hookrightarrow\cV_{i}$ and $\varphi_{1}:V_{1}\hookrightarrow\cV_{j}$ as follows:
\begin{enumerate}
\item For each $s\in \{1,2,\dots,i-1\}$, let $v_{s}=v_{s}'$ be a uniformly random element of $V_{1}$.
\item For each $s\in \{i+1,i+2,\dots,j-1\}$, pick a pair $(v_{s},v_{s}')\in E$ uniformly at random.
\item For each $s\in \{j,j+1,\dots,\ell-1\}$, let $v_{s}=v_{s}'$ be a uniformly random element of $V_{0}$.
\item For each $w\in V_{0}$, let $\varphi_{0}(w)=(v_{1},\dots,v_{i-1},w,v_{i+1},\dots,v_{\ell-1})\in \cV_{i}$.
\item For each $w'\in V_{1}$, let $\varphi_{1}(w')=(v_{1}',\dots,v_{i-1}',w',v_{i+1}',\dots,v_{\ell-1}')\in \cV_{j}$.
\end{enumerate}
If we first sample the random maps $\varphi_{0},\varphi_{1}$ and then independently sample a pair $(w,w')\in E$, it is clear that $(\varphi_{0}(w),\varphi_{1}(w'))$ is uniformly distributed over $\cE_{i,j}$. Therefore, it follows from~\Cref{eq:multi_layer_strategy} that
\begin{equation}\label{eq:multi_layer_strategy_rewrite}
\Exu{(\varphi_{0},\varphi_{1})}{\Pru{(w,w')\in E}{\pi_{(\varphi_{0}(w),\varphi_{1}(w'))}\left(\widetilde{A}_{i}\big(\varphi_{0}(w)\big)\right)=\widetilde{A}_{j}\big(\varphi_{1}(w')\big)}}\geq \frac{\eta}{4c^3}.
\end{equation}

We now define a random pair of (correlated) assignment maps $\cA_{0}:V_{0}\rightarrow \Sigma_{0}$ and $\cA_{1}:V_{1}\rightarrow \Sigma_{1}$ as follows:
\begin{enumerate}
\item Sample the pair of embedding maps $\varphi_{0},\varphi_{1}$ as described earlier.
\item For each $w\in V_{0}$, let $\cA_{0}(w)$ be the $i$-th coordinate of $\widetilde{A}_{i}(\varphi_{0}(w))$.
\item For each $w'\in V_{1}$, let $\cA_{1}(w')$ be the $i$-th coordinate of $\widetilde{A}_{j}(\varphi_{1}(w'))$.
\end{enumerate}
By the definition of $\pi_{(\varphi_{0}(w),\varphi_{1}(w'))}$ in Step 1, it follows from~\Cref{eq:multi_layer_strategy_rewrite} that
\[
\Exu{(\cA_{0},\cA_{1})}{\Pru{(w,w')\in E}{\pi_{(w,w')}\big(\cA_{0}(w)\big)=\cA_{1}(w')\Big.}}\geq \frac{\eta}{4c^3}.
\]
Therefore, there exists a deterministic instantiation $(A_{0},A_{1})$ of the random pair of maps $(\cA_{0},\cA_{1})$ that satisfies at least $\eta/(4c^{3})$ fraction of the constraints in $\Psi$. Thus we have
\[
\mathrm{val}(\Psi)\geq \frac{\eta}{4c^{3}}\geq \frac{6^{-c}}{4c^{4}(c+1)\ln |R_{j}|}\geq \frac{6^{-c}}{4c^{4}(c+1)(\ell-1)\ln(7^{u})}\geq \frac{1}{2^{20c}\cdot u}> 2^{-\gamma u},
\]
where we used $u\geq 40\gamma^{-2}c$ in the last transition. 
\end{proof}

The following combinatorial lemma (adapted from~\cite[Claim 4.5]{DinurRegevSmyth}) is used in Step 7 of the preceding proof.

\begin{lemma}\label{lem:base_game_combinatorics}
Fix a positive integer $c$. Let $R$ be a nonempty finite set with $|R|\geq 6c$, and let $r$ be an integer such that $|R|/3\leq r\leq 2|R|/3$. Suppose $B_{1},\dots,B_{m}\subseteq R$ are a sequence of subsets such that $|B_{i}|\leq c$ for each $i\in [m]$. If at least an $|R|^{-c-1}$ fraction of $r$-element subsets $S\in \binom{R}{r}$ contain none of the sets $B_{1},\dots,B_{m}$, then there exists an element $a\in R$ contained in at least
\[\frac{6^{-c}}{c(c+1)\ln |R|}\]
fraction of the sets $B_{1},\dots,B_{m}$.
\end{lemma}
\begin{proof}
Let
\[
\eta:=\max_{a\in R}\frac{\big|\{i\in[m]\mid a\in B_i\}\big|}{m}.
\]
Without loss of generality, suppose $B_{1},\dots,B_{q}$ are a maximum pairwise disjoint subfamily among the sets $B_1,\dots,B_m$, for some integer $q\in [m]$. By maximality, every $B_i$ contains some element of
\[
U:=B_1\cup\cdots\cup B_q.
\]
Therefore
\[
m
\leq
\sum_{a\in U} \big|\{i\in[m]\mid a\in B_i\}\big|\leq |U|\cdot \eta m\leq cq\cdot \eta m,
\]
and hence $\eta\geq 1/(cq)$.

Next we upper-bound $q$ using the assumptions in the lemma. Let
$S$ be chosen uniformly at random from $\binom{R}{r}$. For each $j\in [q]$, we have
\[
\Pru{S\in \binom{R}{r}}{B_j\subseteq S}
=\binom{|R|-|B_{j}|}{r-|B_{j}|}/\binom{|R|}{r}
=
\prod_{t=0}^{|B_{j}|-1}\frac{r-t}{|R|-t}\geq \left(\frac{1}{6}\right)^{|B_{j}|}\geq 6^{-c},
\]
where we used the assumptions $r\geq |R|/3$ and $|B_{j}|\leq c\leq |R|/6$ in the third transition.

Because the sets $B_1,\dots,B_q$ are pairwise disjoint, the events \(\{B_j\subseteq S\}\), for $j\in [q]$,
are negatively correlated under the uniform distribution over $\binom{R}{r}$. Hence
\[
\Pru{S\in \binom{R}{r}}{B_j\nsubseteq S \text{ for all } j\in[q]\big.}
\leq
\prod_{j=1}^q \Pru{S\in \binom{R}{r}}{B_j\nsubseteq S}
\leq
(1-6^{-c})^q.
\]
Therefore, by the assumption in the lemma, we have
\[
|R|^{-c-1}
\leq
\Pru{S\in \binom{R}{r}}{B_i\nsubseteq S \text{ for all } i\in[m]\big.}
\leq
\Pru{S\in \binom{R}{r}}{B_j\nsubseteq S \text{ for all } j\in[q]\big.}
\leq
(1-6^{-c})^{q}.
\]
Taking logarithms gives
\[
q\leq \frac{(c+1)\ln |R|}{-\ln(1-6^{-c})}\leq 6^{c}(c+1)\ln |R|.
\]
Combining this with $\eta \geq 1/(cq)$, we obtain the desired inequality
\[
\eta\geq \frac{6^{-c}}{c(c+1)\ln|R|}.\qedhere
\]
\end{proof}

\section{Grassmann Decoding Theorems}\label{sec:grassmann_dec}

In this appendix, we explain how the Grassmann decoding theorems (\Cref{thm:grassmann_linear_dec,thm:grdecoding2}) used in this paper can be derived from the slightly different versions in the existing literature.

\subsection{The Expansion Theorem}

To prove Grassmann decoding theorems, the standard approach is to first establish some small-set expansion properties of the Grassmann graph. It is observed in~\cite{KMS1,DKKMS18} that small vertex sets may not expand very well in the Grassmann graph, and the correct approach is to consider sets that remain small even after zoom-ins and zoom-outs. This is formalized in the next definition.

\begin{definition}\label{def:Grassmann_pseudorandom}
Let $n,\ell,q$ be nonnegative integers such that $n-q\geq \ell\geq q+1$, and let $\varepsilon\in (0,1)$. Suppose $\cL\subseteq \mathsf{Gr}(\bF_{2}^{n},\ell)$ is a set of subspaces. We say that $\cL$ is \emph{$(q,\varepsilon)$-pseudorandom} if for any subspaces $Q\in \mathsf{Gr}(\bF_{2}^{n},q)$ and $W\in \mathsf{Gr}(\bF_{2}^{n},n-q)$ such that $Q\subseteq W$, we have
\[
\Pru{L\in \mathsf{Gr}(\bF_{2}^{n},\ell)}{L\in \cL\,\big|\,Q\subseteq L\subseteq W}\leq \varepsilon.
\]
\end{definition}

To facilitate Fourier analysis, it is useful to ``lift'' the Grassmann graph to a Cayley graph over an Abelian group. The ``short-code graph'' is a Cayley graph that serves this purpose very well, with the matrix space $\bF_{2}^{n\times \ell}$ as the additive group and the rank-1 matrices as the generating set. The counterpart of~\Cref{def:Grassmann_pseudorandom} in the short-code graph is given as follows.

\begin{definition}\label{def:short_code_pseudorandom}
Let $n,\ell,q$ be nonnegative integers such that $n-q\geq \ell\geq q+1$, and let $\varepsilon\in(0,1)$. A function $f:\bF_{2}^{n\times \ell}\rightarrow \{0,1\}$ is said to be \emph{$(q,\varepsilon)$-pseudorandom} if for any matrices 
\[Y_{1}\in \bF_{2}^{\ell \times q},\qquad Z_{1}\in \bF_{2}^{n\times q},\qquad Y_{2}\in \bF_{2}^{q\times n}\qquad\text{and}\qquad Z_{2}\in \bF_{2}^{q\times \ell},\] 
such that $\left\{X\in \bF_{2}^{n\times \ell}\,\middle|\, XY_{1}=Z_{1}\text{ and }Y_{2}X=Z_{2}\right\}$ is nonempty, we have
\[
\Exu{X\in \bF_{2}^{n\times \ell}}{f(X)\,\big|\, XY_{1}=Z_{1}\text{ and }Y_{2}X=Z_{2}}\leq \varepsilon.
\]
\end{definition}

It is not hard to relate~\Cref{def:Grassmann_pseudorandom,def:short_code_pseudorandom} together; we use the following lemma from~\cite{minzer2024near}.

\begin{lemma}[{\cite[Lemma A.18]{minzer2024near}}]\label{lem:Grassmann_to_short_code_pseudorandom}
Let $n,\ell,q$ be nonnegative integers such that $n-q\geq \ell\geq q+1$, and let $\varepsilon\in (0,1)$. Suppose $\cL\subseteq \mathsf{Gr}(\bF_{2}^{n},\ell)$ is $(q,\varepsilon)$-pseudorandom (in the sense of~\Cref{def:Grassmann_pseudorandom}). Define a function $f:\bF_{2}^{n\times\ell}\rightarrow\{0,1\}$ by
\[
f(X)=\begin{cases}
1,&\text{if }\operatorname{rank}(X)=\ell\text{ and }\operatorname{col}(X)\in \cL,\\
0,&\text{otherwise}.
\end{cases}
\]
Then $f$ is $(q,4\varepsilon)$-pseudorandom (in the sense of~\Cref{def:short_code_pseudorandom}).
\end{lemma}

The adjacency matrix of the short-code graph is represented by the following linear operator.

\begin{definition}
Given a function $f:\bF_{2}^{n\times\ell}\rightarrow \bR$, we let $\cT f:\bF_{2}^{n\times\ell}\rightarrow \bR$ be the function defined by
\begin{equation}\label{eq:Grassmann_to_short_code}
\cT f(X)=\Exu{Y_{1}\in \bF_{2}^{n\times 1},\,Y_{2}\in \bF_{2}^{1\times \ell}}{f(X+Y_{1}Y_{2})\big.}.
\end{equation}
\end{definition}

We are now ready to state the main expansion theorem from~\cite{KMS23,ellis2023analogue}, which shows that pseudorandom sets have good expansion in the short-code graph.

\begin{lemma}[{\cite[Theorem 69]{ellis2023analogue}}]\label{lem:EKL_theorem}
There is an absolute constant $C>0$ such that the following holds. Let $n,\ell,q$ be nonnegative integers such that $n-q\geq \ell\geq q+1$. Then any $(q,2^{-Cq^{2}})$-pseudorandom function $f:\bF_{2}^{n\times \ell}\rightarrow \{0,1\}$ satisfies\footnote{Here the inner product $\inp{f}{\cT f}$ should be interpreted as $\Exs{X\in \bF_{2}^{n\times \ell}}{f(X)\cT f(X)}$.}
\[
\inp{f}{\cT f}\leq 2^{-q+1}\Ex{f}.
\]
\end{lemma}

We then transfer the expansion theorem from the short-code graph to the Grassmann graph.

\begin{lemma}\label{lem:Grassmann_expansion}
There is an absolute constant $C>0$ such that the following holds. Let $n,\ell,q$ be nonnegative integers such that $n-q-3\geq \ell\geq q+3$. Then any $(q,2^{-Cq^{2}})$-pseudorandom subset $\cL\subseteq \mathsf{Gr}(\bF_{2}^{n},\ell)$ satisfies
\[
\Pru{(L_{1},L_{2})\sim \cG(\bF_{2}^{n},\ell)}{L_{1},L_{2}\in \cL\big.}\leq 2^{-q+2}\cdot \Pru{L\in \mathsf{Gr}(\bF_{2}^{n},\ell)}{L\in \cL\big.}.
\]
\end{lemma}

\begin{proof}
If $q=0$ the conclusion obviously holds. We assume $q\geq 1$, and let $C$ be the constant from~\Cref{lem:EKL_theorem}. Suppose $\cL\subseteq \mathsf{Gr}(\bF_{2}^{n},\ell)$ is $(q,2^{-(C+2)q^{2}})$-pseudorandom. Define a function $f:\bF_{2}^{n\times\ell}\rightarrow\{0,1\}$ from $\cL$ as in \Cref{lem:Grassmann_to_short_code_pseudorandom}. By~\Cref{lem:Grassmann_to_short_code_pseudorandom} we know that $f$ is $(q,2^{-Cq^{2}})$-pseudorandom (in the sense of~\Cref{def:short_code_pseudorandom}).

If $X\in \bF_{2}^{n\times\ell}$, $Y_{1}\in \bF_{2}^{n\times 1}$ and $Y_{2}\in \bF_{2}^{1\times \ell}$ are independent and uniformly random, then conditioned on the event
\begin{equation}\label{eq:Grassmann_edge_event}
\{\operatorname{rank}(X)=\ell\}\cap\{Y_{1}\notin\operatorname{col}(X)\}\cap\{Y_{2}\neq 0\},
\end{equation}
the pair $\bigl(\operatorname{col}(X),\operatorname{col}(X+Y_{1}Y_{2})\bigr)$ is a uniformly random edge of the Grassmann graph $\mathsf{Gr}(\bF_{2}^{n},\ell)$. Note that the probability of the event \Cref{eq:Grassmann_edge_event} is at least \(1-2^{\ell-n}-2^{\ell-n}-2^{-\ell}\geq 1/2.\)
Therefore, we have
\begin{equation}\label{eq:Grassmann_short_code_factor}
\Pru{X\in \bF_{2}^{n\times \ell},\,Y_{1}\in \bF_{2}^{n\times 1},\,Y_{2}\in\bF_{2}^{1\times\ell}}{f(X)=f(X+Y_{1}Y_{2})=1
\big.}\geq \frac{1}{2}\cdot \Pru{(L_{1},L_{2})\sim \cG(\bF_{2}^{n},\ell)}{L_{1},L_{2}\in \cL\big.}
\end{equation}
Since the left-hand side of \Cref{eq:Grassmann_short_code_factor} is exactly $\inp{f}{\cT f}$, using~\Cref{lem:EKL_theorem} we can conclude that
\[
\Pru{(L_{1},L_{2})\sim \cG(\bF_{2}^{n},\ell)}{L_{1},L_{2}\in \cL\big.}\leq 2\inp{f}{\cT f}\leq 2^{-q+2}\cdot \Ex{f}\leq 2^{-q+2}\cdot\Pru{L\in \mathsf{Gr}(\mathbb{F}_2^n,\ell)}{L\in \cL}.\qedhere
\]
\end{proof}

\subsection{From Expansion to Decoding}
We now use an argument, originally due to~\cite{BKS19}, to deduce the decoding statement from the expansion statement.

\begin{theorem}\label{thm:Grassmann_decoding}
There is an absolute constant $C>0$ such that the following holds. Let $n,\ell,q$ be nonnegative integers such that $n-q-3\geq \ell\geq q+3$. Suppose $F[\cdot]$ is a table where $F[L]\in L^*\cup\{\nil\}$ for each $L\in \mathsf{Gr}(\bF_{2}^{n},\ell)$, such that 
\begin{equation}\label{eq:Grassmann_decoding_assumption}
\Pru{(L_{1},L_{2})\sim \cG(\bF_{2}^{n},\ell)}{\left.F[L_{1}]\right|_{L_{1}\cap L_{2}}=\left.F[L_{2}]\right|_{L_{1}\cap L_{2}}\Big.}\geq 2^{-q+4}.
\end{equation}
Then there exist a linear map $f:\bF_{2}^n\rightarrow \bF_{2}$ and subspaces $Q\in \mathsf{Gr}(\bF_{2}^{n},q)$, $W\in \mathsf{Gr}(\bF_{2}^{n},n-q)$ such that $Q\subseteq W$ and
\[
\Pru{L\in \mathsf{Gr}(\bF_{2}^{n},\ell)}{F[L]=\left.f\right|_{L}\,\Big|\, Q\subseteq L\subseteq W}\geq 2^{-Cq^{2}}.
\]
\end{theorem}

\begin{proof}
For each linear functional $f:\bF_{2}^{n}\rightarrow \bF_{2}$, let $\cL_{f}$ denote the set of subspaces $L\in\mathsf{Gr}(\bF_{2}^{n},\ell)$ such that $F[L]=\left.f\right|_{L}$. We have
\begin{equation}\label{eq:BKS_argument_vertices}
\Exu{f\in (\bF_{2}^n)^*}{\Pru{L\in \mathsf{Gr}(\bF_{2}^{n},\ell)}{L\in \cL_{f}\big.}}=\Exu{L\in \mathsf{Gr}(\bF_{2}^{n},\ell)}{\Pru{f\in (\bF_{2}^n)^*}{\left.f\right|_{L}=F[L]\Big.}}\leq 2^{-\ell}.
\end{equation}
Furthermore,
\begin{equation}\label{eq:BKS_argument_edges}
\Exu{f\in (\bF_{2}^n)^*}{\Pru{(L_{1},L_{2})\sim \cG(\bF_{2}^{n},\ell)}{L_{1},L_{2}\in\cL_{f}\big.}}=\Exu{(L_{1},L_{2})\sim \cG(\bF_{2}^{n},\ell)}{\Pru{f\in (\bF_{2}^n)^*}{\left.f\right|_{L_{1}}=F[L_{1}]\text{ and }\left.f\right|_{L_{2}}=F[L_{2}]\Big.}}.
\end{equation}
Note that for any edge $(L_{1},L_{2})$ of the Grassmann graph $\mathsf{Gr}(\bF_{2}^{n},\ell)$ such that \[\left.F[L_{1}]\right|_{L_{1}\cap L_{2}}=\left.F[L_{2}]\right|_{L_{1}\cap L_{2}},\] we have
\[
\Pru{f\in (\bF_{2}^n)^*}{\left.f\right|_{L_{1}}=F[L_{1}]\text{ and }\left.f\right|_{L_{2}}=F[L_{2}]\Big.}=2^{-\ell-1}.
\]
Therefore, combining \Cref{eq:Grassmann_decoding_assumption,eq:BKS_argument_edges} yields
\begin{equation}\label{eq:BKS_argument_edges_2}
\Exu{f\in (\bF_{2}^n)^*}{\Pru{(L_{1},L_{2})\sim \cG(\bF_{2}^{n},\ell)}{L_{1},L_{2}\in\cL_{f}\big.}}\geq 2^{-\ell-1}\cdot 2^{-q+4}=2^{-q-\ell+3}.
\end{equation}
Combining \Cref{eq:BKS_argument_vertices,eq:BKS_argument_edges_2}, we know that there exists a linear functional $f\in (\bF_{2}^{n})^*$ such that
\[
\Pru{(L_{1},L_{2})\sim \cG(\bF_{2}^{n},\ell)}{L_{1},L_{2}\in\cL_{f}\big.}>2^{-q+2}\cdot \Pru{L\in \mathsf{Gr}(\bF_{2}^{n},\ell)}{L\in \cL_{f}\big.}.
\]
It then follows from~\Cref{lem:Grassmann_expansion} that $\cL_{f}$ is not $(q,2^{-Cq^{2}})$-pseudorandom, which is exactly the desired conclusion.
\end{proof}

We now restate and prove the precise decoding result stated in~\Cref{subsubsec:Grassmann_agreement_test}. It is almost identical to~\Cref{thm:Grassmann_decoding}, except that we require the feasible zoom-ins $Q$ to be abundant.

\grdecodingA*
\begin{proof}
We take $C$ to be the same constant as in~\Cref{thm:Grassmann_decoding}. Let $\cQ$ be the collection of subspaces $Q\in \mathsf{Gr}(\bF_{2}^{n},q)$ that satisfy the desired property, and we will prove that \[|\cQ|\geq 2^{-q-\ell^{2}}|\mathsf{Gr}(\bF_{2}^n,q)|.\]
In the rest of the proof, we assume on the contrary that $|\cQ|< 2^{-q-\ell^{2}}|\mathsf{Gr}(\bF_{2}^n,q)|$.

We define a table $F'[\cdot]$, where $F'[L]\in L^*\cup\{\nil\}$ for each $L\in \mathsf{Gr}(\bF_{2}^{n},\ell)$, as follows:
\[
F'[L]=\begin{cases}
\nil, &\text{if }Q\subseteq L\text{ for some }Q\in \cQ,\\
F[L],&\text{otherwise}.\\
\end{cases}
\]
We thus have
\begin{align*}
\Pru{L\in \mathsf{Gr}(\bF_{2}^{n},\ell)}{F[L]\neq F'[L]}&\leq \sum_{Q\in \cQ}\Pru{L\in \mathsf{Gr}(\bF_{2}^{n},\ell)}{Q\subseteq L}=|\cQ|\cdot \qbinom{n-q}{\ell-q}\qbinom{n}{\ell}^{-1}\\
&= |\cQ|\cdot \qbinom{n}{q}^{-1}\qbinom{\ell}{q}\leq 2^{-q-\ell^{2}}\cdot 2^{\ell^{2}}=2^{-q}.
\end{align*}
Therefore, it follows from the assumption \Cref{eq:Grassmann_decoding_2_assumption} that
\[
\Pru{(L_{1},L_{2})\sim \cG(\bF_{2}^{n},\ell)}{\left.F'[L_{1}]\right|_{L_{1}\cap L_{2}}=\left.F'[L_{2}]\right|_{L_{1}\cap L_{2}}\Big.}\geq 2^{-q+5}-2\cdot2^{-q}\geq 2^{-q+4}.
\]
Applying~\Cref{thm:Grassmann_decoding} to the table $F'$ yields subspaces $Q\in \mathsf{Gr}(\bF_{2}^{n},q)$, $W\in \mathsf{Gr}(\bF_{2}^{n},n-q)$ such that $Q\subseteq W$ and
\begin{equation}\label{eq:Grassmann_decoding_2_contradiction}
\Pru{L\in \mathsf{Gr}(\bF_{2}^{n},\ell)}{F'[L]=\left.f\right|_{L}\,\Big|\, Q\subseteq L\subseteq W}\geq 2^{-Cq^{2}}.
\end{equation}
By the definition of $\cQ$, this implies that $Q\in \cQ$. However, as $F'[L]=\nil$ for any $L$ containing a member of $\cQ$, the left-hand side of~\Cref{eq:Grassmann_decoding_2_contradiction} must be 0, which is a contradiction.
\end{proof}

We now move on to the second decoding theorem we need. The following lemma says that if a function $F$ on the vertices of the Grassmann graph is locally constant on its edges, then it is globally constant inside some zoom-in/zoom-out combination.
\begin{lemma}\label{lem:Grassmann_decoding_baby}
There is an absolute constant $C>0$ such that the following holds. Let $n,\ell,q$ be nonnegative integers such that $n-q-3\geq \ell\geq q+3$. Suppose $\Sigma$ is a finite set and $F[\cdot]$ is a table where $F[L]\in \Sigma\cup\{\nil\}$ for each $L\in \mathsf{Gr}(\bF_{2}^{n},\ell)$, such that 
\begin{equation}
\Pru{(L_{1},L_{2})\sim \cG(\bF_{2}^{n},\ell)}{F[L_{1}]=F[L_{2}]\neq \nil\Big.}\geq 2^{-q+3}.
\end{equation}
Then there exist an element $\sigma\in \Sigma$ and subspaces $Q\in \mathsf{Gr}(\bF_{2}^{n},q)$, $W\in \mathsf{Gr}(\bF_{2}^{n},n-q)$ such that $Q\subseteq W$ and
\[
\Pru{L\in \mathsf{Gr}(\bF_{2}^{n},\ell)}{F[L]=\sigma\,\Big|\, Q\subseteq L\subseteq W}\geq 2^{-Cq^{2}}.
\]
\end{lemma}

\begin{proof}
For each element $\sigma\in \Sigma$, let $\cL_{\sigma}$ denote the set of subspaces $L\in \mathsf{Gr}(\bF_{2}^{n},\ell)$ such that $F[L]=\sigma$. We have
\begin{equation}\label{eq:BKS_argument_vertices_baby}
\Exu{\sigma\in \Sigma}{\Pru{L\in \mathsf{Gr}(\bF_{2}^{n},\ell)}{L\in \cL_{\sigma}\big.}}=\Exu{L\in \mathsf{Gr}(\bF_{2}^{n},\ell)}{\Pru{\sigma\in \Sigma}{\sigma=F[L]\big.}}\leq \frac{1}{|\Sigma|}.
\end{equation}
Furthermore, due to the assumption of the lemma, we have
\begin{align}
\Exu{\sigma\in \Sigma}{\Pru{(L_{1},L_{2})\sim \cG(\bF_{2}^{n},\ell)}{L_{1},L_{2}\in\cL_{\sigma}\big.}}&=\Exu{(L_{1},L_{2})\sim\cG(\bF_{2}^{n},\ell)}{\Pru{\sigma\in \Sigma}{\sigma=F[L_{1}]=F[L_{2}]\big.}}\nonumber\\
&=\frac{1}{|\Sigma|}\cdot\Pru{(L_{1},L_{2})\sim\cG(\bF_{2}^{n},\ell)}{F[L_{1}]=F[L_{2}]\neq \nil\Big.}\geq \frac{2^{-q+3}}{|\Sigma|}.\label{eq:BKS_argument_edges_baby}
\end{align}
Therefore, combining~\Cref{eq:BKS_argument_vertices_baby,eq:BKS_argument_edges_baby}, we know that there exists an element $\sigma\in \Sigma$ such that
\[
\Pru{(L_{1},L_{2})\sim \cG(\bF_{2}^{n},\ell)}{L_{1},L_{2}\in\cL_{\sigma}\big.}>2^{-q+2}\cdot \Pru{L\in \mathsf{Gr}(\bF_{2}^{n},\ell)}{L\in \cL_{\sigma}\big.}.
\]
It then follows from~\Cref{lem:Grassmann_expansion} that $\cL_{\sigma}$ is not $(q,2^{-Cq^{2}})$-pseudorandom, which is exactly the desired conclusion.
\end{proof}

Our second decoding theorem (restated below) is a similar variant of~\Cref{lem:Grassmann_decoding_baby} with ``abundant zoom-ins.''

\grdecodingB*
\begin{proof}
One can use \Cref{lem:Grassmann_decoding_baby} to prove~\Cref{thm:grdecoding2} in an entirely analogous way to how~\Cref{thm:Grassmann_decoding} is used to prove~\Cref{thm:grassmann_linear_dec}.
\end{proof}

\section{Linearity Testing and List Decoding Bounds}\label{sec:list_decoding_bound}

In this appendix, we briefly explain how to prove~\Cref{lem:dkkms_list_bound}. The core of the proof is an $\ell$-vs.-$b$-dimensional Grassmann agreement test, whose soundness is given by the following lemma from~\cite[Theorem B.1]{KMS1} (see also~\cite[Lemma 3.1]{DKKMS18}). 

\begin{lemma}[\cite{KMS1,DKKMS18}]\label{lem:list_decoding_linearity_test}
There is an absolute constant $C>0$ such that the following holds. Let $n,\ell,b$ be positive integers such that $n\geq 2\ell$ and $40\leq b\leq \ell/4$. Suppose $F[\cdot]$ is a table where $F[L]\in L^*$ for each $L\in \mathsf{Gr}(\bF_{2}^{n},\ell)$, such that
\[
\Pru{(L_{1},L_{2})\sim \cG(\bF_{2}^{n},\ell,b)}{\left.F[L_{1}]\right|_{L_{1}\cap L_{2}}=\left.F[L_{2}]\right|_{L_{1}\cap L_{2}}\Big.}\geq \varepsilon,
\]
for some parameter $\varepsilon\in [2^{4-b/4},1]$. Then there exists a linear map $f:\bF_{2}^{n}\rightarrow \bF_{2}$ such that
\[
\Pru{L\in \mathsf{Gr}(\bF_{2}^{n},\ell)}{F[L]=\left.f\right|_{L}\Big.}\geq C^{-1}\varepsilon^{3}.
\]
\end{lemma}

The following lemma is a slight variant where some entries are allowed to be $\nil$.
\begin{lemma}\label{lem:list_decoding_linearity_test_nil}
There is an absolute constant $C\geq 40$ such that the following holds. Let $n,\ell,b$ be positive integers such that $n\geq 2\ell$ and $C\leq b\leq \ell/4$. Suppose $F[\cdot]$ is a table where $F[L]\in L^*\cup\{\nil\}$ for each $L\in \mathsf{Gr}(\bF_{2}^{n},\ell)$, such that
\begin{equation}\label{eq:list_decoding_linearity_test}
\Pru{(L_{1},L_{2})\sim \cG(\bF_{2}^{n},\ell,b)}{\left.F[L_{1}]\right|_{L_{1}\cap L_{2}}=\left.F[L_{2}]\right|_{L_{1}\cap L_{2}}\Big.}\geq \varepsilon,
\end{equation}
for some parameter $\varepsilon\in [2^{4-b/4},1]$. Then there exists a linear map $f:\bF_{2}^{n}\rightarrow \bF_{2}$ such that
\[
\Pru{L\in \mathsf{Gr}(\bF_{2}^{n},\ell)}{F[L]=\left.f\right|_{L}\Big.}\geq C^{-1}\varepsilon^{3}.
\]
\end{lemma}

\begin{proof}
We define a random table $F'[\cdot]$ as follows: for each $L\in \mathsf{Gr}(\bF_{2}^{n},\ell)$,
\begin{enumerate}[label=(\arabic*)]
\item if $F[L]\neq \nil$, let $F'[L]=F[L]$;
\item if $F[L]=\nil$, pick a uniformly random linear functional $h\in L^*$ and let $F'[L]=h.$
\end{enumerate}
Therefore, for each $L\in \mathsf{Gr}(\bF_{2}^{n},\ell)$ the entry $F'[L]$ always takes value in the space $L^*$. 

For any fixed linear map $f:\bF_{2}^{n}\rightarrow \bF_{2}$, the random variables
\[
Z_{f}(L)=\ind{F'[L]=\left.f\right|_{L}\text{ and }F[L]=\nil\Big.}\qquad\text{for }L\in \mathsf{Gr}(\bF_{2}^{n},\ell)
\]
take values in $\{0,1\}$, have expected values at most $2^{-\ell}$, and are mutually independent. Therefore, by Chernoff bound we have
\begin{align*}
\Pru{F'}{\frac{1}{\left|\mathsf{Gr}(\bF_{2}^{n},\ell)\right|}\sum_{L\in\mathsf{Gr}(\bF_{2}^{n},\ell)}Z_{f}(L)\geq 2^{1-\ell}}&\leq \exp\left(-\frac{1}{3}\cdot 2^{-\ell}\cdot \left|\mathsf{Gr}(\bF_{2}^{n},\ell)\right|\right)\\
&\leq \exp\left(-\frac{1}{3}\cdot 2^{\ell(n-\ell-1)}\right)<2^{-n},
\end{align*}
where in the last transition we used the condition that $n\geq 2\ell\geq 320$. By a union bound over all linear maps $f:\bF_{2}^{n}\rightarrow \bF_{2}$, it follows that with nonzero probability over $F'$ we have
\[
\Exu{L\in \mathsf{Gr}(\bF_{2}^{n},\ell)}{Z_{f}(L)\big.}\leq 2^{1-\ell}.
\]
In other words, there exists a realization of the random table $F'[\cdot]$, which we denote by $F_{0}[\cdot]$, such that 
\begin{enumerate}[label=(\arabic*)]
\item for each $L\in \mathsf{Gr}(\bF_{2}^{n},\ell)$, we have $F_{0}[L]\in L^*$;
\item for each $L\in \mathsf{Gr}(\bF_{2}^{n},\ell)$, we have $F_{0}[L]=F[L]$ whenever $F[L]\neq \nil$;
\item $\Pru{L\in \mathsf{Gr}(\bF_{2}^{n},\ell)}{F_{0}[L]=\left.f\right|_{L}\text{ and }F[L]=\nil\Big.}\leq 2^{1-\ell}$ for any linear map $f:\bF_{2}^{n}\rightarrow \bF_{2}$.
\end{enumerate}

Since $F_{0}[\cdot]$ only replaces the $\nil$ values in $F[\cdot]$, it is clear that~\Cref{eq:list_decoding_linearity_test} also holds for $F_{0}$, i.e.
\[
\Pru{(L_{1},L_{2})\sim \cG(\bF_{2}^{n},\ell,b)}{\left.F_{0}[L_{1}]\right|_{L_{1}\cap L_{2}}=\left.F_{0}[L_{2}]\right|_{L_{1}\cap L_{2}}\Big.}\geq \varepsilon,
\]
and hence we may apply~\Cref{lem:list_decoding_linearity_test} to obtain a linear map $f:\bF_{2}^{n}\rightarrow \bF_{2}$ such that
\[
\Pru{L\in\mathsf{Gr}(\bF_{2}^{n},\ell)}{F_{0}[L]=\left.f\right|_{L}\Big.}\geq C_{0}^{-1}\varepsilon^{3},
\]
for some absolute constant $C_{0}>0$. Combining this with the second and third properties of $F_{0}[\cdot]$ listed above, it follows that
\[
\Pru{L\in \mathsf{Gr}(\bF_{2}^{n},\ell)}{F[L]=\left.f\right|_{L}\Big.}\geq C_{0}^{-1}\varepsilon^{3}-2^{1-\ell}\geq (2C_{0})^{-1}\varepsilon^{3},
\]
where the last transition holds as long as $\varepsilon^{3}\cdot 2^{\ell}$ is sufficiently large.
\end{proof}

The next corollary is a version of~\Cref{lem:list_decoding_linearity_test_nil} that also takes zoom-ins into account.

\begin{corollary}\label{cor:list_decoding_linearity_test_nil}
There is an absolute constant $C\geq 40$ such that the following holds. Fix positive integers $n,\ell,b$ and a nonnegative integer $q$ such that $n\geq 2\ell$ and $C(q+1)\leq b\leq \ell/4$. Let $Q\in\mathsf{Gr}(\bF_{2}^{n},q)$ be a subspace. Suppose $F[\cdot]$ is a table where $F[L]\in L^*\cup\{\nil\}$ for each $L\in \mathsf{Gr}(\bF_{2}^{n},\ell)$, such that
\begin{equation}
\Pru{(L_{1},L_{2})\sim \cG(\bF_{2}^{n},\ell,b)}{\left.F[L_{1}]\right|_{L_{1}\cap L_{2}}=\left.F[L_{2}]\right|_{L_{1}\cap L_{2}}\,\Big|\,Q\subseteq L_{1},L_{2}}\geq \varepsilon,
\end{equation}
for some parameter $\varepsilon\in [2^{-b/8},1]$. Then there exists a linear map $f:\bF_{2}^{n}\rightarrow \bF_{2}$ such that
\[
\Pru{L\in\mathsf{Gr}(\bF_{2}^{n},\ell)}{F[L]=f|_{L}\,\Big|\,Q\subseteq L}\geq C^{-1}\cdot 2^{-3q}\varepsilon^{3}
.\]
\end{corollary}
\begin{proof}
Since there are exactly $2^{q}$ linear functionals $Q\rightarrow \bF_{2}$, there exists a linear functional $h\in Q^{*}$ such that
\begin{equation}\label{eq:list_decoding_h_on_Q}
\Pru{(L_{1},L_{2})\sim \cG(\bF_{2}^{n},\ell,b)}{\left.F[L_{1}]\right|_{L_{1}\cap L_{2}}=\left.F[L_{2}]\right|_{L_{1}\cap L_{2}}\textup{ and }\left.F[L_{1}]\right|_{Q}=h\,\Big|\,Q\subseteq L_{1},L_{2}}\geq 2^{-q}\varepsilon,
\end{equation}

Let $Q^{\mathsf{c}}\in \mathsf{Gr}(\bF_{2}^{n},n-q)$ be a linear subspace such that $Q\cap Q^{\mathsf{c}}=\{0\}$. Define a table $F'[\cdot]$ as follows: for each $L\in \mathsf{Gr}(Q^{\mathsf{c}},\ell-q)$,
\begin{enumerate}[label=(\arabic*)]
\item if $F[L+Q]\neq \nil$ and $\left.F[L+Q]\right|_{Q}=h$, let $F'[L]=\left.F[L+Q]\right|_{L}$;
\item if $F[L+Q]=\nil$ or $\left.F[L+Q]\right|_{Q}\neq h$, let $F'[L]=\nil$.
\end{enumerate}
It then follows from~\Cref{eq:list_decoding_h_on_Q} that
\[\Pru{(L_{1},L_{2})\sim \cG(Q^{\mathsf{c}},\ell-q,b-q)}{\left.F'[L_{1}]\right|_{L_{1}\cap L_{2}}=\left.F'[L_{2}]\right|_{L_{1}\cap L_{2}}\Big.}\geq 2^{-q}\varepsilon.\]
Applying~\Cref{lem:list_decoding_linearity_test_nil} to the table $F'[\cdot]$ and then translating the conclusion back to $F[\cdot]$ yields the desired result.
\end{proof}

We now explain how to prove~\Cref{lem:dkkms_list_bound}. The proof of~\Cref{lem:dkkms_list_bound} (with $\nil$ entries disallowed) in~\cite{DKKMS18} runs a sunflower argument (similar to the proof of~\Cref{lem:sunflower_transfer}) to show that if we have too many $(f_{i},W_{i})\in \mathsf{LIST}_{Q}^{\vec{\tau}}(F)$ with $W_{i}$'s having the same codimension, then a subcollection of them can be ``patched together'' into a $(f,W)\in \mathsf{LIST}_{Q}^{\vec{\tau}}(F)$ with $\operatorname{codim}W$ smaller than the codimension of $W_{i}$'s. Since this contradicts the maximality requirement in~\Cref{def:LIST_Q}, one can conclude that $\mathsf{LIST}_{Q}^{\vec{\tau}}(F)$ does not contain too many pairs with the same zoom-out dimension. Since there are only $(w+1)$ possible zoom-out dimensions, the size of the list itself can consequently be bounded.

In order to allow $\nil$ entries, the only change one has to make to the proof is to use~\Cref{cor:list_decoding_linearity_test_nil} instead of~\Cref{lem:list_decoding_linearity_test} as the linearity-test input. We also record the following lemma that is used both in the proof of~\Cref{lem:dkkms_list_bound} and in the proof of~\Cref{lem:sunflower_transfer}.

% The next lemma gives an upper bound on the number of global linear functions that can be decoded from a given table $F$.
% \textcolor{red}{Dor: either cite or prove? I believe I attributed such stuff to Blinovsky in earlier papers (but gave a full proof because it was in Russian).}
\begin{lemma}[{\cite[Theorem 2.6]{KMS1}}]\label{lem:list_decoding_bound_baby}
There is an absolute constant $C\geq 40$ such that the following holds. Fix positive integers $n,\ell,b$, a nonnegative integer $q$ and a parameter $\eta\in (0,1)$ such that 
\[n\geq 2\ell,\qquad C(q+1)\leq b\leq \ell/4,\qquad\text{and}\qquad \eta\geq 2^{-b/8}.\] 
Let $Q\in\mathsf{Gr}(\bF_{2}^{n},q)$ be a subspace. Suppose $F[\cdot]$ is a table where $F[L]\in L^*\cup\{\nil\}$ for each $L\in \mathsf{Gr}(\bF_{2}^{n},\ell)$. Then there are at most $2/\eta$ different linear functions $f:\bF_{2}^{n}\rightarrow \bF_{2}$ such that
\begin{equation}\label{eq:list_decoding_baby_agreement}
\Pru{L\in \mathsf{Gr}(\bF_{2}^{n},\ell)}{F[L]=f|_{L}\,\Big|\,Q\subseteq L}\geq \eta.
\end{equation}
\end{lemma}

\begin{proof}
Let $f_{1},\dots,f_{m}:\bF_{2}^{n}\rightarrow \bF_{2}$ be the pairwise distinct linear functions satisfying~\Cref{eq:list_decoding_baby_agreement}. Let
\[\mu=\Pru{L\in \mathsf{Gr}(\bF_{2}^{n},\ell),\; i\in [m]}{F[L]=f_{i}|_{L}\,\Big|\,Q\subseteq L}\geq \eta.\]
To prove $m\leq 2/\eta$, assume on the contrary that $m>2/\eta$. We note that for any distinct indices $i,j\in [m]$, since $f_{i}$ and $f_{j}$ are distinct linear functions, we have
\[
\Pru{L\in\mathsf{Gr}(\bF_{2}^{n},\ell)}{f_{i}|_{L}=f_{j}|_{L}\,\Big|\,Q\subseteq L}=\Pru{L\in\mathsf{Gr}(\bF_{2}^{n},\ell)}{L\subseteq \ker(f_{i}-f_{j})\,\Big|\,Q\subseteq L}\leq 2^{q-\ell}.
\]
Therefore, we have
\begin{align*}
\Exu{L\in \mathsf{Gr}(\bF_{2}^{n},\ell)}{\Pru{i\in [m]}{F[L]=f_{i}|_{L}\big.}^{2}\,\middle|\,Q\subseteq L}&=\Exu{L\in \mathsf{Gr}(\bF_{2}^{n},\ell)}{\Pru{i,j\in [m]}{F[L]=f_{i}|_{L}=f_{j}|_{L}\big.}\,\middle|\,Q\subseteq L}\\
&\leq \frac{\mu}{m}+\frac{m-1}{m}\cdot 2^{q-\ell}\leq \frac{\mu\eta}{2}+2^{q-\ell}< \mu^{2}.
\end{align*}
On the other hand, by Cauchy-Schwarz we have
\[
\Exu{L\in \mathsf{Gr}(\bF_{2}^{n},\ell)}{\Pru{i\in [m]}{F[L]=f_{i}|_{L}\big.}^{2}\,\middle|\,Q\subseteq L}\geq \Pru{L\in \mathsf{Gr}(\bF_{2}^{n},\ell),\; i\in [m]}{F[L]=f_{i}|_{L}\,\Big|\,Q\subseteq L}^{2}=\mu^{2},
\]
leading to a contradiction.
\end{proof}

\section{Covering Properties}\label{sec:covering_prop}
Several covering-property-type statements were left unproved in~\Cref{sec:outer_PCP,sec:middle_PCP,sec:inner_PCP}. In this section, we complete their proofs. 

\subsection{Covering in the Outer PCP}\label{subsec:covering_outer}
We first prove~\Cref{prop:covering_outer} (restated below), which analyzes the distribution of $E'$ in a sample $(E',A_{\inn},A_{\midd})\sim \cJ_{\outt}(E,q_{1},q_{2},\beta)$ conditioned on $(A_{\inn},A_{\midd})$. 
\coveringouterA*
\begin{proof}
We write \(q=q_{1}+q_{2}\). Without loss of generality, we assume 
\[E=[3k]\qquad\text{and}\qquad\operatorname{Triples}(E)=\bigl\{\{3i-2,3i-1,3i\}\,\big|\, i\in [k]\bigr\}\]
for some positive integer $k$. 

Let \(P\subseteq E\) be the union of the supports of all vectors appearing in the fixed advice sequences \(A_{1}\) and \(A_{2}\). For each $i\in [k]$, let 
\[P_{i}=\bigl\{j\in \{1,2,3\}\,\big|\,3i-3+j\in P\bigr\}.\]
Consider a sample
\[
(E',A_{\inn},A_{\midd})\sim \cJ_{\outt}([3k],q_{1},q_{2},\beta)
\]
conditioned on the event \(\{A_{\inn}=A_{1}\}\cap\{A_{\midd}=A_{2}\}\). Under this conditioning, the marginal distribution of \(E'\) is independent over the triples of $E$: for every $i\in [k]$,
\begin{enumerate}[label=(\arabic*)]
\item if $|P_{i}|\geq 2$, then $3i-2,3i-1$ and $3i$ are all included in $E'$ with probability 1;
\item if $P_{i}$ is a singleton set $\{j\}$, then $3i-2,3i-1$ and $3i$ are all included in $E'$ with probability $1-\eta_{1}$, and with probability $\eta_{1}$ only $3i-3+j$ is included in $E'$, where
\begin{equation}\label{eq:covering_eta_1}
\eta_{1}=\frac{(\beta/3)(1-2^{-q})}{(\beta/3)(1-2^{-q})+(1-\beta)(1-2^{-q})2^{-2q}}=\frac{2^{2q}\beta}{2^{2q}\beta+3(1-\beta)}\leq 2^{2q}\beta;
\end{equation}
\item if $P_{i}=\emptyset$, then $3i-2,3i-1$ and $3i$ are all included in $E'$ with probability $1-\eta_{0}$, and with probability $\eta_{0}$ only (a uniformly random) one of them is included in $E'$, where
\begin{equation}\label{eq:covering_eta_2}
\eta_{0}=\frac{\beta\cdot 2^{-q}}{\beta\cdot 2^{-q}+(1-\beta)2^{-3q}}=\frac{2^{2q}\beta}{2^{2q}\beta+(1-\beta)}\leq 2^{2q}\beta. 
\end{equation}
\end{enumerate}
In each of the above three cases, the triple $\{3i-2,3i-1,3i\}$ is included entirely in $E'$ with probability at least $1-2^{2q}\beta=1-4^{q_{1}+q_{2}}\beta$.
\end{proof}
Building on the proof of~\Cref{prop:covering_outer}, we can now prove the main covering property needed for the middle PCP (\Cref{prop:covering_outer_2}, restated below), which asserts that the distribution of $A_{\mathsf{mid}}$ is close to uniform over Alice's space even when conditioned on $A_{\mathsf{in}}$ (similar arguments appeared in~\cite{KhotSafra,KMS1}).
\coveringouterB*
\begin{proof}[Proof]
As in the proof of~\Cref{prop:covering_outer}, we assume
\[E=[3k]\qquad\text{and}\qquad\operatorname{Triples}(E)=\bigl\{\{3i-2,3i-1,3i\}\,\big|\, i\in [k]\bigr\}\]
for some positive integer $k$. Let \(P\subseteq E\) be the union of the supports of all vectors appearing in the fixed advice sequence \(A_{1}\). For each $i\in [k]$, let 
\(P_{i}=\bigl\{j\in \{1,2,3\}\,\big|\,3i-3+j\in P\bigr\}\).
We organize the proof into the following two steps.

\paragraph{Step 1: description of the distribution.} Similarly to~\Cref{eq:covering_eta_1,eq:covering_eta_2}, we let
\[
\eta_{1}=\frac{2^{2q_{1}}\beta}{2^{2q_{1}}\beta+3(1-\beta)}\leq 2^{2q_{1}}\beta\qquad\text{and}\qquad \eta_{0}=\frac{2^{2q_{1}}\beta}{2^{2q_{1}}\beta+(1-\beta)}\leq 2^{2q_{1}}\beta.
\]
Let $\cD_{0}$ be the uniform distribution over
\[(\mathbb F_{2}^{q_{2}})^{3}=\left\{\bigl(x^{(1)},x^{(2)},x^{(3)}\bigr)\,\middle|\,x^{(1)},x^{(2)},x^{(3)}\text{ are vectors in }\bF_{2}^{q_{2}}\right\}.\]
For each $j\in [3]$, let $\cD_{j}$ be the uniform distribution over the subset
\[\left\{\bigl(x^{(1)},x^{(2)},x^{(3)}\bigr)\in (\bF_{2}^{q_{2}})^{3}\,\middle|\,x^{(i)}=0\text{ for all }i\in [3]\setminus \{j\}\right\}\subseteq (\bF_{2}^{q_{2}})^{3}.\]
For each $i\in [k]$, we define a distribution $\cD^{(i)}$ over $(\mathbb F_{2}^{q_{2}})^{3}$ as follows (according to the conditional distribution of $E'\cap \{3i-2,3i-1,3i\}$ described in the proof of~\Cref{prop:covering_outer}):
\begin{enumerate}[label=(\arabic*)]
\item if $|P_{i}|\geq 2$, let $\cD^{(i)}=\cD_{0}$;
\item if $P_{i}$ is a singleton set $\{j\}$, let $\cD^{(i)}$ be the mixture $(1-\eta_{1})\cD_{0}+\eta_{1}\cD_{j}$;
\item if $P_{i}=\emptyset$, let $\cD^{(i)}$ be the mixture $(1-\eta_{0})\cD_{0}+\frac{1}{3}\eta_{0}\cD_{1}+\frac{1}{3}\eta_{0}\cD_{2}+\frac{1}{3}\eta_{0}\cD_{3}$.
\end{enumerate}
It is straightforward to verify that in each of the above three cases, for any tuple $X\in (\bF_{2}^{q_{2}})^{3}$ we have
\begin{equation}\label{eq:covering_3row_pointwise}
\left|\sqrt{\cD^{(i)}(X)/\cD_{0}(X)}-1\right|\leq \left|\cD^{(i)}(X)/\cD_{0}(X)-1\right|\leq 2^{2q_{2}}\max\{\eta_{0},\eta_{1}\}\leq 2^{2q_{2}+2q_{1}}\beta.
\end{equation}

Now consider a sample
\[
(E',A_{\inn},A_{\midd})\sim \cJ_{\outt}([3k],q_{1},q_{2},\beta)
\]
conditioned on the event \(\{A_{\inn}=A_{1}\}\). The conditional distribution of $E'$ is as described in the proof of~\Cref{prop:covering_outer} and does not depend on $q_{2}$. It is then easy to see from the proof of~\Cref{prop:covering_outer} that the conditional distribution of $A_{\midd}$ is the output distribution of the following process:
\begin{enumerate}
    \item For each \(i\in [k]\), sample a tuple of three vectors \(\bigl(x^{(3i-2)},x^{(3i-1)},x^{(3i)}\bigr)\sim \cD^{(i)}\).
    \item Form the $3k$-by-$q_{2}$ matrix \(X\in \mathbb F_{2}^{3k\times q_{2}}\) whose $i$-th row is $x^{(i)}$ for each $i\in [3k]$. Let \(\cM_{P}\) denote the resulting distribution of \(X\).
    \item Output the $q_{2}$ columns of $X$.
\end{enumerate}

\paragraph{Step 2: tensorization of Hellinger affinity.} It now suffices to show that $\cM_{P}$ is close to uniform over $\bF_{2}^{3k\times q_{2}}$ (the space of $3k$-by-$q_{2}$ matrices).

For each $i\in[k]$, by \Cref{eq:covering_3row_pointwise} we have that the squared Hellinger distance between $\cD^{(i)}$ and $\cD_{0}$ is
\[
\frac{1}{2}\cdot\Exu{X\sim \cD_{0}}{\left(\sqrt{\cD^{(i)}(X)/\cD_{0}(X)}-1\right)^{2}}\leq 2^{4q_{1}+4q_{2}-1}\beta^{2},
\]
and therefore 
\[
\sum_{X\in (\bF_{2}^{q_{2}})^{3}}\sqrt{\cD^{(i)}(X)\cD_{0}(X)}=1-\frac{1}{2}\sum_{X\in (\bF_{2}^{q_{2}})^{3}}\left(\sqrt{\cD^{(i)}(X)}-\sqrt{\cD_{0}(X)}\right)^{2}\geq 1-2^{4q_{1}+4q_{2}-1}\beta^{2}.
\]
Consider the distributions
\[\prod_{i=1}^{k}\cD^{(i)}=\cM_{P}\qquad\text{and}\qquad\prod_{i=1}^{k}\cD_{0}=\cD_{0}^{\otimes k}.
\]
over $\bF_{2}^{3k\times q_{2}}$. By the tensorization we have
\begin{align*}
\sum_{X\in \bF_{2}^{3k\times q_{2}}}\sqrt{\cM_{P}(X)\cD^{\otimes k}_{0}(X)}&=\prod_{i=1}^{k}\sum_{X\in (\bF_{2}^{q_{2}})^{3}}\sqrt{\cD^{(i)}(X)\cD_{0}(X)}\\
&\geq (1-2^{4q_{1}+4q_{2}-1}\beta^{2})^{k}\geq 1-2^{4q_{1}+4q_{2}-1}\beta^{2}k.
\end{align*}
Consequently, the squared Hellinger distance between $\cM_{P}$ and $\cD_{0}^{\otimes k}$ is
\[
\frac{1}{2}\sum_{X\in \bF_{2}^{3k\times q_{2}}}\left(\sqrt{\cM_{P}(X)}-\sqrt{\cD^{\otimes k}_{0}(X)}\right)^{2}\leq 2^{4q_{1}+4q_{2}-1}\beta^{2}k,
\]
and hence by Cauchy-Schwarz,\footnote{For any two probability vectors $p,q\in \bR^{n}$, we have $\|p-q\|_{1}\leq \|\sqrt{p}-\sqrt{q}\|_{2}\cdot\|\sqrt{p}+\sqrt{q}\|_{2}\leq 2\|\sqrt{p}-\sqrt{q}\|_{2}$.} the total variation distance between $\cM_{P}$ and $\cD_{0}^{\otimes k}$ is
\[
\frac{1}{2}\sum_{X\in \bF_{2}^{3k\times q_{2}}}\left|\cM_{P}(X)-\cD^{\otimes k}_{0}(X)\right|\leq \sqrt{2}\cdot\sqrt{2^{4q_{1}+4q_{2}-1}\beta^{2}k}=2^{2q_{1}+2q_{2}}\beta\sqrt{k}.
\qedhere\]
\end{proof}

\subsection{Covering in the Middle PCP}\label{subsec:covering_middle}

In the middle PCP, we introduced a new joint distribution $\cJ_{\midd}(E,q_{1},q_{2},\beta,\ell)$ based on the ``thinning with advice'' distribution $\cJ_{\outt}(E,q_{1},q_{2},\beta)$. Specifically, in addition to the tuple $(E',A_{\inn},A_{\midd})$ sampled in the latter joint distribution, a subspace $U\subseteq \bF_{2}^{E'}$ is generated. Since the way we generated $U$ in  $\cJ_{\midd}(E,q_{1},q_{2},\beta,\ell)$ can be viewed as sampling ``advice vectors'' and taking their span, the following proposition follows rather easily from~\Cref{prop:covering_outer_2}.

\coveringmiddleA*
\begin{proof}

We write \(q=q_{1}+q_{2}\) and $\ell'=\ell-q_{2}$.

According to~\Cref{def:Joint_dist_S_A_Aprime_U}, when $(E',A_{\inn},A_{\midd},U)$ is sampled from $\cJ_{\midd}(E,q_{1},q_{2},\beta,\ell)$, a vector sequence $A'\in (\bF_{2}^{E'})^{\ell'}$ is generated, and then $U$ is taken to be the sum $\operatorname{span}(A_{\midd})+\operatorname{span}(A')$. Let \[A_{\inn}\uplus A_{\midd} \in (\bF_{2}^{E'})^{q_{1}+q_{2}}\]
denote the concatenation of the two vector sequences $A_{\inn}$ and $A_{\midd}$. Note that before any conditioning, the distribution of $(E',A_{\inn}\uplus A_{\midd}, A')$ is identical to $\cJ_{\outt}(E,q,\ell',\beta)$. Therefore, it follows from~\Cref{prop:covering_outer_2} that after conditioning on $\{A_{\inn}=A_{1}\}\cap\{A_{\midd}=A_{2}\}$,
\begin{equation}\label{eq:covering_middle_1}
\textup{the distribution of }A'\textup{ is }\Bigl(4^{q+\ell'}\beta\sqrt{|E|/3}\Bigr)\textup{-close to uniform over }(\bF_{2}^{E})^{\ell'}. 
\end{equation}

Since $\dim\operatorname{span}(A_{2})=q_{2}$, it is easy to see that
\begin{equation}\label{eq:covering_middle_2}
\Pru{A\in (\bF_{2}^{E})^{\ell'}}{\dim\bigl(\operatorname{span}(A_{2})+\operatorname{span}(A)\bigr)=\ell}\geq 1-2^{\ell-|E|}.
\end{equation}
Furthermore, if $A\in (\bF_{2}^{E})^{\ell'}$ is a uniformly random vector sequence, then the distribution of $\bigl(\operatorname{span}(A_{2})+\operatorname{span}(A)\bigr)$ conditioned on the event \[\bigl\{\dim\bigl(\operatorname{span}(A_{2})+\operatorname{span}(A)\bigr)=\ell\bigr\}\]
is uniform over the set \(\big\{U\in \mathsf{Gr}(\mathbb F_{2}^{E},\ell)\,\big|\, \operatorname{span}(A_{2})\subseteq U\big\}\). Therefore, it follows from \Cref{eq:covering_middle_1,eq:covering_middle_2} that conditioned on $\{A_{\inn}=A_{1}\}\cap\{A_{\midd}=A_{2}\}$, the distribution of $U=\bigl(\operatorname{span}(A_{\midd})+\operatorname{span}(A')\bigr)$ has total variation distance at most
\[
4^{q+\ell'}\beta\sqrt{|E|/3}+2^{\ell-|E|}=4^{q_{1}+\ell}\beta\sqrt{|E|/3}+2^{\ell-|E|}
\]
to the uniform distribution over \(\big\{U\in \mathsf{Gr}(\mathbb F_{2}^{E},\ell)\,\big|\, \operatorname{span}(A_{2})\subseteq U\big\}\).
\end{proof}

\subsection{Covering in the Inner PCP}
We finish by stating the following general probabilistic fact that establishes the covering property achieved by the lazy parallel repetition step. It is essentially the same as~\Cref{claim:prob_fact_intro}, but stated more generally in a way that is suitable for us. 

\begin{proposition}\label{prop:covering_inner_pre}
Fix positive integers $k,t$ and finite sets $\scrA_{1},\dots,\scrA_{t}$ with $|\scrA_{j}|=k$ for each $j\in [t]$. Suppose that for each $j\in [t]$ there is a probability distribution $\mu_{j}$ over $\scrA_{j}$ and a surjective map $\pi_{j}:\scrA_{j}\rightarrow \scrB_{j}$ onto some finite set $\scrB_{j}$. Let 
\(
\pi:\prod_{j=1}^{t}\scrA_{j}\rightarrow \prod_{j=1}^{t}\scrB_{j}
\)
be the coordinate-wise map defined by
\[\pi(a_{1},\dots,a_{t})=(\pi_{1}(a_{1}),\dots,\pi_{t}(a_{t}))\qquad\text{for all }(a_{1},\dots,a_{t})\in {\textstyle\prod_{j=1}^{t}\scrA_{j}}.\]
Define a probability distribution $\mu$ on $\prod_{j=1}^{t}\scrA_{j}$ to be the output distribution of the following process: 
\begin{enumerate}
\item Sample a uniformly random index $c\in [t]$.
\item Sample $a_{c}\sim \mu_{c}$. For each $j\in [t]\setminus\{c\}$, sample $a_{j}$ uniformly at random from $\scrA_{j}$.
\item Output $(a_{1},\dots,a_{t})$.
\end{enumerate}
For each $b\in \prod_{j=1}^{t}\scrB_{j}$, let $\operatorname{dist}(b)$ be the total variation distance between 
\begin{enumerate}[label=(\arabic*)]
\item the distribution of $a\sim \mu$ conditioned on the event $\{\pi(a)=b\}$ (if $\Prs{a\sim \mu}{\pi(a)=b}=0$, let $\operatorname{dist}(b)=1$), and
\item the uniform distribution over the subset $\pi^{-1}(\{b\})\subseteq \prod_{j=1}^{t}\scrA_{j}$.
\end{enumerate}
Then we have 
\[\Exu{a\sim \mu}{\operatorname{dist}(\pi(a))\big.}\leq \sqrt{\frac{k}{4t}}.\]
\end{proposition}

\begin{proof}
For each $j\in [t]$, let $f_{j}:\scrA_{j}\rightarrow [0,+\infty)$ be the probability density function of $\mu_{j}$; that is, let $f_{j}(a)=k\cdot\mu_{j}(\{a\})$ for all $a\in \scrA_{j}$. It is then easy to see that the probability density function of $\mu$ is given by
\[
f(a_{1},\dots,a_{t}):=\frac{1}{t}\sum_{j=1}^{t}f_{j}(a_{j}),\qquad\text{for all }(a_{1},\dots,a_{t})\in \textstyle\prod_{j=1}^{t}\scrA_{j}.
\]
For each $j\in [t]$, let $g_{j}:\scrB_{j}\rightarrow [0,+\infty)$ be defined by 
\[g_{j}(b_{j})=\begin{cases}
\Exs{a_{j}\in \pi_{j}^{-1}(\{b_{j}\})}{f_{j}(a_{j})},&\text{if }\pi_{j}^{-1}(\{b_{j}\})\neq \emptyset,\\
0,&\text{otherwise}.
\end{cases}\]
For each $b\in \prod_{j=1}^{t}\scrB_{j}$, let
\[
g(b)=\frac{1}{t}\sum_{j=1}^{t}g_{j}(b_{j})=\Exu{a\in \pi^{-1}(\{b\})}{f(a)}.
\]

For any $b\in \prod_{j=1}^{t}\scrB_{j}$ that lies in the support of $g$, the total variation distance between the distribution of $a\sim \mu$ conditioned on $\{\pi(a)=b\}$
and the uniform distribution over $\pi^{-1}(\{b\})$ is
\begin{align*}
\frac{1}{2}\cdot\Exu{a\in \pi^{-1}(\{b\})}{\left|\frac{f(a)}{g(b)}-1\right|}&\leq \frac{1}{2 g(b)}\sqrt{\Varu{a\in \pi^{-1}(\{b\})}{f(a)}}\tag{by Cauchy-Schwarz}\\
&=\frac{1}{2t\cdot g(b)}\sqrt{\sum_{j=1}^{t}\Varu{a_{j}\in \pi_{j}^{-1}(b_{j})}{f_{j}(a_{j})}}\tag{using the fact that $\pi^{-1}(\{b\})$ is a product set}\\
&\leq \frac{1}{2t\cdot g(b)}\sqrt{\sum_{j=1}^{t}k\cdot \Exu{a_{j}\in \pi^{-1}(b_{j})}{f_{j}(a_{j})}}\tag{$f_{j}$ is upper-bounded by $k$}\\
&\leq \frac{1}{2}\sqrt{\frac{k}{t\cdot g(b)}}.\tag{by the definition of $g(b)$}
\end{align*}
Let $\nu$ be the distribution of $\pi(a)$ when $a$ is sampled uniformly from $\prod_{j=1}^{t}\scrA_{j}$.
We then have
\[
\Exu{a\sim \mu}{\operatorname{dist}(\pi(a))^{2}\big.}=\Exu{b\sim \nu}{g(b)\cdot\operatorname{dist}(b)^{2}}\leq \Exu{b\sim \nu}{\frac{k}{4t}}=\frac{k}{4t},
\]
and hence $\Exs{a\sim \mu}{\operatorname{dist}(\pi(a))}\leq \sqrt{k/(4t)}$ by Cauchy-Schwarz.
\end{proof}

\end{document}